\documentclass[a4paper,reqno]{amsart}
\usepackage[utf8]{inputenc}

\usepackage{amsmath}
\usepackage{amssymb}
\usepackage{bbm}
\usepackage[english]{babel}
\usepackage{color}
\usepackage{enumitem}
\usepackage[breaklinks,colorlinks=true, citecolor=blue]{hyperref}
\usepackage{ellipsis}
\usepackage{hyperref}
\usepackage[misc,geometry]{ifsym} \usepackage{dsfont}
\usepackage{mathrsfs}
\usepackage{mathtools}
\usepackage{algorithm}
\usepackage{algpseudocode}
\usepackage{etoolbox} \usepackage{stmaryrd}
\usepackage{thm-restate}
\usepackage{subcaption}
\usepackage{systeme}

\usepackage{tikz}
\usepackage{accents}
\usetikzlibrary{patterns}
\usetikzlibrary{matrix}

\usepackage{todonotes}
\usepackage{lipsum}

 \newtoggle{llncs}

\makeatletter
\@ifclassloaded{llncs}{\toggletrue{llncs}}{\togglefalse{llncs}}
\makeatother

\usepackage{empheq}
\iftoggle{llncs}{}{
	\usepackage{a4wide}
	\usepackage{amsthm}
	\usepackage[foot]{amsaddr}

}

\newcommand{\hardestrate}{0.42}

\newcommand{\mat}[1]{\ensuremath{\boldsymbol{\rm #1}}}

\newcommand{\LPN}{$\mathsf{LPN}$}
\newcommand{\LPNs}{\LPN \ }
\newcommand{\RLPN}{$\mathsf{RLPN}$}
\newcommand{\RLPNs}{\RLPN \ }
\newcommand{\nRLPN}{$\mathsf{double}$-$\mathsf{RLPN}$}
\newcommand{\nRLPNs}{\nRLPN \ }

\newcommand{\prob}[2]{\mathbb{P}_{#1}\left(#2\right)}
\DeclareMathOperator{\Prob}{{\mathbb{P}}}
\newcommand{\expect}[2]{\mathbb{E}_{#1}\left(#2\right)}
\newcommand{\esp}{\mathbb{E}}
\newcommand{\drawn}{{\stackrel{\$}{\gets}}}

\newcommand{\IInt}[2]{\left\llbracket #1, #2 \right\rrbracket}

\newcommand{\fract}[2]{\hbox{\leavevmode
\kern.1em \raise .5ex \hbox{\the\scriptfont0 $#1$}\kern-.1em }/
\hbox{\kern-.15em \lower .25ex \hbox{\the\scriptfont0 $#2$}}
}

\newcommand{\NN}{\mathbb{N}}

\newcommand{\RR}{\mathbb{R}}

\newcommand{\C}{\mathcal{C}}
\newcommand{\CC}{\C}
\newcommand{\DD}{\ensuremath{\mathcal{D}}}

\newcommand{\dgv}{d_{\mathrm{GV}}}

\newcommand{\scal}[2]{\left\langle #1 , #2 \right\rangle}
\newcommand{\scp}[2]{\ensuremath{\scal{#1}{#2}}}
\newcommand{\norm}[1]{\left|#1\right|}
\newcommand{\hw}[1]{\left|#1\right|}

\newcommand{\card}[1]{\left| #1 \right|}

\newcommand{\transp}[1]{#1^\intercal}
\newcommand{\transpose}{\intercal}

\newcommand{\sA}{{\mathscr{A}}}
\newcommand{\sB}{{\mathscr{B}}}
\newcommand{\sE}{{\mathscr{E}}}
\newcommand{\sF}{{\mathscr{F}}}
\newcommand{\sH}{{\mathscr{H}}}
\newcommand{\sI}{{\mathscr{I}}}

\newcommand{\sN}{{\mathscr{N}}}
\newcommand{\sP}{{\mathscr{P}}}

\newcommand{\sR}{{\mathscr{R}}}
\newcommand{\sS}{{\mathscr{S}}}

\newcommand{\sW}{{\mathscr{W}}}

\newcommand{\ie}{{\em i.e. }}
\DeclareMathOperator*{\bias}{bias}

\newcommand{\cS}{\mathcal{S}}

\newcommand{\Niter}{{N}_{\textup{iter}}}
\newcommand{\Naux}{{N}_{\textup{aux}}}
\newcommand{\Nbis}{\Naux}
\newcommand{\Nisd}{{N}_{\textup{ISD}}}

\newcommand{\Teq}{{T}_{\textup{eq}}}

\newcommand{\Psucc}{{P}_{\textup{succ}}}

\newcommand{\Tdumer}{{T}_{\textup{Dumer}}}
\newcommand{\Tisd}{{T}_{\textup{ISD}}}

\newcommand{\Tdec}{{T}_{\textup{dec}}}
\newcommand{\un}{{\mathbbm{1}}}

\newcommand{\oo}[1]{\ensuremath{\mathop{}\mathopen{}o\mathopen{}\left(#1\right)}}
\newcommand{\om}[1]{\ensuremath{\mathop{}\mathopen{} \omega\mathopen{}\left(#1\right)}}
\newcommand{\Th}[1]{\Theta\left( #1 \right)}

\newcommand{\OO}[1]{\ensuremath{\mathop{}\mathopen{}\mathcal{O}\mathopen{}\left(#1\right)}}
\newcommand{\OOt}[1]{\ensuremath{\mathop{}\mathopen{}\widetilde{\mathcal{O}}\mathopen{}\left(#1\right)}}
\newcommand{\Om}[1]{\Omega \left( #1 \right)}
\newcommand{\Omt}[1]{\widetilde{\Omega} \left( #1 \right)}

\renewcommand{\vec}[1]{\mathbf{#1}}
\newcommand{\Gm}{{\mathbf{G}}}
\newcommand{\Hm}{{\mathbf{H}}}
\newcommand{\Am}{{\mathbf{A}}}
\newcommand{\Mm}{{\mathbf{A}}}

\newcommand{\Rm}{{\mathbf{R}}}
\newcommand{\Jm}{{\mathbf{J}}}

\newcommand{\Id}{{\mathbf{I}}}

\newcommand{\av}{{\mathbf{a}}}

\newcommand{\cv}{{\mathbf{c}}}

\newcommand{\ev}{{\mathbf{e}}}
\newcommand{\hv}{{\mathbf{h}}}
\newcommand{\mv}{{\mathbf{m}}}

\newcommand{\rv}{{\mathbf{r}}}
\newcommand{\sv}{{\mathbf{s}}}
\newcommand{\uv}{{\mathbf{u}}}
\newcommand{\vv}{{\mathbf{v}}}
\newcommand{\wv}{{\mathbf{w}}}
\newcommand{\xv}{{\mathbf{x}}}

\newcommand{\yv}{{\mathbf{y}}}
\newcommand{\zv}{{\mathbf{z}}}

\newcommand{\eqdef}{\stackrel{\textrm{def}}{=}}
\renewcommand{\leq}{\leqslant}

\renewcommand{\geq}{\geqslant}

\newcommand\bis[1]{{#1_{\textup{aux}}}}
\newcommand{\CCbis}[1][]{\bis{\CC^{#1}}}
\newcommand{\CCbisperp}{\bis{\CC^{\perp}}}

\newcommand{\kbis}{\bis{k}}
\newcommand{\tbis}{\bis{t}}

\newcommand{\cvbis}{\bis{\cv}}
\newcommand{\cvbisperp}{\bis{\cv^{\perp}}}
\newcommand{\evbis}{\bis{\ev}}
\newcommand{\mvbis}{\bis{\mv}}
\newcommand{\Gmbis}[1][]{\bis{\Gm^{#1}}}
\newcommand{\Gmbistransp}{\bis{\Gm^{\transpose}}}
\newcommand{\Hmbis}[1][]{\bis{\Hm^{#1}}}
\newcommand{\decodebis}{{\mathcal{D}\mathrm{ec}}}

\newcommand{\mysubtitle}[1]{{\bf \noindent #1}}

\iftoggle{llncs}{
\renewcommand{\vec}[1]{\mathbf{#1}}
}{
\renewcommand{\vec}[1]{\mathbf{#1}}
}
\newcommand{\F}{\mathbb{F}}

\newcommand{\charles}[1]{}

\iftoggle{llncs}{

\newtheorem{fact}[theorem]{Fact}
\newtheorem{notation}[theorem]{Notation}
\newtheorem{constraint}[theorem]{Constraint}
\newtheorem{model}[theorem]{Model}

\counterwithin{equation}{section}
\counterwithin{figure}{section}
\counterwithin{table}{section}
\counterwithin{algorithm}{section}

\spnewtheorem{distribution}[theorem]{Distribution}{\bfseries}{\itshape}
\spnewtheorem{pb}[theorem]{Problem}{\bfseries}{\itshape}

\usepackage{thm-restate}
}{

\newtheorem{remark}{Remark}
\newtheorem{lemma}{Lemma}

\newtheorem{definition}{Definition}
\newtheorem{theorem}{Theorem}
\newtheorem{proposition}{Proposition}
\newtheorem{pb}{Problem}
\newtheorem{notation}{Notation}

\newtheorem{conjecture}{Conjecture}
\newtheorem{distribution}{Distribution}
\newtheorem{model}{Model}
\newtheorem{constraint}{Parameter Constraint}
\newtheorem{corollary}{Corollary}
\newtheorem{fact}{Fact}
}

\begin{document}
	
\title{Reduction from sparse LPN to LPN, Dual Attack 3.0}

\iftoggle{llncs}{

\author{ 
}
\institute{ }
}{
\author{K\'{e}vin Carrier$^{1}$} 
\email{kevin.carrier@ensea.fr}
\author{Thomas Debris--Alazard$^{2,3}$} \email{thomas.debris@inria.fr}  
\author{Charles Meyer-Hilfiger$^{4}$} 
\email{charles.meyer-hilfiger@inria.fr}
\author{Jean-Pierre Tillich$^{4}$} 
\email{jean-pierre.tillich@inria.fr}

\address{$^{1}$ Laboratoire ETIS, UMR 8051, CY Cergy-Paris Université, ENSEA, CNRS}
\address{$^{2}$ Project GRACE, Inria Saclay}
\address{$^{3}$ Laboratoire LIX, \'Ecole Polytechnique, Institut
	Polytechnique de Paris, 1 rue Honor\'e d'Estienne d'Orves, 91120
	Palaiseau Cedex}
\address{$^{4}$ Project COSMIQ, Inria de Paris}

\thanks{The work of KC, TDA and JPT was funded by the French Agence Nationale de la
	Recherche through ANR JCJC DECODE (ANR-22-CE39-0004-01) for KC, ANR JCJC COLA (ANR-21-CE39-0011) for TDA and ANR-22-PETQ-0008 PQ-TLS for JPT. The work of CMH was funded by the French Agence de l'innovation de d\'efense and by Inria.} 	
}
\maketitle

\begin{abstract}
The security of code-based cryptography relies primarily on the hardness of decoding generic linear codes. Until very recently, all the best algorithms for solving the decoding problem 
were information set decoders ($\mathsf{ISD}$). However, recently a new algorithm called \RLPN-decoding which relies on a completely different approach was introduced and it has been shown that \RLPN \ outperforms significantly $\mathsf{ISD}$ decoders for a rather large range of rates. This \RLPN \ decoder relies on two ingredients, first reducing decoding to some underlying \LPN{} problem, and then computing efficiently many parity-checks of small weight when restricted to some positions. 
 We revisit \RLPN-decoding by noticing that, in this algorithm, decoding is in fact reduced to a sparse-\LPN{}  problem, namely with a secret whose Hamming weight is small. Our new approach consists this time in making an additional reduction from sparse-\LPN{} to plain-\LPN{} with a coding approach inspired by $\mathsf{coded}$-$\mathsf{BKW}$. It outperforms significantly the $\mathsf{ISD}$'s and \RLPN \ for code rates smaller than $\hardestrate$. This algorithm can be viewed as the code-based cryptography cousin of recent dual attacks in lattice-based cryptography. We depart completely from the traditional analysis of this kind of algorithm which uses a certain number of independence assumptions that have been strongly questioned recently in the latter domain. We give instead a formula for the \LPNs noise relying on duality which allows to analyze the behavior of the algorithm by relying only on the analysis of a certain weight distribution. By using only a minimal assumption whose validity has been verified experimentally we are able to justify the correctness of our algorithm. This key tool, namely the  duality formula, can be readily adapted to the lattice setting and is shown to give a simple explanation for some phenomena observed on dual attacks in lattices in~\cite{DP23}.

\end{abstract}

\section{Introduction}
\label{sec:intro}

\subsection{Background}
\iftoggle{llncs}{}{\mbox{}\medskip}

{\bf \noindent Code-based Cryptography: Decoding and \LPN{} Problems.}
Code-based cryptography relies on the hardness of decoding generic linear codes or sometimes also on a closely related problem, namely the \LPN{} problem. The first one corresponds in the binary case to
\begin{pb}[decoding a fixed error weight in a  linear code]
\label{pb:decodingFq}
Let $\C$ be a binary linear code over $\mathbb{F}_{2}$ of dimension $k$ and length $n$, \ie a subspace of $\mathbb{F}_{2}^n$ of dimension $k$.
We are given $\vec{y} \in \mathbb{F}_{2}^n$, an integer $t$ and we want to find  a codeword $\vec{c} \in \C$ and an error vector $\vec{e} \in \mathbb{F}_{2}^n$ of Hamming weight $|\vec{e}|=t$ for which $\vec{y} = \vec{c}+\vec{e}$. 
\end{pb}
Generally the linear code is specified by a {\em generator matrix}, namely a $k \times n$ binary matrix $\vec{G}$ whose
rows span the vector space $\C$, in other words 
$$
\C = \{\vec{u} \vec{G}: \vec{u} \in \mathbb{F}_{2}^k\}.
$$
The second one is a version of this problem where the length $n$ is basically unbounded; the code is randomly chosen and the
error model is slightly modified to take into account that the length is not fixed. 
\begin{pb}[\LPN{} problem]
Let $\vec{s}$ be a secret chosen uniformly at random in $\mathbb{F}_{2}^k$. We have unbounded access to an oracle such that each query provides a pair 
$(\vec{a},b)$ where $\vec{a}$ is chosen uniformly at random in $\mathbb{F}_{2}^k$ and $b$ is a bit obtained as
$$
b = \scal{\vec{s}}{\vec{a}}+e
$$
where $e \in \mathbb{F}_{2}$ is chosen at random and is equal to $1$ with probability $p$. Quantity $\scal{\vec{s}}{\vec{a}}$ stands for the inner product 
$\sum_{i=1}^k s_i a_i$ between $\vec{s}=(s_i)_{1 \leq i \leq k}$ and $\vec{a}=(a_i)_{1\leq i \leq k}$. The aim is to output $\vec{s}$ after querying 
a certain number of times the oracle.
\end{pb}

Sometimes a variation of the \LPN{} problem is considered, namely the {\em sparse \LPN{}} problem where the only difference is the way $\vec{s}$ is chosen, say uniformly at random among the words of length $n$ and Hamming weight $t'$ small, or the entries like i.i.d. Bernoulli random
variables of parameter $p'$ small. 
\newline

{\bf \noindent The Complexity of the Best Generic Decoding Algorithms and \LPN{}-solvers.}
It is of fundamental importance to study the complexity of these problems, the best state of the art algorithms being those that are used to 
determine secure parameters of code-based cryptosystems. The regime of parameters which is relevant for code-based cryptography depends on the type of primitive, but a large range of parameters is relevant here. For some code-based cryptosystems, $t$ is sublinear in $n$,  \cite{M78,AABBBBDGGGMMPSTZ21,AABBBDGPZ21a,BCLMNPPSSSW19,AABBBDGPZ21a}, for some Stern like signatures schemes
\cite{S93,V96,CVA10,AGS11,GPS22,FJR22} it is precisely decoding at the Gilbert-Varshamov distance that is relevant. It is at this distance that the decoding problem is expected to be the hardest. Recall that the Gilbert-Varshamov  distance  $\dgv(n,k)$ is given by 
$\dgv(n,k) \eqdef n \; h^{-1}(1-R)$, where $R \eqdef \frac{k}{n}$ is the code rate, $h$ is the binary entropy function $h(x) \eqdef - x \log_2 x -(1-x)\log_2(1-x)$ and $h^{-1}(x)$ its inverse ranging over $\left(0,\frac{1}{2}\right)$.
Above this bound, the number of solutions becomes exponential and this helps to devise more efficient decoders.

Concerning now the \LPN{} problem, it has long been recognized that having an unbounded number of queries or codelength while having a fixed
error probability $p$ per bit as in \LPN{} makes the problem really simpler. The best algorithms for solving this problem, 
are $\mathsf{BKW}$ type algorithms \cite{BKW03,EKM17} and are of subexponential complexity $2^{\OO{k/\log k}}$.
However, this is not true anymore if the number of queries is fixed and the error rate $p$ is chosen such that the problem is the hardest, namely when $h(p)=1-k/n$. In this case, the best algorithms behave exponentially in $\min(k,n-k)$ 
 despite many efforts on this issue
\cite{P62,S88,D91,BLP11,MMT11,BJMM12,MO15,BM17,BM18,CDMT22}. 
\newline

{\bf \noindent Reduction from Decoding to an \LPN{} Problem.}
Note that until very recently, all the best algorithms for solving the decoding problem or the \LPN{} problem when it is the hardest have been $\mathsf{ISD}$ algorithms. They all rely crucially on the Prange bet, namely that we have finally found after many trials a subset of positions of size $\approx n-k$ which contains almost all the errors. This was the situation since 1962 \cite{P62}.
There was at some point, just one exception \cite{D86} which relied instead on a collision technique and gave only a slight improvement in a very tiny rate range $R \in (0.98,1)$, but it was soon found out how to incorporate this 
technique in $\mathsf{ISD}$ algorithms \cite{S88,D89} to improve them. 
However in 2022,  a new algorithm called \RLPN-decoding was introduced in \cite{CDMT22}. It relies on a completely different approach following an old  idea called ``statistical decoding'' due to Al Jabri \cite{J01}. The new approach consists in reducing decoding to \LPN{}. For the first time in sixty years a strong competitor for $\mathsf{ISD}$ techniques was found: it outperforms  $\mathsf{ISD}$ techniques in the low rate regime, say $R \in (0,0.3)$ and the improvement is quite significant in the range $R \in (0,0.2)$ say. To explain the idea, assume we are given  an instance of the decoding problem $\yv=\cv+\ev$, where $\cv \in \CC$ and $|\ev|=t$.
As in statistical decoding, decoding relies on low weight parity-check equations, namely vectors $\vec{h}$ such that $\scal{\vec{h}}{\vec{c}}=0$ for any 
$\vec{c} \in \C$ (in other words, such $\vec{h}$'s belong to the dual code $\C^\perp$). However, in the new approach these parity-check equations 
are required to be of low weight only on a subset $\sN$ of positions. The rest of the positions $\sP$ correspond to the entries of $\vec{e}$ we aim to recover and is the secret  in the \LPN{} problem. The point of the whole approach is that 
$$
\langle \vec{y}, \vec{h} \rangle = \langle \vec{e}, \vec{h} \rangle 
 = \sum_{j \in \sP} h_{j}e_{j} + \sum_{j \in \sN} h_{j} e_{j}
 = \underbrace{\langle \vec{e}_{\sP}, \vec{h}_{\sP} \rangle }_{\text{lin. comb.}}+ \underbrace{\langle \vec{e}_{\sN}, \vec{h}_{\sN} \rangle}_{\text{\LPN{} noise}}.
$$
	
Here the notation $\vec{e}_\sP$ means the restriction of $\vec{e}$ to the positions in $\sP$:
	$\vec{e}_\sP=(e_i)_{i \in \sP}$. Vector $\vec{e}_\sP$ is interpreted as the \LPN{} secret $\vec{s}$, \ie $\vec{s} \eqdef \vec{e}_\sP$ and $\vec{h}_\sP$ as the linear 
	combination vector $\vec{a}$ while $\scal{\vec{e}_\sN}{\vec{h}_\sN}$ is the \LPN{} noise. Therefore, by computing ($\vec{h},\scal{\vec{y}}{\vec{h}})$ we really have access to the \LPN{} sample
	$$
	\underbrace{\vec{a}}_{\vec{h}_\sP},\overbrace{\underbrace{\langle \vec{s},\vec{a} \rangle}_{\scal{\vec{e}_\sP}{\vec{h}_\sP}} + \underbrace{e}_
	{\scal{\vec{e}_\sN}{\vec{h}_\sN}}}^{=\scal{\vec{y}}{\vec{h}}}.
	$$
The point of choosing low weight vectors $\vec{h}$ on $\sN$, is that  it is readily verified that this translates into the fact that the binary 
	random variable $\scal{\vec{e}_\sN}{\vec{h}_\sN}$ is biased, say $\Prob(\scal{\vec{e}_\sN}{\vec{h}_\sN}=1) =\frac{1-\varepsilon}{2}$ with 
	a bias $\varepsilon$ which gets bigger when the Hamming weight $|\vec{h}_\sN|$ of $\vec{h}$ on $\sN$ gets smaller.  
	
	Recovering $\vec{e}_\sP$ is then performed by producing enough parity-check equations to have enough information on $\vec{e}_\sP$ (we need about $N \approx 1/\varepsilon^2$ parity-check equations)  and amounts to solve the \LPN{} problem. This is done by the Fast Fourier Transform (FFT) and costs about $s 2^s$ where 
	$s \eqdef \card{\sP}$. We cannot afford more sophisticated techniques like the $\mathsf{BKW}$ algorithm which would give a sub-exponential algorithm, because we are very far away from the constant error probability regime. Here the bias $\varepsilon$ is exponentially small in the codelength, so we are really in the extreme noise regime, where on top of that we have hardly more \LPN{} samples than the number we need to recover the secret. In other words, we are in a situation where we can only use very basic algorithms, and the FFT which saves a factor $N$ when compared to plain exhaustive search over all possible \LPN{} secrets comes in handy here. The low weight parity-check equations are found by using collision techniques which are borrowed from advanced $\mathsf{ISD}$ techniques \cite{D89,BJMM12}. 
	
	The improvement upon statistical decoding given by \RLPN \ is really due to this splitting in two parts. Recall that plain statistical decoding uses parity-checks which are low weight on the {\em whole support}. In both cases, $\frac{1}{\varepsilon^2}$ of such parity-checks are needed, however in \RLPN \ decoding the bias $\varepsilon$ is way bigger because the weight we have on $\sN$ is way smaller for our parity-checks.\\

{\bf \noindent Dual Attacks, Some Negative Results and a New Analysis.} 
Statistical decoding \cite{J01} or its variant, namely \RLPNs decoding, both fall  into the category of {\em dual attacks} meaning a decoding algorithm that computes in a first step low weight codewords
in the dual code  and then computes the inner products of the received word $\yv$ with those parity-checks to infer some information about the error $\ev$. These methods can be viewed as the coding theoretic analogue of the dual attacks in lattice-based cryptography \cite{MR09}. Similarly to what happened in code-based cryptography, they were shown after a sequence of improvements \cite{A17,EJK20,GJ21,M22,CST22} to be able of being competitive with primal attacks, and the crucial improvement came from similar techniques, namely by a splitting strategy. 
Like in \RLPNs  decoding, the point is that this splitting  in two parts really allows to find dual vectors that are of smaller weight/norm on the restricted subset. Note that this idea was already put forward for statistical decoding (but not exploited there) in  \cite[\S8, p.33]{DT17} or \cite[p. 21]{DT17b}. 

However, the analysis in both settings relies on various independence assumptions, see for instance \cite[Ass. 4.4, Ass. 5.8]{M22} for dual attacks in lattices or \cite[Ass. 3.7]{CDMT22} for dual attacks for codes. In lattice-based cryptography, the dual attacks were strongly questioned recently in \cite{DP23} by showing that 
these independence assumptions made for analyzing dual attacks were in contradiction with some theorems in certain regimes or with well-tested heuristics in some other regimes. Note that it was already noticed in \cite[\S 3.4]{CDMT22} that the i.i.d. Bernoulli model implied by the \LPNs model for the $\langle \ev_{\sN}, \hv_{\sN} \rangle$'s is not always accurate, but it was conjectured there that the discrepancy between this ideal model and experiments does not impact the asymptotic analysis of the decoding based on this model. This was proved to be wrong in \cite{MT23} where it was shown that the number of candidates passing the validity test of the \RLPNs decoder given in \cite{CDMT22} is actually exponentially large for the parameters considered there, whereas there should be only one candidate passing the test if the algorithm was correct. However, this paper gave at the same time an approach for analyzing rigorously dual attacks in coding theory by bringing in a duality equation \cite[Prop. 1.3]{MT23} which relates the fundamental quantity manipulated by the decoder and the weight distribution of translates of a shortened version of the code to be decoded. By studying this weight distribution together with an assumption whose validity has been verified experimentally, a slightly modified \RLPNs decoder was introduced there and shown to attain the complexity exponent claimed in \cite{CDMT22}.

\subsection{Our Contribution}

\iftoggle{llncs}{
$(i)$ improving \RLPN-decoding by a reduction from sparse \LPN{} to plain \LPN{},\\
$(ii)$ a rigorous analysis of the decoding algorithm based on a simple assumption verified experimentally.
\mbox{}\medskip
}{
\begin{itemize}\setlength{\itemsep}{5pt}
	\item[$(i)$] improving \RLPN-decoding by a reduction from sparse \LPN{} to plain \LPN{},

	\item[$(ii)$] a rigorous analysis of the decoding algorithm based on a simple assumption verified experimentally.
\end{itemize}}

{\bf \noindent Reduction from Sparse \LPN{} to Plain \LPN{}.}	
Notice that the \LPN{} problem we have to solve is actually a {\em sparse} \LPN{} problem: $\vec{e}_\sP$ is not uniformly distributed among $\mathbb{F}_{2}^s$ since it is of low weight. Indeed, it is the restriction to $\sP$ of a vector which is itself of low weight. Unfortunately, 
the FFT algorithm used for recovering $\vec{e}_\sP$ is unable to exploit this fact. In a sense, what we need here to improve \RLPN \ decoding is an algorithm for solving sparse secret \LPN{} in the very noisy regime (but with an exponential number of samples). This can be 
done by using a $\mathsf{coded}$-$\mathsf{BKW}$ technique that was introduced in \cite{GJL14}. There it was not used as a technique for solving sparse $\mathsf{LPN}$ but as a technique to improve the reduction steps of the $\mathsf{BKW}$ algorithm \cite{BKW03} that put together pairs of vectors $\vec{a}$ and $\vec{a}'$ which are equal on a block of positions and add the corresponding \LPN{} samples to get an \LPN{} sample $(\vec{a}+\vec{a}',\scal{\vec{a}+\vec{a}'}{\vec{s}}+e+e')$ which is  more noisy but with vectors $\vec{a}$ which become sparser and sparser as the number of blocks increases. Asking exact collisions on the block needs a lot of \LPN{} samples and this can be relaxed by the $\mathsf{coded}$-$\mathsf{BKW}$ technique. It basically uses a code of the same length as the block of positions we are considering during the $\mathsf{BKW}$ step and asks only an approximate collision on the block meaning that the closest codewords $\vec{c}$ and $\vec{c}'$ to $\vec{a}$ and $\vec{a}'$ restricted to this block should be the same. 

To explain what we have in mind here, consider an \LPN{} 
sample
which is of the following form~$(\vec{h}_\sP,\scal{\vec{e}_\sP}{\vec{h}_\sP}+e)$. Choose now a linear code $\CCbis$ of length $s$ and dimension $\kbis$ (\ie a subspace of $\F_2^s$) which we know how to decode for any possible entry, meaning here that we can produce for any entry $\vec{y} \in \F_2^s$ a codeword $\cvbis\in \CCbis$ which is close {\em enough} to $\yv$.
Codes with this property are known under the name of {\em lossy source codes} in information theory. In \cite{GJL14} it was proposed to use for instance a product of small codes.
There are almost optimal codes (producing for a given dimension $\kbis$ almost optimal near codewords) using a low complexity decoder. Basically, the best that can be done is to produce codewords at distance  $\dgv(s,\kbis)$. 
For instance polar codes are asymptotically optimal \cite{KU10}, they attain asymptotically this Gilbert-Varshamov distance by using only a decoding algorithm of quasi-linear complexity $\OO{s \log s}$. 

Consider now a parity-check $\vec{h}$ of small weight $w$ on $\sN$ that we use for \RLPN-decoding and decode $\vec{h}_\sP$ with the lossy source code $\CCbis$: $\vec{h}_\sP = \cvbis + \evbis$ where $\cvbis \in \CCbis$ and $|\evbis|$ is small. Consider a generator matrix $\Gmbis$ of $\CCbis$, namely  a $\kbis \times s$  matrix such that
$\CCbis = \{\vec{u} \Gmbis:\;\vec{u} \in \F_2^{\kbis}\}$ (\ie the rows of $\Gmbis$ generate $\CCbis$). Notice now that
\begin{align*}
\scal{\vec{e}_\sP}{\vec{h}_\sP} &= \scal{\vec{e}_\sP}{\cvbis+\evbis} = \scal{\vec{e}_\sP}{\cvbis}+\scal{\vec{e}_\sP}{\evbis}  \\
&=  \scal{\vec{e}_\sP}{\vec{u} \Gmbis}+ \scal{\vec{e}_\sP}{\evbis}\;\; \text{(where $\vec{u} \in \mathbb{F}_{2}^{\kbis}$)}  \\
&=  \scal{\vec{e}_\sP\Gmbistransp}{\vec{u}}+\underbrace{\scal{\vec{e}_\sP}{\evbis}}_{\text{biased}} .
\end{align*}
If we plug this expression in the original \LPN{} sample $(\vec{h}_\sP,\scal{\vec{h}}{\vec{y}}=\scal{\vec{e}_\sP}{\vec{h}_\sP}+\scal{\vec{e}_\sN}{\vec{h}_\sN})$ we~obtain
\begin{equation*}
\label{eq:manip_produit_scalaire}
\scal{\vec{h}}{\vec{y}} = \scal{\vec{e}_\sP\Gmbistransp}{\vec{u}}+ \underbrace{\scal{\vec{e}}{\evbis}}_{\text{noise 1}}+ \underbrace{\scal{\vec{e}_\sN}{\vec{h}_\sN}}_{\text{noise 2}}.
\end{equation*}
In other words, we have a new \LPN{} problem where
	\begin{equation}
        \label{eq:LPN-problem}
	\underbrace{\vec{a}}_{\vec{u}},\overbrace{\underbrace{\langle \vec{s},\vec{a} \rangle}_{\scal{\vec{e}_\sP\Gmbistransp}{\vec{u}}} + \underbrace{e}_
	{\scal{\vec{e}_\sP}{\evbis}+ \scal{\vec{e}_\sN}{\vec{h}_\sN}}}^{\scal{\vec{y}}{\vec{h}}}.
	\end{equation}
The new secret is not anymore a part $\vec{e}_\sP$ of the error but a linear combination $\vec{e}_\sP \Gmbistransp$ of it and the \LPN{} noise has increased somehow. However, now the secret is way smaller, it belongs to $\mathbb{F}_{2}^{\kbis}$. The situation is changed significantly by this. Before, basically the optimal parameters for \RLPN-decoding were such that the cost of FFT decoding the \LPN{} secret, namely $\OO{s2^s}$ is of the same order as $1/\varepsilon^2$ the number of parity-check equations we need. Here $\varepsilon$ is defined by
\begin{equation*}\label{eq:bias_of_e}
\Prob(e =1)=\frac{1-\varepsilon}{2}.
\end{equation*}
Recall that $\varepsilon$ is basically a decreasing function of the weight $w$ of the parity-check equations we are able to produce. 
Here since we do not pay anymore $\OO{s2^s}$ for 
FFT decoding the new \LPN{} secret but $\OO{\kbis2^{\kbis}}$ we can take larger values for $s$ which themselves give a smaller support~$\sN$ resulting in much smaller weight $w$ on $\sN$ and thus the bias term coming from $\scal{\vec{e}_\sN}{\vec{h}_\sN}$ is much smaller. Of course there is an additional noise term now which is $\scal{\vec{e}_\sP}{\evbis}$. However, all in all, the gain we have by being able to use a much larger $s$  outweighs the additional noise term. It can also be observed that we do not recover $\vec{e}_\sP$ but $\kbis$ linear combinations of bits of $\vec{e}_\sP$. This is easy to fix by running a few times more this algorithm with other lossy source codes $\CCbis$ until getting enough linear combinations to be able to recover $\vec{e}_\sP$.

We call this new algorithm \nRLPN-decoding, since it is based on two successive reductions: first we reduce the problem to sparse-\LPN{}, then we reduce the sparse-\LPN{} to a plain-\LPN{} problem as explained above.

\medskip
{\bf \noindent \nRLPN-Decoding and its Analysis.} 
It turns out that the \LPNs problem given in Equation~\eqref{eq:LPN-problem} is more structured than a standard \LPNs problem 
and like what happened in the \RLPN{} algorithm \cite{CDMT22}, producing the most likely candidate for the \LPNs problem does not necessarily produce the right candidate $\ev_\sP \Gmbis[\transpose]$ even if we have enough samples for ensuring that in the ideal i.i.d model of the \LPNs problem the most likely candidate would indeed be the right solution. Again, the i.i.d. model is not accurate. We have to use the whole information given by the FFT and output for $L$ big enough the $L$ most likely solutions to have a chance to have $\ev_\sP \Gmbis[\transpose]$ in the list. However, verifying whether a candidate $\sv$ for 
$\ev_\sP \Gmbis[\transpose]$ is indeed valid is relatively straightforward: 
\iftoggle{llncs}{\\
{\it (a)} we can as in the \RLPNs algorithm make a bet on the weight $|\ev_\sP|$, say 
$|\ev_\sP|=t'$ (and run enough \nRLPNs decoding steps until finding a partition $\sP \cup \sN$ for which this bet is valid),\\
{\it (b)} recover $\ev_\sP$ by solving the decoding problem (in its syndrome form) $\sv = \ev_\sP \Gmbis[\transpose]$ and $|\ev_\sP|=t'$,\\
{\it (c)} check whether the putative candidate $\vv$ for $\ev_\sP$ we get can be extended to a complete solution by solving the decoding problem
$\yv=\cv+\ev$, $\ev_\sP = \vv$, $\cv \in \CC$, $|\ev|=t$ which is much easier to solve than the original decoding problem due to the partial knowledge
about $\ev$ ($\ev_\sP=\vv$).}
{
\begin{itemize}\setlength{\itemsep}{5pt}
	\item[$(a)$] we can as in the \RLPNs algorithm make a bet on the weight $|\ev_\sP|$, say 
	$|\ev_\sP|=t'$ (and run enough \nRLPNs decoding steps until finding a partition $\sP \cup \sN$ for which this bet is valid),

	\item[$(b)$] recover $\ev_\sP$ by solving the decoding problem (in its syndrome form) $\sv = \ev_\sP \Gmbis[\transpose]$ and $|\ev_\sP|=t'$,

	\item[$(c)$] check whether the putative candidate $\vv$ for $\ev_\sP$ we get can be extended to a complete solution by solving the decoding problem
	$\yv=\cv+\ev$, $\ev_\sP = \vv$, $\cv \in \CC$ and $|\ev|=t$ which is much easier to solve than the original decoding problem due to the partial knowledge
	about $\ev$, \ie $\ev_\sP=\vv$.
\end{itemize}
}

The whole problem we face here for analyzing the problem is the same as the one that was faced to analyze the \RLPNs algorithm, the i.i.d. \LPNs model
is not valid and we really have to get rid of the independence assumptions. Part of this work is achieved 
by adapting one of the fundamental tools used for analyzing \RLPNs decoding, namely \cite[Proposition 3.1]{CDMT22} which gives a formula of the bias $\varepsilon$ in terms of Krawtchouk polynomials. We will obtain a generalization of this proposition adapted to the \nRLPNs decoder, namely Proposition \ref{prop:biasCodedRLPN} in \S \ref{sec:reductiontoLPN}. Note that Proposition \ref{prop:biasCodedRLPN} does not rely on unproven assumptions contrarily to what is done in dual attacks in lattice-based cryptography where 
the corresponding result is achieved through independence assumptions.

Estimating the number $L$ of candidates for the \LPNs problem given in Equation \eqref{eq:LPN-problem} we have to keep for being sure to have $\ev_\sP\Gmbis[\transpose]$ is even more delicate. It requires a careful adaptation to our setting of \cite{MT23} that analyzed the \RLPNs decoder.   
Here, we will not be able to avoid completely assumptions for performing the analysis (but this was also the case in \cite{MT23}). However, again we will not resort to independence assumptions which seem in our context not only to be wrong strictly speaking, but also to be unable to be good enough for capturing the size of $L$. We will namely develop some tools analogous to what has been achieved in \cite{MT23}:
\begin{itemize}\setlength{\itemsep}{5pt}
\item[$(i)$] a duality result, namely Proposition \ref{prop:exp_bias} of \S \ref{sec:tools}, which expresses the FFT value of a candidate as a 
 weighted sum of the product of evaluations of Krawtchouk polynomials where the weights come from a certain weight distributions of codes related to $\CC$ and $\CCbis$. This is an adaptation of \cite[Prop. 3.2]{MT23} to our setting and is the key for estimating $L$ as explained in \S \ref{ss:intuition},

\item[$(ii)$] an estimation of this sum with probabilistic considerations. These probabilistic considerations are rigorous for the part of the sum which is most certainly the dominating term. However for the part of the sum which is very likely to be negligible, we lack accurate tail bounds for the number of codewords of a given weight in a random linear code and in this case we just conjecture that the part of the sum which seems negligible and for which we have only partial control with the probabilistic tool at hand, is indeed negligible. This conjecture has been verified experimentally and we even used a very crude approximation of this weighted sum with the help of independent Poisson variables which captures the size of $L$ obtained in our experiments and which implies our conjecture.
\end{itemize}
All in all with the help of a conjecture that we verified experimentally, we are able to capture the size of $L$ and to obtain a formula for the complexity of \nRLPNs decoding. The key tool for performing this analysis, namely the duality result, can be readily adapted to lattices (see \S \ref{sec:lattice}). It turns out that even a crude use of this duality result gives a good explanation of the part of the experimental curve departing from the theoretical curve based on the standard independence assumption  found in \cite[Fig. 3]{DP23}. This substantiates the claim made in \cite[\S 6]{MT23} that the code duality result of \cite{MT23} carries over to the lattice setting and can be used to predict dual attacks without using the independence assumption.

\medskip
{\bf \noindent The Results Obtained by this New Approach.}
This new approach results in a very significant gain compared to \RLPN \ decoding. Our most advanced version of \nRLPN-decoding algorithm performs better than the current state of the art $\mathsf{ISD}$ algorithm for all rates $R \leq \hardestrate$ \iftoggle{llncs}{}{as shown in Figure \ref{fig:compISDDoubleRLPNRLPN}}.
\begin{figure}[h!]
	\centering
	\includegraphics[height=7.5cm]{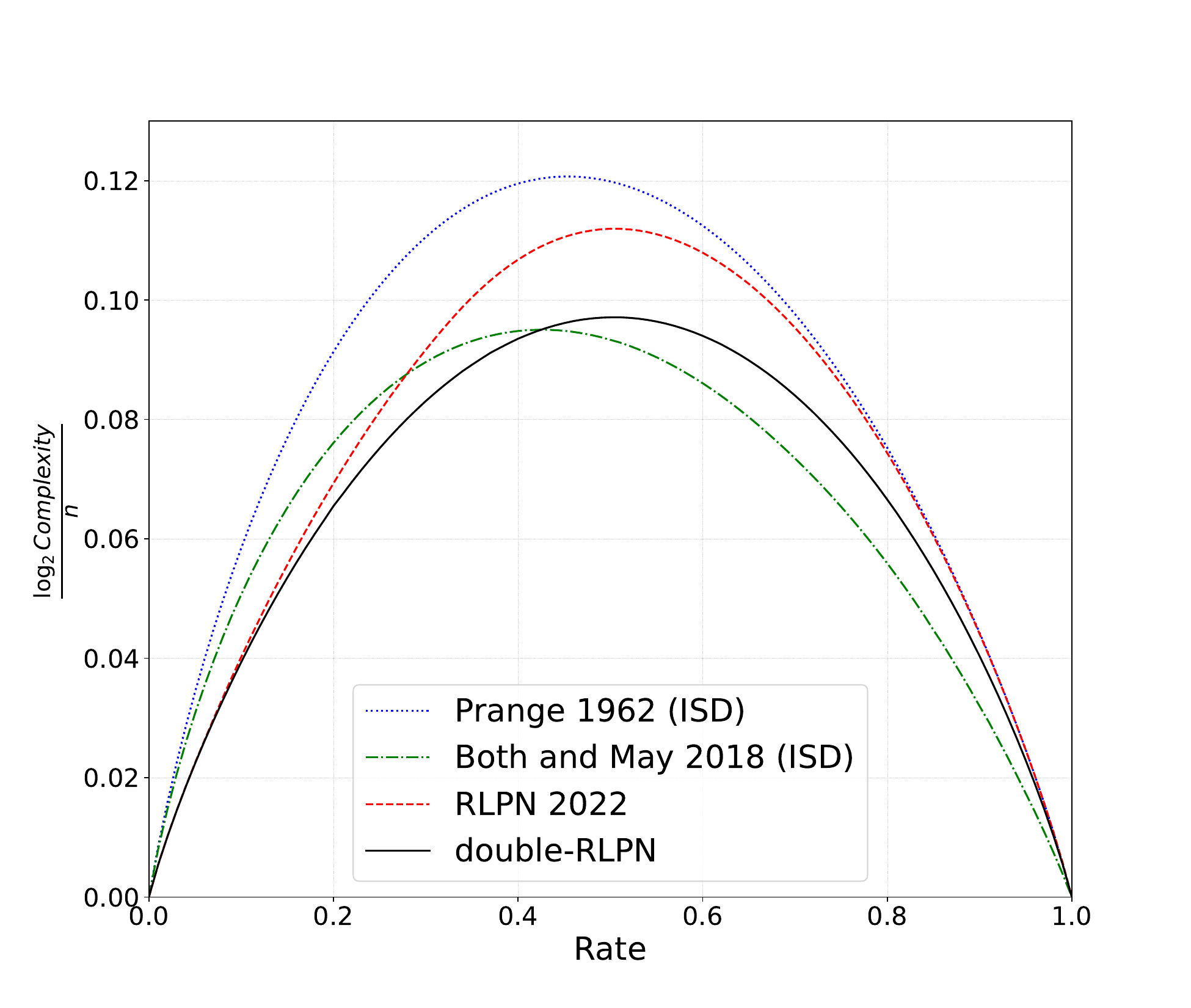}
	\caption{\label{fig:compISDDoubleRLPNRLPN} Asymptotic complexity exponent of some decoding algorithms: 
		our new \nRLPN \ decoder, 
the \RLPN \ decoder,  
Both and May algorithm \cite{BM18} (with the correction of \cite{CDMT22,E22}) which is the state-of-the-art of $\mathsf{ISD}$ decoders and the Prange decoder \cite{P62}.
}
\end{figure}

\medskip
{\bf \noindent Concurrent/related work.}
Very recently, we became aware that the prediction we have made on the score function for lattices by using our duality result and crude estimates of the relevant sum
(see  \S \ref{sec:lattice}) has also been obtained by using as we do here Bessel functions and related tools in \cite{DP23a}. This paper provides a much more in depth study as we do here.

 \section{Notation and Coding Theory Background}
\label{sec:notation}

{\bf \noindent Basic Notation.}
Vectors and  matrices are respectively denoted in bold letters 
and bold capital letters such as $\av$ and $\Am$. The entry at index $i$ of the vector $\xv$ is denoted by $x_i$ or $x(i)$. The canonical inner product $\sum_{i=1}^n x_i y_i$ between two vectors $\xv$ and $\yv$ of $\mathbb{F}_{2}^n$ is denoted by $\scal{\xv}{\yv}$ where $\F_2$ denotes the binary field. 
Let  $\sI$ be a list of indexes. We denote by $\xv_{\sI}$ the vector $(x_i)_{i \in \sI}$. In the same way, we denote by $\Mm_{\sI}$ the sub-matrix made of the columns of $\Mm$ which are indexed by $\sI$. We denote by $\vec{0}_n \in \F_2^{n \times n}$ and $\vec{Id}_n \in \F_2^{n \times n}$ the null matrix and the identity matrix of size $n$ respectively.
The concatenation of two vectors $\xv$ and $\yv$ is denoted by $\xv || \yv$. 
The Hamming weight of a vector $\xv$ and the cardinality of a finite set $\sA$ are denoted in the same way by $|\xv|$ and $\card{\sA}$ respectively. There will be no confusion since they apply to different objects. 
Notation $\IInt{a}{b}$ stands for the set of the integers between $a$ and $b$, both included. 
Furthermore, we let $\mathcal{S}_{w}^{n}$ denote the Hamming sphere of $\F_{2}^{n}$ with radius $w$ and centered at $\vec{0}$, namely \iftoggle{llncs}{$\mathcal{S}_w^n \eqdef \{ \xv \in \F_2^n \: : |\xv|=w\}$.}{
$$
\mathcal{S}_w^n \eqdef \{ \xv \in \F_2^n \: : |\xv|=w\}.
$$
}
\iftoggle{llncs}{\newline}{}

{\bf \noindent Probabilistic Notation.}
For a finite set $\sS$, we write $X \drawn \sS$ when $X$ is an element of $\sS$ drawn uniformly at random in it.
 For a Bernoulli random variable $X$, denote by $\bias(X)$ the quantity 
	\iftoggle{llncs}{$\bias(X) \eqdef \Prob(X=0)-\Prob(X=1)$.}{
	$$
	\bias(X) \eqdef \Prob(X=0)-\Prob(X=1).
	$$}
	For a Bernoulli random variable $X$ of parameter $p=\frac{1-\varepsilon}{2}$, \ie $\Prob(X=1)=\frac{1- \varepsilon}{2}$, we have $\bias(X)=\varepsilon$. 
\newline 

\mysubtitle{Fourier Transform.}
Let $f \: : \F_2^{n} \rightarrow \RR$ be a function. We define its Fourier transform $\widehat{f} : \F_2^n \rightarrow \RR$ as
\begin{equation} \label{eq:def_fourier}
\widehat{f}\left(\xv\right) = \sum_{\uv \in \F_2^n} f(\uv) \left(-1\right)^{\scp{\xv}{\uv}}.
\end{equation}
{\bf \noindent Soft-O Notation.}
For real valued functions defined over $\RR$ or $\NN$ we define $o()$, $\OO{}$, $\Om{}$, $\Th{}$, in the usual way and also use the less common notation $\OOt{}$ and $\Omt{}$, where 
$f = \OOt{g}$ means that $f(x) = \OO{g(x) \log^kg(x)}$ and $f=\Omt{g}$ means that $f(x)=\Om{g(x)\log^k g(x)}$ for some $k$. We will use this 
for functions which have an exponential behavior, say $g(x)=e^{\alpha x}$, in which case $f(x)=\OOt{g(x)}$ means that 
$f(x)=\OO{P(x)g(x)}$ where $P$ is a polynomial in $x$.
We also use $f = \omega(g)$ when $f$ dominates $g$ asymptotically; that is when $\mathop{\lim}\limits_{x \rightarrow \infty} \frac{|f(x)|}{g(x)} = \infty$.
\newline

{\bf \noindent Coding Theory.} A binary linear code $\C$ of length $n$ and dimension $k$ is a subspace of $\mathbb{F}_{2}^n$ of dimension $k$. We say that it has parameters $[n, k]$ or that it is an $[n, k]$-code. Its {\em rate} $R$ is defined as $R \eqdef \frac{k}{n}$. 
A generator matrix $\Gm$ for $\C$ is a full rank $k \times n$ matrix over $\mathbb{F}_{2}$ such that
$$
\C =\left\{\uv \Gm:\uv \in \F_{2}^k \right\}.
$$
A parity-check matrix $\Hm$ for $\C$ is a full-rank $(n-k)\times n$ matrix over $\mathbb{F}_{2}$ such that
$$
\C = \left\{\cv \in \F_{2}^n: \Hm \transp{\cv} = \mathbf{0} \right\}.
$$
In other words, $\C$ is the null space of $\Hm$. The dimension of the code is given by $\dim\left(\CC\right) \eqdef k$.
The code whose generator matrix is the parity-check matrix of $\CC$ is called the dual code of $\CC$. It might be seen as the subspace of parity-checks of $\CC$ and is defined equivalently as
\begin{definition}[Dual Code]
The {\em dual code} $\CC^\perp$ of an $[n,k]$-code $\CC$ is an $[n,n-k]$-code which is defined by
\iftoggle{llncs}{$
\CC^\perp \eqdef \left\{\vec{h} \in \mathbb{F}_{2}^n: \forall \cv \in \CC,\; \scal{\cv}{\vec{h}}=0 \right\}.
$
}
{
$$
\CC^\perp \eqdef \left\{\vec{h} \in \mathbb{F}_{2}^n: \forall \cv \in \CC,\; \scal{\cv}{\vec{h}}=0 \right\}.
$$
}
\end{definition}

Sometimes it is considered in the literature the following equivalent version of the decoding problem (see Problem \ref{pb:decodingFqSynd} as defined in the introduction) by using instead the parity-check matrix and syndrome point of view

\begin{pb}[Decoding a fixed error weight via syndromes]
	\label{pb:decodingFqSynd}
	Let $\C$ be an $\lbrack n,k \rbrack$-code with parity-check matrix $\vec{H}\in\F_{2}^{(n-k)\times n}$. 
	We are given a syndrome $\vec{s} \in \mathbb{F}_{2}^{n-k}$, an integer $t$ and we want to find an error vector $\vec{e} \in \mathbb{F}_{2}^n$ of Hamming weight $|\vec{e}|=t$ for which $\vec{H}\vec{e}^{\transpose} = \vec{s}^{\transpose}$. 
\end{pb}

It is readily seen that both Problems \ref{pb:decodingFq} and \ref{pb:decodingFqSynd} are equivalent: given $\CC$ with parity-check matrix $\vec{H}$, then decoding $\vec{c} + \vec{e}$ with a codeword $\vec{c} \in \CC$ and $\vec{e} \in \cS_{t}^{n}$ amounts to recover $\vec{e}$ from $\vec{H}\vec{y}^{\transpose} = \vec{H}\vec{e}^{\transpose}$ as by definition $\vec{H}\vec{c}^{\transpose} = \vec{0}$.

When $\CC$ is an $[n,k]$-code and $\xv \in \F_2^n$ we let \iftoggle{llncs}{$\CC + \xv \eqdef \{ \cv + \xv \: , \cv \in \CC\}$}{
	$$
	\CC + \xv \eqdef \{ \cv + \xv \: , \cv \in \CC\}
	$$} 
	denote a coset of $\CC$ and we denote by $N_i\left(\CC + \xv \right)$ the number of words of hamming weight $i$ in the coset $\CC + \xv$, namely \iftoggle{llncs}{$N_i\left(\CC + \xv\right) \eqdef \card{\CC + \xv \cap \mathcal{S}_i^n}$.}{
$$
N_i\left(\CC + \xv\right) \eqdef \card{\CC + \xv \cap \mathcal{S}_i^n}.
$$
}
An important quantity is the {\em Gilbert-Varshamov} distance which is defined as
\begin{definition}[Gilbert-Varshamov distance] The Gilbert-Varshamov distance $\dgv(n,k)$ associated to a length $n$ and dimension $k$ is defined as
 the largest integer $d$ such that
$$
2^k \card{\sB_d} < 2^n
$$
where $\sB_d$ is the Hamming ball centered at $\vec{0}$ in $\F_2^n$ and radius $d$, that is 
$\sB_d \eqdef \{\xv \in \F_2^n:\;|\xv|\leq d\}$.
\end{definition}
This quantity has two different interpretations. On one hand, it corresponds up to a constant term to the {\em typical minimum distance} of a linear code of length $n$ and dimension $k$, but it is also related to the expected number of solutions of the decoding problem for a random linear $[n,k]$-code which is defined as follows.

\begin{pb}[\textup{$(n,k,t)$ Decoding Problem - $\mathsf{DP}(n,k,t)$}]\label{pb:DP}\iftoggle{llncs}{\vspace*{-0.2cm}}{\mbox{ }}
	\begin{itemize}\iftoggle{llncs}{}{\setlength{\itemsep}{5pt}}
		\item \textup{Given:}  $(\vec{G},\vec{y}\eqdef \vec{m}\vec{G} + \vec{x})$ where $\vec{m},\vec{G}$ and $\vec{x}$ are respectively picked uniformly at random over $\mathbb{F}_{2}^{k}$, $\mathbb{F}_{2}^{k \times n}$ and $\cS_{t}^{n}$.

		\item \textup{Aim:} an error $\vec{e}\in \mathbb{F}_{2}^{n}$ of Hamming weight $t$ such that $\vec{y}-\vec{e} = \vec{z}\vec{G}$ for some $\vec{z} \in \mathbb{F}_{2}^{k}$. 
	\end{itemize}
\end{pb}

This problem really corresponds to decoding at distance $t$ the $\lbrack n,k \rbrack$-code admitting $\vec{G}$ as generator matrix. 
The largest weight $t$ for which we might hope for having a single solution (strictly speaking when we look for solutions of weight $\leq t$ and not exactly $t$, but the difference
between these two notions is generally irrelevant) is given by the Gilbert-Varshamov distance $\dgv(n,k)$. At this distance, the expected number of solutions
is readily seen to be $\Th{1}$ whether we look 
at codewords at distance exactly $t$ from the received word $\yv$ or 
at distance $\leq t$.

It will also be very convenient to consider the operation of puncturing a code, \ie keeping only a subset of entries in a codeword.
\begin{definition}[Punctured Code]
For a code $\CC$ and a subset $\sI$ of code positions, we denote by $\CC_\sI$ the punctured code obtained from $\CC$ by keeping only the positions in $\sI$, \ie
$$
\CC_{\sI}= \{\cv_\sI: \cv \in \CC\}.
$$
\end{definition}
\begin{definition}[Shortened Code]
	For a code $\CC$ and a subset $\sI$ of code positions, we denote by $\CC^\sI$ the shortened code is defined by
	$$
	\CC^{\sI}= \{\cv_\sI: \cv \in \CC \mbox{ and } \cv_{ \llbracket 1,n \rrbracket \setminus \sI } = \vec{0}\}.
	$$
\end{definition}
It is readily seen that we have 
\begin{equation}\label{eq:shortPunct}  
\left(\CC^{\sI}\right)^{\perp} = \left(\CC^{\perp}\right)_{\sI} \quad \mbox{and} \quad \left(\CC_{\sI}\right)^{\perp} = \left(\CC^{\perp}\right)^{\sI}.
\end{equation} 

{\bf \noindent Krawtchouk Polynomial.}
We recall here some properties about Krawtchouk polynomial that will be useful in the article. Many useful properties can be found in \cite[\S 2.2]{KS21}
\begin{definition}{(Krawtchouk polynomial)}
	We define the Krawtchouk polynomial $K_w^{(n)}$ of degree $w$ and of order $n$ as 
	\iftoggle{llncs}{$K_w^{(n)}\left( X \right)  \eqdef  \sum_{j = 0}^w \left(-1\right)^j \binom{X}{j} \binom{n-X}{w-j}$.}{
		$$
		K_w^{(n)}\left( X \right)  \eqdef  \sum_{j = 0}^w \left(-1\right)^j \binom{X}{j} \binom{n-X}{w-j}.
		$$
}
\end{definition}
The following fact is well known: it gives an alternate expression of the Krawtchouk polynomial (see for instance \cite[Lemma 5.3.1]{L99}).
\begin{fact} \label{fact:krawtchouk_bias}
	For any $\vec{x} \in \F_2^n$,
	\begin{equation}
		K_w^{(n)} \left( |\xv| \right) = \widehat{\un_w}(\xv)=\sum_{\yv \in \F_2^n : |\yv|= w } \left(-1\right)^{\langle \xv, \yv \rangle}.
	\end{equation}
where $\un_w$ is the characteristic function of the Hamming sphere $\cS^n_w$ of radius $w$.
\end{fact}
We recall here the summary of some known results about Krawtchouk polynomials made in~\cite{CDMT22}.
\begin{proposition}{\cite[Prop. 3.5, Prop. 3.6]{CDMT22}}\label{prop:expansion}
	\begin{enumerate}\iftoggle{llncs}{}{\setlength{\itemsep}{7pt}}
		\item {\bf Value at 0.} For all $0 \leq w \leq n$, $K_w^{(n)}(0)=\binom{n}{w}$.
		\item {\bf Reciprocity.} For all $0 \leq t,w \leq n$, $\binom{n}{t}K_w^{(n)}(t)=\binom{n}{w}K_t^{(n)}(w)$.
		\item {\bf Roots.} The polynomials $K_w^{(n)}$'s have $w$ distinct roots which lie in the interval 
		$\llbracket n/2-\sqrt{w(n-w)},n/2+\sqrt{w(n-w)} \rrbracket.$ The distance between roots is at least $2$ and at most $o(n)$.
		\item {\bf Magnitude in and out the root region.} Let $\tau$ and $\omega$ be two reals in $[0,1]$. Let $\omega^\perp \eqdef \frac{1}{2}-\sqrt{\omega(1-\omega)}$, and let $z \eqdef \frac{1-2\tau  - \sqrt{D}}{2 (1-\omega)}$ where $D \eqdef \left(1- 2 \tau \right)^2-4 \omega(1-\omega)$.
		
		Define $\widetilde{\kappa}(\tau,\omega) \eqdef \left\{
		\begin{array}{ll}
			\tau \log_2(1-z)+ (1-\tau) \log_2(1+z) - \omega\log_2 z  & \mbox{ if } \tau \in [0,\omega^\perp], \\
			\frac{1 - h(\tau)+h(\omega)}{2} & \mbox{ otherwise.}
		\end{array}
		\right.$
		\begin{itemize}\iftoggle{llncs}{}{\setlength{\itemsep}{5pt}}
			\item 4.1. If $\tau \leq \frac{1}{2}-\sqrt{\omega(1-\omega)}$, then for all $t$ and $w$ such that $\lim\limits_{n \to \infty} \frac{t}{n} = \tau$ and $\lim\limits_{n \to \infty} \frac{w}{n} = \omega$ we have $K_w^{(n)}(t) = 2^{n \left(\widetilde{\kappa} \left(\tau, \omega\right) + o(1)\right)}$.
			\item 4.2.	If $\tau > \frac{1}{2}-\sqrt{\omega(1-\omega)}$, then there exists $t(n)$ and $w(n)$ such that $\lim\limits_{n \to \infty} \frac{t}{n} = \tau$, $\lim\limits_{n \to \infty} \frac{w}{n} = \omega$ and $\left| K_w^{(n)}(t) \right| = 2^{n \left(\widetilde{\kappa} \left(\tau, \omega\right) + o(1)\right)}$.
		\end{itemize}
	\end{enumerate}
\end{proposition}

 \section{Reduction from Sparse to Plain \LPN{}}
\label{sec:reductiontoLPN}

The purpose of this section is  to explain in detail the reduction from sparse to plain \LPN{} and to give an important result about the bias of the resulting \RLPNs samples.
We assume from now on that we are given and $\lbrack n,k \rbrack$-code $\CC$ and a $\yv \in \F_2^n$ such that
$$
\yv \eqdef \cv + \ev, \quad \cv \in \CC, \; |\ev| =t, 
$$ 
and we want to find $\cv$ and $\ev$.
\subsection{The Approach}

First, we randomly select a subset $\sP \subseteq \llbracket 1,n \rrbracket$ of $s$ positions,
where $s$ is a parameter that will be chosen later. Let $\sN \eqdef \llbracket 1,n \rrbracket \setminus \sP$ be the complementary set of $\sP$. Here~$\sP$ corresponds to the entries of $\ev$ we aim to recover. As explained in the introduction, the basic step of the decoding algorithm is to compute a large set $\sW$  of parity-check equations of low weight $w$ on $\sN$ and to compute all the
$\scp{\yv}{\hv}$ with $\hv$ ranging over $\sW$. In \RLPNs decoding, the approach is 
to exploit directly that we have a number $\card{\sW}$ of \LPN{} samples $(\hv_\sP,\scp{\hv}{\yv})$  which can be viewed as an 
\LPNs sample $\left( \av, \scp{ \av}{\sv} + e \right)$ by letting $\av \eqdef \hv_\sP$, $\sv \eqdef \ev_{\sP}$,  $e \eqdef \scp{\hv_{\sN}}{\ev_{\sN}}$. Indeed,
$$\scp{\hv}{\yv} = \scp{\hv}{\cv+\ev} = \underbrace{\scp{\hv}{\cv}}_{=0} + \scp{\hv}{\ev} = \underbrace{\scp{\hv_\sP}{\ev_{\sP}}}_{=\scp{\av}{\sv}}+\underbrace{\scp{\hv_\sN}{\ev_{\sN}}}_{\text{\LPN{} noise $e$}}.$$ Notice that we really have a sparse \LPN{} problem because of the sparseness of the secret $\ev_\sP$ which is not exploited in \cite{CDMT22} and only exploited to verify the solution in the corrected 
\RLPN{} algorithm of \cite{MT23}. The point of this article is to exploit the sparseness of $\ev_{\sP} \in \mathbb{F}_{2}^{s}$ right away in order to reduce the dimension $s$ of the secret. This is obtained by introducing an auxiliary code 
$\CCbis$ of length $s$ and dimension $\kbis$ which will be instrumental for reducing the dimension $s$ of the secret down to $\kbis$. This is obtained as follows. We will assume that $\CCbis$ is chosen as a code with an {\em efficient list-decoding procedure} at distance $\tbis$. 

\begin{definition}[Efficiently list decodable code]\label{def:effListDeco} A code $\CC$ of length $n$ is said to be efficiently decodable code at distance $t$ if it outputs 
for any $\yv \in \F_2^n$ a non empty list of codewords of $\CC$ at distance $t$ in time $2^{o(n)}$.
\end{definition} 
Moreover from now on, we assume that
\begin{notation}\label{notat:codeAux} 
$\CCbis$ is an $[s,\kbis]$ efficiently list decodable for some distance $\tbis$. We denote by
$\decodebis(\zv)$ the set of all codewords of $\CCbis$ at distance $\tbis$ from $\zv \in \F_2^s$, namely
		$$
		\decodebis(\zv) = \left\{ \cv_{\textup{aux}} \in \CCbis \; : \; \left|\cv_{\textup{aux}} + \zv\right| = \tbis  \right\}.
		$$
\end{notation} 

\begin{remark}
	In our instantiation, $\tbis$ is chosen such that $\tbis \approx \dgv(s,\kbis)$, thus we typically have~$\card{\decodebis(\hv_\sP)}=\Th{1}$.
\end{remark}

Now, let us consider $\cvbis \in \decodebis(\hv_{\sP})$, a codeword of $\CCbis$ at distance $\tbis$ of $\hv_{\sP}$. It is readily seen that 
$\scp{\yv}{\hv}$ decomposes as:

$$
\langle \yv,\hv\rangle = \underbrace{\langle \ev_{\sP},\cv_{\textup{aux}} \rangle}_{\mbox{linear comb.}}  + \underbrace{\langle \ev_{\sP},\hv_{\sP} + \cvbis \rangle + \langle \ev_{\sN}, \hv_{\sN} \rangle}_{\mbox{``new'' \LPN{} noise}}. $$

Let us start by defining $\CCbis$ with a generator matrix $\Gmbis \in \F_2^{s \times \kbis}$.
Then, knowing~$\cv_{\textup{aux}} \in \CCbis$ is equivalent to know $\mvbis \in \F_{2}^{\kbis}$ such that 
$$
\cvbis = \mvbis \Gmbis.
$$ 

We can therefore rewrite $\scp{\ev_\sP}{\cvbis}$ as
\begin{equation*}
	\langle \ev_{\sP},\cvbis \rangle = \scp{\ev_\sP}{\mvbis \Gmbis}
	= \langle  \ev_{\sP}  \Gmbis[\transpose], \mvbis \rangle. \label{eq:scalProduct} 
\end{equation*}
We have therefore for each parity-check equation $\hv$ of weight $w$ on $\sN$ that we have computed (\ie \ for all $\hv \in \sW$) and each 
codeword $\mvbis \Gmbis$ of $\CCbis$ at distance $\tbis$ from $\hv_\sP$, an \LPNs sample $(\mvbis,\scp{\yv}{\hv})$ which can be viewed as such by noticing that it is indeed equal to
\begin{equation}\label{eq:LPN}
	\left( \av, \langle \av,\sv \rangle + e \right) \quad \mbox{with} \quad \left\{ \begin{array}{l}
		\av \eqdef \mvbis \\
		\sv \eqdef \ev_{\sP} \Gmbis[\transpose] \\
		e \eqdef \langle \ev_{\sP},\hv_{\sP} + \cvbis \rangle + \langle \ev_{\sN}, \hv_{\sN} \rangle
	\end{array}
	\right. 
\end{equation} 
Notice here that, if $ \decodebis(\hv_{\sP})$ contains more than one element, we can compute such \LPN{} samples for each different $\cvbis \in \decodebis(\hv_{\sP})$.
The secret in the above \LPN{} sample is no longer given by $\ev_{\sP}$ that we want to recover (contrarily
to \RLPN-decoding \cite{CDMT22}), but is given by $\ev_{\sP}  \Gmbis[\transpose] \in \F_{2}^{\kbis}$ which are $\kbis<s$ linear equations involving the $s = \card{\sP}$ bits of the vector $\ev$ we are looking for.

The main advantage of our new technique is that we end up with an \LPN{} problem whose dimension of the secret has decreased from $s$ to $\kbis$. However, the noise has increased; let us describe how it behaves in the following paragraph.

\subsection{Estimating the New Noise} The error $e$ in Equation \eqref{eq:LPN} is biased toward zero and its bias is a function of $n,s,t$ and $u,w,\tbis$ which are respectively 
$$
u \eqdef \left|\ev_{\sN}\right|, \quad w \eqdef \left|\hv_{\sN}\right| \quad \mbox{and} \quad \tbis \eqdef |\hv_{\sP} + \cvbis|.
$$
In the following statement we compute the bias of $e$ over all the possible LPN samples, that is we compute 
$$
\bias_{\left(\hv,\cvbis\right) \drawn \widetilde{\sH}} \left(\langle \ev_{\sP}, \vec{h}_{\textup{aux}}+\cv_{\textup{aux}}  \rangle + \langle \ev_{\sN}, \hv_{\sN} \rangle \right) 
= \frac{1}{\card{\widetilde{\sH}}} \sum_{\left(\hv,\cvbis\right) \in \widetilde{\sH}} (-1)^{\langle \ev_{\sP}, \hv_{\sP} + \cvbis  \rangle + \langle \ev_{\sN}, \hv_{\sN} \rangle} 
$$
where $\widetilde{\sH}$ is defined by
\begin{definition}
	\begin{equation}
			\widetilde{\sH} \eqdef \{ \left( \hv, \cvbis \right) \in \CC^{\perp} \times \CCbis\: : \: |\hv_{\sN}| = w \mbox{ and } |\hv_{\sP} + \cvbis| = \tbis  \}.
	\end{equation}
\end{definition}
It is tempting to conjecture that this bias is well approximated by the bias of a Bernoulli variable 
$X \eqdef \scal{\ev_{\sP}}{\vec{e}_{\textup{aux}}} + \scal{\ev_{\sN}}{\wv}$ where $\vec{e}_{\textup{aux}}$ and $\wv$ 
are respectively drawn uniformly at random in the Hamming spheres
$\cS_{\tbis}^{s}$ and $\cS_{w}^{n-s}$.
The sum $ \scal{\ev_{\sP}}{\vec{e}_{\textup{aux}}} + \scal{\ev_{\sN}}{\wv}$
is performed over $\F_2$ and all the vectors are independent random variables. Because of the independence of the random variables, from the
straightforward fact that $\bias(X_1 + X_2) = \bias(X_1)\bias(X_2)$ when $X_1$ and $X_2$ are independent Bernoulli variables (and the addition is performed modulo $2$). Therefore, 
\begin{eqnarray*}
\bias(X) &= & \bias\left(\scal{\ev_{\sP}}{\ev_{\textup{aux}}}\right) \bias\left(\scal{\ev_{\sN}}{\wv}\right) \\
& = & \frac{K_{\tbis}^{(s)}(t-u)}{\binom{s}{\tbis} } \frac{K_w^{(n-s)}(u)}{\binom{n-s}{w}}\iftoggle{llncs}{\;\;}{\qquad}\text{(by Fact \ref{fact:krawtchouk_bias})}.
\end{eqnarray*}
This kind of approximation was done in the early days of statistical decoding \cite{J01,O06,DT17a}, until \cite[Prop. 3.1]{CDMT22} which has shown that under certain conditions, \ie when there are enough available parity-check equations of weight $w$ (essentially when the number is of order $\om{1/\delta^2}$ where $\delta$ is the bias), then this approximation can indeed be shown to hold with overwhelming probability. \iftoggle{llncs}{}{It turns out that} \cite[Prop. 3.1]{CDMT22} can be adapted to our setting with some additional technicalities and conditions. It can be shown that with overwhelming probability we indeed have
$$
\bias_{\left(\hv,\cvbis\right) \drawn \widetilde{\sH}} \left(\langle \ev_{\sP}, \vec{h}_{\sP} + \vec{c}_{\textup{aux}}  \rangle + \langle \ev_{\sN}, \hv_{\sN} \rangle \right) = (1+o(1)) \bias(X).
$$
This is in essence what the following proposition shows.
\begin{restatable}{proposition}{propbiasCodedRLPN}\label{prop:biasCodedRLPN} 
	Suppose that the parameters are such that for some constant $\alpha > 0$
	\begin{equation}\label{eq:cstPropo35}
		\frac{\binom{n-s}{w} \binom{s}{\tbis}}{2^{k-\kbis}} = \om{\frac{n^\alpha}{\delta^2}} \quad \mbox{ where } \quad  \delta \eqdef \frac{K_w^{(n-s)} \left(u\right)K_{\tbis}^{(s)}(t-u) }{ \binom{n-s}{w} \binom{s}{\tbis}}. 
	\end{equation}
Moreover suppose that
	\begin{equation}\label{eq:constraints} 
		\frac{\binom{n-s}{w} \binom{s}{\tbis}}{2^{k}} = \OO{n^{\alpha}} \quad \mbox{and} \quad \frac{\binom{s}{\tbis}}{2^{s-\kbis}} = \OO{n^{\alpha}}.
	\end{equation}
		Let $\sN$ be a set of $n-s$ positions in $\IInt{1}{n}$ and $\sP \eqdef \llbracket 1,n \rrbracket \backslash \sN$. Let $\ev$ be a vector of weight $u$ on $\sN$ and $t-u$ on $\sP$.
		Let $\CC$ and $\CCbis$ be $[n,k]$ and $[s,\kbis]$ linear codes respectively.
Let us choose $(\cvbis,\hv)$ uniformly at random in
	$$ 
	\widetilde{\sH} = \{ \left( \hv, \cvbis \right) \in \CC^{\perp} \times \CCbis\: : \: |\hv_{\sN}| = w \mbox{ and } |\hv_{\sP} + \cvbis| = \tbis  \}.
	$$
	Then for a proportion $1-o(1)$ of codes $\CCbis$ and $\CC$ we have that
	$$ 
	\bias_{\left(\hv,\cvbis\right) \drawn \widetilde{\sH}}\left(  \langle \cvbis + \hv_{\sP},\ev_{\sP} \rangle +  \langle \ev_{\sN},  \hv_{\sN}\rangle \right) = \delta(1+o(1)).
	$$
\end{restatable}
\iftoggle{llncs}{This proposition is proved in \S \ref{app:a} of the appendix.}{
	\begin{proof}
		See Appendix \S\ref{app:a} .
	\end{proof}
}

	 \section{The \nRLPN \ Algorithm} \label{ss:doubleRLPN}
We first going to explain the four main ingredients of the \nRLPNs algorithm:
\begin{itemize}\setlength{\itemsep}{5pt}
\item computing suitable LPN samples,
\item FFT decoding,
\item recovering $\ev_\sP$,
\item the bet ensuring that there are $u$ errors on $\sN$ at some point.
\end{itemize}
Let us detail each of these ingredients (or steps of the algorithm).
\newline

{\bf \noindent Computing the LPN Samples.}
First, our algorithm computes a certain number of LPN samples  by computing a set $\sW$ of elements of $\CC^{\perp}$ of weight $w$ on $\sN$ by using a procedure {\tt ParityCheckEquations}$(w,\sN,\CC)$ that uses low-weight codewords search techniques to produce a bunch of parity-check equations of $\CC$ of weight $w$ on $\sN$. Then a random code $\CCbis$ is chosen in a family of codes over $\F_2^{\card{\sP}}$ and dimension $\kbis$ that we know how to decode efficiently at distance $\tbis$. 
For an element $\hv$ in $\sW$, each $\hv_\sP$ is decoded at distance $\tbis$  to finally compute the set $\sH$ containing pairs $\left(\hv,\cvbis\right)$ in $\sW \times \CCbis$ satisfying  $|\hv_{\sN}|=w$ and $|\hv_{\sP} + \cvbis|=\tbis$. 
Algorithm \ref{alg:sample} gives the pseudo-code of the procedure.
\begin{algorithm}[h!]
\begin{algorithmic}[1]
\caption{The function computing the LPN samples associated to $\sN$ \label{alg:sample}}
\Function{LPN-Samples}{$\CC,\sN$}
\State $\sP \gets \IInt{1}{n} \setminus \sN$
\State $\sW \gets \Call{ParityCheckEquations}{w,\sN,\CC}$\label{state:PCE}
\Comment{returns a set of parity-check equations of $\CC$ of weight $w$ on $\sN$}
\State$\CCbis \drawn \sF(\sP,\kbis,\tbis)$
\Comment{returns a code $\CCbis$ in a family of codes $\sF$ over $\F_2^{\card{\sP}}$ and dimension $\kbis$ that we know how to decode efficiently at distance $\tbis$}
\State $\sH \gets \emptyset$
\ForAll{$\hv \in \sH$}
\State $\sH \gets \sH \cup \{\hv\} \times \Call{Decode}{\hv_\sP,\CCbis,\tbis}$\label{state:computeH}
\Comment{$\Call{Decode}{\hv_\sP,\CCbis,\tbis}$ outputs a set of codewords of $\CCbis$ at distance $\tbis$ of $\hv_\sP$} \label{state:Dec} 
\EndFor 
\State{$\Gmbis \gets$ generating matrix of $\CCbis$}\\
\Return{$\left( \sH, \Gmbis \right)$}
\EndFunction
\end{algorithmic}
\end{algorithm}

{\noindent \bf FFT Decoding.}
Computing $\sH$ gives a number $\card{\sH}$ of \LPNs samples, which from the interpretation given in Equation \eqref{eq:LPN}, leads us to think that 
the right choice $\xv \in \F_2^\kbis$ for $\ev_{\sP}\Gmbis[\transpose]$ is the one for which \iftoggle{llncs}{$\bias_{(\hv,\mvbis \Gmbis) \drawn \sH} \left( \scp{\yv}{\hv} + \scp{\xv}{\mvbis}\right)$}{
$$
\bias_{(\hv,\mvbis \Gmbis) \drawn \sH} \left( \scp{\yv}{\hv} + \scp{\xv}{\mvbis}\right)
$$}would be given by  Proposition \ref{prop:biasCodedRLPN}. It should namely be of order $\delta$ which is defined in this proposition. Natural candidates for being equal to $\ev_\sP\Gmbis[\transpose]$ are those for which this bias is say $\geq \delta/2$. This leads to compute all those biases. This can be done rather efficiently by factoring the common 
computations made for computing all those biases for $\xv \in \F_2^\kbis$ by an FFT trick which is standard  in the LPN context. It dates back in this context to \cite{LF06}, but it can be traced back to decoding the first-order Reed-Muller code (which is another way to view the decoding task
in case of the LPN problem) which was already suggested in \cite{G66}. 
The link between the bias of the random variables we are interested in and the Fourier transform is based on the following simple observation that follows right away from the very definition of the Fourier transform. Before we give this observation, let us bring in a notation that will be helpful for describing it and which will be used throughout the paper from now on.
\begin{notation}
For any $\yv \in \F_2^n$, $\sH \subseteq \CC^\perp \times \CCbis$ and a generator matrix $\Gmbis$ of $\CCbis$ we define the function
$f_{\yv,\sH,\Gmbis}$ on $\F_2^\kbis$ by
\iftoggle{llncs}{
	\begin{eqnarray} \label{eq:def_f_H}
		f_{\yv,\sH,\Gmbis[]}: \F_2^{\kbis} & \rightarrow & \RR \nonumber \\
		\xv & \mapsto & \sum_{\hv : \left(\hv, \xv \Gmbis \right) \in \sH} (-1)^{\langle \yv, \hv\rangle} \;\; \text{if this sum is not empty,}\qquad \label{eq:def_f}\\
                \xv & \mapsto & \qquad\qquad\qquad0 \qquad\qquad\text{otherwise.} \nonumber
	\end{eqnarray}
}{
	\begin{eqnarray} \label{eq:def_f_H}
	f_{\yv,\sH,\Gmbis[]}: \F_2^{\kbis} & \rightarrow & \RR \nonumber \\
	\vec{u} & \mapsto & \left\{\renewcommand{\arraystretch}{1.2}\begin{array}{cl} 
	\sum_{\hv : \left(\hv, \vec{u} \Gmbis \right) \in \sH} (-1)^{\langle \yv, \hv\rangle}\label{eq:def_f} & \text{if this sum is not empty,} \\
	0 & \text{otherwise.}
	\end{array}\right. 
\end{eqnarray}
}

\end{notation}

With this notation at hand, the link between the biases and the Fourier transform of this function is given by the following lemma\iftoggle{llncs}{ proved in Appendix \ref{app:proofLemmaFund}.}{.} 
	\begin{restatable}{lemma}{lemFundamental}\label{lem:fundamental}
We have for any $\uv \in \F_2^\kbis$ and any $\xv \in \F_2^s$ such that $\xv \Gmbis[\transpose] = \uv$
		\begin{eqnarray*}
			\widehat{f_{\yv,\sH,\Gmbis}} \left(\uv\right)
&= & \card{\sH}\; \bias_{(\hv,\mvbis \Gmbis)\drawn \sH}(\langle \yv, \hv \rangle  + \langle \uv, \mvbis \rangle)\\
 & = & \card{\sH} \; \bias_{(\hv,\cvbis)\drawn \sH}(\langle \yv, \hv \rangle  + \langle \xv, \cvbis \rangle).
			\end{eqnarray*}
	\end{restatable}
	\iftoggle{llncs}{}{
		\begin{proof} We have the following computation,
		\begin{align*}
			\widehat{f_{\yv,\sH,\Gmbis[]}} \left(\uv\right) &= \sum_{\vv \in \F_2^\kbis} (-1)^{\scp{\uv}{\vv}} f_{\yv,\sH,\Gmbis[]}(\vv) \\
			&= \sum_{(\hv,\vv)\in \CC^\perp\times \F_2^\kbis : \left(\hv,\vv \Gmbis \right) \in \sH} (-1)^{\langle \yv, \hv \rangle + \langle \uv, \vv \rangle} && \mbox{(Equations \eqref{eq:def_fourier} and \eqref{eq:def_f})}\\
			&=  \sum_{(\hv,\vv)\in \CC^\perp\times \F_2^\kbis : \left(\hv,\vv \Gmbis \right) \in  \sH} (-1)^{\langle \yv, \hv \rangle + \langle \xv \Gmbis[\transpose], \vv \rangle} \\
			&= \sum_{(\hv,\vv)\in \CC^\perp\times \F_2^\kbis:\left(\hv,\vv \Gmbis \right) \in  \sH} (-1)^{\langle \yv, \hv \rangle + \langle \xv, \vv \Gmbis \rangle} \\
			&=  \sum_{\left(\hv,\cvbis \right) \in  \sH} (-1)^{\langle \yv, \hv \rangle + \langle \xv, \cvbis \rangle}
		\end{align*}	
		which concludes the proof by definition of the bias. 
	\end{proof}
	}

\begin{remark}
The probabilistic notation hides the fact that computing all these Fourier coefficients and taking the maximum of them allows to decode in a certain code.
Indeed let,
\iftoggle{llncs}{$
\DD \eqdef \left\{ (\scp{\uv}{\mvbis})_{(\hv,\mvbis\Gmbis) \in \sH}: \uv \in \F_2^\kbis \right\}
$}
{$$
	\DD \eqdef \left\{ (\scp{\uv}{\mvbis})_{(\hv,\mvbis\Gmbis) \in \sH}: \uv \in \F_2^\kbis \right\}
$$}
which is under very mild assumptions a linear code of dimension $\kbis$ and length $\card{\sH}$. If we let $c(\uv) \eqdef (\scp{\uv}{\mvbis})_{(\hv,\mvbis\Gmbis) \in \sH}$ be the codeword associated to $\uv$ and $\vv = (\scp{\yv}{\hv})_{(\hv,\mvbis\Gmbis) \in \sH}$ then since 
$$
\card{\sH}\; \bias_{(\hv,\mvbis \Gmbis)\drawn \sH}(\langle \yv, \hv \rangle  + \langle \uv, \mvbis \rangle) = 
\card{\sH} - 2 |\vv + c(\uv)|,
$$
it follows from Lemma \ref{lem:fundamental} that $c(\uv_0)$ is the codeword of $\DD$ which is the closest to $\vv$, where $\uv_0 = \arg\max \widehat{f_{\yv,\sH,\Gmbis}} \left(\uv\right)$. Therefore, vector $\uv_0$ is here a likely candidate for being equal to $\ev_\sP\Gmbis[\transpose]$ when $\sH$ is big enough.
\end{remark}

We give the pseudo-code of the FFT decoding algorithm producing a list $\sS$ of putative candidates for being equal to $\ev_\sP \Gmbis[\transpose]$ 
in Algorithm \ref{algo:FFT}.

\begin{algorithm}[h!]
\caption{FFT algorithm  producing a list of candidates for $\ev_\sP \Gmbis[\transpose]$\label{algo:FFT} }
\begin{flushleft}
\textbf{Input:} $\sH$, $\Gmbis$\\
\textbf{Output:} $\sS$ a list of candidates for $\ev_\sP \Gmbis[\transpose]$
\end{flushleft}
\begin{algorithmic}[1]
\Function{FFT-Decode}{$\sH,\;\Gmbis$}
\State $\widehat{f_{\yv,\sH,\Gmbis }} \gets$\Call{FFT}{$f_{\yv,\sH,\Gmbis}$}
\State $\sS \gets  \left\{ \uv \in \F_2^{\kbis} \: : \widehat{f_{\yv,\sH,\Gmbis }}(\uv) > \frac{\delta}{2} \card{\sH} \right\} $ \label{lst:line:set_candidates} \iftoggle{llncs}{
	\Comment{{\scriptsize $\delta \eqdef \frac{K_{w}^{(n-s)}(u) K_{\tbis}^{(s)}(t-u)}{\binom{n-s}{w} \binom{s}{\tbis}}$}}}{\Comment{ $\delta \eqdef \frac{K_{w}^{(n-s)}(u) K_{\tbis}^{(s)}(t-u)}{\binom{n-s}{w} \binom{s}{\tbis}}$}}
\State \Return $\sS$
\EndFunction
\end{algorithmic}
\end{algorithm}
The point of using the FFT for computing all these biases is that its complexity is of order $\OO{\kbis 2^\kbis F}$ where $F$ is the complexity of computing $f_{\yv,\sH,\Gmbis}$ which can be bounded by $\OO{\max\left(1,\frac{\card{\sH}}{2^\kbis}\right)}$. On the other hand, if we had computed directly all those biases we would have a much bigger complexity of $\OO{\card{\sH}2^\kbis}$  because $\sH$ is of exponential size for the problem at hand.
\medskip

{\noindent \bf Recovering $\ev_\sP$ and then $\ev$.}
If we have a candidate ${\sv}$ for $\ev_\sP\Gmbis[\transpose]$, then since we expect $|\ev_\sP|=t-u$, recovering $\ev_\sP$ from the equality ${\sv} = \ev_\sP \Gmbis[\transpose]$ is nothing but solving a decoding problem, namely to decode $t-u$ errors in the code of parity-check matrix $\Gmbis$, \ie $\CC_{\textup{aux}}^\perp$. In other words, we have to solve $\textsf{DP}(s,s-\kbis,t-u)$. 
This approach can be generalized by taking  $\Nbis$ different sets of \LPNs samples associated respectively to the codes $\CCbis[(1)], \cdots, \CCbis[(\Nbis)]$. For $i$ in $\IInt{1}{\Nbis}$, let $\Gmbis[(i)]$ be the generating matrix which is chosen for $\CCbis[(i)]$. Then each of these sets of \LPNs samples brings candidates for $\ev_\sP\Gmbis[\transpose]$. By choosing an $\Nbis$-tuple of candidates $(\sv^{(1)},\cdots,\sv^{(\Nbis)})$, where $\sv^{(i)}$ is a candidate for $\Gmbis[(i)] \transp{\ev_\sP}$ (we have taken the transpose to have a more readable form) given by the $i$-th \LPNs samples set, we get to solve the 
set of simultaneous equations 
$$
\transp{(\sv^{(1)})} = \Gmbis[(1)]\transp{\ev_\sP},\cdots, \; \transp{(\sv^{(\Nbis)})} = \Gmbis[(\Nbis)] \transp{\ev_\sP}
$$
with the constraint $|\ev_\sP|=t-u$. In other words if we set 
\iftoggle{llncs}{
\begin{align*}
\transp{\Hm} \eqdef \begin{pmatrix}
\Gmbis[(1)]^{\transpose} & \cdots &
\Gmbis[(\Naux)]^{\transpose}
\end{pmatrix} &\quad &
\sv \eqdef \begin{pmatrix} \sv^{(1)} & \hdots & \sv^{(\Nbis)}\end{pmatrix}
\end{align*}
}
{
$$
\transp{\Hm} \eqdef \begin{pmatrix}
	\Gmbis[(1)]^{\transpose} & \cdots &
	\Gmbis[(\Naux)]^{\transpose}
\end{pmatrix} \quad \mbox{and} \quad 
\sv \eqdef \begin{pmatrix} \sv^{(1)} & \hdots & \sv^{(\Nbis)}\end{pmatrix}
$$
}
then we have to solve the decoding problem $\Hm \transp{\ev_\sP} = \transp{\sv}$ with $|\ev_\sP|=t-u$, \iftoggle{llncs}{that is}{in other words we have to solve}
$\textsf{DP}\left(s,s-\Nbis \kbis,t-u\right)$.  We are going to choose a simple ISD algorithm to solve this problem, namely 
Dumer's algorithm \cite{D89} which is a good compromise between efficiency and simple formula for its complexity. We denote by 
\Call{Decode-Dumer}{$\Hm,\sv,t$} the call to Dumer's algorithm to decode the syndrome $\sv$ of an error of weight $t$ associated to the 
parity-check matrix $\Hm$. We assume here that this call produces {\em all} solutions to this decoding problem.

Once we have recovered $\ev_\sP$, say we know that it is equal to some $\vv$ of weight $t-u$ in $\F_2^s$, we face a much simpler problem. We namely  have to solve the problem
$$\yv = \cv + \ev,\; \cv \in \CC,\; \ev_\sP = \vv,\;|\ev_\sN|=u.$$
This is nothing but $\textsf{DP}(n-s,k-s,u)$ which is much simpler. Here we might just use \iftoggle{llncs}{$\Call{Decode-Dumer}{}$}{algorithm ~$\Call{Decode-Dumer}{}$} on it. 
Let us call $\Call{Solve-SubProblem}{\CC,\sN,\yv,\vv,u}$ the routine which performs this task and which returns a candidate for $\ev_\sN$ 
and returns $\bot$ otherwise.
If this problem has no solution we have of course a false candidate for $\ev_\sP$ and if we have a solution, then we have solved our decoding problem. To verify that we have indeed such a decoding problem,  suppose without loss of generality that $\sP = \llbracket 1,s \rrbracket$\iftoggle{llncs}{,}{ and} $\sN = \llbracket s+1,n \rrbracket$. We can also assume that $\CC_{\sP}$ is of full rank dimension $s$ (this holds with overwhelming probability).  We can compute $\Gm$ a generator matrix of $\CC$ of the form $\Gm =  \begin{pmatrix}
	\vec{Id}_{s} & \Rm \\
	\vec{0}_{k-s} & \Rm' 
\end{pmatrix}$
by applying partial Gaussian elimination on a generator matrix of $\CC$.
Then $\Call{Solve-SubProblem}{\CC,\sN,\yv,\xv,u}$ decodes at distance $u$ the word $\yv' \eqdef \yv_{\sN} - (\yv_{\sP} - \xv) \Rm$ onto the code $\CC^{\sN}$ of generator matrix $\Rm'$.

 With this notation at hand, the pseudo-code describing the algorithm for recovering $\ev_\sP$ and then returning $\ev$ if a suitable solution is found, is given in Algorithm \ref{algo:recovering_ep}.
 \medskip

\begin{algorithm}[h!]
\caption{algorithm recovering $\ev_\sP$ and then $\ev$ \label{algo:recovering_ep}} 
\begin{flushleft}
\textbf{Input:} $\sS^{(1)},\cdots,\sS^{(\Nbis)}\subset \F_2^{\kbis}$, $\Gmbis[(1)],\cdots,\Gmbis[(\Nbis)]\in \F_2^{\kbis \times s}$
\end{flushleft}
\begin{algorithmic}[1]
\Function{Recover-$\ev$}{$\sS^{(1)},\cdots,\sS^{(\Nbis)},\Gmbis[(1)],\cdots,\Gmbis[(\Nbis)]$}
\State $\transp{\Hm} \gets \begin{pmatrix}
	\Gmbis[(1)]^{\transpose} & \cdots &
	\Gmbis[(\Naux)]^{\transpose}
\end{pmatrix}$
\For{$ \left(\sv^{(1)}, ..., \sv^{(N_{\textup{aux}})} \right)\in  \prod_{j=1}^{N_{\textup{aux}}}\mathcal{S}^{(j)}$}
		\State $\sv \gets \begin{pmatrix} \sv^{(1)} & \hdots & \sv^{(\Nbis)}\end{pmatrix}$
        \ForAll{$\vv \in \Call{Decode-Dumer}{\Hm,\sv,t}$}
		\State $\ev' \gets \Call{Solve-SubProblem}{\CC,\sN,\yv,\vv,u}$
		\If{$\ev'  \neq \bot$} \label{lst:line:redo1}
		\State \Return $\ev$ such that $\ev_{\sP} = \xv$ and $\ev_{\sN} = \ev'$\label{lst:line:redo2}
		\EndIf \label{lst:line:redo3}
		\EndFor
		\EndFor
	
\EndFunction
\end{algorithmic}
\end{algorithm}

{\bf \noindent Testing Enough Candidates $\sN$.}
Now, it may also happen that when choosing  $\sN$, we might not have that $|\ev_\sN|=u$. For this, we have to check enough candidates. The probability that a set $\sN$ of size $n-s$ satisfies this property is given by 
\iftoggle{llncs}{$\Psucc \eqdef \frac{\binom{t}{u}\binom{n-t}{n-u-s}}{\binom{n}{n-s}}$.}{$$\Psucc \eqdef \frac{\binom{t}{u}\binom{n-t}{n-u-s}}{\binom{n}{n-s}}.$$} Performing a number $\Niter$ of trials for $\sN$ which is of order $\Th{1/\Psucc}$ will succeed with constant probability.
Putting all these ingredients together leads to the whole \nRLPNs algorithm given in Algorithm \ref{alg:codedRLPN}.
\medskip

\iftoggle{llncs}{}{\mbox{}\medskip}
\begin{algorithm}
	\caption{\nRLPN \ decoder \label{alg:codedRLPN}}
	\begin{flushleft}
\textbf{Input:} $\yv$, $t$, $\CC$ an $\lbrack n,k \rbrack$-code \\ 
\textbf{Parameters:} $s, u,\kbis,\tbis, N, \left( \CCbis[(j)] \right)_{j \in \mathcal{F}}$ \\ 
\textbf{Output:} $\ev$ such that $|\ev|=t$ and $\yv - \ev \in \CC$.
	\end{flushleft}
	\begin{algorithmic}[1]
		\Function{\nRLPN}{$\yv$, $\CC$, $t$}
\For{$i$ from $1$ to $\Niter$}  \label{lst:line:while}
		\Comment{$\Niter$ such that w.o.p one iteration is s.t $|\ev_{\sN}|=u$}
		\State $\sN \stackrel{\$}{\gets} \left\{\sI \subseteq \llbracket 1, n \rrbracket \ : \ \card{\sI} = n-s\right\}$
\Comment{Hope that $|\ev_{\sN}|=u$}
\For{$j = 1,\dots,N_{\textup{aux}}$}  \label{lst:line:while_auxcode} 
                \State $(\sH^{(j)},\Gmbis[(j)]) \gets \Call{LPN-Samples}{\CC,\sN} $
                \State $\cS^{(j)} \gets \Call{FFT-Decode}{\sH^{(j)},\Gmbis[(j)]} $
		\EndFor
                \State \Call{Recover-$\ev$}{$\sS^{(1)},\cdots,\sS^{(\Nbis)},\Gmbis[(1)],\cdots,\Gmbis[(\Nbis)]$}
		\EndFor
		\EndFunction
	\end{algorithmic}
\end{algorithm}

{\noindent \bf Complexity of the Algorithm.}
It is sufficient to take $\Nbis=\OO{1}$ for the parameters we are interested in which will correspond to a choice of $\kbis$ of the form
$\kbis = \Om{s}$. With this choice we immediately get the following complexity for the \nRLPNs algorithm
\begin{proposition}{}\label{prop:complexity}
The complexity $C$ of the \nRLPNs algorithm is given by
\begin{eqnarray*}
C & = & \OOt{\frac{1}{\Psucc}\left(\Teq + N \Tdec + \kbis \max\left( 2^\kbis,\card{\sH}\right)+S^{\Nbis} \Tisd\right)}\;\;\text{where}\\
\Psucc & = & \frac{\binom{t}{u}\binom{n-t}{n-u-s}}{\binom{n}{n-s}}\\
\Tisd & = & \Tdumer(s,s-\Nbis \kbis,t-u)+\Nisd \Tdumer(n-s,k-s,u) 
\end{eqnarray*}
and $\Teq$ is the time complexity of $\Call{ParityCheckEquations}{}$, 
$N$ is the number of parity-check equations produced by this procedure, $\Tdec$ is the complexity of 
decoding $\CCbis$, \ie it is the complexity of a call to $\Call{Decode}$, $S$ is the size of a list output by
$\Call{FFT-Decode}$, $\Nisd$ is the number of solutions to the $(s,s-\Nbis \kbis,t-u)$ decoding problem and $\Tdumer(n,k,t)$ stands for the complexity of solving the $(n,k,t)$ decoding problem with Dumer's algorithm when we want to find all solutions to the problem.
\end{proposition}
The asymptotic complexity formula for the \nRLPNs algorithm, including also the constraints required on our parameters, is given in Appendix \ref{app:complexity} Proposition \ref{prop:asym_comp_doubleRLPN} which was used to generate Figure \ref{fig:compISDDoubleRLPNRLPN}. %
 
\section{Estimating the Number of False Candidates}\label{sec:bias}

		The goal of this section is to introduce the main tool necessary to make a rigorous analysis of Algorithm \ref{alg:codedRLPN} and to 
give a formula for the number of false candidates which is proved by making a certain conjecture whose validity has then been verified experimentally.

	\subsection{Main Duality Tool}
	 The fundamental quantity when analyzing dual attacks is the bias of $\langle \yv, \hv \rangle  + \langle \xv, \cvbis \rangle$ which tells us whether $\xv \Gmbis[\transpose]$  has to be put in the list $\sS$ of candidates output by Algorithm \ref{algo:FFT}. While initially standard independence assumptions were made to analyze its distribution \cite[Ass. 3.7]{CDMT22} (which are very similar to analyze dual attacks in lattice based cryptography), recently \cite{MT23}  showed that these assumptions were erroneous and, gave for the first time a dual expression \cite[Prop 1.]{MT23} for this quantity which seems a key step to understand  its behavior and gave with an additional assumption a rigorous analysis of the \RLPNs dual attack. The proposition given there to estimate the number of false candidates turns out to match accurately the experiments.
  The following proposition is a generalization of \cite[Prop 1.]{MT23} and gives a dual expression for the aforementioned bias. 
		\begin{restatable}{proposition}{propDualityBias} \label{prop:exp_bias}
			Let $\sP$ and $\sN$ be two complementary subsets of $\llbracket 1 , n \rrbracket$ of size $s$ and $n-s$ respectively. Let $\CC$ be an~$[n,k]$-code such that $\CC_{\sP}$ is of dimension $s$ and let $\CCbis$ be an~$[s,\kbis]$-code. We have for any $\xv \in \F_2^s$
		\begin{equation} \label{eq:prop_bias_doubleRLPN}
			\bias_{\left(\hv,\cvbis \right) \drawn \widetilde{\sH}} \left(\langle \yv, \hv \rangle  + \langle \xv, \cvbis \rangle\right) = \frac{1}{2^{k - \kbis}}\: \frac{1}{\card{\widetilde{\sH}}} \sum_{i = 0}^{n-s} \sum_{j = 0}^{s}  N_{i,j} K_w^{(n-s)} \left(i\right) K_{\tbis}^{(s)} \left(j\right)  
		\end{equation}
		where
		\begin{eqnarray*}
			N_{i,j} &\eqdef &\card{ \left\{  \left(\rv,\cv^{\sN}\right) \in  (\xv +\CCbisperp)  \times \CC^{\sN} \: : \left| \rv \right|=j \mbox{ and } \left| \left( \rv + \ev_{\sP}\right) \Rm + \ev_{\sN} + \cv^{\sN} \right| = i\right\} },\\
\widetilde{\sH} & = &\{ \left( \hv, \cvbis \right) \in \CC^{\perp} \times \CCbis\: : \: |\hv_{\sN}| = w \mbox{ and } |\hv_{\sP} + \cvbis| = \tbis  \}
		\end{eqnarray*}
			and where $\Rm \in \F_2^{s \times (n-s)}$ is such that for any $\vec{h} \in \CC^{\perp}$ we have $\hv_{\sP}= \hv_{\sN} \Rm^{\transpose}$.
\end{restatable}
	\begin{proof}
		This proposition is proved in Appendix \ref{app:proofBias}.
	\end{proof}
\subsection{Intuition on How this Formula Allows to Estimate $\card{\sS}$}
\label{ss:intuition}
As a preliminary remark, notice that $\bias_{\left(\hv,\cvbis \right) \drawn \widetilde{\sH}} \left(\langle \yv, \hv \rangle  + \langle \xv, \cvbis \rangle\right)$ is the same for all $\xv$ belonging to a same coset of $\CCbis[\perp]$ and therefore only possibly allows to distinguish the values $\xv \Gmbis[\transpose]$. Second, observe that the expected value of $\card{\widetilde{\sH}}$ is $\frac{\binom{n-s}{w}}{2^{k-s}}\frac{\binom{s}{\tbis}}{2^{s-\kbis}}$ so that we expect
\iftoggle{llncs}{$
\frac{1}{2^{k-\kbis} \; \card{\widetilde{\sH}}} \approx \frac{1}{\binom{n-s}{w}\binom{s}{\tbis}}.
$}
{$$
\frac{1}{2^{k-\kbis} \; \card{\widetilde{\sH}}} \approx \frac{1}{\binom{n-s}{w}\binom{s}{\tbis}}.
$$}
Third, observe that Proposition \ref{prop:biasCodedRLPN} means in essence that the bias corresponding to $\ev_\sP$, namely 
$\bias_{\left(\hv,\cvbis \right) \drawn \widetilde{\sH}} \left(\langle \yv, \hv \rangle  + \langle \ev_\sP, \cvbis \rangle\right)$
should be $\approx \frac{K_w^{(n-s)} \left(u\right) K_{\tbis}^{(s)}(t-u)}{\binom{n-s}{w}\binom{s}{\tbis}}$, \ie it corresponds roughly to the ``first'' pair $(i,j)$ 
(where we range the values according to the product $K_w^{(n-s)} \left(i\right) K_{\tbis}^{(s)} \left(j\right) $) for which $N_{i,j} \neq 0$, namely 
$(i,j)=(u,t-u)$ where the pair $(\ev_\sP,\vec{0})$ is likely to be the only pair $(\rv,\cv^\sN)$ in $(\ev_\sP +\CCbisperp)  \times \CC^{\sN}$ such that
$\left| \rv \right|=t-u$ and $\left| \left( \rv + \ev_{\sP}\right) \Rm + \ev_{\sN} + \cv^{\sN} \right| = u$ (and therefore we likely have
$N_{u,t-u}=1$). Therefore the behavior of the sum appearing in \eqref{eq:prop_bias_doubleRLPN} is dominated by this first 
term $N_{i,j}$ which is non zero, namely $(i,j)=(u,t-u)$ since we really have in this case that the corresponding term in the 
sum is nothing but \iftoggle{llncs}{$\frac{K_w^{(n-s)} \left(u\right) K_{\tbis}^{(s)}(t-u)}{\binom{n-s}{w}\binom{s}{\tbis}}$.}{$$
\frac{K_w^{(n-s)} \left(u\right) K_{\tbis}^{(s)}(t-u)}{\binom{n-s}{w}\binom{s}{\tbis}}.
$$}

This kind of phenomenon appears to be much more general than this: the $\xv \in \F_2^s$ which give a high bias (and are therefore the ones we put in $\sS$) are those for which there is an $N_{i,j}$ which is unexpectedly non zero (and therefore most likely equal to $1$) 
in the low values of $(i,j)$ for which the term $ K_w^{n-s} \left(i\right) K_{\tbis}^{(s)} \left(j\right)$ can compete or even supersede the term
dominating in the expression~\eqref{eq:prop_bias_doubleRLPN} of the bias of $\ev_{\sP}$, namely $K_w^{(n-s)} \left(u\right) K_{\tbis}^{(s)}(t-u)$
(we have ignored the common denominator $2^{k-\kbis}\card{\widetilde{\sH}}$ appearing in both sums). Similarly we expect that the bias of those
$\xv$ is of order in this case
$$
\frac{K_w^{(n-s)} \left(i\right) K_{\tbis}^{(s)}(j)}{\binom{n-s}{w}\binom{s}{\tbis}}.
$$
This intuition is formalized by Conjecture \ref{ass:bias_dom} what we make later on.

	\subsection{Main Proposition}
	\label{sec:tools}

	The key step of the analysis is to estimate the number of candidates, namely the size of $\sS$ (Instruction \ref{lst:line:set_candidates} of Algorithm \ref{alg:codedRLPN}).
	 Provided that the bet  ($|\ev_{\sN}| = u$)  on the error is valid we expect that the secret vector $\ev_{\sP} \Gmbis[\transpose]$ belongs to $\sS$. But, as we will show in this section this set also contains some false positives, namely any element of $\sS \setminus \{ \ev_{\sP} \Gmbis[\transpose]\}$. Testing if an element of $\sS$ is a false positive (Algorithm \ref{algo:recovering_ep}) will be of exponential cost. Estimating their number is therefore crucial to predict the complexity of our algorithm.
	 The following proposition bounds the expected number of candidates in a typical iteration of Algorithm \ref{alg:codedRLPN}.
	\begin{restatable}{proposition}{propExpSizeS} \label{prop:exp_sizeS}
		Using Distribution \ref{notation:randomCodes} for $\CC,\CCbis,\ev$ and $\yv$ and given that our parameters verify Parameter Constraint \ref{ass:parameters_doubleRLPN}, that the number of computed \LPNs samples is the total number of available \LPNs samples, \ie{}
$\sH = \widetilde{\sH}$ and under Conjecture \ref{ass:bias_dom} we have that the expected number of candidates per iteration is bounded by
		\begin{equation}
			\mathbb{E }_{\CC,\CCbis}\left(\card{\sS }\right) = \OOt{\max_{(i,j) \in \mathcal{A}} \frac{\binom{s}{j} \binom{n-s}{i} }{2^{n-k}}} + 1
		\end{equation}
		where 
			\begin{equation} \label{def:setijA}
				\mathcal{A} \eqdef \left\{ \left(i,j\right) \in \llbracket 0, n-s \rrbracket \times \llbracket0,s \rrbracket, \iftoggle{llncs}{}{\;} \left| \frac{ K_w^{(n-s)}\left(u\right) K_{\tbis}^{(s)}\left(t-u\right) }{ K_w^{(n-s)}\left(i\right) K_{\tbis}^{(s)}\left(j\right)}\right| \leq n^{3.2}\right\}.
			\end{equation}
			The set $\sS$ of candidates is defined by
			\begin{equation} \label{eq:prop_def_S}
				\sS \eqdef \left\{\sv \in \F_2^\kbis \: : \widehat{f_{\yv,\widetilde{\sH},\Gmbis}}\left(\sv\right) \geq  \frac{\delta}{2} \: \widetilde{H}\right\},
			\end{equation}
			where 
			\begin{equation}
				\widetilde{H} \eqdef \frac{\binom{n-s}{w} \binom{s}{\tbis}}{2^{k-\kbis}} \qquad \mbox{and} \qquad \delta \eqdef \frac{K_w^{(n-s)}\left(u\right) K_{\tbis}^{(s)}\left(t-u\right) }{\binom{n-s}{w} \binom{s}{\tbis}}.
			\end{equation}
\end{restatable}
	\begin{remark}
		The additional constraint that $\sH=\widetilde{\sH}$ is only here to simplify the proof. One could make a similar proposition without this constraint. In our instantiation of Algorithm \ref{alg:codedRLPN} and with our optimal parameters this constraint is de-facto verified. Note also that $\widetilde{H}$ appearing in the expression of the threshold is  the expected number of available \LPNs samples, namely $\expect{\CC,\CCbis}{\left|\widetilde{\sH}\right|}$.
	\end{remark}
\begin{distribution} \label{notation:randomCodes}\mbox{ }
	\begin{itemize}\setlength{\itemsep}{5pt}
		\item $\sP$ and $\sN$ are two fixed complementary subsets of $\llbracket 1 ,n \rrbracket$ of size $s$ and $n-s$ respectively.
		\item The code $\CC$ of generator matrix $\Gm$ is chosen uniformly at random among $[n,k]$ linear codes which are such that $\CC_{\sP}$ is of dimension $s$.
		\item The code $\CCbis$ of generator matrix $\Gmbis\in \F_2^{\kbis \times s}$  is chosen uniformly at random among the~$[s,\kbis]$-codes.
		\item $\ev \in \F_2^n$ is a fixed vector of $\mathcal{S}_{t}^n$, $\cv \in \CC$ is a random codeword of $\CC$ and we define $\yv \eqdef \cv + \ev$.
	\end{itemize}
\end{distribution} 
Correctness of our algorithm is ensured by the following constraints.
\begin{constraint} \label{ass:parameters_doubleRLPN}
	We suppose that the parameters $n, k , t, s, \kbis, \tbis, w, u$ are such that there exists a constant $\alpha > 0$ that is such that
	\begin{equation} \label{constraints:constraint}
(i)\;\;		\frac{\binom{n-s}{w} \binom{s}{\tbis}}{2^{k-\kbis}} = \om{\frac{n^{\alpha + 8}}{\delta^2}},\;(ii)\;\frac{\binom{n-s}{w} \binom{s}{t_{\textup{aux}}}}{2^{k}} = \OO{n^{\alpha}},\;(iii)\;\;\frac{\binom{s}{t_{\textup{aux}}}}{2^{s-\kbis}} = \OO{n^{\alpha}}.
	\end{equation}
where, 
$$		
\delta \eqdef \frac{K_w^{(n-s)}\left(u\right) K_{\tbis}^{(s)}\left(t-u\right) }{\binom{n-s}{w} \binom{s}{\tbis}}.
$$
\end{constraint}
\begin{remark}
	Note that these constraints are in fact, up to a polynomial factor, the minimal constraints required for our algorithm to work. Indeed, there are precisely the constraints required in Proposition \ref{prop:biasCodedRLPN} which estimates the bias of the error of the \LPNs samples \eqref{eq:LPN}.
\end{remark}
The difficulty of proving Proposition \ref{prop:exp_sizeS} is similar to the difficulties encountered in analyzing the \RLPNs algorithm in \cite{MT23}, we know too little about the tails of the distribution of $N_{i,j}$. As such, we will make the following conjecture which formalizes the discussion in \S \ref{ss:intuition}.
\begin{restatable}{conjecture}{conjectureBiasDom} \label{ass:bias_dom}
	Using Distribution \ref{notation:randomCodes}, under Parameter Constraint \ref{ass:parameters_doubleRLPN},
\iftoggle{llncs}{
	\begin{multline*}
\text{\footnotesize{\(\prob{}
{ \sum_{j = 0}^{s} \sum_{i = 0}^{n-s} K_{\tbis}^{(s)} \left(j\right) K_w^{(n-s)} \left(i\right) N_{i,j} \geq \frac{1}{2} K_w^{(n-s)}\left(u\right) K_{\tbis}^{(s)}\left(t-u\right)} =
\OOt{ \max_{(i,j) \in \mathcal{A}}\prob{}{N_{i,j} \neq 0}  + 2^{-n}}\)}} 
\end{multline*}
	}
	{
		\begin{multline*}
	\mathbb{P}\left( \sum_{j = 0}^{s} \sum_{i = 0}^{n-s} K_{\tbis}^{(s)} \left(j\right) K_w^{(n-s)} \left(i\right) N_{i,j} \geq \frac{1}{2} K_w^{(n-s)}\left(u\right) K_{\tbis}^{(s)}\left(t-u\right)\right) \\ =
			\OOt{ \max_{(i,j) \in \mathcal{A}}\prob{}{N_{i,j} \neq 0} + 2^{-n} }
		\end{multline*} 
	}
	where $\mathcal{A}$ is given in Equation \eqref{def:setijA} and $\xv$ is taken uniformly at random in $\F_2^{s} \setminus \{ \CCbisperp + \ev_{\sP}\}$.
\end{restatable}
\iftoggle{llncs}{
\begin{remark}
	Conjecture \ref{ass:bias_dom} is discussed in Section \ref{sec:conjecture} where we give experimental evidences that our analysis holds. In Appendix \ref{sec:min_ass} we show that this conjecture is in fact a consequence of a more minimalistic conjecture.
\end{remark}
}{	Conjecture \ref{ass:bias_dom} is discussed in the following section where we give experimental evidences that our analysis holds. In Appendix \ref{sec:min_ass} we show that this conjecture is in fact a consequence of a more minimalistic conjecture.} 

\section{Experimental Evidence for Our Analysis} \label{sec:conjecture}
The goal of this section is  to provide experimental evidence supporting Proposition \ref{prop:exp_sizeS}. 
We will propose a convenient probabilistic model for the $N_{i,j}$'s and show that this model does not change the output distribution of our algorithm. We will essentially use the same model as in \cite[Appendix D]{MT23} and model the weight distribution of the coset of a random linear code as a Poisson distribution of the right expected value. Recall that $N_{i,j}$ can be written as
$$
N_{i,j} = \sum_{u = 0}^{N_j\left(\CCbisperp + \xv\right)} N_i \left(\left(\rv^{(u)} + \ev_{\sP}\right) \Rm + \ev_{\sN} + \CC^{\sN}\right)
$$
where $\Rm \in \F_2^{s \times (n-s)}$ is such that for any $\vec{h} \in \CC^{\perp}$ we have $\hv_{\sP}= \hv_{\sN} \Rm^{\transpose}$, $\rv^{(u)}$ is the $u$'th codeword of weight~$j$ of $\CCbisperp + \xv$ and $N_j(\CCbisperp + \xv)$ counts the number of elements in $\CCbisperp + \xv$ of Hamming weight~$j$. With our model, we first draw $N_j\left(\CCbisperp + \xv\right)$ according to a Poisson distribution of expected value $\frac{\binom{s}{j}}{2^{\kbis}}$ and then we draw each $N_i \left(\left(\rv^{(u)} + \ev_{\sP}\right) \Rm + \ev_{\sN} + \CC^{\sN}\right)$ according to independent Poisson distributions of expected values $\frac{\binom{n-s}{i}}{2^{n-k}}$
(see appendix \ref{app:prop_Nij}, Lemma \ref{lem:Nbar} where we compute $\mathbb{E}\left( N_j\left(\CCbisperp + \xv\right) \right)$ and $\mathbb{E }\left(N_i \left(\left(\rv^{(u)} + \ev_{\sP}\right) \Rm + \ev_{\sN} + \CC^{\sN}\right) \right)$ under Distribution~\ref{notation:randomCodes}).
Finally, we get the following model for $N_{i,j}$ by using the fact that the sum of independent Poisson random variables is a Poisson random variable:
	\begin{model}[Poisson Model] \label{model:poisson}
	Under Distribution \ref{notation:randomCodes} and when $\xv$ is taken uniformly at random in $\F_2^s \setminus\{\CCbisperp + \ev_{\sP}\}$ we make the model that
	$$ N_{i,j} \sim \mathrm{Poisson}\left(N_j \:  \frac{\binom{n-s}{i}}{2^{n-k}}	 \right) \mbox{, where }  N_j \sim \mathrm{Poisson}\left( \frac{\binom{s}{j}}{2^{\kbis}} \right). $$
\end{model}
Under Poisson Model \ref{model:poisson}, the following proposition proves Conjecture \ref{ass:bias_dom} and thus it shows that Proposition \ref{prop:exp_sizeS} holds.
		\begin{proposition} \label{prop:model_conjecture}
			Under the Poisson Model \ref{model:poisson}, Conjecture \ref{ass:bias_dom} holds.
		\end{proposition}
		\begin{proof}
			The proof is given Appendix \ref{app:model_conjecture}.
		\end{proof}

		In Figure \ref{fig:experiment} we computed the expected number of $\xv$'s 
		whose bias multiplied by $\card{\widetilde{H}}$ is bigger than some prescribed quantity $T$ according to 
		\begin{itemize}\setlength{\itemsep}{5pt}
			\item the standard independence model in dual attacks where the $\langle \vec{e},\vec{h}\rangle$'s are supposed to be independent,

			\item some experiments,

			\item the case were we replace 
			the right-hand term of $N_{i,j}$ (given in Equation \eqref{eq:prop_bias_doubleRLPN}) by their Poisson model.
		\end{itemize}
		
As it is shown by Figure  \ref{fig:experiment}, the Poisson model matches remarkably well with the experiments.
This shows, as was the case in the analysis \cite{MT23} of the \RLPNs algorithm, that the Poisson model allows to 
predict accurately the size of $\sS$.
	\begin{figure}[h]
		\centering
\includegraphics[width=\textwidth]{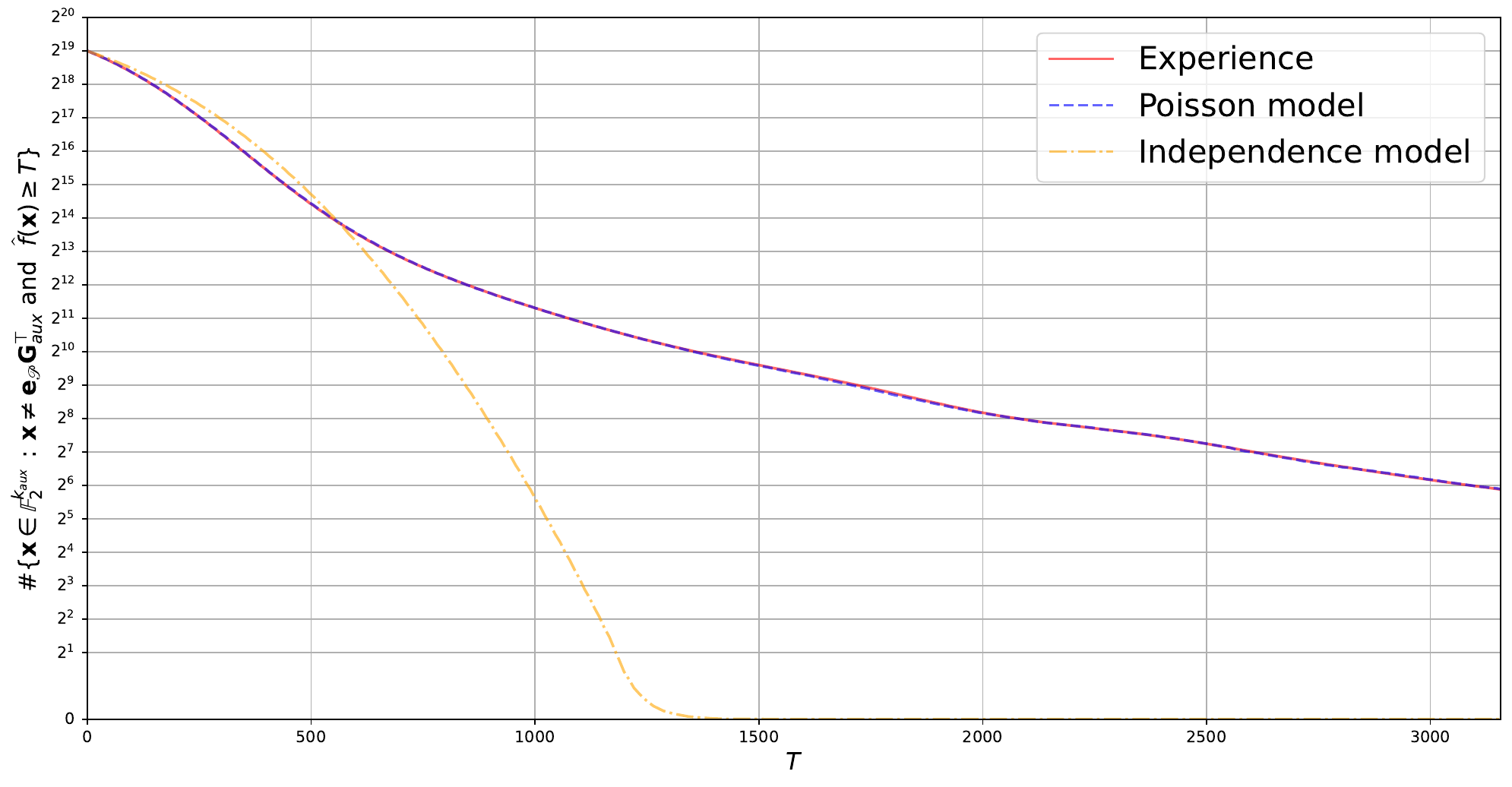}
	\caption{
	 Size of the set 
		$  \{ \xv \in \F_2^{\kbis} \setminus \{\ev_{\sP} \Gmbis \} \; : \widehat{f_{\yv,\sH}}\left(\xv\right) \geq T \} $ as a function of $T$ 
when $[w,\tbis,\kbis,s,k,n,u,t]=[5,2,20,28,30,60,8,8]$, number of \LPNs samples $N = 65536$.
Here the curve ``Independence model'' has been replaced when 
modelling the $\langle \yv, \hv \rangle$'s by i.i.d Bernouilli random variables of parameter $\frac{1}{2}$ (standard independence model in dual attacks).
\label{fig:experiment} }
	\end{figure}

 \section{Instantiating the Auxiliary Code $\CCbis$ with an Efficient Decoder}
\label{sec:SDC}

In \nRLPN\; we need to choose an auxiliary code $\CCbis$ which is efficiently list-decodable (Definition \ref{def:effListDeco}) at the smallest as possible distance $\tbis = \dgv(s,\kbis)$. We propose to use the following product of small random codes (other choices may be more suitable but they are harder to analyze, like polar codes \cite{A09,KU10,S11g,TV12}),
$$
\CCbis \eqdef \CC_1 \times \cdots \times \CC_b \iftoggle{llncs}{\;\mbox{  where the $\CC_i$'s are random  $\left[\frac{s}{b}, \frac{\kbis}{b} \right]$-codes.}}{}
$$
where the $\CC_i$'s are random  $\left[\frac{s}{b}, \frac{\kbis}{b} \right]$-codes. 
Notice that \iftoggle{llncs}{$\decodebis(\zv) \eqdef \left\{ \cv_{\textup{aux}} \in \CCbis \; : \; \left|\cv_{\textup{aux}} + \zv\right| = \tbis  \right\}$}{
$$
\decodebis(\zv) \eqdef \left\{ \cv_{\textup{aux}} \in \CCbis \; : \; \left|\cv_{\textup{aux}} + \zv\right| = \tbis  \right\}
$$} 
does not look exactly like how it should with a random code for which our analysis given in Propositions \ref{prop:biasCodedRLPN} and \ref{prop:exp_sizeS} hold. Furthermore, we will compute
$$
{\sH} \subseteq \left\{ \left( \hv, \cvbis \right) \in \CC^{\perp} \times \CCbis\: : \: \forall i \in \IInt{1}{b}, \; |\hv_{\sN}(i)| = \tfrac{w}{b} \mbox{ and } |\hv_{\sP}(i) + \cv_i| = \tfrac{\tbis}{b}  \right\}
$$
in Instruction \ref{state:computeH} of Algorithm \ref{alg:sample}. To this aim we will perform exhausting search on the random codes.
By choosing the number $b$ of blocs as,
\begin{equation}\label{eq:b} 
b = \Theta(\log n)
\end{equation}
the above decoding algorithm costs for any parity-check equation $O\left( 2^{n/\log n}\right)$ (recall that $\tbis \approx \dgv(n,k)$). Therefore, as Algorithm \ref{alg:sample} running time is exponential (in $n$) for our considered parameters, it won't affect it. Furthermore, there are false candidates when computing $\sH$ and it is crucial to estimate their numbers.

Our analysis of Sections \ref{sec:reductiontoLPN} and \ref{sec:bias} has been made in the idealized-model where $\CCbis$ is a random code equipped with genie aided decoders. But, by choosing $b$ as in Equation \eqref{eq:b}, analysis of Propositions \ref{prop:biasCodedRLPN}  and \ref{prop:exp_sizeS} is still verified with our particular choice of $\CCbis$ (up to negligible factors) as justified in Appendix \ref{app:Caux}.

 \section{Links with Dual Attacks in Lattice Based Cryptography}
\label{sec:lattice}

The purpose of this section is to give more details about the close connection between dual attacks in coding theory ({\it a.k.a}
``statistical decoding'' after the pioneering work of \cite{J01}) and dual attacks in lattice based cryptography. Basically, with some slight differences
highlighted in \cite[App. A]{PS23}, the lattice based analogue of the dual attack presented here is the slight improvement \cite{CST22} of the Matzov attack \cite{M22}.
The improvement in \cite{CST22} is based on the fact that the modulus switching technique used in \cite{M22} can be viewed as a suboptimal source distortion code for the Euclidean metric which can 
be replaced by an almost optimal polar code. The approach followed here should carry over to this lattice setting as well and in particular,
the fundamental duality Proposition \ref{prop:exp_bias}. Let us just observe now that a simple duality equality (together with a gross approximation based on the considerations of 
\S \ref{ss:intuition}) can be used to explain the results observed in \cite[Fig. 3]{DP23}.
It was shown there that predictions of the score function based on standard independence assumptions made for dual attacks in lattice based cryptography seem to be off 
in some parameter region (what can be called
the ``error-floor'' region due to its similarity with the Low-Density-Parity-Check codes literature).
To explain this point, we will use the same notation as in \cite{DP23} and will not redefine the quantities appearing here.

Let us first observe that an immediate corollary of Proposition \ref{prop:exp_bias} is
\begin{corollary}\label{cor:simple}
Consider an $[n,k]$ linear code $\CC$ and consider some word $\yv=\cv+\ev \in \F_2^n$ where $\cv$ is in $\CC$. Let 
$\sW$ be the set of codewords of weight $w$ in $\CC^\perp$ and let $f_\sW(\yv) \eqdef \sum_{\wv \in \sW} (-1)^{\scp{\yv}{\wv}}$.
We have
$$
f_\sW(\yv) = \frac{1}{2^k} \sum_{i=0}^n N_i K_w^n(i)
$$
where $N_i$ is the number of words of weight $i$ in $\CC+\ev$.
\end{corollary} 
It is insightful to view these Krawtchouk polynomials as the 
 Fourier transform of the indicator function of a Hamming sphere, see Fact \ref{fact:krawtchouk_bias}. Similarly, the lattice based analogue of Corollary \ref{cor:simple} 
 involving the lattice based analogue of $f_\sW$, which is called the score function in 
 \cite{DP23} will involve the Bessel function of the first kind (see for instance \cite[Fact 4.9]{DDRT23}). We namely obtain
\begin{restatable}{proposition}{propdualityLattice}\label{propo:dualityLattice}
	Consider a lattice $\Lambda\subseteq \mathbb{R}^{n}$ and consider some word $\vec{y} = \vec{x} + \vec{e} \in \mathbb{R}^{n}$ where $\vec{x}$ is in $\Lambda$. Let $\widetilde{\sW}$ be the set of dual lattice vectors of Euclidean weights in $(w-\varepsilon,w+\varepsilon)$ in $\Lambda^{\vee}$ and let 
	$f_{\widetilde{\sW}}(\yv) \eqdef \sum_{\vec{w} \in \widetilde{\sW}} \cos(2\pi \langle \vec{x},\vec{y} \rangle)$. We have
	\begin{align} 
	f_{\widetilde{\sW}}(\vec{y}) &= \frac{1}{|\Lambda^{\vee}|} \; \sum_{j \geq 0}  \frac{N_{j}}{j^{n}\; (2\pi)^{n/2}} \;  \Big( (2\pi(w+\varepsilon)j)^{n/2}J_{n/2}(2\pi (w+\varepsilon )j) \\
	&\qquad\qquad\qquad\qquad\qquad\qquad\qquad\quad - (2\pi(w-\varepsilon)j)^{n/2}J_{n/2}(2\pi (w-\varepsilon )j) \Big)\nonumber  
\end{align} 
	where $N_{j}$ is the number of words of Euclidean norm $j$ in $\Lambda + \vec{e}$ and $J_{\nu}$ is the Bessel function of the first kind of order\footnote{Here the $j$'s belong to the discrete set of all possible norms in the lattice and should not be viewed as an integer value.} $\nu$.
\end{restatable}
The proof is given in Appendix \ref{sec:app_lattice}.
Let us take some subset $\sW$ of $\widetilde{\sW}$ of size $N$ say. We make the approximation $f_{{{\sW}}}(\yv) \approx \frac{N}{\card{\widetilde{\sW}}} \; f_{\widetilde{\sW}}(\yv)$ and by using the
Gaussian heuristic and some computations that are detailed in Appendix \ref{sec:app_lattice}:
\begin{equation}
\label{eq:sum}
f_{\sW}(\vec{y}) \approx N \frac{\sqrt{n \pi}}{e} \sum_{j\geq 0} N_{j} \left( \frac{n}{2\pi e wj} \right)^{n/2-1} J_{n/2-1}(2\pi wj).
\end{equation}
We can use now a similar heuristic as the one described in \S \ref{ss:intuition} and predict that the abnormal large values of 
the score function $f_{\sW}(\vec{y})$ appear when $\yv$ is abnormally close to $\Lambda$, say $N_{\leq x} \neq 0$, where
$N_{\leq x}  =\card{\{ \cv \in \Lambda : \norm{ \yv - \cv} \leq x\}}$ when $\prob{}{N_{\leq x} \neq 0} \ll 1$. 
In this case, we make the crude approximation that the sum \eqref{eq:sum} is dominated by the term $j_0$ which is the smallest in it:
\begin{equation} \label{eq:lattice_approx_f_first_term}
f_{\sW}(\vec{y}) \approx N \frac{\sqrt{n \pi}}{e} N_{j_0} \left( \frac{n}{2\pi e wj_0} \right)^{n/2-1} J_{n/2-1}(2\pi wj_0).
\end{equation}
The survival function $\prob{}{f_{\sW}(\vec{y}) \geq X}$ is then crudely approximated as the probability that such an event happens
\begin{eqnarray*}
\prob{}{f_{\sW}(\vec{y}) \geq N \frac{\sqrt{n \pi}}{e} N_{x} \left( \frac{n}{2\pi e wx} \right)^{n/2-1} J_{n/2-1}(2\pi w x)} & \approx & \prob{}{N_{\leq x} \geq 1} \\
& \approx & \esp(N_{\leq x}) \\
& \approx  & \left( \frac{x}{\sqrt{\tfrac{n}{2 \pi e} \cdot (\pi n)^{1/n}} \cdot | \Lambda|^{1/n}}\right)^n
\end{eqnarray*}
where the last approximation is the Gaussian heuristic.
In the context of the experiments described in \cite[\S 5]{DP23}, $\Lambda \eqdef \mathcal{L}\left( \mat{B} \right)$ is given by\footnote{In \cite{DP23}, $\Lambda$ and $\mat{B}$ are actually respectively $\Lambda'$ and $\mat{B}'$}
	\begin{equation}
	\label{eq:B}
	\mat{B} \eqdef 
	\left[ 
	\begin{array}{c|c}
	\Id_{n/2} & \mat{A} \\
	\hline
	\mat{0} & q \cdot \Id_{n/2}
	\end{array}
	\right] \cdot 
	\left[ 
	\begin{array}{c|c}
	2 \cdot \Id_{k_{\mathrm{fft}}} & \mat{0} \\
	\hline
	\mat{0} & \Id_{n - k_{\mathrm{fft}}}
	\end{array}
	\right]
	\end{equation}
Then we use the same \emph{full sieve} algorithm as in \cite{DP23} to produce short vectors ${{\sW}} \subset \Lambda^{\vee}$. 
In what follows, we use the practical values of $N$ and $w$ that we obtained by experiments. We have reused the implementation for the experiments in \cite[§5]{DP23}\footnote{\url{https://github.com/ludopulles/DoesDualSieveWork}}. 
This very crude estimation seems to capture the error-floor behavior of the survival function as shown in Figure \ref{fig:dual_attack}. The point is that it is precisely this part of the curve which is not predicted by the standard independence assumption and which had no explanation so far. It can also be observed that the duality result is nothing but a straighforward use of the Poisson formula which has also be 
used very recently in \cite{WE23} to predict the abnormal variance of the BDD score distribution observed in \cite[Table~1]{DP23}.
\begin{figure}[h]
\centering
\begin{minipage}[c]{0.45\linewidth}
\includegraphics[width=\textwidth]{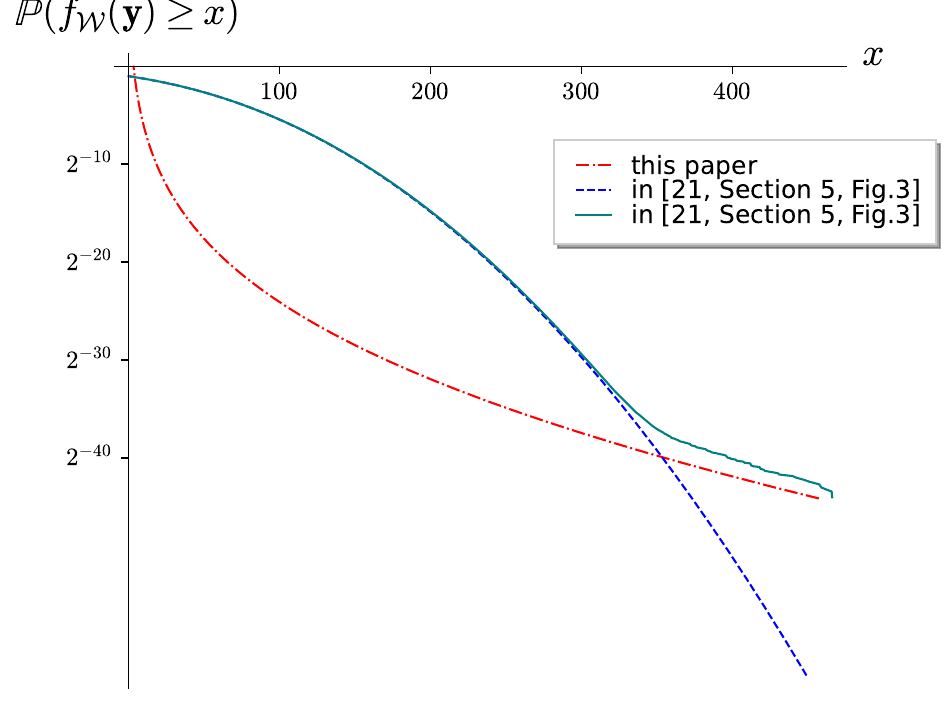}
	
\end{minipage} \hfill
\begin{minipage}[c]{0.45\linewidth}
\includegraphics[width=\textwidth]{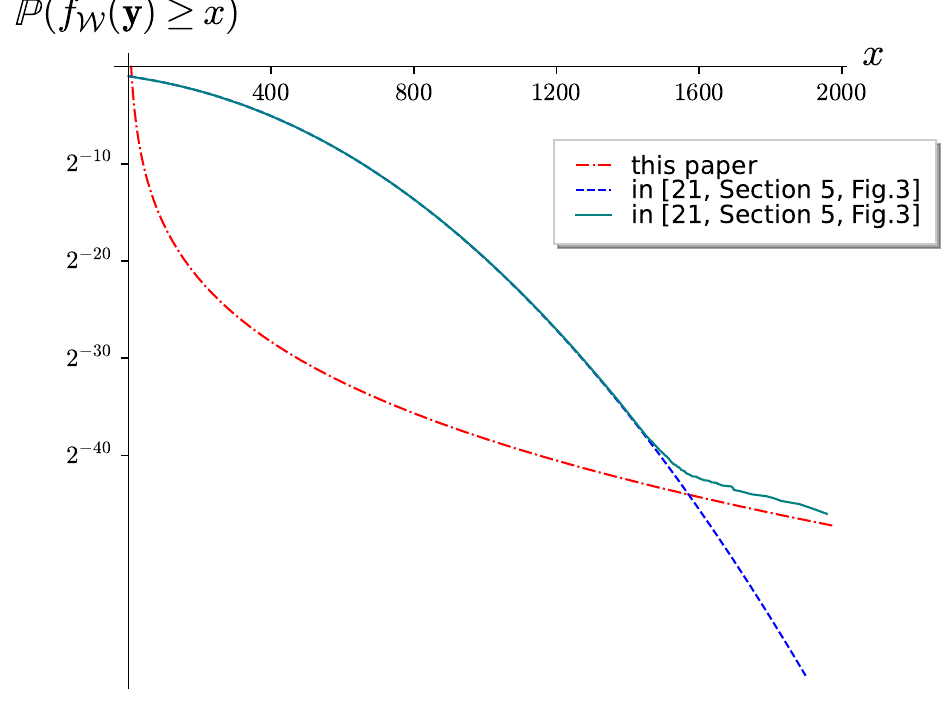}
	
\end{minipage}
		\caption{Crude estimation of the survival function (red, dash-dot line) compared to the experiments in \cite[§5]{DP23} (green, full line) and the prediction with the standard independence assumption (blue, dashed line). (left) $q = 3329$, $n = 60$, $T = 2^{45}$, $N = 5040$ and $w = 0.0320$;  (right) $q = 3329$, $n = 80$, $T = 2^{48}$, $N = 89494$ and $w = 0.0376$. \label{fig:dual_attack} }
\end{figure}
\newline

{\bf \noindent A more precise prediction.} One can remark (see Figure \ref{fig:dual_attack}) that $i)$ our newly introduced approximate distribution for $f_{\sW}$ given in Equation \eqref{eq:lattice_approx_f_first_term} matches the experimental curves specifically in the waterfall-floor zone and $ii)$ the distribution of $f_{\sW}$ given by the independence heuristic \cite[Heuristic 3]{DP23} matches the experimental curves up to this waterfall-floor zone. A natural idea to predict the experimental curve on the whole support is therefore to take the convolution of these two distributions. Indeed, for any support point, there is always one distribution which exponentially dominates the other one.
\newline

\noindent \textbf{Distribution in the waterfall-floor zone.}
Let us denote by $\vec{X}_{\textup{floor}}$ the random variable given by Equation \eqref{eq:lattice_approx_f_first_term}, namely:
\begin{equation}\label{eq:XFloor}  
\vec{X}_{\textup{floor}} \eqdef G(j_0)
\end{equation} 
where 
$$
G(j) = N \frac{\sqrt{n \pi}}{e} \left( \frac{n}{2\pi e w j} \right)^{n/2-1} J_{n/2-1}(2\pi w j)
$$ and $j_0$ is the length of the shortest vector of the lattice.
We only have to compute the distribution of $j_0$ to compute the distribution of $\vec{X}_{\textup{floor}}$. To that extend we make the following classic model.
\begin{model}{Model for the number of lattice points in a ball.} \label{model:lattice}
	Let $\Lambda$ be a random lattice of full rank $n$ and of volume $V$. We make the model that:
	\begin{equation*}
		|\Lambda \setminus \{0\} \cap \mathcal{B}_z| \sim \mathrm{Poisson}\left(\frac{\mathrm{Vol}\left( \mathcal{B}_1\right)  z^{n} }{V}\right)
	\end{equation*}
	where $\mathcal{B}_z$ denotes the Euclidean ball of center $\vec{0}$ and radius $z$ and 
	$\mathrm{Vol}\left( \mathcal{B}_1\right) = \frac{\sqrt{\pi^{n}}}{\Gamma\left(\frac{n}{2} + 1\right)}$.
\end{model}
This model allows us to write the following fact regarding the distribution of the length of the $j$'th shortest vector:
\begin{fact}\label{fact:gamma}
	Under Model \ref{model:lattice}, and for $k=0\cdots\infty$ the distribution of $j_k$, \ie the $k'th$ non-zero shortest vector of a random lattice $\Lambda$ of full rank $n$ and of volume $V$, is given by:
	$$j_k^n \sim \mathrm{Gamma}\left(k + 1, \; \frac{\mathrm{Vol}\left( \mathcal{B}_1\right) }{V} \right)$$
	where $\vec{Z} \sim \mathrm{Gamma}\left(b,\theta\right)$ has the following survival function when $b$ is an integer:
	$$ 
	\prob{}{\vec{Z} \geq \alpha} \eqdef  e^{-\theta \alpha} \sum_{i = 0}^{b-1}\frac{\left(\theta \alpha\right)^i}{i!}.
	$$
\end{fact}
\begin{proof}
	For any $z > 0$, we have that
	\begin{align*}
		\prob{}{j_k^n> z^n} &=\prob{}{j_k> z} \\
		&= \prob{}{\bigcup_{i=0}^{k} ``|\Lambda \setminus \{0\} \cap \mathcal{B}_z| =  i" } \\
		&=\sum_{i = 0}^{k} \prob{}{|\Lambda \setminus \{0\} \cap \mathcal{B}_z| =  i} && \mbox{(disjoint union)} \\
		&=\sum_{i= 0}^{k} \frac{(\mathrm{Vol}\left(\mathcal{B}_1\right) \; z^{n} / V)^i \: e^{- (\mathrm{Vol}\left(\mathcal{B}_1 \right) \; z^{n} / V)}}{i \: ! } && \mbox{(Model \ref{model:lattice})}
	\end{align*}
	which completes the proof. 
\end{proof}

\noindent \textbf{Distribution in the waterfall zone.}
Let us denote by $\vec{X}_{\textup{fall}}$ the random variable $f_{\sW}$ under the independence heuristic. As given by \cite[Lemma 3]{DP23} and the discussion that follows their lemma we have
\begin{fact}
	$\vec{X}_{\textup{fall}}$ follows a normal distribution of mean $0$ and variance $\frac{1}{2} N$. More precisely, its probability density function $p$ is given by
	$$
	p(x) = \frac{1}{\sqrt{\pi N} } e^{-x^2/N}. 
	$$
\end{fact}
We now make the refined model that $f_{\sW}\left(\yv\right)$ follows the same distribution as $\vec{X}_{\textup{fall}} +  \vec{X}_{\textup{floor}}$ where the distribution of $\vec{X}_{\textup{floor}}$ is computed numerically by using Fact \ref{fact:gamma} along with Equation~\eqref{eq:XFloor}. We show in Figure~\ref{fig:dual_attack2} that the distribution of this refined model well approximates the behavior of the experimental distribution of $f_{\sW}$ on the whole support.
\begin{figure}[h]
	\centering
	\begin{minipage}[c]{0.45\linewidth}
		\includegraphics[width=\textwidth]{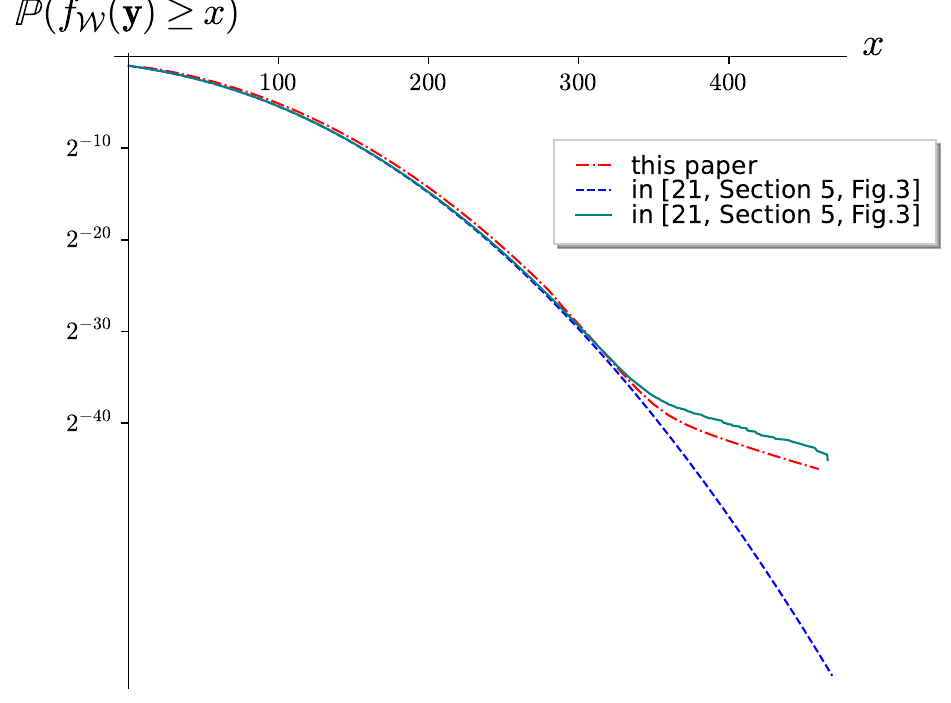}
		
	\end{minipage} \hfill
	\begin{minipage}[c]{0.45\linewidth}
		\includegraphics[width=\textwidth]{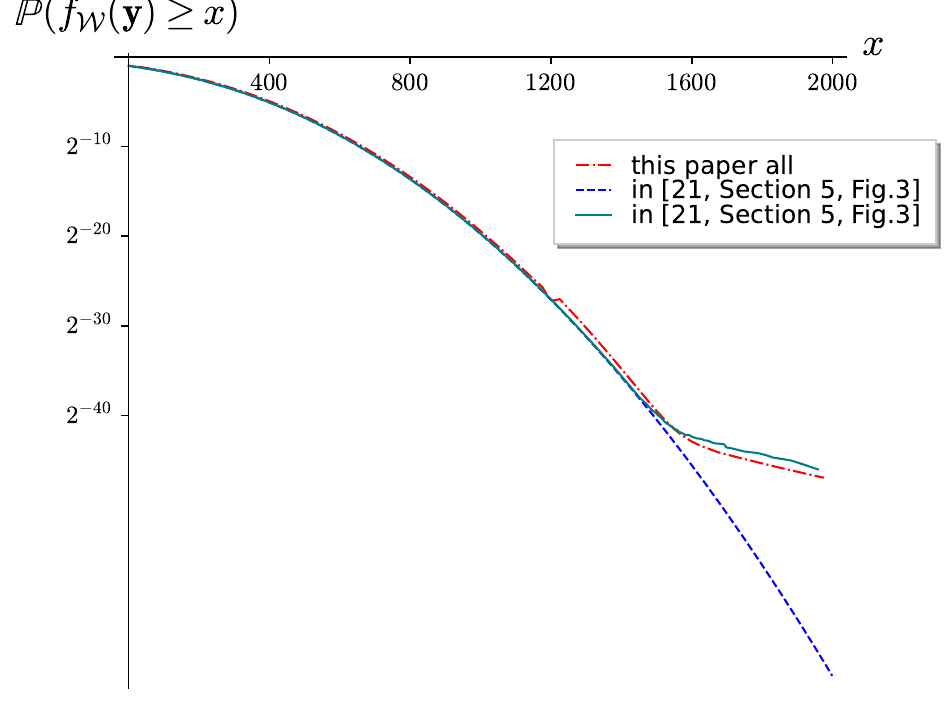}
		
	\end{minipage}
	\caption{Refined estimation of the survival function (red, dash-dot line) compared to the experiments in \cite[§5]{DP23} (green, full line) and the prediction with the standard independence assumption (blue, dashed line). (left) $q = 3329$, $n = 60$, $T = 2^{45}$, $N = 5040$ and $w = 0.0320$; (right) $q = 3329$, $n = 80$, $T = 2^{48}$, $N = 89494$ and $w = 0.0376$ \label{fig:dual_attack2} }
\end{figure}
\medskip

\noindent \textbf{Concurrent work.}
Note that very recently we have been made aware of the concurrent work \cite{DP23a} which similarly to what we do here, uses Bessel functions to predict the score with a related approach 
and similar predictions (see \cite[\S 4.3]{DP23a}).

\iftoggle{llncs}{
	\bibliographystyle{splncs04}
	
}{
	\bibliographystyle{alpha}
}
\newcommand{\etalchar}[1]{$^{#1}$}

\iftoggle{llncs}{
\newpage 
\appendix
\section{Proof of Lemma \ref{lem:fundamental}}\label{app:proofLemmaFund}

	\lemFundamental* 

	\begin{proof}
	\begin{align*}
		\widehat{f_{\yv,\sH,\Gmbis[]}} \left(\uv\right) &= \sum_{\vv \in \F_2^\kbis} (-1)^{\scp{\uv}{\vv}} f_{\yv,\sH,\Gmbis[]}(\vv) \\
		&= \sum_{(\hv,\vv)\in \CC^\perp\times \F_2^\kbis : \left(\hv,\vv \Gmbis \right) \in \sH} (-1)^{\langle \yv, \hv \rangle + \langle \uv, \vv \rangle} && \mbox{(Eq. \eqref{eq:def_f} and \eqref{eq:def_fourier})}\\
		&=  \sum_{(\hv,\vv)\in \CC^\perp\times \F_2^\kbis : \left(\hv,\vv \Gmbis \right) \in  \sH} (-1)^{\langle \yv, \hv \rangle + \langle \xv \Gmbis[\transpose], \vv \rangle} \\
		&= \sum_{(\hv,\vv)\in \CC^\perp\times \F_2^\kbis:\left(\hv,\vv \Gmbis \right) \in  \sH} (-1)^{\langle \yv, \hv \rangle + \langle \xv, \vv \Gmbis \rangle} \\
		&=  \sum_{\left(\hv,\cvbis \right) \in  \sH} (-1)^{\langle \yv, \hv \rangle + \langle \xv, \cvbis \rangle}
	\end{align*}	
\end{proof} \section{Proof of Proposition \ref{prop:biasCodedRLPN}}\label{app:a}
 
 Let us first recall this proposition
 \propbiasCodedRLPN*
 
The proof of this Proposition \ref{prop:biasCodedRLPN} will be a consequence of the following lemma. 
\begin{restatable}{lemma}{lemmaVar}\label{lemma:varianceEb}
	Let $w,s,k,\kbis,t_{\textup{aux}},n \in \mathbb{N}$ be such that (for some constant $a > 0$)
	\begin{equation}\label{eq:constraintsVar} 
			\frac{\binom{n-s}{w} \binom{s}{t_{\textup{aux}}}}{2^{k}} = \OO{n^{\alpha}} \quad \mbox{and} \quad \frac{\binom{s}{t_{\textup{aux}}}}{2^{s-\kbis}} = \OO{n^{\alpha}}
		\end{equation}

	Let $\sN$ be a fixed set of $n-s$ positions in $\IInt{1}{n}$ and $\sP \eqdef \llbracket 1,n \rrbracket \backslash \sN$. Let $\vec{e}\in\F_{2}^{n}$ be some error of weight $u$ on $\sN$.

	Assume that $\CC$ is an $[n,k]$-linear code chosen by picking an $(n-k) \times n$ binary parity-check matrix $\Hm$ uniformly at random and that $\CCbis$ is chosen by picking an $(s-\kbis) \times s$ binary parity-check matrix $\Hmbis$ uniformly at random.
	Let us define for $b \in \{0,1\}$,
	\begin{align*}
			E_b &\eqdef \card{\left\{(\ev',\hv) \in   \mathcal{S}_{\tbis}^s \times \CC^\perp:\; \;\left|\hv_{\sN}\right| = w \; \mbox{ and } \; \ev' + \hv_{\sP} \in \CCbis, \; \langle \ev',\vec{e}_{\sP} \rangle +  \langle \vec{e}_{\sN},  \vec{h}_{\sN}\rangle=b \right\}}, \\
			E'_b &\eqdef \card{\left\{ \left(\vec{e}',\vec{h}\right) \in  \mathcal{S}_{\tbis}^{s} \times  \F_{2}^{n} :\; |\vec{h}_{\sN}| = w \; \mbox{ and } \; \;\langle \vec{e}',\vec{e}_{\sP} \rangle +  \langle \vec{e}_{\sN},  \vec{h}_{\sN}\rangle =b\right \}}. 
		\end{align*}
	
	Then,
	\begin{eqnarray}
			\mathbb{E}_{\vec{H},\Hmbis}(E_b) &=&  \Th{1}\frac{E'_b}{2^{k+s-\kbis}} \label{eq:espL}  \\
			\mathbf{Var}_{\vec{H},\Hmbis}(E_b) &=& \OO{n^{\alpha}}\frac{E'_b}{2^{k+s-\kbis}}  \label{eq:varL}
		\end{eqnarray}
\end{restatable}
	\begin{proof}
	Let $\mathbf{1}_{\ev',\hv}$ be the indicator function of the event ``$\hv \in \CC^{\bot}$ and $\ev'+\hv_{\sP} \in \CCbis$''. Define 
	$$
	\mathcal{E}_{b} \eqdef \left\{ \left( \vec{e}',\vec{h} \right) \in \mathcal{S}_{t_{\textup{aux}}}^{s} \times \F_{2}^{n} \mbox{: }  |\vec{h}_{\sN}| = w \mbox{ and }  \langle \ev',\vec{e}_{\sP} \rangle + \langle \vec{e}_{\sN},\vec{h}_{\sN} \rangle=b \right\}.
	$$
	Notice that $E_{b}'  = \card{\mathcal{E}_{b}}$. 
By definition and linearity of the expectation
	\begin{align}
		\mathbb{E}_{\vec{H},\Hmbis}\left( E_{b} \right) &=\mathbb{E}_{\vec{H},\Hmbis} \left( \sum_{(\vec{e}',\vec{h})\in \mathcal{E}_{b}} \mathbf{1}_{\ev',\hv} \right) \nonumber \\
		&= \sum_{(\vec{e}',\vec{h})\in \mathcal{E}_{b}} \mathbb{E}_{\vec{H},\Hmbis}(\mathbf{1}_{\ev',\hv}) \nonumber \\ 
		&=  \sum_{(\vec{e}',\vec{h})\in \mathcal{E}_{b}} \mathbb{P}_{\vec{H},\Hmbis}\left( \hv \in \CC^{\bot}, \hv_{\sP} + \ev' \in \CCbis \right)\label{eq:expectation}
	\end{align}
	We have that,
	\begin{equation}\label{eq:probBasic}
		\mathbb{P}_{\vec{H},\Hmbis}\left( \hv \in \CC^{\bot}, \hv_{\sP} + \ev' \in \CCbis \right) = \left\{
		\begin{array}{cl}
		\frac{1}{2^{k + s - \kbis}} & \;\; \mbox{if } \vec{h}_{\sP}+\vec{e}' \neq \vec{0} \\
		\frac{1}{2^{k}} & \;\; \mbox{otherwise.}
		\end{array}
		\right.
	\end{equation}
Indeed, notice that events ``$\hv \in \CC^{\bot}$'' and ``$\hv_{\sP} + \ev' \in \CCbis $'' are independent. Furthermore, $\vec{h}$ cannot be equal to $\vec{0}$ as $|\vec{h}_{\sN}| = w > 0$.  Therefore,
	\begin{align*}
		\mathbb{P}_{\vec{H},\Hmbis}\left( \hv \in \CC^{\bot},\hv_{\sP} + \ev' \in \CCbis \right) &= \mathbb{P}_{\vec{H}}\left( \hv \in \CC^{\bot}\right)  \; \mathbb{P}_{\Hmbis}\left(\hv_{\sP} + \ev' \in \CCbis \right) \nonumber \\
		&=  \frac{1}{2^{k}}\;\mathbb{P}_{\Hmbis}\left( \hv_{\sP} + \ev' \in \CCbis  \right)
\end{align*}
	Now, plugging Equation \eqref{eq:probBasic} in \eqref{eq:expectation} leads to,
	\begin{align}
		\mathbb{E}_{\vec{H},\Hmbis}\left( E_{b} \right) &= \sum_{\substack{(\vec{e}',\vec{h})\in \mathcal{E}_{b}\\ \vec{h}_{\sP} \neq \vec{e}' }} \frac{1}{2^{k + s - \kbis}} + \sum_{\substack{(\vec{e}',\vec{h})\in \mathcal{E}_{b}\\ \vec{h}_{\sP} = \vec{e}' }} \frac{1}{2^{k}} \nonumber \\
		&= \sum_{\vec{e}' \in \cS_{\tbis}^{s}} \left( \frac{\card{\sE_{b,1}^{\vec{e}'} }}{2^{k+s-\kbis}} + \frac{\card{\sE_{b,2}^{\vec{e}'} }}{2^k} \right) \label{eq:Eb12}
\end{align} 
	where for a fixed $\vec{e}'$, 
	\begin{equation}\label{eq:Eb1} 
	\sE_{b,1}^{\vec{e}'} \eqdef \left\{ \vec{h} \in \mathbb{F}_{2}^{n}: \; (\vec{e}',\vec{h}) \in \sE_{b} \mbox{ and } \vec{h}_{\sP} \neq \vec{e}' \right\},
	\end{equation}
	\begin{equation}\label{eq:Eb2} 
		\sE_{b,2}^{\vec{e}'} \eqdef \left\{ \vec{h} \in \mathbb{F}_{2}^{n}: \; (\vec{e}',\vec{h}) \in \sE_{b} \mbox{ and } \vec{h}_{\sP} = \vec{e}' \right\}.
	\end{equation} 
	Notice that for any $\vec{e}'$, 
	\begin{equation}\label{eq:cardEb12}  
	\card{\sE_{b,1}} = \frac{\card{\sE_{b,2}}}{2^{s}} 
	\end{equation} 
	It is readily seen that,
	$$
	\card{\sE_b} = \sum_{\vec{e}' \in \cS_{\tbis}^{s}} \left( \card{\sE_{b,1}^{\vec{e}'}} +  \card{\sE_{b,2}^{\vec{e}'}} \right) 
	$$
	which implies that 
	$$
	\sum_{\vec{e}' \in \cS_{\tbis}^{s}} \card{\sE_{b,1}^{\vec{e}'}} = \frac{\card{\sE_{b}}}{1+\frac{1}{2^{s}}} = \Th{1} \card{\sE_b} \quad \mbox{and} \quad \sum_{\vec{e}' \in \cS_{\tbis}^{s}} \card{\sE_{b,2}^{\vec{e}'}} = \frac{\card{\sE_{b}}}{2^{s}-1} = \Th{2^{-s}} \card{\sE_{b}}
	$$
	Plugging this into Equation \eqref{eq:Eb12} leads to,
	\begin{equation}\label{eq:expE12b} 
	\mathbb{E}_{\vec{H},\Hmbis}\left( E_{b} \right) = \Th{1} \frac{\card{\sE_{b}}}{2^{k+s-\kbis}} + \Th{1} \frac{\card{\sE_{b}}}{2^{k+s}} = \Th{1}\frac{\card{\sE_{b}}}{2^{k+s-\kbis}}
	\end{equation} 
	which shows Equation \eqref{eq:espL}. 	
	Let us show now Equation \eqref{eq:varL}. By definition 
	\begin{align}
		\mathbf{Var}\left( E_{b} \right) &= \sum_{(\vec{e}',\vec{h}) \in \mathcal{E}_{b}}\mathbf{Var}\left( \mathbf{1}_{\vec{e}',\vec{h}} \right) + \sum_{\substack{(\vec{e}_{0}',\vec{h}^{0}),(\vec{e}_{1}',\vec{h}^{1}) \in \mathcal{E}_{b} \\ (\vec{e}_{0}',\vec{h}^{0}) \neq (\vec{e}_{1}',\vec{h}^{1})}} \mathbb{E}\left( \mathbf{1}_{\vec{e}_{0}',\vec{h}^{0}} \; \mathbf{1}_{\vec{e}_{1}',\vec{h}^{1}} \right) - \mathbb{E}\left( \mathbf{1}_{\vec{e}_{0}',\vec{h}^{0}} \right) \mathbb{E}\left( \mathbf{1}_{\vec{e}_{1}',\vec{h}^{1}} \right) \nonumber \\
		&\leq \frac{E_{b}'}{2^{k+s-\kbis}}+ \sum_{\substack{(\vec{e}_{0}',\vec{h}^{0}),(\vec{e}_{1}',\vec{h}^{1}) \in \mathcal{E}_{b} \\ (\vec{e}_{0}',\vec{h}^{0}) \neq (\vec{e}_{1}',\vec{h}^{1})}} \mathbb{E}\left( \mathbf{1}_{\vec{e}_{0}',\vec{h}^{0}} \; \mathbf{1}_{\vec{e}_{1}',\vec{h}^{1}} \right) - \mathbb{E}\left( \mathbf{1}_{\vec{e}_{0}',\vec{h}^{0}} \right) \mathbb{E}\left( \mathbf{1}_{\vec{e}_{1}',\vec{h}^{1}} \right) \nonumber
	\end{align}
	where we used that $\mathbf{Var}(\mathbf{1}_{\vec{e}',\vec{h}}) \leq \mathbb{E}\left( \mathbf{1}_{\vec{e},\vec{h}}^{2} \right) =  \mathbb{E}\left( \mathbf{1}_{\vec{e},\vec{h}} \right)$ and Equation \eqref{eq:expE12b}.

	To compute the above expectations, we will split in two cases according to $\sE_{b,1}^{\vec{e}'}$ and $\sE_{b,2}^{\vec{e}'}$ which are respectively defined in Equations \eqref{eq:Eb1} and \eqref{eq:Eb2}. More precisely, we will fix $\vec{e}_{b}'$ and suppose that $\vec{h}^{b}$ belongs to $\sE_{b,1}$ or $\sE_{b,2}$. We will treat the following disjoint cases.

	\begin{enumerate}[label=\textcolor{blue}{\arabic*.}, ref=\arabic*]
		\item\label{case1} $\sF_{0}^{\vec{e}_{0}',\vec{e}_{1}'} = \left\{(\vec{h}^{0},\vec{h}^{1}): \;  \vec{h}^{0} \in \sE_{b,1}^{\vec{e}_{0}'} \mbox{ and }  \vec{h}^{1} \in \sE_{b,1}^{\vec{e}_{1}'}\right\}$,

		\item\label{case2} $\sF_{1}^{\vec{e}_{0}',\vec{e}_{1}'} = \left\{(\vec{h}^{0},\vec{h}^{1}): \;  \vec{h}^{0} \in \sE_{b,2}^{\vec{e}_{0}'} \mbox{ and }  \vec{h}^{1} \in \sE_{b,1}^{\vec{e}_{1}'}\right\}$

		\item\label{case3} $\sF_{2}^{\vec{e}_{0}',\vec{e}_{1}'} = \left\{(\vec{h}^{0},\vec{h}^{1}): \;  \vec{h}^{0} \in \sE_{b,1}^{\vec{e}_{0}'} \mbox{ and }  \vec{h}^{1} \in \sE_{b,2}^{\vec{e}_{1}'}\right\}$

		\item\label{case4} $\sF_{3}^{\vec{e}_{0}',\vec{e}_{1}'} = \left\{(\vec{h}^{0},\vec{h}^{1}): \;  \vec{h}^{0} \in \sE_{b,2}^{\vec{e}_{0}'} \mbox{ and }  \vec{h}^{1} \in \sE_{b,2}^{\vec{e}_{1}'}\right\}$
	\end{enumerate}
	In particular,
	\begin{multline}\label{eq:ineqVariant} 
		\mathbf{Var}\left( E_{b} \right)  \leq \frac{E_{b}'}{2^{k+s-\kbis}} + \sum_{\vec{e}_{0}',\vec{e}_{1}'}\Big( \underbrace{\sum_{\substack{\vec{h}^{0},\vec{h}^{1} \in \sF_{0}^{\vec{e}_{0}',\vec{e}_{1}'} \\ (\vec{e}_{0}',\vec{h}^{0}) \neq (\vec{e}_{1}',\vec{h}^{1}) }}  \mathbb{E}\left( \mathbf{1}_{\vec{e}_{0}',\vec{h}^{0}} \; \mathbf{1}_{\vec{e}_{1}',\vec{h}^{1}} \right) - \mathbb{E}\left( \mathbf{1}_{\vec{e}_{0}',\vec{h}^{0}} \right) \mathbb{E}\left( \mathbf{1}_{\vec{e}_{1}',\vec{h}^{1}} \right)}_{\eqdef F_0}  \\
		+ \underbrace{\sum_{\substack{\vec{h}^{0},\vec{h}^{1} \in \sF_{1}^{\vec{e}_{0}',\vec{e}_{1}'} \\ (\vec{e}_{0}',\vec{h}^{0}) \neq (\vec{e}_{1}',\vec{h}^{1}) }}  \mathbb{E}\left( \mathbf{1}_{\vec{e}_{0}',\vec{h}^{0}} \; \mathbf{1}_{\vec{e}_{1}',\vec{h}^{1}} \right) - \mathbb{E}\left( \mathbf{1}_{\vec{e}_{0}',\vec{h}^{0}} \right) \mathbb{E}\left( \mathbf{1}_{\vec{e}_{1}',\vec{h}^{1}} \right)}_{\eqdef F_1} \\
		+
		\underbrace{\sum_{\substack{\vec{h}^{0},\vec{h}^{1} \in \sF_{2}^{\vec{e}_{0}',\vec{e}_{1}'} \\ (\vec{e}_{0}',\vec{h}^{0}) \neq (\vec{e}_{1}',\vec{h}^{1}) }}  \mathbb{E}\left( \mathbf{1}_{\vec{e}_{0}',\vec{h}^{0}} \; \mathbf{1}_{\vec{e}_{1}',\vec{h}^{1}} \right) - \mathbb{E}\left( \mathbf{1}_{\vec{e}_{0}',\vec{h}^{0}} \right) \mathbb{E}\left( \mathbf{1}_{\vec{e}_{1}',\vec{h}^{1}} \right)}_{\eqdef F_2}  \\
		+
		\underbrace{\sum_{\substack{\vec{h}^{0},\vec{h}^{1} \in \sF_{3}^{\vec{e}_{0}',\vec{e}_{1}'} \\ (\vec{e}_{0}',\vec{h}^{0}) \neq (\vec{e}_{1}',\vec{h}^{1}) }}  \mathbb{E}\left( \mathbf{1}_{\vec{e}_{0}',\vec{h}^{0}} \; \mathbf{1}_{\vec{e}_{1}',\vec{h}^{1}} \right) - \mathbb{E}\left( \mathbf{1}_{\vec{e}_{0}',\vec{h}^{0}} \right) \mathbb{E}\left( \mathbf{1}_{\vec{e}_{1}',\vec{h}^{1}} \right)}_{\eqdef F_3} 
		\Big) 
	\end{multline}
	Let,
	\begin{equation}\label{eq:cov} 
		\textup{Cov} \eqdef F_{0} + F_{1} + F_{2} + F_{3}
	\end{equation}

	\noindent For each cases, we will split according to the following subcases. 
		\begin{enumerate}[label=\textcolor{blue}{\roman*.}, ref=\roman*]\setlength{\itemsep}{5pt}
		\item \label{subcase:1} $ \vec{e}_{0}' = \vec{e}_{1}', \; \vec{h}^{0} \neq \vec{h}^{1}$ and $\vec{h}_{\sP}^{0} = \vec{h}_{\sP}^{1}$,

		\item \label{subcase:2} $\vec{e}_{0}' = \vec{e}_{1}', \; \vec{h}^{0} \neq \vec{h}^{1}$ and $\vec{h}_{\sP}^{0} \neq \vec{h}_{\sP}^{1}$,

		\item\label{subcase:3} $\vec{e}_{0}' \neq \vec{e}_{1}'$ and $\vec{h}^{0} = \vec{h}^{1}$

		\item\label{subcase:4} $\vec{e}_{0}' \neq \vec{e}_{1}', \; \vec{h}^{0} \neq \vec{h}^{1}$ and $\vec{h}_{\sP}^{0} = \vec{h}_{\sP}^{1}$

		\item\label{subcase:5} $\vec{e}_{0}' \neq \vec{e}_{1}',  \vec{h}^{0} \neq \vec{h}^{1}, \hv_{\sP}^{0} \neq \vec{h}_{\sP}^{1}$ and $\vec{h}_{\sP}^{0} + \vec{e}_{0}' = \vec{h}_{\sP}^{1} + \vec{e}_{1}'$

		\item\label{subcase:6} $\vec{e}_{0}' \neq \vec{e}_{1}',  \vec{h}^{0} \neq \vec{h}^{1}, \hv_{\sP}^{0}\neq \vec{h}_{\sP}^{1}$ and $\vec{h}_{\sP}^{0} + \vec{e}_{0}' \neq  \vec{h}_{\sP}^{1} + \vec{e}_{1}'$ 
		\newline
	\end{enumerate}

	{\bf\noindent Case \ref{case1}:} Recall that in this case we have 
	\begin{equation}\label{eq:case1}  
		\vec{h}^{0}_{\sP} \neq \vec{e}_{0}' \quad \mbox{and} \quad \vec{h}^{1}_{\sP} \neq \vec{e}_{1}'
	\end{equation}   
	We have, 
	\begin{equation*} 
	\mathbb{E}\left( \mathbf{1}_{\vec{e}_{0}',\vec{h}^{0}} \right) \mathbb{E}\left( \mathbf{1}_{\vec{e}_{1}',\vec{h}^{1}} \right) = \left( \frac{1}{2^{k+s-\kbis}} \right)^{2}
	\end{equation*} 
	Let us compute $\mathbb{E}_{\vec{H},\Hmbis}\left( \mathbf{1}_{\vec{e}_{0}',\vec{h}^{0}} \; \mathbf{1}_{\vec{e}_{1}',\vec{h}^{1}} \right)$ when $(\vec{e}_{0}',\vec{h}^{0}) \neq (\vec{e}_{1}',\vec{h}^{1})$. By definition
	\begin{align*}
		\mathbb{E}_{\vec{H},\Hmbis}\left( \mathbf{1}_{\vec{e}_{0}',\vec{h}^{0}} \; \mathbf{1}_{\vec{e}_{1}',\vec{h}^{1}} \right) &= \mathbb{P}_{\vec{H},\Hmbis}\left( \hv^{0},\vec{h}^{1} \in \CC^{\bot}, \vec{h}^{0}_{\sP} + \ev_{0}' \in \CCbis, \vec{h}^{1}_{\sP} + \ev_{1}' \in \CCbis \right) \\
		&= \mathbb{P}_{\vec{H},\Hmbis}\left( \hv^{0},\vec{h}^{1} \in \CC^{\bot}\right) \mathbb{P}_{\vec{H},\Hmbis} \left(  \vec{h}^{0}_{\sP} + \ev_{0}' \in \CCbis, \vec{h}^{1}_{\sP} + \ev_{1}' \in \CCbis\right)
	\end{align*}
Therefore,
	\begin{align} 
	\sum_{\vec{e}_{0}',\vec{e}_{1}'} F_{0} &= \sum_{\vec{e}_{0}',\vec{e}_{1}'}\sum_{\substack{\vec{h}^{0},\vec{h}^{1} \in \sF_{0}^{\vec{e}_{0}',\vec{e}_{1}'} \\ (\vec{e}_{0}',\vec{h}^{0}) \neq (\vec{e}_{1}',\vec{h}^{1}) }}  \mathbb{E}\left( \mathbf{1}_{\vec{e}_{0}',\vec{h}^{0}} \; \mathbf{1}_{\vec{e}_{1}',\vec{h}^{1}} \right) - \left( \frac{1}{2^{k+s-\kbis}}\right)^{2} \nonumber\\
	&= \sum_{\vec{e}_{0}',\vec{e}_{1}'}\sum_{\substack{\vec{h}^{0},\vec{h}^{1} \in \sF_{0}^{\vec{e}_{0}',\vec{e}_{1}'} \\ (\vec{e}_{0}',\vec{h}^{0}) \neq (\vec{e}_{1}',\vec{h}^{1}) }} \textup{Cov}^{(0)}\label{eq:cov0} 
	\end{align}
	where, 
	\begin{equation*}
	\textup{Cov}^{(0)} \eqdef \mathbb{P}_{\vec{H}}\left( \hv^{0},\vec{h}^{1} \in \CC^{\bot}\right) \mathbb{P}_{\Hmbis} \left( \vec{h}_{\sP}^{0} + \ev_{0}' \in \CC_{\textup{aux}}, \vec{h}_{\sP}^{1} + \ev_{1}' \in \CC_{\textup{aux}} \right) - \left(\frac{1}{ 2^{k+s-\kbis}}\right)^{2}
	\end{equation*} 
	Our aim now it to upper-bound $\textup{Cov}^{(0)}$ according to the above $6$ sub-cases.
	\newline

	{\bf \noindent Sub-case \ref{subcase:1}.} Suppose that $\vec{e}_{0}' = \vec{e}_{1}'$, $\vec{h}^{0} \neq \vec{h}^{1}$ and $\vec{h}_{\sP}^{0} = \vec{h}_{\sP}^{1}$. We have
	\begin{align*}
		\textup{Cov}^{(0)} &= \frac{1}{2^{2k}} \; \mathbb{P}_{\Hmbis}\left( \vec{h}^{0}_{\sP} + \ev_{0}' \in \CCbis, \vec{h}^{1}_{\sP} + \ev_{1}' \in \CCbis \right) - \left(\frac{1}{ 2^{k+s-\kbis}}\right)^{2} \\
		&= \frac{1}{2^{2k}}  \; \mathbb{P}_{\Hmbis}\left( \vec{h}^{0}_{\sP} + \ev_{0}' \in \CCbis \right) 	- \left(\frac{1}{ 2^{k+s-\kbis}}\right)^{2} \qquad (\mbox{as }\vec{h}^{0}_{\sP} + \ev_{0}'  = \vec{h}^{1}_{\sP} + \ev_{1}' )  \\
		&= \frac{1}{2^{2k}}  \; \frac{1}{2^{s-\kbis}} - \left(\frac{1}{ 2^{k+s-\kbis}}\right)^{2}
	\end{align*}
		where in the last line we used Equation \eqref{eq:case1}.
	Therefore, in that case 
	$$
	\textup{Cov}^{(0)} \leq   \frac{1}{2^{2k}} \; \frac{1}{2^{s-\kbis}}.
	$$

	{\bf \noindent Sub-case \ref{subcase:2}.} Suppose that $\vec{e}_{0}'=\vec{e}_{1}', \vec{h}^{0}\neq \vec{h}^{1}$ and $\vec{h}_{\sP}^{0} \neq \vec{h}_{\sP}^{1}$. We have 
	\begin{align}
		\textup{Cov}^{(0)} &= \frac{1}{2^{2k}} \;\mathbb{P}_{\Hmbis} \left(\vec{h}^{0}_{\sP} + \ev_{0}' \in \CCbis, \vec{h}^{1}_{\sP} + \ev_{1}' \in \CCbis \right) - \left(\frac{1}{ 2^{k+s-\kbis}}\right)^{2}  \nonumber \\
		&=  \frac{1}{2^{2k}} \;\mathbb{P}_{\Hmbis} \left(\vec{h}^{0}_{\sP} + \ev_{0}' \in \CCbis \right) \mathbb{P}_{\Hmbis} \left( \vec{h}^{1}_{\sP} + \ev_{1}' \in \CCbis \right) - \left(\frac{1}{ 2^{k+s-\kbis}}\right)^{2}\label{eq:cov2}
	\end{align}
	where in the last line we used that $\vec{h}^{0}_{\sP} + \ev_{0}'  \neq \vec{h}^{1}_{\sP} + \ev_{1}'$showing that these two vectors are linearly independent (we work in $\F_{2}$), and thus that both events are independent. But, as they are different from $\vec{0}$ (according to Equation \eqref{eq:case1}) we have for ($b \in \{0,1\}$)
	$$
	\mathbb{P}_{\Hmbis} \left( \vec{h}^{b}_{\sP} + \ev_{b}' \in \CCbis \right)  = \frac{1}{2^{s-k_{\textup{aux}}}}
	$$
	Therefore, plugging this in Equation \eqref{eq:cov2} leads to 
	$$
	\textup{Cov}^{(0)} = 0.
	$$

	{\bf \noindent Sub-case \ref{subcase:3}.} Suppose that $\vec{e}_{0}' \neq \vec{e}_{1}'$, $\vec{h}^{0} = \vec{h}^{1}$. In that case, 
	\begin{align*} 
		\textup{Cov}^{(0)} &=	\mathbb{P}_{\Hmbis} \left(\hv^0 \in \CC^\perp \right) \mathbb{P}_{\Hmbis} \left(\vec{h}^{0}_{\sP} + \ev_{0}' \in \CCbis, \vec{h}^{0}_{\sP} + \ev_{1}' \in \CCbis  \right)  - \left(\frac{1}{ 2^{k+s-\kbis}}\right)^{2}\\
		&=\frac{1}{2^{k}} \; \mathbb{P}_{\Hmbis} \left(\vec{h}^{0}_{\sP} + \ev_{0}' \in \CCbis \right) \mathbb{P}_{\Hmbis} \left(\vec{h}^{0}_{\sP} + \ev_{1}' \in \CCbis \right)  - \left(\frac{1}{ 2^{k+s-\kbis}}\right)^{2}\\
		&= \frac{1}{2^{k}} \; \left(\frac{1}{2^{s-\kbis}} \right)^2  - \left(\frac{1}{ 2^{k+s-\kbis}}\right)^{2}
	\end{align*} 
	where in the second equality we used that $\vec{h}_{\sP}^{0} + \vec{e}_{0}' \neq \vec{h}_{\sP}^{1} + \vec{e}_{1}'$ and are different from $\vec{0}$ (according to Equation \eqref{eq:case1}) which implies that both events are independent.  
	Therefore, in that case  
	\begin{equation}\label{eq:subcase03} 
	\textup{Cov}^{(0)} \leq   \frac{1}{2^{k + 2s-2\kbis}} 
	\end{equation}

	{\bf \noindent Sub-case \ref{subcase:4}.} Suppose that $\vec{e}_{0}' \neq \vec{e}_{1}'$, $\vec{h}^{0} \neq \vec{h}^{1}$ and $\vec{h}_{\sP}^{0} = \vec{h}_{\sP}^{1}$. In that case, 
	\begin{align*}
		\textup{Cov}^{(0)} &=\frac{1}{2^{2k}} \; \mathbb{P}_{\Hmbis} \left(\vec{h}^{0}_{\sP} + \ev_{0}' \in \CCbis, \vec{h}^{0}_{\sP} + \ev_{1}' \in \CCbis  \right) - \left(\frac{1}{ 2^{k+s-\kbis}}\right)^{2} \\
		&=\frac{1}{2^{2k}}\mathbb{P}_{\Hmbis} \left(\vec{h}^{0}_{\sP} + \ev_{0}' \in \CCbis \right) \mathbb{P}_{\Hmbis} \left(\vec{h}^{0}_{\sP} + \ev_{1}' \in \CCbis \right) - \left(\frac{1}{ 2^{k+s-\kbis}}\right)^{2} \\
		&=  \left(\frac{1}{2^{k+ s-\kbis}} \right)^2 - \left(\frac{1}{ 2^{k+s-\kbis}}\right)^{2} \\
		&= 0
	\end{align*}
	where in the second equality we used that $\vec{h}_{\sP}^{0} + \vec{e}_{0}' = \vec{h}_{\sP}^{1} + \vec{e}_{0}' \neq \vec{h}_{\sP}^{1} + \vec{e}_{1}'$ which implies that both events are independent.  
	\newline

	{\bf \noindent Sub-case \ref{subcase:5}.} Suppose that $\vec{e}_{0}' \neq \vec{e}_{1}'$, $\vec{h}^{0} \neq \vec{h}^{1}$, $\vec{h}_{\sP}^{0} \neq \vec{h}_{\sP}^{1}$ and $\hv_{\sP}^0 + \ev_0' = \hv_{\sP}^1 + \ev_1'$. We have
	\begin{align*}
		\textup{Cov}^{(0)} &=\frac{1}{2^{2k}}\; \mathbb{P}_{\Hmbis} \left(\vec{h}^{0}_{\sP} + \ev_{0}' \in \CCbis \right) - \left(\frac{1}{ 2^{k+s-\kbis}}\right)^{2} \\
		&=\frac{1}{2^{2k}}\; \mathbb{P}_{\Hmbis} \left(\vec{h}^{0}_{\sP} + \ev_{0}' \in \CCbis \right)  - \left(\frac{1}{ 2^{k+s-\kbis}}\right)^{2} \\
		&=  \frac{1}{2^{2k}} \; \frac{1}{2^{s-\kbis}}  - \left(\frac{1}{ 2^{k+s-\kbis}}\right)^{2}
	\end{align*}
	Therefore, 
	$$
	\textup{Cov}^{(0)} \leq \frac{1}{2^{2k + s - \kbis}} 
	$$

	{\bf \noindent Sub-case \ref{subcase:6}.} Suppose that $\vec{e}_{0}' \neq \vec{e}_{1}'$, $\vec{h}^{0} \neq \vec{h}^{1}$, $\vec{h}_{\sP}^{0} \neq \vec{h}_{\sP}^{1}$ and $\hv_{\sP}^0 + \ev_0' \neq \hv_{\sP}^1 + \ev_1'$.
	In that case we can write
	\begin{align*}
		\textup{Cov}^{(0)} &=\frac{1}{2^{2k}} \; \mathbb{P}_{\Hmbis} \left(\vec{h}^{0}_{\sP} + \ev_{0}' \in \CCbis, \vec{h}^{1}_{\sP} + \ev_{1}' \in \CCbis  \right) - \left(\frac{1}{ 2^{k+s-\kbis}}\right)^{2} \\
		&=\frac{1}{2^{2k}} \; \mathbb{P}_{\Hmbis} \left(\vec{h}^{0}_{\sP} + \ev_{0}' \in \CCbis \right) \mathbb{P}_{\Hmbis} \left(\vec{h}^{1}_{\sP} + \ev_{1}' \in \CCbis \right) - \left(\frac{1}{ 2^{k+s-\kbis}}\right)^{2} \\
		&=  \left(\frac{1}{2^{k+ s-\kbis}} \right)^2 - \left(\frac{1}{ 2^{k+s-\kbis}}\right)^{2}
	\end{align*}
	Therefore we obtain,  
	$$
	\textup{Cov}^{(0)} = 0.
	$$

		{\bf\noindent Case \ref{case2}:} Recall that in this case we have 
	\begin{equation}\label{eq:case2}  
		\vec{h}^{0}_{\sP} = \vec{e}_{0}' \quad \mbox{and} \quad \vec{h}^{1}_{\sP} \neq \vec{e}_{1}'
	\end{equation}  
	We have, 
	\begin{equation*} 
		\mathbb{E}\left( \mathbf{1}_{\vec{e}_{0}',\vec{h}^{0}} \right) \mathbb{E}\left( \mathbf{1}_{\vec{e}_{1}',\vec{h}^{1}} \right) =  \frac{1}{2^{k}}  \; \frac{1}{2^{k+s-\kbis}}= \frac{1}{2^{2k+s-\kbis}}
	\end{equation*} 
	Let us compute $\mathbb{E}_{\vec{H},\Hmbis}\left( \mathbf{1}_{\vec{e}_{0}',\vec{h}^{0}} \; \mathbf{1}_{\vec{e}_{1}',\vec{h}^{1}} \right)$ when $(\vec{e}_{0}',\vec{h}^{0}) \neq (\vec{e}_{1}',\vec{h}^{1})$. By definition
	\begin{align*}
		\mathbb{E}_{\vec{H},\Hmbis}\left( \mathbf{1}_{\vec{e}_{0}',\vec{h}^{0}} \; \mathbf{1}_{\vec{e}_{1}',\vec{h}^{1}} \right) &= \mathbb{P}_{\vec{H},\Hmbis}\left( \hv^{0},\vec{h}^{1} \in \CC^{\bot}, \vec{h}^{0}_{\sP} + \ev_{0}' \in \CCbis, \vec{h}^{1}_{\sP} + \ev_{1}' \in \CCbis \right) \\
		&= \mathbb{P}_{\vec{H},\Hmbis}\left( \hv^{0},\vec{h}^{1} \in \CC^{\bot}\right) \mathbb{P}_{\vec{H},\Hmbis} \left(  \vec{h}^{0}_{\sP} + \ev_{0}' \in \CCbis, \vec{h}^{1}_{\sP} + \ev_{1}' \in \CCbis\right)
	\end{align*}
	Therefore,
	\begin{align} 
		\sum_{\vec{e}_{0}',\vec{e}_{1}'} F_{1} &= \sum_{\vec{e}_{0}',\vec{e}_{1}'} \mathbb{E}\left( \mathbf{1}_{\vec{e}_{0}',\vec{h}^{0}} \; \mathbf{1}_{\vec{e}_{1}',\vec{h}^{1}} \right) -  \frac{1}{2^{2k+s-\kbis}}\nonumber \\
		&= \sum_{\vec{e}_{0}',\vec{e}_{1}'}\sum_{\substack{\vec{h}^{0},\vec{h}^{1} \in \sF_{1}^{\vec{e}_{0}',\vec{e}_{1}'} \\ (\vec{e}_{0}',\vec{h}^{0}) \neq (\vec{e}_{1}',\vec{h}^{1}) }}  \textup{Cov}^{(1)} \label{eq:cov1} 
	\end{align}
	where, 
		\begin{equation*}
	\textup{Cov}^{(1)} \eqdef \mathbb{P}_{\vec{H}}\left( \hv^{0},\vec{h}^{1} \in \CC^{\bot}\right) \mathbb{P}_{\Hmbis} \left( \vec{h}_{\sP}^{0} + \ev_{0}' \in \CC_{\textup{aux}}, \vec{h}_{\sP}^{1} + \ev_{1}' \in \CC_{\textup{aux}} \right) - \frac{1}{ 2^{2k+s-\kbis}}
	\end{equation*} 
	Our aim now it to upper-bound $\textup{Cov}^{(1)}$ according to the above $6$ sub-cases.
	\newline

	{\bf \noindent Sub-case \ref{subcase:1}.} Suppose that $\vec{e}_{0}' = \vec{e}_{1}'$, $\vec{h}^{0} \neq \vec{h}^{1}$ and $\vec{h}_{\sP}^{0} = \vec{h}_{\sP}^{1}$. This subcase is impossible according to Equation \eqref{eq:case2}. Therefore, 
	$$
	\textup{Cov}^{(1)} =0.
	$$

	{\bf \noindent Sub-case \ref{subcase:2}.} Suppose that $\vec{e}_{0}'=\vec{e}_{1}', \vec{h}^{0}\neq \vec{h}^{1}$ and $\vec{h}_{\sP}^{0} \neq \vec{h}_{\sP}^{1}$. We have 
	\begin{align}
		\textup{Cov}^{(1)} &= \frac{1}{2^{2k}} \;\mathbb{P}_{\Hmbis} \left(\vec{h}^{0}_{\sP} + \ev_{0}' \in \CCbis, \vec{h}^{1}_{\sP} + \ev_{1}' \in \CCbis \right) - \frac{1}{ 2^{2k+s-\kbis}}  \nonumber \\
		&=  \frac{1}{2^{2k}}\; \mathbb{P}_{\Hmbis} \left( \vec{h}^{1}_{\sP} + \ev_{1}' \in \CCbis \right) - \frac{1}{ 2^{2k+s-\kbis}}\label{eq:cov22}
	\end{align}
	where in the last line we used that $\vec{h}^{0}_{\sP} + \ev_{0}'  \neq \vec{h}^{1}_{\sP} + \ev_{1}'$ showing that these two vectors are linearly independent (we work in $\F_{2}$), and thus that both events are independent. Furthermore, we also used that $\vec{h}^{0}_{\sP} + \vec{e}_{0}' = \vec{0}$ according to Equation \eqref{eq:case2}. But,
	$$
	\mathbb{P}_{\Hmbis} \left( \vec{h}^{1}_{\sP} + \ev_{1}' \in \CCbis \right)  = \frac{1}{2^{s-k_{\textup{aux}}}}
	$$
	Therefore, plugging this in Equation \eqref{eq:cov22} leads to 
	$$
	\textup{Cov}^{(1)} = 0.
	$$

	{\bf \noindent Sub-case \ref{subcase:3}.} Suppose that $\vec{e}_{0}' \neq \vec{e}_{1}'$, $\vec{h}^{0} = \vec{h}^{1}$. In that case, 
	\begin{align*} 
		\textup{Cov}^{(1)} &=	\mathbb{P}_{\Hmbis} \left(\hv^0 \in \CC^\perp \right) \mathbb{P}_{\Hmbis} \left(\vec{h}^{0}_{\sP} + \ev_{0}' \in \CCbis, \vec{h}^{1}_{\sP} + \ev_{1}' \in \CCbis  \right)  - \frac{1}{ 2^{2k+s-\kbis}}\\
		&=\frac{1}{2^{k}} \; \mathbb{P}_{\Hmbis} \left(\vec{h}^{0}_{\sP} + \ev_{0}' \in \CCbis \right) \mathbb{P}_{\Hmbis} \left(\vec{h}^{0}_{\sP} + \ev_{1}' \in \CCbis \right)  - \frac{1}{ 2^{2k+s-\kbis}} \\
		&= \frac{1}{2^{k}} \;\frac{1}{2^{s-\kbis}}  - \frac{1}{ 2^{2k+s-\kbis}}
	\end{align*} 
	where in the second equality we used that $\vec{h}_{\sP}^{0} + \vec{e}_{0}' \neq \vec{h}_{\sP}^{1} + \vec{e}_{1}'$ which implies that both events are independent. Furthermore, we also used that $\vec{h}^{0}_{\sP} + \vec{e}_{0}' = \vec{0}$ according to Equation \eqref{eq:case2}.
	Therefore, in that case  
	\begin{equation}\label{eq:subcase13}
	\textup{Cov}^{(1)} \leq \frac{1}{2^{k + s-\kbis}}
	\end{equation}

	{\bf \noindent Sub-case \ref{subcase:4}.} Suppose that $\vec{e}_{0}' \neq \vec{e}_{1}'$, $\vec{h}^{0} \neq \vec{h}^{1}$ and $\vec{h}_{\sP}^{0} = \vec{h}_{\sP}^{1}$. In that case, 
	\begin{align*}
		\textup{Cov}^{(1)} &=\frac{1}{2^{2k}} \; \mathbb{P}_{\Hmbis} \left(\vec{h}^{0}_{\sP} + \ev_{0}' \in \CCbis, \vec{h}^{0}_{\sP} + \ev_{1}' \in \CCbis  \right) - \frac{1}{ 2^{2k+s-\kbis}} \\
		&=\frac{1}{2^{2k}} \mathbb{P}_{\Hmbis} \left(\vec{h}^{0}_{\sP} + \ev_{1}' \in \CCbis \right) - \frac{1}{ 2^{2k+s-\kbis}}\\
		&= \frac{1}{ 2^{2k+s-\kbis}} - \frac{1}{ 2^{2k+s-\kbis}} \\
		&= 0
	\end{align*}
	where in the second equality we used that $\vec{h}_{\sP}^{0} + \vec{e}_{0}' = \vec{h}_{\sP}^{1} + \vec{e}_{0}' \neq \vec{h}_{\sP}^{1} + \vec{e}_{1}'$ which implies that both events are independent. Furthermore, we also used that $\vec{h}^{0}_{\sP} + \vec{e}_{0}' = \vec{0}$ according to Equation \eqref{eq:case2}.
	\newline

	{\bf \noindent Sub-case \ref{subcase:5}.} Suppose that $\vec{e}_{0}' \neq \vec{e}_{1}'$, $\vec{h}^{0} \neq \vec{h}^{1}$, $\vec{h}_{\sP}^{0} \neq \vec{h}_{\sP}^{1}$ and $\hv_{\sP}^0 + \ev_0' = \hv_{\sP}^1 + \ev_1'$. This subcase is impossible according to Equation \eqref{eq:case2}. Therefore,
	$$
	\textup{Cov}^{(1)} = 0
	$$

	{\bf \noindent Sub-case \ref{subcase:6}.} Suppose that $\vec{e}_{0}' \neq \vec{e}_{1}'$, $\vec{h}^{0} \neq \vec{h}^{1}$, $\vec{h}_{\sP}^{0} \neq \vec{h}_{\sP}^{1}$ and $\hv_{\sP}^0 + \ev_0' \neq \hv_{\sP}^1 + \ev_1'$.
	In that case we can write
	\begin{align*}
		\textup{Cov}^{(0)} &=\frac{1}{2^{2k}} \; \mathbb{P}_{\Hmbis} \left(\vec{h}^{0}_{\sP} + \ev_{0}' \in \CCbis, \vec{h}^{1}_{\sP} + \ev_{1}' \in \CCbis  \right) - \frac{1}{ 2^{2k+s-\kbis}}\\
		&=\frac{1}{2^{2k}} \; \mathbb{P}_{\Hmbis} \left(\vec{h}^{1}_{\sP} + \ev_{1}' \in \CCbis \right) - \frac{1}{ 2^{2k+s-\kbis}} \\
		&=  \frac{1}{ 2^{2k+s-\kbis}} - \frac{1}{ 2^{2k+s-\kbis}}
	\end{align*}
	Therefore we obtain,  
	$$
	\textup{Cov}^{(1)} = 0.
	$$	
	{\bf\noindent Case \ref{case3}:} This situation is symmetric to Case $2$.  
	\newline

	{\bf\noindent Case \ref{case4}:} Recall that in this case we have 
	\begin{equation}\label{eq:case4}  
		\vec{h}^{0}_{\sP} = \vec{e}_{0}' \quad \mbox{and} \quad \vec{h}^{1}_{\sP} = \vec{e}_{1}'
	\end{equation}  
	We have, 
	\begin{equation*} 
		\mathbb{E}\left( \mathbf{1}_{\vec{e}_{0}',\vec{h}^{0}} \right) \mathbb{E}\left( \mathbf{1}_{\vec{e}_{1}',\vec{h}^{1}} \right) =  \frac{1}{2^{2k}} 
	\end{equation*} 
	Let us compute $\mathbb{E}_{\vec{H},\Hmbis}\left( \mathbf{1}_{\vec{e}_{0}',\vec{h}^{0}} \; \mathbf{1}_{\vec{e}_{1}',\vec{h}^{1}} \right)$ when $(\vec{e}_{0}',\vec{h}^{0}) \neq (\vec{e}_{1}',\vec{h}^{1})$. By definition
	\begin{align*}
		\mathbb{E}_{\vec{H},\Hmbis}\left( \mathbf{1}_{\vec{e}_{0}',\vec{h}^{0}} \; \mathbf{1}_{\vec{e}_{1}',\vec{h}^{1}} \right) &= \mathbb{P}_{\vec{H},\Hmbis}\left( \hv^{0},\vec{h}^{1} \in \CC^{\bot}, \vec{h}^{0}_{\sP} + \ev_{0}' \in \CCbis, \vec{h}^{1}_{\sP} + \ev_{1}' \in \CCbis \right) \\
		&= \mathbb{P}_{\vec{H},\Hmbis}\left( \hv^{0},\vec{h}^{1} \in \CC^{\bot}\right) \mathbb{P}_{\vec{H},\Hmbis} \left(  \vec{h}^{0}_{\sP} + \ev_{0}' \in \CCbis, \vec{h}^{1}_{\sP} + \ev_{1}' \in \CCbis\right)
	\end{align*}
	Therefore,
	\begin{align} 
		\sum_{\vec{e}_{0}',\vec{e}_{1}'} F_{3} &= \sum_{\vec{e}_{0}',\vec{e}_{1}'}\sum_{\substack{\vec{h}^{0},\vec{h}^{1} \in \sF_{3}^{\vec{e}_{0}',\vec{e}_{1}'} \\ (\vec{e}_{0}',\vec{h}^{0}) \neq (\vec{e}_{1}',\vec{h}^{1}) }}  \mathbb{E}\left( \mathbf{1}_{\vec{e}_{0}',\vec{h}^{0}} \; \mathbf{1}_{\vec{e}_{1}',\vec{h}^{1}} \right) -  \frac{1}{2^{2k}}\nonumber \\
		&= \sum_{\vec{e}_{0}',\vec{e}_{1}'} \sum_{\substack{\vec{h}^{0},\vec{h}^{1} \in \sF_{3}^{\vec{e}_{0}',\vec{e}_{1}'} \\ (\vec{e}_{0}',\vec{h}^{0}) \neq (\vec{e}_{1}',\vec{h}^{1}) }}\textup{Cov}^{(3)}\label{eq:cov3} 
	\end{align}
	where, 
		\begin{equation*}
	\textup{Cov}^{(3)} \eqdef \mathbb{P}_{\vec{H}}\left( \hv^{0},\vec{h}^{1} \in \CC^{\bot}\right) \mathbb{P}_{\Hmbis} \left( \vec{h}_{\sP}^{0} + \ev_{0}' \in \CC_{\textup{aux}}, \vec{h}_{\sP}^{1} + \ev_{1}' \in \CC_{\textup{aux}} \right) - \frac{1}{2^{2k}}
	\end{equation*} 
	Our aim now it to upper-bound $\textup{Cov}^{(1)}$ according to the above $6$ sub-cases.
	\newline

	{\bf \noindent Sub-case \ref{subcase:1}.} Suppose that $\vec{e}_{0}' = \vec{e}_{1}'$, $\vec{h}^{0} \neq \vec{h}^{1}$ and $\vec{h}_{\sP}^{0} = \vec{h}_{\sP}^{1}$. We have
	\begin{align*}
		\textup{Cov}^{(3)} &= \frac{1}{2^{2k}} \; \mathbb{P}_{\Hmbis}\left( \vec{h}^{0}_{\sP} + \ev_{0}' \in \CCbis, \vec{h}^{1}_{\sP} + \ev_{1}' \in \CCbis \right) - \frac{1}{2^{2k}} \\
		&= \frac{1}{2^{2k}}  \; \mathbb{P}_{\Hmbis}\left( \vec{0} \in \CCbis \right) 	- \frac{1}{2^{2k}}\\
		&= \frac{1}{2^{2k}}  - \frac{1}{ 2^{2k}}
	\end{align*}
	where in the second equality we used Equation \eqref{eq:case4}.
	Therefore, in that case 
	$$
	\textup{Cov}^{(3)} = 0.
	$$
	
	{\bf \noindent Sub-case \ref{subcase:2}.} Suppose that $\vec{e}_{0}'=\vec{e}_{1}', \vec{h}^{0}\neq \vec{h}^{1}$ and $\vec{h}_{\sP}^{0} \neq \vec{h}_{\sP}^{1}$. According to Equation \eqref{eq:case4} this sub-case is impossible. Therefore, 
	$$
	\textup{Cov}^{(3)} = 0.
	$$

	{\bf \noindent Sub-case \ref{subcase:3}.} Suppose that $\vec{e}_{0}' \neq \vec{e}_{1}'$, $\vec{h}^{0} = \vec{h}^{1}$. According to Equation \eqref{eq:case4} this sub-case is impossible. Therefore, 
	$$
	\textup{Cov}^{(3)} = 0.
	$$

	{\bf \noindent Sub-case \ref{subcase:4}.} Suppose that $\vec{e}_{0}' \neq \vec{e}_{1}'$, $\vec{h}^{0} \neq \vec{h}^{1}$ and $\vec{h}_{\sP}^{0} = \vec{h}_{\sP}^{1}$. According to Equation \eqref{eq:case4} this sub-case is impossible. Therefore, 
	$$
	\textup{Cov}^{(3)} = 0.
	$$

	{\bf \noindent Sub-case \ref{subcase:5}.} Suppose that $\vec{e}_{0}' \neq \vec{e}_{1}'$, $\vec{h}^{0} \neq \vec{h}^{1}$, $\vec{h}_{\sP}^{0} \neq \vec{h}_{\sP}^{1}$ and $\hv_{\sP}^0 + \ev_0' = \hv_{\sP}^1 + \ev_1'$. We have
	\begin{align*}
		\textup{Cov}^{(3)} &=\frac{1}{2^{2k}}\; \mathbb{P}_{\Hmbis} \left(\vec{h}^{0}_{\sP} + \ev_{0}' \in \CCbis \right) - \frac{1}{2^{2k}} \\
		&=\frac{1}{2^{2k}}\; \mathbb{P}_{\Hmbis} \left( \vec{0} \in \CCbis \right)  - \frac{1}{2^{2k}} \\
		&=  \frac{1}{2^{2k}}  - \frac{1}{2^{2k}}
	\end{align*}
	where in the second equality we used Equation \eqref{eq:case4}. 
	Therefore, 
	$$
	\textup{Cov}^{(3)} = 0
	$$

	{\bf \noindent Sub-case \ref{subcase:6}.} Suppose that $\vec{e}_{0}' \neq \vec{e}_{1}'$, $\vec{h}^{0} \neq \vec{h}^{1}$, $\vec{h}_{\sP}^{0} \neq \vec{h}_{\sP}^{1}$ and $\hv_{\sP}^0 + \ev_0' \neq \hv_{\sP}^1 + \ev_1'$. According to Equation \eqref{eq:case4} this sub-case is impossible. Therefore, 
	$$
	\textup{Cov}^{(3)} = 0.
	$$

	We are now ready to gather Cases \ref{case1}, \ref{case2}, \ref{case3} and \ref{case4} according to Subcases \ref{subcase:1}, \ref{subcase:2}, \ref{subcase:3}, \ref{subcase:4}, \ref{subcase:5} and \ref{subcase:6}. Our aim is to bound 
	$
	\textup{Cov}
	$
	that were defined in Equation \ref{eq:cov}.
	We can already notice that Case \ref{case4} has no impact on this sum while Cases \ref{case2} and \ref{case3} have an influence only in Subcase \ref{subcase:3}. Furthermore, \ref{subcase:2}, \ref{subcase:4} and \ref{subcase:6} have no contribution to this sum, whatever is the considered case.

	Let us upper-bound $\textup{Cov}$ according to the different subcases where $\textup{Cov}_{i}$ denotes the terms involved in $\textup{Cov}$ coming from Subcase $i$ (in particular $\textup{Cov}_{i}$ is defined as a certain sum of $\textup{Cov}^{(i)}$, see Equations \eqref{eq:cov0}, \eqref{eq:cov1} and \eqref{eq:cov2}).
	\newline

	{\bf \noindent Subcase \ref{subcase:1}:} We have,  
		\begin{align*}
		\textup{Cov}_{1} 
		&\leq \sum_{(\vec{e}_{0}',\vec{h}^{0}) \in \mathcal{E}_{b}} \sum_{\substack{\vec{h}^{1}: (\vec{e}_{0}',\vec{h}^{1})\in \mathcal{E}_{b} \\ \vec{h}^{1} \neq \vec{h}^{0}, \vec{h}^{0}_{\sP} = \vec{h}^{1}_{\sP}}} \frac{1}{2^{2k}} \frac{1}{2^{s-\kbis}}
		&\leq \sum_{(\vec{e}_{0},\vec{h}^{0}) \in \mathcal{E}_{b}} \frac{\binom{n-s}{w}}{2^{2k + s - \kbis}} 
		&= \frac{\binom{n-s}{w}}{2^{k}} \frac{E_{b}'}{2^{k+s-\kbis}}
	\end{align*}

	{\bf \noindent Subcase \ref{subcase:2}:} We have,
	$$
	\textup{Cov}_{2} = 0. 
	$$

	{\bf \noindent Subcase \ref{subcase:3}:} We have here to split the computation here between Cases \ref{case1} and \ref{case2}. Recall that they are given by
		$$\sF_{0}^{\vec{e}_{0}',\vec{e}_{1}'} = \left\{(\vec{h}^{0},\vec{h}^{1}): \;  \vec{h}^{0} \in \sE_{b,1}^{\vec{e}_{0}'} \mbox{ and }  \vec{h}^{1} \in \sE_{b,1}^{\vec{e}_{1}'}\right\},
	$$
	$$
	\sF_{1}^{\vec{e}_{0}',\vec{e}_{1}'} = \left\{(\vec{h}^{0},\vec{h}^{1}): \;  \vec{h}^{0} \in \sE_{b,2}^{\vec{e}_{0}'} \mbox{ and }  \vec{h}^{1} \in \sE_{b,1}^{\vec{e}_{1}'}\right\}
	$$
	where,
	$$
		\sE_{b,1}^{\vec{e}'} \eqdef \left\{ \vec{h} \in \mathbb{F}_{2}^{n}: \; (\vec{e}',\vec{h}) \in \sE_{b} \mbox{ and } \vec{h}_{\sP} \neq \vec{e}' \right\},
	$$
	$$
		\sE_{b,2}^{\vec{e}'} \eqdef \left\{ \vec{h} \in \mathbb{F}_{2}^{n}: \; (\vec{e}',\vec{h}) \in \sE_{b} \mbox{ and } \vec{h}_{\sP} = \vec{e}' \right\}.
	$$
	But recall according to Equation \eqref{eq:cardEb12} that, 
	$$
	\card{\sE_{b,1}} = \frac{\card{\sE_{b,2}}}{2^{s}}
	$$  
	Therefore, in Subcase \ref{subcase:3},
	\begin{align*} 
	 \textup{Cov}_{3} &= \sum_{\vec{e}_{0}'\neq\vec{e}_{1}'}\sum_{\substack{\vec{h}^{0}=\vec{h}^{1} \in \sF_{0}^{\vec{e}_{0}',\vec{e}_{1}'}}} \textup{Cov}^{(0)} + 2\sum_{\vec{e}_{0}'\neq\vec{e}_{1}'}\sum_{\substack{\vec{h}^{0}=\vec{h}^{1} \in \sF_{1}^{\vec{e}_{0}',\vec{e}_{1}'}  }} \textup{Cov}^{(1)}\\ 
	 &= \sum_{\vec{e}_{0}'\neq\vec{e}_{1}'}\sum_{\substack{\vec{h}^{0}=\vec{h}^{1} \in \sF_{0}^{\vec{e}_{0}',\vec{e}_{1}'} }} \frac{1}{2^{k+2s-2\kbis}} + 2\sum_{\vec{e}_{0}'\neq\vec{e}_{1}'}\sum_{\substack{\vec{h}^{0}=\vec{h}^{1} \in \sF_{1}^{\vec{e}_{0}',\vec{e}_{1}'} }} \frac{1}{2^{k+s-\kbis}} \quad \mbox{(Equations \eqref{eq:subcase03}) and \eqref{eq:subcase13}} \\
	 &=  \sum_{\vec{e}_{0}'\neq\vec{e}_{1}'}\sum_{\substack{\vec{h}^{0}=\vec{h}^{1} \in \sF_{0}^{\vec{e}_{0}',\vec{e}_{1}'} }} \frac{1}{2^{k+2s-2\kbis}}  + \frac{1}{2^{s}}  \sum_{\vec{e}_{0}'\neq\vec{e}_{1}'}\sum_{\substack{\vec{h}^{0}=\vec{h}^{1} \in \sF_{0}^{\vec{e}_{0}',\vec{e}_{1}'} }} \frac{1}{2^{k+s-\kbis}} \\  
	 &= \OO{1}  \sum_{\vec{e}_{0}'\neq\vec{e}_{1}'}\sum_{\substack{\vec{h}^{0}=\vec{h}^{1} \in \sF_{0}^{\vec{e}_{0}',\vec{e}_{1}'}}} \frac{1}{2^{k+2s-2\kbis}} 
	\end{align*} 
	where we basically use the same reasoning than for proving Equation \eqref{eq:expE12b}. There in this subcase,
	\begin{align*} 
		\textup{Cov}_{3}
		&\leq \sum_{(\vec{e}_{0}',\vec{h}^{0}) \in \mathcal{E}_{b}} \sum_{\substack{ (\vec{e}_{1}',\vec{h}^{0})\in \mathcal{E}_{b} \\ \vec{h}^{1} = \vec{h}^{0}, \ev_0' \neq \ev_1' }} \frac{1}{2^{k+2s-2\kbis}} 
		&\leq \sum_{(\vec{e}_{0}',\vec{h}^{0}) \in \mathcal{E}_{b}} \frac{\binom{s}{t_{\textup{aux}}}}{2^{k + 2s - 2\kbis}} 
		&= \frac{\binom{s}{t_{\textup{aux}}}}{2^{s-\kbis}} \; \frac{E_{b}'}{2^{k+s-\kbis}}
	\end{align*}
	\newline

	{\bf \noindent Subcase \ref{subcase:4}:} We have,
	$$
	\textup{Cov}_{4} = 0. 
	$$

	{\bf \noindent Subcase \ref{subcase:5}:} We have,
			\begin{align*}
		\textup{Cov}_{5} &\leq \sum_{(\vec{e}_{0}',\vec{h}^{0}) \in \mathcal{E}_{b}} \sum_{\substack{(\vec{e}_{1}',\vec{h}^{1})\in \mathcal{E}_{b} \\ \vec{h}^{1} \neq \vec{h}^{0},\; \vec{h}^{0}_{\sP} \neq \vec{h}^{1}_{\sP}, \;\ev_0' \neq \ev_1' \\ \ev_0' + \hv^0_{\sP} =  \ev_1' + \hv^1_{\sP}}}   \frac{1}{ 2^{2k+s-\kbis}}
		\leq  \sum_{(\vec{e}_{0},\vec{h}^{0}) \in \mathcal{E}_{b}}  \frac{\binom{s}{t_{\textup{aux}}} \binom{n-s}{w}}{ 2^{2k+s-\kbis}} \\
		&= \frac{\binom{n-s}{w}\binom{s}{t_{\textup{aux}}}}{ 2^{k}} \; \frac{E_b'}{2^{k+s-\kbis}}
	\end{align*}

	{\bf \noindent Subcase \ref{subcase:6}:} We have,
	$$
	\textup{Cov}_{6} = 0.
	$$

		Plugging all these bounds on the $\textup{Cov}_{i}$ together and using that $\textup{Cov} = \sum_{i} \textup{Cov}_{i}$ leads to 
	\begin{align*} 
		\textup{Cov} &\leq \left( \frac{\binom{n-s}{w}  }{2^{k}} +  \frac{\binom{s}{t_{\textup{aux}}}}{2^{s - \kbis}}+ \frac{\binom{n-s}{w}  \binom{s}{t_{\textup{aux}}}  }{2^{k}}\right)\; \frac{E_{b}'}{2^{k+s-\kbis}} \\
		&= \OO{n^{\alpha}} \frac{E_{b}'}{2^{k+s-\kbis}} 
	\end{align*} 
	where in the last lines we used the constraints \eqref{eq:constraintsVar} given in the proposition. Plugging this equation in  Equation \eqref{eq:ineqVariant} concludes the proof.
	\end{proof}

We are now ready to prove our proposition:
\begin{proof}[Proof of Proposition \ref{prop:biasCodedRLPN}] Let $E_{b}$ and $E_{b}'$ (for $b \in \{0,1\}$) be defined as in Lemma \ref{lemma:varianceEb}. 
	By using the Bienaym\'e-Tchebychev inequality, we obtain for any function $f$ mapping the positive integers to positive real numbers:
	\begin{equation*}
		\mathbb{P}_{\vec{H},\Hmbis}\left( |E_b- \mathbb{E}\left( E_b \right)|\geq  \sqrt{f(n) \mathbb{E}\left(E_b\right)}\right) \leq  \frac{\mathbf{Var}(E_b)}{f(n) \mathbb{E}(E_b)} 
		= \frac{\OO{n^{\alpha}}}{f(n)}
	\end{equation*}
	where the last inequality is a consequence of Lemma \ref{lemma:varianceEb}.		
	Since,
	$$
	\bias_{\left(\hv,\cvbis\right) \drawn \widetilde{\sH}}\left(  \langle \cvbis + \hv_{\sP},\ev_{\sP} \rangle +  \langle \ev_{\sN},  \hv_{\sN}\rangle \right) = \frac{E_0-E_1}{E_0+E_1},
	$$ 
	we have with probability greater than
	$1 - \frac{\OO{n^{\alpha}}}{f(n)}$ that
	\begin{multline}\label{eq:complicated}
		\frac{\mu_0-\mu_1 - \sqrt{2f(n)} \sqrt{ \mu_0+\mu_1} }{\mu_0+\mu_1 + \sqrt{2f(n)} \sqrt{ \mu_0+\mu_1}}	 \leq \bias_{\left(\hv,\cvbis\right) \drawn \widetilde{\sH}}\left(  \langle \cvbis + \hv_{\sP},\ev_{\sP} \rangle +  \langle \ev_{\sN},  \hv_{\sN}\rangle \right) \\
		  \leq  
		\frac{\mu_0-\mu_1 + \sqrt{2f(n)}\sqrt{ \mu_0+\mu_1} }{\mu_0+\mu_1 - \sqrt{2f(n)}\sqrt{ \mu_0+\mu_1}}	 \end{multline}
	where 
	$$
	\mu_i \eqdef \mathbb{E}\left(E_i\right)
	$$	
	and where we used that for all positive $x$ and $y$, $\sqrt{x}+\sqrt{y} \leq \sqrt{2(x+y)}$. 
	Let, 
	$$
	N = \mu_{0} + \mu_{1} 
	$$
	It is readily seen that,
	$$
	N = \frac{\binom{n-s}{w} \binom{s}{\tbis}}{2^{k-\kbis}} 
	$$
	We let 
	$f(n) = \delta\sqrt{N}/2$. Since $N= \mu_0+\mu_1$ this implies $f(n) = \delta\sqrt{\mu_0+\mu_1}/2$. By Equation \eqref{eq:cstPropo35}, note that $\frac{f(n)}{\OO{n^{\alpha}}}$ tends to infinity as 
	$n$ tends to infinity.  We notice that 
	\begin{align*}
		\sqrt{2f(n)} \sqrt{ \mu_0+\mu_1} &=  \delta^{1/2}(\mu_0+\mu_1)^{3/4} \nonumber \\
		&= o\left( \delta(\mu_0+\mu_1)\right) 
	\end{align*}
	because 
	\begin{equation*}
		\frac{\delta^{1/2}(\mu_0+\mu_1)^{3/4}}{\delta(\mu_0+\mu_1)} 
		= 	\frac{1}{\sqrt{\delta\sqrt{\mu_0+\mu_1}}} 
		=  \frac{1}{\sqrt{2f(n)}} 
		\mathop{\longrightarrow}\limits_{n\to +\infty}  0.
	\end{equation*}
	Equation \eqref{eq:complicated} can now be rewritten as
	\begin{multline}
		\label{eq:simpler}
		\frac{\mu_0-\mu_1 - o\left(\delta(\mu_0+\mu_1)\right) }{\mu_0+\mu_1 + o\left(\delta(\mu_0+\mu_1)\right) }	 \leq \bias\left( \langle \ev'',\vec{e}_{\sP} \rangle +  \langle \vec{e}_{\sN},  \vec{h}_{\sN}\rangle  \right) \\
		 \leq  
		\frac{\mu_0-\mu_1 + o\left(\delta(\mu_0+\mu_1)\right)  }{\mu_0+\mu_1 - o\left(\delta(\mu_0+\mu_1)\right)}	
	\end{multline}Now on the other hand
	\begin{align*}
		\delta  &= \bias_{\left(\hv,\cvbis\right) \drawn \widetilde{\sH}}\left(  \langle \cvbis + \hv_{\sP},\ev_{\sP} \rangle +  \langle \ev_{\sN},  \hv_{\sN}\rangle \right)   \\
		&= \frac{E'_0-E'_1}{E'_0+E'_1} \\
		&= \frac{\frac{E'_0}{2^{k+s-\kbis}}-\frac{E'_1}{2^{k+s-\kbis}}}{\frac{E'_0}{2^{k+s-\kbis}}+\frac{E'_1}{2^{k+s-\kbis}}} \\
		&= \frac{\mu_0-\mu_1}{\mu_0+\mu_1}
	\end{align*}
	where the last equality is a consequence of Lemma \ref{lemma:varianceEb}, in particular Equation \eqref{eq:espL}. 
	From this it follows that we can rewrite \eqref{eq:simpler} as
	\begin{multline*}
		\frac{\delta}{1+o(\delta)}- o(\delta) \leq \bias_{\left(\hv,\cvbis\right) \drawn \widetilde{\sH}}\left(  \langle \cvbis + \hv_{\sP},\ev_{\sP} \rangle +  \langle \ev_{\sN},  \hv_{\sN}\rangle \right) \\
		   \leq 
		\frac{\delta}{1-o(\delta)}+ o(\delta) 
	\end{multline*}
	from which it follows immediately that 
	$$
	\bias_{\left(\hv,\cvbis\right) \drawn \widetilde{\sH}}\left(  \langle \cvbis + \hv_{\sP},\ev_{\sP} \rangle +  \langle \ev_{\sN},  \hv_{\sN}\rangle \right)  = \delta(1+o(1))
	$$
	which concludes the proof. 
\end{proof} \section{Correctness and Running-Time of the \nRLPNs algorithm (Algorithm \ref{alg:codedRLPN})}\label{app:complexity}
In this section we prove the correctness of Algorithm \ref{alg:codedRLPN} in Subsection \ref{subsec:Correc}. Furthermore, we give its running-time in Subsection \ref{sec:app_final_comp}. To this aim, we instantiate Instructions \ref{state:PCE} ($\Call{ParityCheckEquations}{}$) \ref{state:Dec} ($\Call{Decode}{}$) of Algorithm \ref{alg:sample}.
\begin{notation}
	\begin{itemize}
		\item Instantiation of Algorithm \ref{alg:codedRLPN}.
	\begin{itemize}
		\item We instantiate $\Call{ParityCheckEquations}{}$ instruction with the technique devised in \cite[\S 5]{CDMT22} to compute all the parity-checks of a given weight in a code (which is inspired from \cite{BJMM12}). Its asymptotic complexity is recalled in Proposition \ref{prop:asym_BJMM}.
		\item The family of auxiliary $[s,\kbis]$ linear codes $\CCbis$ used will be product of $\log s$ random codes as devised in Section \ref{sec:SDC} . Using these the $\Call{Decode}{}$ procedure outputs almost all codeword at distance $\tbis$ in time $2^{o(s)} \max\left(\frac{\binom{s}{\tbis}}{2^{s-\kbis}} , 1 \right)$.
\end{itemize}
	\item Framework for the analysis of Algorithm \ref{alg:codedRLPN}
	\begin{itemize}
			\item We prove the correctness (Proposition \ref{prop:correctness}) and we make the complexity analysis (Proposition \ref{prop:asym_comp_doubleRLPN}) in the framework of Proposition \ref{prop:exp_sizeS}. More specifically, analysis is made for 
$\CC$ and $\CCbis$ being random $\lbrack n,k \rbrack$ and $\lbrack s,k\bis \rbrack$ codes. We argue in Section \ref{sec:SDC} that the proof would be roughly the same (but more complicated) if we were to make it using $\CCbis$ being random product codes. The complexity of Algorithm \ref{alg:codedRLPN} would only grow by a factor of $2^{o(s)}$ when using these codes.
		\item Note that with the $\Call{ParityCheckEquations}{}$ and $\Call{Decode}{}$ procedures we have chosen, Algorithm \ref{alg:codedRLPN} computes in fact all the possible LPN samples, namely we have (as required in Proposition \ref{prop:exp_sizeS}):
		$$ \sH = \widetilde{\sH} $$
		\item We reuse notation introduced in Proposition \ref{prop:exp_sizeS}: the set $\sS^{(j)}$ of candidates for the $j$'th auxiliary code $\CCbis[(j)]$ is defined by
		\begin{equation} \label{eq:prop_def_S2}
			\sS^{(j)} \eqdef \{\sv \in \F_2^\kbis \: : \widehat{f_{\yv,\widetilde{\sH},\Gmbis[(j)] }  }\left(\sv\right) \geq  \frac{\delta}{2} \: \widetilde{H}\},
		\end{equation}
		where 
		\begin{equation}
			\widetilde{H} \eqdef \frac{\binom{n-s}{w} \binom{s}{\tbis}}{2^{k-\kbis}}, 	\qquad \delta \eqdef \frac{K_w^{(n-s)}\left(u\right) K_{\tbis}^{(s)}\left(t-u\right) }{\binom{n-s}{w} \binom{s}{\tbis}}.
		\end{equation}
	\end{itemize}
	\end{itemize}
\end{notation}
\subsection{Correctness of the algorithm}\label{subsec:Correc}
The goal of this section is to prove that  \nRLPN, namely Algorithm \ref{alg:codedRLPN}, outputs the desired error vector $\ev$ after essentially $\Niter \approx \frac{\binom{n}{t}}{\binom{s}{t-u} \binom{n-s}{u} }$ iterations of the outer loop (Line \eqref{lst:line:while} of Algorithm \ref{alg:codedRLPN}). This is given by the following proposition.
\begin{proposition} \label{prop:correctness}
	Let $\CC$ be a code taken uniformly at random among the $[n,k]$ linear codes and $\CCbis[(1)],\dots, \CCbis[(N_{\textup{aux}})]$ which are $N_{\textup{aux}}$ codes taken uniformly at random among the $[s,\kbis]$ linear codes.  Let $\yv \eqdef \cv + \ev$ where $\cv \in \CC$ and where $\ev \in \mathcal{S}_t^{n}$ is a fixed error vector of weight $t$.
	As long as the parameters $s,\kbis,\tbis,w,u$ verify the Parameters constraint \eqref{ass:parameters_doubleRLPN} and as long as $N_{\textup{iter}} = \omega \left(\frac{\binom{n}{t}}{\binom{s}{t-u} \binom{n-s}{u} }\right)$ and $N_{\textup{aux}} = \OO{1}$, Algorithm \ref{alg:codedRLPN} outputs the error vector $\ev$ with probability $1-o(1)$. \end{proposition}

It is readily seen that when $\Niter = \omega \left(\frac{\binom{n}{t}}{\binom{s}{t-u} \binom{n-s}{u} }\right)$ then, with probability $1-o(1)$ over the choice of $\sP$ there exists an iteration such that $|\ev_{\sN}|=u$. We only have left to show that for such an iteration we have with high probability that $\ev_{\sP} \left(\Gmbis[(j)]\right)^{\top} \in \sS^{(j)}$ for $j=1, \dots,\Naux$ which is the purpose of the following lemma.
\begin{lemma}
	Let us reuse the setting of Proposition \ref{prop:correctness}. Moreover, let us fix $\sP$ and $\sN$ two complementary sets of $\llbracket 1,n\rrbracket$ of size $s$ and $n-s$ respectively and such that $|\ev_{\sN} | = u$. Let us denote by  $\Gmbis[(1)] , \dots, \Gmbis[(N_{\textup{aux}})]$  the generators matrices of the codes $\CCbis[(1)] , \dots, \CCbis[(N_{\textup{aux}})]$ respectively. Then,\begin{equation} \label{eq:eq_correctness}
		\prob{}{\bigcap_{j=1}^{N_{\textup{aux}}} ``\ev_{\sP} \: \transp{\Gmbis[(j)]}\in \sS^{(j)}"} = 1-o(1).
	\end{equation}
\end{lemma}
\begin{proof}
	First, notice that
	\begin{align*}
		\prob{}{\bigcap_{j=1}^{\Naux}  ``\ev_{\sP} \: \transp{\Gmbis[(j)]}\in \sS^{(j)}"} &= 1 - \prob{}{\bigcup_{j=1}^{\Naux}  ``\ev_{\sP} \: \transp{\Gmbis[(j)]}\notin \sS^{(j)}"} \\
		&\geq 1 - \sum_{j=1}^{\Naux} \prob{}{``\ev_{\sP} \: \transp{\Gmbis[(j)]}\notin \sS^{(j)}"}
	\end{align*}
		where we used the union-bound. 
	 Now, as $\Naux = \OO{1}$ we only have to show that $\prob{}{``\ev_{\sP} \: \transp{\Gmbis[(j)]}\notin \sS^{(j)}"}=o(1)$ to prove Equation \eqref{eq:eq_correctness}. By using Fact  \ref{fact:fundamental1},
	\begin{equation*}
		"\ev_{\sP} \: \transp{\Gmbis[(j)]}\notin \sS^{(j)}" \Longleftrightarrow  \bias_{\left(\hv,\cvbis\right) \drawn \widetilde{\sH}} \left(\langle \yv, \hv \rangle  + \langle \ev_{\sP}, \cvbis \rangle\right)  < \frac{\delta}{2} \frac{\widetilde{H}}{\card{\widetilde{\sH}}}
\end{equation*}
	Our aim is to show,
	$$
	\prob{}{\bias_{\left(\hv,\cvbis\right) \drawn \widetilde{\sH}} \left(\langle \yv, \hv \rangle  + \langle \ev_{\sP}, \cvbis \rangle\right)  < \frac{\delta}{2} \frac{\widetilde{H}}{\card{\widetilde{\sH}}} } = o(1).
	$$ 
	To simplify notation let $b \eqdef \bias_{\left(\hv,\cvbis\right) \drawn \widetilde{\sH}} \left(\langle \yv, \hv \rangle  + \langle \ev_{\sP}, \cvbis \rangle\right)$.
It is readily seen that $\expect{\CC,\CCbis}{\card{\widetilde{\sH}}} = \widetilde{H} \left(1+o(1)\right)$ and $\mathbf{Var}_{\CC,\CCbis}\left( \card{\widetilde{\sH}} \right) \leq \widetilde{H}(1+o(1))$. Therefore, by using Bienaymé-Tchebychev inequality, for any $a$,
	$$
	\mathbb{P}\left( \left| \card{\widetilde{\sH}} - (1+o(1))\widetilde{H} \right| \geq \sqrt{a \widetilde{H}} \right) \leq \frac{1+o(1)}{a}
	$$
	We have the following computation,
	\begin{multline*}
			\prob{}{b < \frac{\delta}{2} \frac{\widetilde{H}}{\card{\widetilde{\sH}}}} = 	\mathbb{P}\left( b < \frac{\delta}{2} \frac{\widetilde{H}}{\card{\widetilde{\sH}}} \mid  \left| \card{\widetilde{\sH}} - (1+o(1))\widetilde{H}\right| \geq \sqrt{a\widetilde{H}} \right) 	\mathbb{P}\left( \left| \card{\widetilde{\sH}} - (1+o(1))\widetilde{H} \right| \geq \sqrt{a \widetilde{H}} \right) \\
			+ 	\mathbb{P}\left( b < \frac{\delta}{2} \frac{\widetilde{H}}{\card{\widetilde{\sH}}} \mid  \left| \card{\widetilde{\sH}} - (1+o(1))\widetilde{H}\right| < \sqrt{a\widetilde{H}} \right) 	\mathbb{P}\left( \left| \card{\widetilde{\sH}} - (1+o(1))\widetilde{H} \right| < \sqrt{a \widetilde{H}} \right)
	\end{multline*}
	Therefore,
	\begin{align*}
			\prob{}{b < \frac{\delta}{2} \frac{\widetilde{H}}{\card{\widetilde{\sH}}}} &\leq 	\mathbb{P}\left( \left| \card{\widetilde{\sH}} - (1+o(1))\widetilde{H} \right| \geq \sqrt{a \widetilde{H}} \right) + \mathbb{P}\left( b < \frac{\delta}{2} \frac{\widetilde{H}}{(1+o(1))\widetilde{H} - \sqrt{a\widetilde{H}} } \right) \\
			&\leq \frac{1+o(1)}{a} + \mathbb{P}\left( b < \frac{\delta}{2} \frac{1}{(1+o(1)) - \sqrt{\frac{a}{\widetilde{H}} }} \right)
	\end{align*}
	Let us choose $a = \widetilde{H}^{1/2}$. Recall that $\widetilde{H} = \om{\frac{n^{\alpha + 8}}{\delta^2}}$ where $\delta \leq 1$. Therefore, 
	\begin{align*} 
	\prob{}{b < \frac{\delta}{2} \frac{\widetilde{H}}{\card{\widetilde{\sH}}}}  &= o(1) + \mathbb{P}\left( b <  \frac{\delta}{2}(1+ o(1)) \right) \\
	&=o(1)
	\end{align*} 
	where in the last equality we used Proposition \ref{eq:cstPropo35}. It concludes the proof. 
\end{proof}
\subsection{Asymptotic complexity of \nRLPNs} \label{sec:app_final_comp}
We now have every tool to give the complexity of our algorithm, namely, we can compute the expected number of candidates at each iteration given by Proposition \ref{prop:exp_sizeS} and we have the correctness of our algorithm which is given by Proposition \ref{prop:correctness}.
\begin{proposition}{Asymptotic complexity exponent of the \nRLPNs algorithm.} \label{prop:asym_comp_doubleRLPN} Define
	\begin{align*}
	 &R \eqdef \lim\limits_{n \to \infty}\frac{k}{n}, \quad \tau \eqdef \lim\limits_{n \to \infty}\frac{t}{n},   \quad \sigma \eqdef  \lim\limits_{n \to \infty} \frac{s}{n}, \quad R_{\textup{aux}} \eqdef  \lim\limits_{n \to \infty} \frac{\kbis}{n}\\ &  \quad \tau_{\textup{aux}} \eqdef  \lim\limits_{n \to \infty} \frac{\tbis}{n}, \omega \eqdef  \lim\limits_{n \to \infty} \frac{w}{n}, \quad \mu \eqdef  \lim\limits_{n \to \infty} \frac{u}{n}
	\end{align*}
	Suppose that de $\Call{Decode}{}$ procedure has an expected time complexity of $2^{n \cdot \sigma \cdot o(1)}$.
	\noindent The expected complexity of the \nRLPNs algorithm to decode a code of rate $R$ at relative distance $\tau$ is upper bounded by
		$2^{n \: \left(\alpha_{\textup{\nRLPNs}} + o(1)\right)}$
where
\begin{eqnarray*}
\alpha_{\textup{\nRLPNs}}&\leq & - \pi + \max\left( \left(1-\sigma\right) \: \cdot \:\alpha_{\textup{eq}}\left(\frac{R-\sigma}{1-\sigma}, \frac{\omega}{1-\sigma}\right) , \; \nu_{\textup{sample}}, \; R_{\textup{aux}}, \; \nu_{\textup{candidate}} \: \cdot \:  N_{\textup{aux}} +  \alpha_{\textup{ISD}} \right)
	\end{eqnarray*}
	where
	\begin{eqnarray*}
	\pi & \eqdef & h \left(\tau \right) - \sigma  \: \cdot \: h_2 \left(\frac{\tau - \mu }{\sigma}\right) -  (1- \sigma)  \: \cdot \: h \left(\frac{\mu }{1-\sigma}\right), \\
	\nu_{\textup{samples}} &\eqdef& \left(1 - \sigma\right)  \: \cdot \: h_2 \left(\frac{\omega}{1-\sigma}\right)  + \sigma  \; \cdot \; h_2 \left(\frac{\tau_{\textup{aux}}}{\sigma}\right)- (R - R_{\textup{aux}}), \\
	\alpha_{\textup{ISD}} & \eqdef & \max\left(\sigma \: \cdot \: \alpha_{\textup{ISD-Dumer}}\left( 1 - \frac{N_{\textup{aux}} \cdot \tau_{\textup{aux}} }{\sigma}, \frac{\tau - \mu}{\sigma}\right), \nu_{\textup{ISD}} + \left(1- \sigma\right)\; \cdot \; \alpha_{\textup{ISD-Dumer}}\left(\frac{R-\sigma}{1- \sigma}, \frac{\mu}{1-\sigma}\right) \right), \\
	 \nu_{\textup{ISD}} &\eqdef& \max\left(\sigma  \: \cdot \: h_2 \left(\frac{\tau - \mu }{\sigma}\right) -N_{\textup{aux}} \: \cdot \: \tau_{\textup{aux}},\: 0  \right)\\
	 		\nu_{\textup{candidates}} &\eqdef& \max \left(\max_{(\eta,\zeta) \in \mathcal{A}} \sigma \: \cdot \: h_2\left(\frac{\zeta}{\sigma}\right) + (1-\sigma) \: \cdot \:  h_2 \left(\frac{\eta}{1-\sigma}\right) - (1-R), \; 0\right) ,
\end{eqnarray*}
with
\begin{multline*}
	\mathcal{A}  \eqdef \Bigg\{ (\eta,\zeta) \in \left[0, 1-\sigma \right]\times \left[ 0,\sigma \right] \: : \sigma \left[\widetilde{\kappa} \left(\frac{\tbis}{\sigma},\frac{\tau-\mu}{\sigma}\right)  - \widetilde{\kappa} \left(\frac{\tbis}{\sigma},\frac{\zeta}{\sigma}\right) \right] +  \\   (1-\sigma) \left[\widetilde{\kappa} \left(\frac{\omega}{1-\sigma},\frac{\mu}{1-\sigma}\right)  - \widetilde{\kappa} \left(\frac{\omega}{1-\sigma},\frac{\eta}{1-\sigma}\right)  \right] \leq 0\Bigg\},
\end{multline*}
and
\begin{itemize}
	\item $\alpha_{\textup{eq}}(R',\tau')$ is the complexity exponent of $\Call{ParityCheckEquations}{}$ to compute all parity-checks of relative weight $\tau'$ of a code of rate $R'$. It is instantiated here with a technique devised in \cite[\S 5 Eq. (5.4)]{CDMT22} and its complexity is recalled in Proposition \ref{prop:asym_BJMM}.
	\item  $\alpha_{\textup{ISD-Dumer}}(R',\tau')$ is the complexity exponent of $\Call{Decode-Dumer}{}$ to return all the solutions to the decoding problem in a code of rate $R'$ at relative distance $\tau'$. Its complexity is recalled in Proposition \ref{prop:comp_ISD_decoder_dumer_asym}.
\end{itemize}
Moreover, $\sigma, \: R_{\textup{aux}}, \: \tau_{\textup{aux}}, \: \omega, \: \mu,$ are non-negative and such that
$$
\sigma \leq R, \quad \tau- \sigma \leq \mu\leq \tau, \quad \omega \leq 1 - \sigma,
$$
\begin{align}
\nu_{\textup{samples}} &\geq -2 \: \varepsilon_{\textup{bias}} , \\
0 &\geq (1-\sigma) \:  h_2 \left(\frac{\omega}{1-\sigma}\right) + \sigma h_2 \left(\frac{\tau_{\textup{aux}}}{\sigma}\right)  - R \\
0 &\geq \sigma \:  h_2 \left(\frac{\tau_{\textup{aux}}} {\sigma}\right) - (\sigma - R_{\textup{aux}})
\end{align}
where $\widetilde{\kappa}$ is the function defined in Proposition \ref{prop:expansion} and $$\varepsilon_{\textup{bias}} \eqdef  \: \sigma \: \left[ \widetilde{\kappa} \left(\frac{\tbis}{\sigma},\frac{\tau-\mu}{\sigma}\right) - h_2 \left(\frac{\tau_{\textup{aux}}}{\sigma}\right)\right] + \: (1-\sigma) \: \left[\widetilde{\kappa} \left(\frac{\omega}{1-\sigma},\frac{\mu}{1-\sigma}\right) - h_2 \left(\frac{\omega}{1-\sigma}\right)\right].$$ 
Finally, we require that $N_{\textup{aux}}=\OO{1}$.
\end{proposition}
\begin{remark}
	In practice our parameters are such that we decode the auxiliary code at Gilbert-Varshamov distance, namely $\tau_{\textup{aux}} = \sigma \: h_2^{-1} \left(1-R_{\textup{aux}}\right)$.
\end{remark}
While initially Dumer's decoder \cite{D91} is designed to produce only one solution to the decoding problem it suffices to re-run it as many time as the number of solutions we expect from the decoding problem to find all of them. We get the following proposition giving the asymptotic complexity of the $\Call{Decode-Dumer}{}$ procedure.
\begin{proposition}[Asymptotic time complexity of ISD Decoder \cite{D91} to produce \textbf{all} solutions to the decoding problem] \label{prop:comp_ISD_decoder_dumer_asym}
	Let
	$ R \eqdef \lim\limits_{n \to \infty}\frac{k}{n}, \tau \eqdef \lim\limits_{n \to \infty}\frac{t}{n}$. Let $\ell$ and $w$ be two (implicit) parameters of the algorithm and define $ \lambda \eqdef \lim\limits_{n \to \infty}\frac{\ell}{n}, \omega \eqdef \lim\limits_{n \to \infty}\frac{w}{n}$.
	The time and space complexities of \cite{D91} to find a proportion $1-o(1)$ of all solutions to the decoding problem at distance $t$ on an $[n,k]$ linear code are given by $2^{n \: \left(\alpha_{\textup{ISD-Dumer}} + o(1)\right)}$ and $2^{n \: \left(\beta_{\textup{ISD-Dumer}}+ o(1)\right)}$ respectively where
	\begin{align}
		\alpha_{\textup{ISD-Dumer}} &\eqdef \min_{\omega,\lambda} \left( \pi +  \max\left(\frac{R + \lambda}{2}h_2 \left(\frac{\omega}{ R+\lambda}\right) , \left(R + \lambda\right)h_2 \left(\frac{\omega}{ R+\lambda}\right) - \lambda\right) \right), \\
		\pi &\eqdef h_2(\tau) - (1- R-\lambda) h_2 \left(\frac{\tau - \omega}{1- R-\lambda}\right)-  (R + \lambda) h_2 \left(\frac{\omega}{R+\lambda}\right), \\
		\nu_{\textup{sol}} &\eqdef \max \left(h_2\left(\tau\right) - (1-R), 0 \right), \\
			\beta_{\textup{ISD-Dumer}} &\eqdef \frac{R + \lambda}{2}h_2 \left(\frac{\omega}{R+\lambda}\right).
	\end{align}
	Moreover $\lambda$ and $\omega$ must verify the following constraints:
	$$ 0 \leq \lambda \leq 1- R, \qquad \max \left(R + \lambda + \tau - 1, 0\right)\leq \omega \leq \min \left(\tau,R + \lambda \right).$$
	The expected number of solutions is given by $2^{n \: \left(\nu_{\textup{sol}} + o(1)\right)}$.
\end{proposition}
We recall here the asymptotic complexity of the technique devised in \cite[\S 5, Equation (5.4)]{CDMT22} based on \cite{BJMM12} decoder to produce all parity-checks of low weight of a code.
\begin{proposition}{Asymptotic time complexity exponent of BJMM technique \cite{BJMM12}, \cite[\S 5, Equation (5.4)]{CDMT22} to produce all parity-checks of relative weight $\omega$ of a code of rate $R$} \label{prop:asym_BJMM}
	\begin{equation}
		\alpha_{\text{BJMM}}(R,\omega) \eqdef \min_{\pi_1,\pi_2, \lambda_1,\lambda_2} \gamma
	\end{equation}
	where
	\begin{align*}
		&\gamma \eqdef \max(\gamma_1,\gamma_2,\gamma_3) \\
			&\gamma_1 \eqdef \max\left(\nu_0, \: 2 \: \nu_0 - \lambda_1 \right), \quad \gamma_2\eqdef  \max\left(\nu_1, \: 2 \: \nu_1 - \left( \lambda_2 - \lambda_1\right) \right), \quad	\gamma_3 \eqdef  \max\left(\nu_2, \: 2 \: \nu_2 - \left( \lambda - \lambda_2\right) \right) \\
		&\nu_0 \eqdef \frac{h(\pi_1)}{2}, \quad \nu_1 \eqdef h(\pi_1) - \lambda_1, \quad\nu_2 \eqdef h(\pi_2) - \lambda_2, \quad \nu_3 \eqdef h(\omega) - \lambda.
	\end{align*}

	and the constraint region $\sR$ is defined by the sub-region of nonnegative tuples $(\pi_1,\pi_2,\lambda_1,\lambda_2)$ such that
	$$
	\lambda_1 \leq \lambda_2 \leq \lambda, \quad \pi_1 \leq \pi_2 \leq \pi, \quad
	\pi_2 \leq 2 \pi_1, \quad \pi \leq 2\pi_2, \quad 
	\pi_2 \leq \lambda_1, \quad \pi \leq \lambda_2, 
	$$
	and
	\begin{eqnarray}
		\lambda_1 & = &  \pi_2 + (1 - \pi_2)h\left( \frac{\pi_1-\pi_2/2}{1-\pi_2}\right),\\
		\lambda_2 & = &  \omega + (1 - \omega)h\left( \frac{\pi_2-\omega/2}{1-\omega}\right)
	\end{eqnarray}
		The expected number of parity-checks computed is given by $2^{n \: \left(\nu_3 + o(1) \right)}$
\end{proposition}

 \section{Proof of Proposition \ref{prop:exp_bias}} \label{app:proofBias}

Let us recall Proposition \ref{prop:exp_bias}.
\propDualityBias*
Let us devise a more convenient expression for $\langle \yv, \hv \rangle  + \langle \xv, \cvbis \rangle$. By noticing that $\hv_{\sP}$ and $\hv_{\sN}$ are linearly linked we get the following lemma.
\begin{lemma} \label{lem:expr_error_doubleRLPN}
	Let $\sP$ and $\sN$ be two complementary subsets of $\llbracket 1 , n \rrbracket$ of size $s$ and $n-s$ respectively. Let $\CC$ and $\CCbis$ be two $[n,k]$ linear and $[s,\kbis]$ linear codes respectively such that $\CC_{\sP}$ is of dimension $s$. Let $\xv \in \F_2^s$ and $\left(\hv,\cvbis\right) \in \widetilde{\sH}$ where recall that 
	$$ 	\widetilde{\sH} \eqdef \{ \left(\hv,\cvbis\right) \in \CC^{\perp} \times \CCbis \: : |\hv_{\sN}|=w \mbox{ and } |\hv_{\sP} + \cvbisperp|=\tbis \}. $$
	We have that \begin{equation} \label{rel:noise_doubleRLPN}
		\langle \yv, \hv \rangle + \langle \xv, \cvbis \rangle = \langle \left(\xv + \ev_{\sP} \right) \Rm + \ev_{\sN}, \hv_{\sN}\rangle + \langle \xv , \hv_{\sP} + \cvbis \rangle
	\end{equation} 
	where $\Rm \in \F_2^{s \times (n-s)}$ is (independently of the parity-check $\vec{h}$) such that
	\begin{equation}\label{rel:parity}
		\hv_{\sP}= \hv_{\sN} \Rm^{\transpose}.
	\end{equation}
\end{lemma}
\begin{proof}
	First, let us show Equation \eqref{rel:parity}. Suppose without loss of generality that $\sP = \llbracket 1,s \rrbracket$ and $\sN = \llbracket s+1,n \rrbracket$. Let $\Gm \in \F_2^{k \times n}$ be a generator matrix of $\CC$. Because $\CC_{\sP}$ is of dimension $s$ there exists an invertible $\Jm \in \F_2^{k \times k}$ such that
	$$\Jm \Gm =  \begin{pmatrix}
		\vec{Id}_{s} & \Rm \\
		\vec{0}_{k-s} & \Rm'
	\end{pmatrix} $$
	where $\Rm \in \F_2^{s \times (n-s)}$ and $\Rm' \in \F_2^{(k-s)\times (n-s)}$. Furthermore, $\Jm \Gm$ is another generator matrix for $\CC$. Therefore for any $\hv \in \CC^{\perp}$ we have $\Jm \Gm \: \transp{\hv} = \vec{0}$. Since $\Jm \Gm \transp{\hv}= \transp{\hv_{\sP}}+\Rm \transp{\hv_{\sN}}$, this gives \eqref{rel:parity}.
	Now, let us prove \eqref{rel:noise_doubleRLPN}.
	Recall that, using Equation \eqref{rel:parity} we have:
	\begin{align*}
		\langle \yv, \hv\rangle &= \langle \ev_{\sP}, \hv_{\sP}\rangle + \langle \ev_{\sN}, \hv_{\sN}\rangle \\
		&= \langle \ev_{\sP}, \hv_{\sN} \Rm^{\top} \rangle + \langle \ev_{\sN}, \hv_{\sN}\rangle  \\
		&= \langle \ev_{\sP} \Rm + \ev_{\sN}, \hv_{\sN}\rangle,
	\end{align*}
	and
	\begin{align*}
		\langle \xv, \cvbis \rangle &=  \langle \xv, \hv_{\sP}\rangle +  \langle \xv,  \hv_{\sP} + \cvbis \rangle \\
		&= \langle \xv \Rm , \hv_{\sN} \rangle + \langle \xv ,  \hv_{\sP} + \cvbis \rangle
	\end{align*}
	where in the last equality we used Equation \eqref{rel:parity}. 
	This concludes the proof.
\end{proof}

\begin{proof}[Proof of Proposition \ref{prop:exp_bias}]
	Let us consider $\Rm \in \F_2^{s \times (n-s)}$ as in Lemma \ref{lem:expr_error_doubleRLPN} and let us prove Equation \eqref{eq:prop_bias_doubleRLPN}.
	By definition of the bias and $\widetilde{\sH}$ given Equation \eqref{rel:parity}  we have the following computation and ,
	\begin{align*}
	\bias_{\left(\hv,\cvbis \right)) \drawn \widetilde{\sH}}&   \left(	\langle \yv, \hv \rangle + \langle \xv, \cvbis \rangle \right) \\
	&= \frac{1}{\card{\widetilde{\sH}}} \sum_{ \left( \hv, \cvbis \right) \in \widetilde{\sH}} (-1)^{\langle \yv, \hv \rangle + \langle \xv, \cvbis \rangle} \\
	&=  \frac{1}{\card{\widetilde{\sH}}} \sum_{ \substack{ \left( \hv_{\sN}, \cvbis\right) \in \left(\CC^{\perp}\right)_{\sN} \times \CCbis  \\   \left|\hv_{\sN}\right| = w ,\left| \Rm \hv_{\sN} + \cvbis \right| = \tbis}} (-1)^{\langle \left(\xv + \ev_{\sP} \right) \Rm + \ev_{\sN}, \hv_{\sN}\rangle + \langle \xv , \hv_{\sN} \transp{\Rm} + \cvbis \rangle} 
	\end{align*}
	where in the last equality we used Lemma \ref{lem:expr_error_doubleRLPN}. Therefore, 
	\begin{equation}\label{eq:bias} 
		\bias_{\left(\hv,\cvbis \right)) \drawn \widetilde{\sH}} = \frac{1}{\card{\widetilde{\sH}}} \sum_{ \substack{ \left( \hv_{\sN}, \cvbis\right) \in \left(\CC^{\perp}\right)_{\sN} \times \CCbis  } } f(\hv_{\sN}, \cvbis) 
	\end{equation}
	where,
	$$
	f(\hv_{\sN}, \cvbis) \eqdef   (-1)^{\langle \left(\xv + \ev_{\sP} \right) \Rm + \ev_{\sN}, \hv_{\sN} \rangle + \langle \xv ,  \hv_{\sN} \transp{\Rm} + \cvbis \rangle} \mathbf{1}_{ \{\left|\hv_{\sN}\right| = w ,\left| \hv_{\sN} \transp{\Rm} + \cvbis\right| = \tbis \}}.
	$$
Using Equation \ref{eq:shortPunct} we have that $\left(\CC^{\perp}\right)_{\sN} =  \left(\CC^{\sN}\right)^{\perp}$ and thus $\left( \left(\CC^{\sN}\right)^{\perp} \times \CCbis \right)^{\perp} =\CC^{\sN} \times \CCbisperp $. By using the Poisson formula
	(see \cite[Lemma 2, Ch. 5.2]{MS86}), together with the fact that $\dim\left(\CC^{\sN} \times \CCbisperp \right) = k-\kbis$, 
	we get
	\begin{align}\label{eq:Poisson} 
		\sum_{ \substack{ \left( \hv_{\sN}, \cvbis\right) \in \left(\CC^{\sN} \right)^{\perp} \times \CCbis  } } f(\hv_{\sN}, \cvbis) &= \frac{1}{2^{k-\kbis}}\sum_{ \substack{ \left( \cv^{\sN}, \cvbisperp\right) \in \CC^{\sN} \times \CCbisperp  } } \widehat{f}(\cv^{\sN}, \cvbisperp).
	\end{align}
	Let us compute the right-hand term. By definition of $f$, it is readily seen that 
	\begin{align*} 
	\widehat{f}(\vec{y}_{1},\vec{y}_{2}) &=\sum_{\substack{\vec{z}_{1}\in \F_{2}^{n-1},\vec{z}_{2}\in \F_{2}^{s} \\ |\vec{z}_{1}| = w, |\vec{z}_{1}\vec{R}^{\top}+\vec{z}_{2}| = \tbis } } (-1)^{\langle \vec{y}_{1},\vec{z}_{1} \rangle + \langle \vec{y}_{2},\vec{z}_{2}\rangle} (-1)^{\langle \left(\xv + \ev_{\sP} \right) \Rm + \ev_{\sN}, \vec{z}_{1} \rangle + \langle \xv ,  \vec{z}_{1} \transp{\Rm} + \vec{z}_{2} \rangle} \\ 
	&= \sum_{\vec{z}_{1}\in \F_{2}^{n-s}:|\vec{z}_{1}|=w}  (-1)^{\langle \vec{y}_{1}+\left(\xv + \ev_{\sP}+\vec{y}_{2} \right) \Rm + \ev_{\sN}, \vec{z}_{1} \rangle} \sum_{\substack{\vec{z}_{2}\in \F_{2}^{s}: |\vec{z}_{1}\transp{\vec{R}}+\vec{z}_{2}}| = \tbis } (-1)^{\langle\vec{y}_{2}+\vec{x},\vec{z}_{1}\transp{\vec{R}} +\vec{z}_{2} \rangle} \\
	&= K^{(n-s)}_{w}\left(\left| \vec{y}_{1} + (\vec{y}_{2}+\vec{x}+\vec{e}_{\sP})\vec{R}+\vec{e}_{\sN} \right| \right) \; K^{(s)}_{\tbis}\left( \left|\vec{y}_{2}+\vec{x} \right| \right)
	\end{align*} 
	where in the last equality we used Fact \ref{fact:krawtchouk_bias}. Plugging this into Equation \eqref{eq:Poisson}  and then into Equation \eqref{eq:bias} concludes the proof. 
\end{proof} \section{Proof of Proposition \ref{prop:exp_sizeS}} \label{app:appendix_sec_tools}
The proof of this Appendix is to prove Proposition \ref{prop:exp_sizeS} which we recall is given by
\propExpSizeS*
The proof is divided in the following steps.
\newline

{\bf \noindent Step 1:} in Lemma \ref{lem:link_S_bias} we show that the expected size of $\sS$ is related to the probability that the bias of $\langle \yv, \hv \rangle  + \langle \xv, \hv_{\sP} \rangle$ is superior to the threshold $\approx \frac{\delta}{2}$.
\newline

{\bf \noindent Step 2:}  We give an exponential bound on the aforementioned probability by using Poisson summation formula as it was done in the proof of Proposition \ref{prop:exp_bias}.

\subsubsection*{\textbf{Step 1}} \label{sec:link_S_bias}
Recall that we have from Equation $\ref{eq:prop_def_S}$ that 
$$ 	
\sS \eqdef \left\{\sv \in \F_2^\kbis \: : \widehat{f_{\yv,\widetilde{\sH},\Gmbis}}\left(\sv\right) \geq  \frac{\delta}{2} \: \widetilde{H}\right\}. 
$$
By using Lemma \ref{lem:fundamental} we get the following condition for an element $\sv$ to be a candidate:
\begin{fact} \label{fact:fundamental1} 	 Let $\sv \in \F_2^{\kbis}$ and $\xv \in \F_2^{s}$ such that $\xv \Gmbis[\transpose] = \sv$. We have,
$$
	\sv \in \sS \; \Longleftrightarrow  \bias_{\left(\hv,\cvbis\right) \drawn \widetilde{\sH}}\left(\langle \yv,\hv \rangle + \langle \xv,\cvbis \rangle\right)  \geq \frac{\delta}{2} \; \frac{\widetilde{H}}{\card{\widetilde{\sH}}}.
	$$
\end{fact}
From there, we can derive the following lemma linking the expected size of $\sS$ and the previous bias. 
\begin{lemma} \label{lem:link_S_bias}
	Under Distribution \ref{notation:randomCodes}, using notation of Proposition \ref{prop:exp_sizeS} and under the constraint of Proposition \ref{prop:exp_sizeS} that $\sH = \widetilde{\sH}$ we have
	$$ 
	\mathbb{E}_{\CC,\CCbis}\; \left(\card{\sS}\right) \leq 2^{\kbis} \prob{\CC,\CCbis,\xv}{\bias_{(\hv,\cvbis)\drawn \widetilde{\sH}}(\langle \yv, \hv \rangle  + \langle \xv, \cvbis \rangle) \geq \frac{\delta}{2} \; \frac{\widetilde{H}}{\card{\widetilde{\sH}}} } + 1$$
	where $\xv$ is taken uniformly at random in $\F_2^s \setminus \{\CCbisperp + \ev_{\sP}\}$.
\end{lemma}
\begin{proof} We have the following computation, 
	\begin{align*} 
		\expect{\CC,\CCbis}{\card{\sS}}&= \expect{}{\sum_{\sv \in \F_2^{\kbis}} \mathbf{1}_{\sv \in \sS}} \\
		&\leq1 + \expect{}{\sum_{\sv \in \F_2^{\kbis} \: : \sv \neq \ev_{\sP} \Gmbis[\transpose]} \mathbf{1}_{\sv \in \sS}} \\
		&=1 + \expect{}{\sum_{\sv \in \F_2^{\kbis} \: : \sv \neq \ev_{\sP} \Gmbis[\transpose]} \frac{1}{2^{s-\kbis}} \sum_{ \zv \in \F_2^s \: : \zv \: \Gmbis[\transpose] = \sv}\mathbf{1}_{\zv \Gmbis[\transpose] \in \sS}}
	\end{align*} 
where in the last equality we used that $\Gmbis$, which is a generator matrix of $\CCbis$, has rank $s-\kbis$ according to Distribution \ref{notation:randomCodes}.
Now, from the linearity of the expectation we get,
	\begin{align*}
		\expect{\CC,\CCbis}{\card{\sS}}&\leq  1 + \frac{1}{2^{s-\kbis}} \sum_{\zv \in \F_2^s \: : \zv \notin \CCbisperp + \ev_{\sP}} \prob{\CC,\CCbis}{\zv \Gmbis[\transpose] \in \sS} \\
&= 1 + 2^{\kbis} \:  \prob{\CC,\CCbis,\xv}{\xv \Gmbis[\transpose] \in \sS} 
	\end{align*}
	where in the last line we used that $\Gmbis$ has full rank, $\card{\CCbisperp + \ev_{\sP}} = 2^{s}$ and $\xv$ is taken uniformly at random in $\F_2^s \setminus \{\CCbisperp + \ev_{\sP}\}$. Using Fact \ref{fact:fundamental1} concludes the proof.
\end{proof}
\subsubsection*{Step 2}\label{sec:distrib_bias} The following lemma relates the upper-bound given in Lemma \ref{lem:link_S_bias} to the involved probability in Conjecture \ref{ass:bias_dom}. 
\begin{lemma} \label{lem:simpl_bias}
		Using Distribution \ref{notation:randomCodes} and notation of Proposition \ref{prop:exp_sizeS} we have
		\begin{multline*}
\prob{\CC,\CCbis,\xv}{\bias_{\left(\hv,\cvbis \right) \drawn \widetilde{\sH}} \left(\langle \yv, \hv \rangle  + \langle \xv, \cvbis \rangle\right) \geq \frac{\delta}{2} \frac{\widetilde{H}}{\card{\widetilde{\sH}}}  }   = \\ \prob{\CC,\CCbis,\xv}{ \sum_{j = 0}^{s} \sum_{i = 0}^{n-s} K_{\tbis}^{(s)} \left(j\right) K_w^{(n-s)} \left(i\right) N_{i,j} \geq \frac{1}{2} \: K_w^{(n-s)}\left(u\right) K_{\tbis}^{(s)}\left(t-u\right) }.
		\end{multline*}
		where $\xv$ is taken uniformly at random in $\F_2^s \setminus \{ \CCbisperp + \ev_{\sP}\}$.
	\end{lemma}
	\begin{proof}
		Recall that we have that from notation of Proposition \ref{prop:exp_sizeS},
		\begin{equation} \label{lem:simpl_bias_bias_def}
			\widetilde{H} \eqdef \frac{\binom{n-s}{w} \binom{s}{\tbis}}{2^{k-\kbis}}, \qquad \delta \eqdef \frac{K_w^{(n-s)}(u) K_{\tbis}^{(s)}(t-u) }{\binom{n-s}{w} \binom{s}{\tbis}}.
		\end{equation}
		According to Proposition \ref{prop:exp_bias} we have the following computation,
		\begin{align*}
			\mathbb{P}&\left( \bias_{\left(\hv,\cvbis \right) \drawn \widetilde{\sH}} \left(\langle \yv, \hv \rangle  + \langle \xv, \cvbis \rangle\right) \geq \frac{\delta}{2} \frac{\widetilde{H}}{\card{\widetilde{\sH}}} \right) \\ 
			&= \prob{}{\frac{1}{2^{k-\kbis}} \frac{1}{\card{\widetilde{\sH}}} \sum_{j = 0}^{s} \sum_{i = 0}^{n-s} K_{\tbis}^{(s)} \left(j\right) K_w^{n-s} \left(i\right) N_{i,j} \geq \frac{\delta}{2} \frac{\widetilde{H}}{\card{\widetilde{\sH}}}  } \\
			&=   \prob{}{\sum_{j = 0}^{s} \sum_{i = 0}^{n-s} K_{\tbis}^{(s)} \left(j\right) K_w^{n-s} \left(i\right) N_{i,j} \geq \frac{K_w^{(n-s)}\left(u\right) K_{\tbis}^{(s)}\left(t-u\right)}{2}}
		\end{align*}
		where in the last equality we used Equation \eqref{lem:simpl_bias_bias_def}. It concludes the proof. 
	\end{proof}
We will now use Conjecture \ref{ass:bias_dom} to bound the right-hand term of Lemma \ref{lem:simpl_bias} by the probability of the  event ``$N_{i,j} \neq 0$'' for some low $i$ and $j$ (more precisely when $K_w^{(n-s)}(u) K_{\tbis}^{(s)}(t-u) \leq K_w^{(n-s)}(i) K_{\tbis}^{(s)}(j)$). 
\begin{restatable}{lemma}{lemprobNij}\label{lem:probNij}
		Under Distribution \ref{notation:randomCodes} we have that for $(i,j) \in \llbracket 0, n-s  \rrbracket \times \llbracket 0, s\rrbracket$,
		$$ 
		\prob{\CC,\CCbis,\xv}{N_{i,j} \neq 0} = \OO{ \frac{\binom{s}{j} \binom{n-s}{i}}{2^{\kbis + n-k}} }
		$$
		where $\xv$ is taken uniformly at random in $\F_2^s \setminus \{ \CCbisperp + \ev_{\sP}\}$.
	\end{restatable}
	\begin{proof}
		This is proved in the first lemma of Appendix \ref{app:prop_Nij}.
	\end{proof}
	We are now ready to prove Proposition \ref{prop:exp_sizeS}.
	\begin{proof}[proof of Proposition \ref{prop:exp_sizeS}]
		Recall that we want to show, 
		$$  \mathbb{E}_{\CC,\CCbis} \left(\card{\sS} \right)  = \OOt{ \max_{(i,j) \in \mathcal{A}} \frac{\binom{s}{j} \binom{n-s}{i}}{2^{n-k}}} + 1. 
			$$
		Lemma \ref{lem:link_S_bias} gives us that 
		\begin{equation}\label{eq:expBia} 
			\mathbb{E}_{\CC,\CCbis} \left(\card{\sS }\right) \leq 2^{\kbis} \prob{\CC,\CCbis,\xv}{\bias_{(\hv,\cvbis)\drawn \widetilde{\sH}}(\langle \yv, \hv \rangle  + \langle \xv, \cvbis \rangle) \geq \frac{\delta}{2} \; \frac{\widetilde{H}}{\card{\widetilde{\sH}}} } + 1
		\end{equation} 
		Now using Lemma \ref{lem:simpl_bias} we get that
		\begin{multline}\label{eq:biasConj}
			\prob{}{\bias_{\left(\hv,\cvbis \right) \drawn \widetilde{\sH}} \left(\langle \yv, \hv \rangle  + \langle \xv, \cvbis \rangle\right) \geq \frac{\delta}{2} \; \frac{\widetilde{H}}{\card{\widetilde{\sH}}}  }  = \\ \prob{}{ \sum_{j = 0}^{s} \sum_{i = 0}^{n-s} K_{\tbis}^{(s)} \left(j\right) K_w^{(n-s)} \left(i\right) N_{i,j} \geq \frac{1}{2} \: K_w^{(n-s)}\left(u\right) K_{\tbis}^{(s)}\left(t-u\right) }.
		\end{multline}
		From Conjecture \ref{ass:bias_dom},
		\begin{align*}
			\prob{}{ \sum_{j = 0}^{s} \sum_{i = 0}^{n-s} K_{\tbis}^{(s)} \left(j\right) K_w^{(n-s)} \left(i\right) N_{i,j} \geq \frac{1}{2} \: K_w^{(n-s)}\left(u\right) K_{\tbis}^{(s)}\left(t-u\right) } = \OOt{ \max_{(i,j) \in \mathcal{A}}\prob{\CC,\CCbis,\xv}{N_{i,j} \neq 0} }
		\end{align*}
		Therefore, using Lemma \ref{lem:probNij} we get 
		$$
		\prob{}{ \sum_{j = 0}^{s} \sum_{i = 0}^{n-s} K_{\tbis}^{(s)} \left(j\right) K_w^{(n-s)} \left(i\right) N_{i,j} \geq \frac{1}{2} \: K_w^{(n-s)}\left(u\right) K_{\tbis}^{(s)}\left(t-u\right) } = \OOt{ \frac{\binom{s}{j} \binom{n-s}{i}}{2^{\kbis + n-k}} }
		$$
		Plugging this into Equation \eqref{eq:biasConj} and then in \eqref{eq:expBia} concludes the proof. 
		which proves our result.
	\end{proof}
 \section{About the distribution of the $N_{i,j}$}\label{app:prop_Nij}
This appendix is dedicated to studying the distribution of $N_{i,j}$ and in particular to prove Lemma~\ref{lem:probNij} which was used in Appendix \ref{app:proofBias} to prove Proposition \ref{prop:exp_bias}. We recall that it is given by
 \lemprobNij*
 First, using the union bound we can devise the following upper bound on our target probability:
\begin{fact} \label{fact:Nij_toexpect}
	$$
	\prob{}{N_{i,j} \neq 0}\leq \mathbb{E}\left(N_{i,j}\right).
	$$
\end{fact}
To compute this expected value, we first need to give the following lemma giving the probability that a word belongs to a random code.
\begin{lemma} \label{lem:probaRandom}
	Let $\CC$ be chosen uniformly at random among the $[n,k]$ linear codes. Let $\cv \in \F_2^n \setminus \{ \vec{0}\}$ we have
	\begin{align}
		\prob{\CC}{\cv \in \CC} &= \frac{2^k - 1}{2^{n} - 1}, \\
		\prob{\CC}{\vec{0} \in \CC} &= 1.
	\end{align}
\end{lemma}
\begin{proof}
	This lemma directly follows from \cite[\S 3, Lemma 9.3.2, (iii)]{BCN89}.
\end{proof}
We now give the preliminary lemma which breaks down the expected value of $N_{i,j}$ on the $\CCbisperp$ part and on the $\CC^{\sN}$ part.
\begin{restatable}{lemma}{lemmaNbar} \label{lem:Nbar}
Under Distribution \ref{notation:randomCodes} we denote by
\begin{align*}
	\overline{N_j^{(\CCbisperp)} }&\eqdef \expect{\CC,\CCbis,\xv} { N_j \left(\CCbisperp + \xv\right)}
\end{align*}
where $\xv$ is taken uniformly at random in $\F_2^{s} \setminus \{\CCbisperp + \ev_{\sP}\}$ and $\vec{r} \in \CCbisperp + \vec{x}$. We denote by fixing~$\vec{x}$,
\begin{align}
	\overline{N_i^{(\CC^{\sN})} }&\eqdef \expect{\CC} {N_i \left( \CC^{\sN} + \left(\rv+\ev_{\sP}\right)\Rm + \ev_{\sN} \right) }
\end{align}
We have that
	\begin{align}
\overline{N_i^{(\CC^{\sN})}} &=\frac{\binom{n-s}{i}}{2^{n-k}} \left(1 + \OO{2^{- (s-k)}} \right) \label{eq:expected_ni_bar}\\
	 \overline{N_j^{(\CCbisperp)}}  &\leq  	 \frac{\binom{s}{j}}{2^{\kbis}} \left(1 + \OO{ 2^{-  \left(s-\kbis\right)}}\right) \label{eq:expected_nj_bar}
\end{align}
Furthermore, the term $\OO{2^{- (s-\kbis)}}$ in Equation \eqref{eq:expected_ni_bar} does not depend on $i$·
\end{restatable}
\begin{proof}

	Under Distribution \ref{notation:randomCodes}, $\CC$ is taken at random among the $[n,k]$-codes that are such that $\CC_{\sP}$ is of full rank dimension $s$. Therefore, it is the same as if $\CC$ was chosen by taking its generator matrix $\Gm \in \F_2^{k \times n}$ as follows: 
	
$$ 
				\Gm_{\sP} \eqdef \begin{pmatrix}
						\Id_{s} \\
						\vec{0}_{k-s}
					\end{pmatrix}, \qquad \Gm_{\sN} \eqdef \begin{pmatrix}
					 \Rm \\
					 \Gm^{\sN}
					\end{pmatrix}. 
					$$
				where $\Rm$ is chosen uniformly at random among matrices of $\F_2^{s \times (n-s)}$ and $\Gm^{\sN}$ is any generator matrix of a code chosen uniformly at random among the $[n-s,k-s]$-codes. In particular, $\Gm^{\sN}$ and $\vec{R}$ are independent.

		 		Now, by definition $\xv\in \F_2^{s} \setminus \{\CCbisperp + \ev_{\sP}\}$ and $\vec{r} \in \CCbisperp + \vec{x}$, therefore $\rv + \ev_{\sP} \neq \vec{0}$. We deduce that, $ \left(\rv+\ev_{\sP}\right)\Rm$ is uniformly distributed in $\F_2^{n-s}$ as a non-zero sum of uniformly distribution vectors.
To simplify the notations let us define the uniformly distributed vector $$\vv \eqdef  \left(\rv+ \ev_{\sP}\right)\Rm + \ev_{\sN}$$ which is independent of $\CC^{\sN}$ (by construction $\Gm^{\sN}$ and $\vec{R}$ are independent). Now, let us show Equation \eqref{eq:expected_ni_bar}. We have by linearity of the expected value that
	\begin{align*}
		\overline{N_i^{(\CC^{\sN})}} &= \sum_{\zv \in \mathcal{S}_i^{n-s}} \prob{ \CC,\vv}{\zv \in \CC^{\sN} + \vv} \\
		&= \sum_{\zv \in \mathcal{S}_i^{n-s}} \prob{ \CC,\vv}{\zv \in \CC^{\sN} + \vv \middle| \vv = \zv} \prob{\CC,\vv}{\vv = \zv} + \prob{ \CC,\vv}{\zv \in \CC^{\sN} + \vv \middle| \vv \neq \zv} \prob{\CC,\vv}{\vv \neq \zv} \\
		&=  \sum_{\zv \in \mathcal{S}_i^{n-s}} \frac{1}{2^{n-s}}  + \frac{2^{k-s} - 1}{2^{n-s} - 1} \left( 1 - \frac{1}{2^{n-s}} \right) \qquad \mbox{(Lemma \ref{lem:probaRandom})} \\
		&= \frac{\binom{n-s}{i}}{2^{n-k}} \left( \frac{2^{n-k}}{2^{n-s}} + 2^{n-k} \; \frac{2^{k-s} -1}{2^{n-s}-1} \OO{1} \right) \\
		&= \frac{\binom{n-s}{i}}{2^{n-k}} \left( \frac{1}{2^{k-s}} + \frac{2^{-s} - 2^{-k}}{2^{-s}-2^{-n}} \; \OO{1}\right) \\
		&= \frac{\binom{n-s}{i}}{2^{n-k}} \left( \frac{1}{2^{k-s}} + \frac{1- 2^{-k+s}}{1-2^{-n+s}} \OO{1} \right) \\
		&= \frac{\binom{n-s}{i}}{2^{n-k}} \left( 1 + \frac{1}{2^{k-s}} + \OO{2^{-k+s}} \right) \quad \mbox{($s \leq k \leq n$)} 
\end{align*}
	Let us now show Equation \eqref{eq:expected_nj_bar}. Recall that 
	$$
	\overline{N_j^{(\CCbisperp)}}  \eqdef \expect{\CCbis, \xv} {N_j \left(\CCbisperp + \xv \right)}
	$$ 
	where $\xv$ is taken uniformly at random in $\F_2^{s} \setminus \{\CCbisperp + \ev_{\sP} \}$. By definition,
	\begin{align*}
		\expect{\CCbis, \xv} {N_j \left(\CCbisperp + \xv \right)} &= \sum_{\vec{z} \in  \mathcal{S}_{j}^{s}} \mathbb{P}_{\CCbis, \xv} \left( \vec{z} \in \CCbisperp + \xv \right) \\
		&= \sum_{\vec{z} \in  \mathcal{S}_{j}^{s}, \vec{x}_{0} \in \F_2^{s}} \mathbb{P}_{\CCbis} \left( \vec{z} - \vec{x}_{0} \in \CCbisperp \right)\mathbb{P}_{\vec{x},\CCbis}\left( \vec{x} = \vec{x}_{0} \right) \\
		&\leq \sum_{\vec{z} \in  \mathcal{S}_{j}^{s}}\left(  \sum_{\vec{x}_{0} \in \F_2^{s}} \frac{2^{s-\kbis}-1}{2^{s}-1} \;  \mathbb{P}_{\vec{x},\CCbis}\left( \vec{x} = \vec{x}_{0} \right) + \mathbb{P}_{\vec{x},\CCbis}\left( \vec{x} = \vec{z} \right) \right) \\
		&\leq \binom{s}{j} \frac{2^{s-\kbis}-1}{2^{s}-1} + \frac{\binom{s}{j}}{2^{s} - 2^{\kbis}} \\
		&\leq \frac{\binom{s}{j}}{2^{\kbis}} \;   2^{\kbis} \; \frac{2^{s-\kbis}-1}{2^{s}-1} + \frac{\binom{s}{j}}{2^{\kbis}} \; \frac{2^{\kbis}}{2^{s} - 2^{\kbis}} \\
		&\leq \frac{\binom{s}{j}}{2^{\kbis}} \;   \frac{2^{s}-2^{\kbis}}{2^{s}-1} + \frac{\binom{s}{j}}{2^{\kbis}} \; \frac{1}{2^{s-\kbis} - 1}  \\
		&\leq \frac{\binom{s}{j}}{2^{\kbis}} \; \left( 1 + \OO{2^{\kbis-s}} \right)
	\end{align*}
	which completes the proof. 
\end{proof}
	We can now show that the expected value of $N_{i,j}$ is the product of the two previously computed quantities:
	\begin{lemma}  \label{fact:product_expect}
		Under Distribution \ref{notation:randomCodes} and when $\xv$ is taken uniformly at random in $\F_2^{s} \setminus \{ \CCbisperp + \ev_{\sP}\}$ we have
		$$ 
		\expect{\CC,\CCbis,\xv}{N_{i,j}} = \overline{N_j^{(\CCbisperp)}}  \; 	\overline{N_i^{(\CC^{\sN})}}
		$$
		where, 
		$$
		N_{i,j} \eqdef \card{ \left\{  \left(\rv,\cv^{\sN}\right) \in  (\xv +\CCbisperp)  \times \CC^{\sN} \: : \left| \rv \right|=j \mbox{ and } \left| \left( \rv + \ev_{\sP}\right) \Rm + \ev_{\sN} + \cv^{\sN} \right| = i\right\} }
		$$
	and,
		$$
			\overline{N_j^{(\CCbisperp)} }\eqdef \expect{\CC,\CCbis,\xv} { N_j \left(\CCbisperp + \xv\right)} \quad \mbox{;} \quad 	\overline{N_i^{(\CC^{\sN})} }\eqdef \expect{\CC} {N_i \left( \CC^{\sN} + \left(\rv+\ev_{\sP}\right)\Rm + \ev_{\sN} \right) }.
		$$
	\end{lemma}
	\begin{proof} By definition,
		$$
		N_{i,j} = \sum_{u = 0}^{N_j\left(\CCbisperp + \xv\right)} N_i \left(\left(\rv^{(u)} + \ev_{\sP}\right) \Rm + \ev_{\sN} + \CC^{\sN}\right)
		$$
		where $\rv^{(u)}$ is the $u$'th codeword of weight~$j$ of $\CCbisperp + \xv$ and $N_j(\CCbisperp + \xv)$ counts the number of elements in $\CCbisperp + \xv$ of Hamming weight~$j$. Therefore,
		\begin{align*}
		N_{i,j} &= \sum_{\vec{z}\in \mathcal{S}_{j}^{s}} N_i \left(\left(\vec{z}+ \ev_{\sP}\right) \Rm + \ev_{\sN} + \CC^{\sN}\right)  \;\mathds{1}_{``\vec{z} \in \CCbisperp + \xv''}  \\
			&= \sum_{\vec{z}\in \mathcal{S}_{j}^{s},\vec{w} \in \mathcal{S}_{i}^{n-s}}  \mathds{1}_{``\vec{w} \in \left(\vec{z}+ \ev_{\sP}\right) \Rm + \ev_{\sN} + \CC^{\sN}"}  \;\mathds{1}_{``\vec{z} \in \CCbisperp + \xv''}  \\
		\end{align*} 
		We deduce that,
		\begin{align*}
			\expect{\CC,\CCbis,\xv}{N_{i,j}} &= \sum_{\vec{z}\in \mathcal{S}_{j}^{s},\vec{w} \in \mathcal{S}_{i}^{n-s}} \mathbb{P}\left( \vec{w} \in \left(\vec{z}+ \ev_{\sP}\right) \Rm + \ev_{\sN} + \CC^{\sN}, \; \vec{z} \in \CCbisperp + \xv \right) \\
			 &= \sum_{\vec{z}\in \mathcal{S}_{j}^{s},\vec{w} \in \mathcal{S}_{i}^{n-s}} \mathbb{P}\left( \vec{w} \in \left(\vec{z}+ \ev_{\sP}\right) \Rm + \ev_{\sN} + \CC^{\sN} \mid \vec{z} \in \CCbisperp + \xv  \right) \mathbb{P}\left(  \vec{z} \in \CCbisperp + \xv \right) \\
			 &= \sum_{\vec{z}\in \mathcal{S}_{j}^{s}}\left( \sum_{\vec{w} \in \mathcal{S}_{i}^{n-s}} \mathbb{P}\left( \vec{w} \in \left(\vec{z}+ \ev_{\sP}\right) \Rm + \ev_{\sN} + \CC^{\sN} \mid \vec{z} \in \CCbisperp + \xv  \right) \right)\mathbb{P}\left(  \vec{z} \in \CCbisperp + \xv \right) \\
			 &= \sum_{\vec{z}\in \mathcal{S}_{j}^{s}} \overline{N_i^{(\CC^{\sN})} }\; \mathbb{P}\left(  \vec{z} \in \CCbisperp + \xv \right) \\
			 &=  \overline{N_i^{(\CC^{\sN})} } \; 	\overline{N_j^{(\CCbisperp)} }
		\end{align*}
		which completes the proof. 
\end{proof}
\begin{proof}[Proof of Lemma \ref{lem:probNij}]
	We prove our result by using Fact \ref{fact:Nij_toexpect}, then Lemmas \ref{lem:Nbar} and \ref{fact:product_expect}.
\end{proof} \section{Instantiating the Auxiliary Code $\CCbis$ with an Efficient Decoder}\label{app:Caux}

We use here notation from \S\ref{sec:SDC}. In particular, we suppose the auxiliary code $\CCbis$ is a product of $b$ small random codes where
\begin{equation}
	b = \Theta(\log n).
\end{equation}

We have to show that for such $b$, the 
analyses 
from Propositions \ref{prop:biasCodedRLPN} and \ref{prop:exp_sizeS} still hold. Indeed, these analyses were done as if $\CCbis$ were a random code equipped with genie aided decoders. Here we compute $\sH$ as a subset of
$$
\widetilde{\sH} \subseteq \left\{ \left( \hv, \cvbis \right) \in \CC^{\perp} \times \CCbis\: : \: \forall i \in \IInt{1}{b}, \; |\hv_{\sN}(i)| = \tfrac{w}{b} \mbox{ and } |\hv_{\sP}(i) + \cv_i| = \tfrac{\tbis}{b}  \right\}
$$
by decoding each parity-check performing an exhaustive search on each block. We have therefore to show that in this case,

\begin{enumerate}[label=\textcolor{blue}{\roman*.}, ref=\roman*]
\item\label{SDC1} the bias
	$$
	\bias_{\left(\hv,\cvbis\right) \drawn \widetilde{\sH}}\left(  \langle \cvbis + \hv_{\sP},\ev_{\sP} \rangle +  \langle \ev_{\sN},  \hv_{\sN}\rangle \right)
	$$
is of the same order as that given by Proposition \ref{prop:biasCodedRLPN},
\item\label{SDC2} and the number of candidates to test 
$$
\mathbb{E }_{\CC,\CCbis}\left(\card{\sS }\right)
$$
is of the same order as that given by Proposition \ref{prop:exp_sizeS}.
\end{enumerate}

To prove item \ref{SDC1}., we first suppose that $\ev_{\sP}$ and $\ev_{\sN}$ have a weight which is fairly distributed, that is:
\begin{equation}
\label{eq:distrib_e}
\forall i \in \IInt{1}{b}, \; \hw{\ev_{\sP}(i)} = \tfrac{t - u}{b} \mbox{ and } \hw{\ev_{\sN}(i)} = \tfrac{u}{b}.
\end{equation}
This happens with a probability:
\begin{align}
\mathbb{P}_{\mathrm{succ}} = \dfrac{ \binom{s/b}{(t - u)/b}^b\binom{(n-s)/b}{u/b}^b}{\binom{s}{t - u}\binom{n-s}{u}} &= \Om{n \left(\tfrac{b}{c  n}\right)^{b}}
\end{align}
where $c$ is constant in $n$.
So we only need to iterate the whole double-RLPN algorithm a sub-exponential number of times, namely at most $\tfrac{1}{\mathbb{P}_{\mathrm{succ}}}$ times. Note that \eqref{eq:distrib_e} is not a necessary condition to achieve our decoding so this overcost is overestimated.

Now, assuming Condition \eqref{eq:distrib_e}, then we can see the bias above as the product of $b$ independent biases involving smaller vectors. More formally, we have
\begin{equation}
\begin{array}{l}
\displaystyle{\bias_{\left(\hv, \cvbis\right) \drawn \overline{\sH}}\left(  \langle \cvbis + \hv_{\sP},\ev_{\sP} \rangle +  \langle \ev_{\sN},  \hv_{\sN}\rangle \right)} \\
\ \ \ \ \ \ \ \ \ \ \ \ \ \ \ \ \ \ \ = \displaystyle{\prod_{i = 1}^{b}\bias_{\left(\hv(i), \cvbis(i)\right) \drawn \widetilde{\sH}_i}\left(  \langle \cvbis(i) + \hv_{\sP}(i),\ev_{\sP}(i) \rangle +  \langle \ev_{\sN}(i),  \hv_{\sN}(i)\rangle \right)}
\end{array}
\end{equation}
where 
\begin{equation}
\begin{array}{l}
\widetilde{\sH}_i \eqdef \Big\{ \left( \hv_{\sP}(i), \hv_{\sN}(i), \cvbis(i) \right) \in \left(\CC_{\sP(i) \cup \sN(i)}\right)^{\perp} \times \CC_i \; \\
\ \ \ \ \ \ \ \ \ \ \ \ \ \ \ \ \ \ \ \ \ \ \ \ \ \ \ \ \ \ : \; |\hv_{\sN}(i)| = \tfrac{w}{b} \mbox{ and } |\hv_{\sP}(i) + \cv_i| = \tfrac{\tbis}{b}  \Big\}
\end{array}
\end{equation}

Moreover, let us degrade the Constraints \eqref{eq:cstPropo35} of Proposition \ref{prop:biasCodedRLPN} by replacing the polynomial factor $n^{\alpha}$ by a super-polynomial 
\begin{equation}
A \eqdef \frac{n^{\alpha - 1 + \log(c) + \log(n)}}{\log(n)^{\log(n)}}.
\end{equation}
On the one hand, this super-polynomial factor is multiplied to the final complexity, but on the other hand the new constraint (with the original one \eqref{eq:constraints}) induces:
\begin{align}
\frac{\binom{(n-s)/b}{w/b} \binom{s/b}{\tbis/b}}{2^{(k-\kbis)/b}}  
&= \om{\frac{n^{\alpha/\log(n)}}{\delta^{2/\log(n)}}}
\end{align}
and
\begin{equation}
\frac{\binom{(n-s)/b}{w/b} \binom{s/b}{\tbis/b}}{2^{k/b}} = \OO{n^{\alpha/\log(n)}} \quad \mbox{and} \quad \frac{\binom{s/b}{\tbis/b}}{2^{(s-\kbis)/b}} = \OO{n^{\alpha/\log(n)}}.
\end{equation}

Which allows us to say, using Proposition \ref{prop:biasCodedRLPN}, that for all $i \in \IInt{1}{b}$ and for a proportion $1 - \oo{1}$ of codes $\CC_i$ and $\CC_{\sP(i) \cup \sN(i)}$:

\begin{equation}
\bias_{\left(\hv(i), \cvbis(i)\right) \drawn \overline{\sH}_i}\left(  \langle \cvbis(i) + \hv_{\sP}(i),\ev_{\sP}(i) \rangle +  \langle \ev_{\sN}(i),  \hv_{\sN}(i)\rangle \right) = \delta^{1/\log(n)}(1-\oo{1}).
\end{equation}
By specifying the values of both $\oo{1}$ (see proof of Proposition \ref{prop:biasCodedRLPN} in Appendix \ref{app:a}), we can deduce that \ref{SDC1}. is verified.\\

To verify item \ref{SDC2}., we can adapt Section \ref{sec:distrib_bias} to show that
\begin{align}
\mathbb{E }_{\CC,\CCbis}\left(\card{\sS}\right) & = \displaystyle{\OOt{ \max_{(i,j) \in \mathcal{A}} \frac{\binom{s/b}{j}^b \binom{(n-s)/b}{i}^b}{2^{n-k}} } + 1} 
\end{align}
where
\begin{equation}
				\mathcal{A} \eqdef \left\{ \left(i,j\right) \in \llbracket 0, \tfrac{n-s}{b} \rrbracket \times \llbracket0,\tfrac{s}{b} \rrbracket, \frac{ K_{w/b}^{((n-s)/b)}\left(u/b\right) K_{\tbis/b}^{(s/b)}\left((t-u)/b\right) }{ K_{w/b}^{((n-s)/b)}\left(i\right) K_{\tbis/b}^{(s/b)}\left(j\right)} \leq n^{2/b} \right\},
\end{equation}
\begin{equation}
\sS \eqdef \{\sv \in \F_2^\kbis \: : \widehat{f_{\yv,\widetilde{\sH},\Gmbis}}\left(\sv\right) \geq  \frac{\delta}{2} \: \widetilde{H}\},
\end{equation}
and
\begin{equation}
\widetilde{H} \eqdef \frac{\binom{(n-s)/b}{w/b}^b \binom{s/b}{\tbis/b}^b}{2^{k-\kbis}}.
\end{equation}
		
Finally, up to a sub-exponential factor, the above expectation is of the same order as in Proposition \ref{prop:exp_sizeS}. \section{Proofs of the statements made in Section \ref{sec:lattice}}\label{sec:app_lattice}

{\bf \noindent Proof of Proposition \ref{propo:dualityLattice}}
\propdualityLattice*

It is helpful to notice before the following link between the Bessel functions and the Fourier transform of the indicator function 
$\un_{\leq w}$
of the words of Euclidean norm $\leq w$
in $\mathbb{R}^n$ (see \cite[Fact 4.11]{DDRT23})
\begin{lemma}
	We have for any positive integer $n$, any $w\geq 0$, any $\xv$ in $\mathbb{R}^n$ 
	$$
	\widehat{\un_{\leq w}}(\xv) = \left( \frac{w}{\| \vec{x}\|_{2}} \right)^{n/2} J_{n/2}\left(2\pi w \|\vec{x}\|_{2}\right)
	$$
where $\widehat{f}(\vec{x}) = \int_{\mathbb{R}^{n}} f(\vec{x})e^{-2i\pi \langle \vec{x},\vec{y} \rangle} d\vec{y}$ for $f : \mathbb{R}^{n} \rightarrow \mathbb{C}$. 
\end{lemma}

\begin{proof}[Proof of Proposition \ref{propo:dualityLattice}]
	First, notice that,
	\begin{align}
		f_{\widetilde{\sW}}(\yv) &= \frac{1}{2} \sum_{\vec{w} \in \widetilde{\sW}} \left( e^{2i\pi \langle \vec{w},\vec{y} \rangle} + e^{-2i\pi \langle \vec{x},\vec{y} \rangle} \right) \nonumber \\
		&=  \sum_{\vec{w} \in \widetilde{\sW}} e^{-2i\pi \langle \vec{w},\vec{y} \rangle} \quad (\mbox{$\vec{w} \mapsto -\vec{w}$ is a bijection in $\sW$}) \nonumber \\
		&= \sum_{\vec{w} \in \Lambda^{\vee}} \left( \un_{\leq w+\varepsilon}(\vec{w}) - \un_{\leq w-\varepsilon}(\vec{w}) \right) e^{-2i\pi \langle \vec{w},\vec{y} \rangle} \label{eq:diffUn} 
	\end{align}
	Recall now the Poisson summation formula, for any $\vec{y} \in \Lambda + \vec{e}$ and sufficiently regular function $f$,
	$$
	\sum_{\vec{x} \in \Lambda^{\vee}} f(\vec{x})e^{-2i\pi \langle \vec{x},\vec{y} \rangle} = \frac{1}{|\Lambda^{\vee}|} \; \sum_{\vec{x}\in \Lambda + \vec{e}} \widehat{f}(\vec{x})
	$$
	Plugging this formula into Equation \eqref{eq:diffUn} yields to, 
	\begin{align}
			f_{\widetilde{\sW}}(\yv) &= \frac{1}{|\Lambda^{\vee}|} \sum_{\vec{x \in \Lambda+\vec{e}}} \left( \widehat{\un_{\leq w+\varepsilon}}(\vec{x}) - \widehat{\un_{\leq w-\varepsilon}}(\vec{x}) \right) \nonumber \\
			&= \frac{1}{|\Lambda^{\vee}|}\sum_{j\geq 0}  \frac{N_{j}}{j^{n/2}} \left( (w+\varepsilon)^{n/2} J_{n/2}\left( 2\pi (w+\varepsilon) j \right) - (w-\varepsilon)^{n/2}J_{n/2}(2\pi (w-\varepsilon) j) \right) \nonumber \\
			&= \frac{1}{|\Lambda^{\vee}|}\sum_{ j \geq 0} \frac{N_j}{(2\pi)^{n/2}\; j^{n}} \Big( \left( 2\pi (w+\varepsilon) j \right)^{n/2} J_{n/2}\left( 2\pi (w+\varepsilon) j \right) \nonumber\\
			&\qquad\qquad\qquad\qquad\qquad\qquad - \left( 2\pi (w-\varepsilon)j \right)^{n/2}J_{n/2}(2\pi (w-\varepsilon) j) \Big)   \label{eq:concl}
	\end{align}
	which 
	 concludes the proof. 
\end{proof} 

{\noindent \bf An Approximation.}
We have also an approximate form for $f_{\widetilde{\sW}}(\yv)$ which is given by
\begin{equation}
f_{\widetilde{\sW}}(\yv) 
\approx \frac{4\pi \varepsilon}{|\Lambda^{\vee}|}\; \sum_{j \geq 0} j N_{j} \;\left( \frac{w}{j} \right)^{n/2} J_{n/2-1}(2\pi  w j) \label{eq:approxDerivBessel}.
\end{equation}
This follows from the fact that 
\begin{equation}\label{eq:derivative} 
	\frac{d}{dx} \left( x^{n/2} J_{n/2}(x) \right) = x^{n/2} J_{n/2-1}(x). 
	\end{equation} 
	Let,
	$$
	X \eqdef 2\pi w j \quad \mbox{ and } \quad h \eqdef 2\pi \varepsilon j
	$$	
	Notice that,
	\begin{multline*}
		\left( \left( 2\pi (w+\varepsilon) j \right)^{n/2} J_{n/2}\left( 2\pi (w+\varepsilon) j \right) - \left( 2\pi (w-\varepsilon)j \right)^{n/2}J_{n/2}(2\pi (w-\varepsilon) j) \right)   \\
		= (X+h)^{n/2}J_{n/2}(X+h) - (X-h)^{n/2} J_{n/2}(X-h)
	\end{multline*}
	From Equation \eqref{eq:derivative},
	\begin{align*} 
	(X+h)^{n/2} J_{n/2}(X+h) - (X-h)^{n/2} J_{n/2}(X-h) &\approx 2h \; \frac{d}{dX} \left( X^{n/2} J_{n/2}(X) \right) \\
	&= 2h \; X^{n/2} J_{n/2-1}(X) \\
	&= (4\pi \varepsilon j) \; (2\pi w j)^{n/2}J_{n/2-1}(2\pi w j)
	\end{align*} 
	Plugging this into Equation \eqref{eq:concl} yields \eqref{eq:approxDerivBessel}.

We also recall that we make the approximation
\begin{equation}\label{eq:fW}
f_{{{\sW}}}(\yv) \approx \frac{N}{\card{\widetilde{\sW}}} \; f_{\widetilde{\sW}}(\yv)
\end{equation}

The number $N_{\leq x}^{\vee}$ of dual vectors of length $\leq x$ can be approximated using the Gaussian heuristic:
\begin{equation*} 
N_{\leq x}^{\vee} \approx \frac{x^{n}}{\sqrt{n\pi}\; \left(\frac{n}{2\pi e} \right)^{n/2}}\;  \frac{1}{|\Lambda^{\vee}|}
\end{equation*} 
Thus we have:
\begin{align}
\card{\widetilde{\sW}} &= N_{\leq w + \varepsilon}^{\vee} -  N_{\leq w - \varepsilon}^{\vee} \nonumber \approx \frac{(w + \varepsilon)^n - (w - \varepsilon)^n}{\sqrt{\tfrac{n}{2 \pi e}}^n \cdot \sqrt{\pi n} \cdot | \Lambda^{\vee}|} \nonumber\approx \frac{2 n \varepsilon w^{n-1}}{\sqrt{\tfrac{n}{2 \pi e}}^n\;  \sqrt{\pi n} \; | \Lambda^{\vee}|}
\end{align}
Putting this into Equation \eqref{eq:fW} shows,
\begin{equation*}
	f_{\sW}(\vec{y}) \approx N \; \frac{\card{\Lambda^{\vee}}\; \sqrt{n\pi}\left( \frac{n}{2\pi e} \right)^{n/2}}{2n \varepsilon w^{n-1} } \; f_{\widetilde{\sW}}(\vec{y})
\end{equation*}
and after some further computation
\begin{align*}
	f_{\sW}(\vec{y}) &\approx   N \; \card{\Lambda^{\vee}} \frac{\sqrt{n\pi}\left( \frac{n}{2\pi e} \right)^{n/2}}{2n \varepsilon w^{n-1} } \; \frac{4\pi \varepsilon}{|\Lambda^{\vee}|}\; \sum_{j \geq 0} j N_{j} \;\left( \frac{w}{j} \right)^{n/2} J_{n/2-1}(2\pi  w j)  \\ 
	&= N \; \frac{\sqrt{n \pi}}{2n w^{n-1}} \left( \frac{n}{2\pi e} \right)^{n/2} \;4\pi \sum_{j\geq 0} j N_{j} \;\left( \frac{w}{j} \right)^{n/2} J_{n/2-1}(2\pi  w j)  \\ 
	&= N \; \frac{\sqrt{n}\; n^{n/2}}{n} \; \frac{1}{w^{n-1}}\; \frac{\sqrt{\pi} \; \pi}{\pi^{n/2}} \; \frac{4}{2 \; 2^{n/2}} \; \frac{1}{e^{n/2}} \; \sum_{j\geq 0} N_{j} \; w^{n/2} \frac{1}{j^{n/2-1}} \; J_{n/2-1}(2\pi wj) \\
	&= N \;\sqrt{n}\; \sqrt{\pi} \;e^{-1}\; n^{n/2-1} \; \pi^{-(n/2-1)} \; 2^{-(n/2-1)} e^{-(n/2-1)} \sum_{j\geq 0} N_{j} \left( \frac{1}{wj} \right)^{n/2-1} \; J_{n/2-1}(2\pi w j) \\
	&= N \frac{\sqrt{n \pi}}{e} \sum_{j\geq 0} N_{j} \left( \frac{n}{2\pi e wj} \right)^{n/2-1} J_{n/2-1}(2\pi wj)
\end{align*}

 \section{Proof of Proposition \ref{prop:model_conjecture}}\label{app:model_conjecture}

\subsection{A more minimalistic conjecture} \label{sec:min_ass}
The goal of this section is to devise a more minimalistic but stronger conjecture which implies Conjecture \ref{ass:bias_dom} (in the sense that it involves only concentration bound of the weight enumerator of some random linear codes). Furthermore, the aforementioned implication is a key step of the proof of Proposition \ref{prop:model_conjecture} which shows that the experimental model implies Conjecture \ref{ass:bias_dom}.
\begin{conjecture} \label{ass:min}
	For any $(i,j) \in \llbracket 0, n-s \rrbracket \times \llbracket 0, s \rrbracket $ and any~$v \in \left \llbracket 0, \overline{V_j}  +  n^{1.1} \max\left(\sqrt{\overline{V_j}},1\right) \right \rrbracket$, we have under Distribution \ref{notation:randomCodes},
	\begin{align}
		&\prob{\CCbis, \xv}{\left|V_j - \overline{V_j}\right|>  n^{1.1} \max\left(\sqrt{\overline{V_j}},1\right)} &&= \OOt{2^{-n}}, \label{eq:ass_min_1} \\
&\prob{\CC,\CCbis,\xv}{ \left|N_{i,j} - v \: \overline{N_i} \right| > n^{1.1} \max \left(\sqrt{v \: \overline{N_i}}, \:1\right) \middle | \: N_j = v} &&= \OOt{2^{-n}}. \label{eq:ass_min_2}
	\end{align}
	where
	\begin{equation*}
		V_j \eqdef N_j \left(\CCbisperp + \xv\right)\end{equation*}
and
	\begin{align}
		\overline{N_i} \eqdef \frac{\binom{n-s}{i}}{2^{n-k}}, \quad \overline{V_j} \eqdef \frac{\binom{s}{j}}{2^{\kbis}} \label{def:N_ibar}.
	\end{align}
\end{conjecture}
\begin{remark}
	Recall that $N_{i,j} \eqdef \sum_{u = 1}^{V_j} N_i^{(u)}$ where $N^{(u)}_i \eqdef N_i\left(\left(\rv^{(u)} + \ev_{\sP}\right) \Rm + \ev_{\sN} + \CC^{\sN}\right)$ and $\vec{r}^{(u)}$ is the $u$'th codeword of weight $j$ of $\CCbisperp + \xv$. From this and using Lemma \ref{lem:Nbar} it is readily seen that we have that $$\expect{\CC,\CCbis,\xv}{N_{i,j} \middle| \: V_j = v} = v \: \overline{N_i} \left(1+o(1)\right).$$ As such, Equation \eqref{eq:ass_min_2} in the previous Conjecture can also really be seen as a concentration inequality.
\end{remark}
\begin{proposition} \label{prop:imply_conj2_conj1}
	Conjecture \ref{ass:min} implies Conjecture \ref{ass:bias_dom}.
\end{proposition}
The following lemmas will be useful to prove the this proposition.
\begin{lemma}[Centering Lemma.]\label{lem:centering} We have,
	$$ 
	\sum_{j = 0}^{s} \sum_{i = 0}^{n-s} K_{\tbis}^{(s)} \left(j\right) K_w^{(n-s)} \left(i\right) N_{i,j} = \sum_{j = 0}^{s} \sum_{i = 0}^{n-s} K_{\tbis}^{(s)} \left(j\right) K_w^{(n-s)} \left(i\right) \left(N_{i,j} - V_j \overline{N_i} \right).
	$$
\end{lemma}
\begin{proof}
	From the orthogonality of Krawtchouk polynomials relatively to the measure $\mu\left(i\right) = \binom{n-s}{i}$ \cite[Ch. 5. \S 7. Theorem 16]{MS86} we have:
	$$
	\sum_{i = 0}^{n-s} \binom{n-s}{i} K_w^{(n-s)}\left(i\right) = 0.
	$$
	By using the definition of $\overline{N_i}$ in Equation \eqref{def:N_ibar}, we get,
	\begin{align*}
		0 &= \sum_{j = 0}^s V_j \: K^{(s)}_{\tbis} \left(j\right)  \sum_{i = 0}^{n-s} K^{(n-s)}_w \left(i\right) \overline{N_{i}}\\
		&= \sum_{j = 0}^s  \sum_{i = 0}^{n-s} K^{(n-s)}_w \left(i\right) K^{(s)}_{\tbis} \left(j\right) V_j \: \overline{N_{i}}.
	\end{align*}
	Therefore,
	$$  
	\sum_{j = 0}^{s} \sum_{i = 0}^{n-s} K_{\tbis}^{(s)} \left(j\right) K_w^{(n-s)} \left(i\right) N_{i,j} = \sum_{j = 0}^{s} \sum_{i = 0}^{n-s} K_{\tbis}^{(s)} \left(j\right) K_w^{(n-s)} \left(i\right) \left(N_{i,j} - V_j \overline{N_i} \right)
	$$
	which completes the proof. 
\end{proof}
\begin{corollary}\label{coro:probaKrawNij}We have,
	\begin{multline*}
		\mathbb{P}\left( \sum_{j = 0}^{s} \sum_{i = 0}^{n-s} K_{\tbis}^{(s)} \left(j\right) K_w^{(n-s)} \left(i\right) N_{i,j} \geq \frac{1}{2} K_w^{(n-s)}\left(u\right) K_{\tbis}^{(s)}\left(t-u\right)\right) = \\
		\OOt{\max_{(i,j) \in \llbracket 0,n-s\rrbracket \times \llbracket 0,s \rrbracket} \prob{}{\left| N_{i,j} -  V_j \: \overline{N_i}\right| \geq \frac{1}{2 \: (n+1)^2}  \left|\frac{K_w^{(n-s)}\left(u\right) K_{\tbis}^{(s)}\left(t-u\right)}{ K_w^{(n-s)}\left(i\right) K_{\tbis}^{(s)}\left(j\right) }\right| } }
	\end{multline*}
	
\end{corollary}
\begin{proof}
	The event 
	$$ 
	\sum_{j = 0}^{s} \sum_{i = 0}^{n-s} K_{\tbis}^{(s)} \left(j\right) K_w^{n-s} \left(i\right) \left( N_{i,j} - V_j \overline{N_i}\right) \geq \frac{1}{2} K_w^{(n-s)}\left(u\right) K_{\tbis}^{(s)}\left(t-u\right) 
	$$ 
	implies that it exists $j \in  \llbracket 0, s \rrbracket$ and $i \in \llbracket 0, n-s \rrbracket$ such that 
	\begin{equation} \label{eq:proof_bias_to_NI_doubleRLPN}
		K_{\tbis}^{(s)} \left(j\right) K_w^{n-s}\left(i\right) \left( N_{i,j} - V_j \overline{N_i}\right)  \geq \frac{K_w^{(n-s)}\left(u\right) K_{\tbis}^{(s)}\left(t-u\right) }{2 \left(n-s+1\right) \left(s+1\right)} .
	\end{equation}
	Therefore,
	\begin{align*}
		&	\mathbb{P}\left( \sum_{j = 0}^{s} \sum_{i = 0}^{n-s} K_{\tbis}^{(s)} \left(j\right) K_w^{(n-s)} \left(i\right) N_{i,j} \geq \frac{1}{2} K_w^{(n-s)}\left(u\right) K_{\tbis}^{(s)}\left(t-u\right)\right) \\ 
		&=	\mathbb{P}\left( \sum_{j = 0}^{s} \sum_{i = 0}^{n-s} K_{\tbis}^{(s)} \left(j\right) K_w^{n-s}\left(i\right) \left( N_{i,j} - V_j \overline{N_i}\right)   \geq \frac{1}{2} K_w^{(n-s)}\left(u\right) K_{\tbis}^{(s)}\left(t-u\right)\right)  \quad (\mbox{Lemma } \ref{lem:centering}) \\ 
		&\leq 	 \prob{\CC,\CCbis}{ \bigvee_{j = 0}^{s} \bigvee_{i = 0}^{n-s} \left( K_{\tbis}^{(s)} \left(j\right) K_w^{n-s} \left(i\right) \left( N_{i,j} - V_j \overline{N_i} \right)  \geq \frac{K_w^{(n-s)}\left(u\right) K_{\tbis}^{(s)}\left(t-u\right) }{2 \left(n-s+1\right) \left(s+1\right)}  \right)} \quad (\mbox{using } \eqref{eq:proof_bias_to_NI_doubleRLPN})\\
		&\leq \sum_{j = 0}^s \sum_{i = 0}^{n-s}  \prob{}{ K_{\tbis}^{(s)} \left(j\right) K_w^{n-s} \left(i\right) \left( N_{i,j} - V_j \overline{N_i} \right) \geq \frac{K_w^{(n-s)}\left(u\right) K_{\tbis}^{(s)}\left(t-u\right) }{2 \left(n-s+1\right) \left(s+1\right)} } \quad (\mbox{union bound}) \\
&= \OO{n^2} \max_{\substack{i = 0 \dots n-s\\j = 0\dots s}}  \prob{}{  \left| N_{i,j} - V_j \overline{N_i} \right|\geq  \frac{1}{2 \: (n+1)^2}  \left|\frac{K_w^{(n-s)}\left(u\right) K_{\tbis}^{(s)}\left(t-u\right)}{ K_w^{(n-s)}\left(i\right) K_{\tbis}^{(s)}\left(j\right) }\right| }
	\end{align*}
	which completes the proof. 
\end{proof}
\begin{lemma} \label{lem:ineq_NI}
	Under Parameters constraint \ref{ass:parameters_doubleRLPN} we have that
	$$ \left(  \frac{ K_w^{(n-s)}\left(u\right) K_{\tbis}^{(s)} \left(t-u\right)}{K_w^{(n-s)}\left(i\right) K_{\tbis}^{(s)} \left(j\right)  } \right)^2 = \omega \left(n^{\alpha + 8}\right)\overline{N_i} \max \left(\overline{V_{j}},\frac{1}{\OO{n^{\alpha}}}\right) 
	$$
\end{lemma}
\begin{proof}
	First, we simplify $ \left(K_w^{(n-s)}\left(u\right) K_{\tbis}^{(s)} \left(t-u\right) \right)^2$ by using Constraint $(i)$ 
of Parameters constraint \ref{ass:parameters_doubleRLPN} which we recall is given by:
	$$ 
	\frac{\binom{n-s}{w} \binom{s}{\tbis}}{2^{k-\kbis}} = \frac{\omega \left(n^{\alpha+8} \right)}{\delta^2}, \quad \mbox{where }\; \delta \eqdef \frac{K_w^{(n-s)}\left(u\right) K_{\tbis}^{(s)} \left(t-u\right) }{\binom{n-s}{w}\binom{s}{\tbis}}.
	$$
	By reordering the terms in the previous equation we get:
	\begin{align*}
		\left(K_w^{(n-s)}\left(u\right) K_{\tbis}^{(s)} \left(t-u\right) \right)^2= 2^{k-\kbis}\binom{s}{\tbis} \binom{n-s}{w} \omega(n^{\alpha + 8}).
	\end{align*}
	Now, to show the lemma we only have to show that
	\begin{align}
		\frac{ 2^{k-\kbis} \binom{s}{\tbis} \binom{n-s}{w} }{ \left( K_w^{(n-s)}\left(i\right) K_{\tbis}^{(s)} \left(j\right) \right)^2 } )\geq \overline{N_i} \max \left(\overline{V_{j}},\frac{1}{\OO{n^{\alpha}}}\right). \label{eq:proof_ineq_Ni}
	\end{align}
	First ,let us lower bound $\frac{1}{\left(K_w^{(n-s)}(i) \right)^2}$.  From the orthonormality relations of the Krawtchouk polynomials \cite[Ch. 5. \S 7. Theorem 16]{MS86} with the measure $\mu_1 \left(v\right) = \frac{\binom{n-s}{v}}{2^{n-s}}$ we have that
	\begin{align*}
		\sum_{v=0}^{n-s} \binom{n-s}{v} \left(K_w^{(n-s)}\left(v\right)\right)^2 &= \binom{n-s}{w} 2^{n-s},
	\end{align*}
	and thus, as the previous sum is composed of positive terms, we have that
	\begin{align}
		\frac{1}{\left(K_w^{(n-s)}(i) \right)^2} \geq \frac{\binom{n-s}{i}}{\binom{n-s}{w}2^{n-s}}. \label{eq:proof_ineq_NI_eq1}
	\end{align}
	Now, let us lower-bound $\frac{1}{\left(K_{\tbis}^{(s)}(j) \right)^2} $. Using the same orthonormality argument relatively to the measure $\mu_2 \left(v\right) = \frac{\binom{s}{v}}{2^{s}}$ we get
	\begin{align*}
		\frac{1}{\left(K_{\tbis}^{(s)}(j) \right)^2} \geq \frac{\binom{s}{j}}{\binom{s}{\tbis}2^{s}}.
	\end{align*}
	Furthermore, from Fact \ref{fact:krawtchouk_bias} we can also deduce the following inequality \begin{equation*}  \label{eq:proof_lem_ineq_NI_3}
		\left(K_{\tbis}^{(s)}\left( j\right)\right)^2 \leq \binom{s}{\tbis}^2.
	\end{equation*} 
	Combining the last two equations we get that
	\begin{align}
		\frac{1}{\left(K_{\tbis}^{(s)}(j) \right)^2} \geq \min \left(\frac{\binom{s}{j}}{\binom{s}{\tbis}2^{s}} , \frac{1}{\binom{s}{\tbis}^2} \right)  \label{eq:proof_ineq_NI_eq2}.
	\end{align}
	Finally let us show Equation \eqref{eq:proof_ineq_Ni} by using Equation \eqref{eq:proof_ineq_NI_eq1} and \eqref {eq:proof_ineq_NI_eq2} :
	\begin{align*}
		\frac{ 2^{k-\kbis} \binom{s}{\tbis} \binom{n-s}{w} }{ \left( K_w^{(n-s)}\left(i\right) K_{\tbis}^{(s)} \left(j\right) \right)^2 }  &\geq  2^{k-\kbis} \binom{s}{\tbis} \binom{n-s}{w}  \frac{\binom{n-s}{i}}{\binom{n-s}{w}2^{n-s}} \:  \min \left(\frac{\binom{s}{j}}{\binom{s}{\tbis}2^{s}} , \frac{1}{\binom{s}{\tbis}^2} \right) \\
		&= \frac{\binom{n-s}{i}}{2^{n-k}} \min \left(\frac{\binom{s}{j}}{2^{\kbis}}, \quad \frac{2^{s-\kbis}}{\binom{s}{\tbis}}\right) \\
		&= \frac{\binom{n-s}{i}}{2^{n-k}} \min \left(\frac{\binom{s}{j}}{2^{\kbis}}, \frac{1}{\OO{n^{\alpha}}}\right) \quad \mbox{(Constraint $(iii)$ of Eq. \eqref{constraints:constraint})}\\
		&= \overline{N_i} \min \left(\overline{V_j}, \frac{1}{\OO{n^{\alpha}}} \right).
	\end{align*}
	This completes the proof. 
\end{proof}
\begin{proof}[Proof of Proposition \ref{prop:imply_conj2_conj1}.]
	Let us suppose that the Parameter constraint \ref{ass:parameters_doubleRLPN} is verified and that Conjecture \ref{ass:min} is true. We want to show Conjecture \ref{ass:bias_dom} holds, namely that
	\begin{multline*} 
	\mathbb{P}\left( \sum_{j = 0}^{s} \sum_{i = 0}^{n-s} K_{\tbis}^{(s)} \left(j\right) K_w^{(n-s)} \left(i\right) N_{i,j} \geq \frac{1}{2} K_w^{(n-s)}\left(u\right) K_{\tbis}^{(s)}\left(t-u\right)\right) \\ 
	=
	\OOt{ \max_{(i,j) \in \mathcal{A}}\prob{}{N_{i,j} \neq 0}  + 2^{-n}}. 
	\end{multline*} 
	Using Corollary \ref{coro:probaKrawNij}, we only have to show that for any $(i,j) \in \llbracket 0,n-s\rrbracket \times \llbracket 0,s \rrbracket$ we have
	$$ \prob{}{\left| N_{i,j} -  V_j \: \overline{N_i}\right| \geq \frac{1}{2 \: (n+1)^2}  \left|\frac{K_w^{(n-s)}\left(u\right) K_{\tbis}^{(s)}\left(t-u\right)}{ K_w^{(n-s)}\left(i\right) K_{\tbis}^{(s)}\left(j\right) }\right| }  = \OOt{ \max_{(i^*,j^*) \in \mathcal{A}}\prob{}{N_{i^*,j^*} \neq 0} + 2^{-n}}.$$
	To ease up the notations let us denote by 
	\begin{equation} \label{def:R_ij}
		R_{i,j} \eqdef  \frac{1}{2 \: (n+1)^2}  \left|\frac{K_w^{(n-s)}\left(u\right) K_{\tbis}^{(s)}\left(t-u\right)}{ K_w^{(n-s)}\left(i\right) K_{\tbis}^{(s)}\left(j\right) }\right|.
	\end{equation}
	Thus, we only have to show that 
	\begin{equation} \label{eq:proof_conj2_to_prove}
		\prob{}{\left| N_{i,j} -  V_j \: \overline{N_i}\right| \geq R_{i,j} }  = \OOt{ \max_{(i^*,j^*) \in \mathcal{A}}\prob{}{N_{i^*,j^*} \neq 0} + 2^{-n}}.
	\end{equation}
	We prove the previous equality for each cases: $(i,j) \in \mathcal{A}$ or $(i,j) \notin \mathcal{A}$, where recall that
	$$
		\mathcal{A} \eqdef \left\{ \left(i,j\right) \in \llbracket 0, n-s \rrbracket \times \llbracket0,s \rrbracket, \iftoggle{llncs}{}{\;} \left| \frac{ K_w^{(n-s)}\left(u\right) K_{\tbis}^{(s)}\left(t-u\right) }{ K_w^{(n-s)}\left(i\right) K_{\tbis}^{(s)}\left(j\right)}\right| \leq n^{3.2}\right\}.
	$$

	\noindent \textbf{Cases 1:} Here we suppose that  $(i,j) \in \mathcal{A}$. Let us prove Equation \eqref{eq:proof_conj2_to_prove}. Using the law of total probability we have that
	\begin{align*}
		\prob{}{\left| N_{i,j} -  V_j \: \overline{N_i}\right| \geq R_{i,j} } \leq \prob{}{N_{i,j} \neq 0} +  \prob{}{\left| N_{i,j} -  V_j \: \overline{N_i}\right| \geq R_{i,j}, \; N_{i,j} = 0}.
	\end{align*}
	As $(i,j) \in \mathcal{A}$, 
	$$  
	\prob{}{N_{i,j} \neq 0} = \OOt{\max_{(i^*,j^*) \in \mathcal{A}}\prob{}{N_{i^*,j^*} \neq 0}}
	$$
	we only have left to show that:
	\begin{equation*}
		\prob{}{\left| N_{i,j} -  V_j \: \overline{N_i}\right| \geq R_{i,j} , \: \; N_{i,j} = 0} = \OOt{\max_{(i^*,j^*) \in \mathcal{A}}\prob{}{N_{i^*,j^*} \neq 0} + 2^{-n}}.
	\end{equation*}
	We now show the previous equation, by proving that,
		\begin{equation} \label{proof:conj2_target_1}
		\prob{}{\left| N_{i,j} -  V_j \: \overline{N_i}\right| \geq R_{i,j} , \: \; N_{i,j} = 0} = \OOt{ 2^{-n}}.
	\end{equation}
	We have:
	\begin{align}
		\prob{}{\left| N_{i,j} -  V_j \: \overline{N_i}\right| \geq R_{i,j} , \: \; N_{i,j} = 0} &= \prob{}{ V_j \: \overline{N_i} \geq R_{i,j}} \nonumber \quad \mbox{(we used that $V_{j},\overline{N_{i}} > 0$)} \\
		&= \prob{}{ V_j\geq \frac{R_{i,j}}{\overline{N_i}}} \nonumber \\
		&= \prob{}{V_j - \overline{V_j}\geq \frac{R_{i,j}}{\overline{N_i}} - \overline{V_j}} \nonumber \\
		&\leq \prob{}{\left|V_j - \overline{V_j} \right|\geq \frac{R_{i,j}}{\overline{N_i}} - \overline{V_j}}. \nonumber
	\end{align}
	Recall that from Equation \eqref{eq:ass_min_1} of Conjecture \ref{ass:min} we have that
	\begin{equation*} \label{redit:conj2_eq_1}
		\prob{\CCbis, \xv}{\left|V_j - \overline{V_j}\right|>  n^{1.1} \max\left(\sqrt{\overline{V_j}},1\right)} = \OOt{2^{-n}}.
	\end{equation*}
	Therefore, we only have to show that for $n$ big enough we have 
	\begin{equation} \label{proof:conj2_ineq_1}
		\frac{R_{i,j}}{\overline{N_i}} - \overline{V_j} \geq n^{1.1} \max \left(\sqrt{\overline{V_j}, \: 1} \right)
	\end{equation}
	to prove Equation \eqref{proof:conj2_target_1}. Let us prove Equation \eqref{proof:conj2_ineq_1}. By definition of $R_{i,j}$ in Equation \eqref{def:R_ij} and using Lemma~\ref{lem:ineq_NI} we have that
	\begin{align}
		R_{i,j}^2 &=   \frac{1}{4 \: (n+1)^4}  \left(\frac{K_w^{(n-s)}\left(u\right) K_{\tbis}^{(s)}\left(t-u\right)}{ K_w^{(n-s)}\left(i\right) K_{\tbis}^{(s)}\left(j\right) }\right)^2 \nonumber \\
		&= \frac{1}{4 \: (n+1)^4} \; \omega \left(n^{\alpha + 8}\right)\overline{N_i} \max \left(\overline{V_{j}},\frac{1}{\OO{n^{\alpha}}}\right) \nonumber \\
&=  f(n) \max \left(\overline{N_i} \: \overline{V_j}, n^{-\alpha}\overline{N_i}\right) \label{eq:fn}
	\end{align}
	where $f(n) = \omega\left(n^{\alpha + 4}\right)$.
Therefore,
	\begin{align}
		\frac{R_{i,j}^2}{\overline{N_i}^2} \; \frac{1}{n^{2.4} \max \left(\overline{V_j}^2, \: 1\right)}& =  \frac{f(n) \max \left(\overline{N_i} \: \overline{V_j}, n^{-\alpha} \: \overline{N_i}\right)}{n^{2.4} \; \overline{N_i}^2 \max \left(\overline{V_j}^2, \: 1\right)} \nonumber \\
		&=   \frac{f(n) \max \left(\frac{\overline{V_j}}{\overline{N_i}}, \frac{n^{-\alpha}}{\overline{N_i}} \: \right)}{ n^{2.4} \: \max \left(\overline{V_j}^2, \: 1\right)} \nonumber\\
		&= \begin{cases}
			\frac{1}{n^{2.4}}f(n)\max \left(\frac{1 }{\overline{N_i} \overline{V_j}}, \frac{n^{-\alpha}}{\overline{N_i} \overline{V_j}^2} \: \right) & \mbox{ if } \overline{V_j} > 1\\
			\frac{1}{n^{2.4}} f(n)\max \left(\frac{\overline{V_j}}{\overline{N_i}}, \frac{n^{-\alpha}}{\overline{N_i}} \: \right) & \mbox{ if } \overline{V_j} \leq 1
		\end{cases} \nonumber \\
		&\geq  	 \frac{1}{n^{2.4}} f(n) \min \left(\frac{1}{\overline{N_i} \overline{V_j}}, \frac{n^{-\alpha}}{\overline{N_i}}\right) \nonumber\\
		&\geq  \frac{n^{-2 \alpha}}{n^{2.4}} \; f(n) \; \min \left(\frac{1}{\overline{N_i} \overline{V_j}}, \frac{1}{n^{-\alpha} \: \overline{N_i}}\right) \nonumber \\
		&= \frac{n^{-2 \alpha}}{n^{2.4}} \; \frac{f(n)}{\max \left( \overline{N_i} \overline{V_j}, n^{-\alpha} \: \overline{N_i}\right)} \nonumber \\
		&=   \frac{n^{-2 \alpha}}{n^{2.4}} \;  \frac{f(n)^2}{R_{i,j}^2}\nonumber \qquad (\mbox{By Equation \eqref{eq:fn}}) \\
		&=  \frac{\omega\left(n^{5.6}\right)}{R_{i,j}^2} \qquad (f(n) = \omega\left(n^{\alpha + 4}\right)) \nonumber\\
		&= \omega(1) \label{eq:proof_conj2_c1_f}
	\end{align}
	where in the last line we used the fact that $(i,j) \in \mathcal{A}$: by definition, 
	$$
	R_{i,j} = \frac{1}{2 \: (n+1)^2} \; \left|\frac{K_w^{(n-s)}\left(u\right) K_{\tbis}^{(s)}\left(t-u\right)}{ K_w^{(n-s)}\left(i\right) K_{\tbis}^{(s)}\left(j\right) }\right|\leq \frac{1}{2 \: (n+1)^2} \; n^{3.2}
	$$ 
	and thus
	\begin{align*}
		\frac{1}{R_{i,j}^2} &\geq \frac{ 2 \: \left(n+1\right)^2}{n^{3.2}}.
	\end{align*}
	Finally, Equation \eqref{eq:proof_conj2_c1_f} shows that for $n$ big enough
	$$  \frac{R_{i,j}}{\overline{N_i}} \geq n^{1.2} \max \left(\overline{V_j}, \: 1 \right)$$
	And as such, for $n$ big enough
	\begin{align*}
		\frac{R_{i,j}}{\overline{N_i}} - \overline{V_j} &\geq n^{1.1}\max \left(\overline{V_j}, \: 1 \right) \\
		&\geq n^{1.1}\max \left(\sqrt{\overline{V_j}}, \: 1 \right)
	\end{align*}
	which proves Equation \eqref{proof:conj2_ineq_1}. Therefore we have just proved Equation \eqref{eq:proof_conj2_to_prove} in the case where $(i,j) \in \mathcal{A}$.
	\newline
	
	\noindent \textbf{Case 2: }. Here we suppose that  $(i,j) \notin \mathcal{A}$. Let us prove Equation \eqref{eq:proof_conj2_to_prove}.
	We only have to prove that:
	\begin{equation*}
		\prob{}{\left| N_{i,j} -  V_j \: \overline{N_i}\right| \geq R_{i,j} }  = \OOt{2^{-n}}.
	\end{equation*}
	 Let $M$ be defined as
	\begin{equation} \label{eq:def_M}
		M \eqdef \overline{V_j} + n^{1.1} \max \left(\sqrt{\overline{V_j}}, 1\right).
	\end{equation} 
	By the law of total probability we have that
	\begin{multline*}
		\prob{\CC,\CCbis,\xv}{  \left| N_{i,j} -  V_j\overline{N_{i}}  \right| > R_{i,j} } = \\  \prob{}{ \left| N_{i,j} -  V_j \overline{N_{i}}  \right| > R_{i,j} | V_j > M}\prob{}{V_j > M} + \\ \sum_{v = 0}^{M} \prob{}{ \left| N_{i,j} -  V_j \overline{N_{i}}  \right| > R_{i,j} | V_j =  v}\prob{}{V_j= v}.
	\end{multline*}
	Which we can upper bound by
	\begin{multline*}
		\prob{\CC,\CCbis,\xv}{  \left| N_{i,j} -  V_j \overline{N_{i}}  \right| > R_{i,j} } \leq \prob{}{V_j > M} +  \max_{v= 0...M }\prob{}{ \left| N_{i,j} -  v \:  \overline{N_{i}}  \right|> R_{i,j} | V_j =  v}.
	\end{multline*}
	By definition of $M$ in Equation \eqref{eq:def_M} and using Equation \eqref{eq:ass_min_1} of Conjecture \ref{ass:min} we get that 
	$$ \prob{}{V_j > M}  = \OOt{2^{-n}}.$$
	Now, we only have left to prove that for any $v \in \llbracket 0, M \rrbracket$ we have
	$$ 
	\prob{}{ \left| N_{i,j} -  v\; \overline{N_{i}}  \right|> R_{i,j} |V_j =  v}  = \OOt{2^{-n}}.
	$$
	Let us consider $v \in \llbracket 0, M \rrbracket$. Let us first show that for $n$ big enough we have that
	\begin{align} \label{proof:ineq_Rij}
		n^{1.1}\max \left(\sqrt{v \: \overline{N_i}}, \:1\right) \leq R_{i,j}.
	\end{align}We have
	\begin{align*}
		n^{2.2}\max \left(v \: \overline{N_i}, \:1\right)  &\leq n^{2.2}\max \left(M \: \overline{N_i}, \:1\right)  \\
		&\leq n^{2.2}\max \left(\overline{V_j} \: \overline{N_i} + n^{1.1} \overline{N_i} \max \left(\sqrt{\overline{V_j}},1\right), \:1\right) && \mbox{(Using Equation \eqref{eq:def_M})} \\
		&\leq  \max \left(2 \: n^{2.2} \: \overline{V_j} \: \overline{N_i}, \; 2\: n^{3.3} \overline{N_i} \sqrt{\overline{V_j}}, \: 2\: n^{3.3} \overline{N_i} ,\:  n^{2.2}\right) 
	\end{align*}
	To show Equation \eqref{proof:ineq_Rij}, ne only have left to prove that, for $n$ big enough, each term in the previous maximum is smaller than $R_{i,j}^2$.
	First let us recall that by definition of $R_{i,j}$ in Equation \eqref{def:R_ij} and from Lemma \ref{lem:ineq_NI},
$$R_{i,j}^2 = \max \left(\omega \left(n^{\alpha + 4}\right) \overline{N_i} \: \overline{V_j}, \omega \left(n^{4}\right) \overline{N_i}\right).$$ 
	For $n$ big enough we have that:
	\begin{align*}
		&2 \: n^{2.2} \: \overline{V_j} \: \overline{N_i} &&\leq  n^{\alpha + 4}\overline{N_i} \: \overline{V_j} \leq R_{i,j}^2, \\
		&2\: n^{3.3} \overline{N_i} \sqrt{\overline{V_j}} &&\leq \begin{cases}
			n^{4}\overline{N_i}\leq R_{i,j}^2 &\mbox{ when } \overline{V_j} \leq 1 \\
			n^{\alpha + 4}\overline{N_i} \: \overline{V_j} \leq R_{i,j}^2 &\mbox{ when } \overline{V_j} > 1
		\end{cases},  \\
		&	2 \: n^{3.3} \: \overline{N_i} &&\leq  n^{4} \overline{N_i} \leq R_{i,j}^2, \\
		& n^{2.2} &&\leq R_{i,j}^2.
	\end{align*}
	Where in the last equation we used the fact that $(i,j) \notin \mathcal{A}$, thus $R_{i,j} \geq \frac{n^{3.2}}{2 \left(n+1\right)^2}$ and thus $R_{i,j}^2 \geq n^{2.3}$ for $n$ big enough. We have shown that $$n^{2.2}\max \left(v \: \overline{N_i}, \:1\right) \leq R_{i,j}^2$$ and thus we have shown Equation \eqref{proof:ineq_Rij}.
	Finally we have
	\begin{align*}
		\prob{}{ \left| N_{i,j} -  V_j \overline{N_{i}}  \right|> R_{i,j} |V_j =  v} &=	\OO{\prob{}{ \left| N_{i,j} -  V_j \overline{N_{i}}  \right|> n^{1.1}\max \left(\sqrt{v \: \overline{N_i}}, \:1\right) |V_j =  v}} \\
		&= \OOt{2^{-n}}
	\end{align*}
	where in the last line we used Equation \eqref{eq:ass_min_2} of Conjecture \ref{ass:min}. This concludes the proof.
\end{proof}

\begin{lemma} \label{lem:poisson_conj2}
	The Poisson Model \ref{model:poisson} imply Conjecture \ref{ass:min}.
\end{lemma}
\begin{proof}
	Let $M$ be defined as
	\begin{equation} \label{lemma:proof_bound_bias_1_Mn}
	M \eqdef \overline{V_j} + n^{1.1} \max \left(\sqrt{\overline{V_j}}, 1\right)
	\end{equation} 
	Recall that to show Conjecture \ref{ass:min} we only have to show that
	\begin{align}
	\prob{\CCbis, \xv}{ \left| V_j - \overline{V_j} \right| \geq  n^{1.1} \max\left(\sqrt{\overline{V_j}},1\right)} &&=& \OOt{2^{-n}}, \label{eq:lem_ass_min_1} \\
	\forall v \in \llbracket 0, M \rrbracket, \:\; \prob{\CC,\CCbis,\xv}{ \left|N_{i,j} - v \: \overline{N_i} \right| > n^{1.1} \max\left(\sqrt{v \: \overline{N_i}}, 1\right)\: \middle | \: N_j = v} &&=& \OOt{2^{-n}} \label{eq:lem_ass_min_2}
\end{align}
	Under the Poisson Model \ref{model:poisson} we have that
	$$ N_{i,j} \sim \mathrm{Poisson}\left(V_j \: \overline{N_i}\right), \quad V_j \sim  \mathrm{Poisson}\left( \overline{V_j}\right).$$
	We will use the following fact: when $\vec{X}$ follows a Poisson distribution of parameter $\lambda$ and $g(n)=\omega(n)$, then we have that
	\begin{align} \label{lemma:proof_bound_bias_1_poisson_bound}
		\prob{}{|\vec{X} - \lambda| > g(n) \max \left(\sqrt{\lambda},1\right)} = 2^{-\omega(n)} .
		\end{align}
		Let us prove this claim. It is known \cite[Prop 11.15]{G17b} that we have the following exponential tail bound for $\vec{X}$:
		\begin{equation*}
			\prob{}{|\vec{X} - \lambda| > r} \leq 2 \: e^{\frac{- \: r^2}{2 \: (\lambda + r) }}.
		\end{equation*}
	Thus,
	\begin{align*}
		\prob{}{|\vec{X} - \lambda| > g(n) \max \left(\sqrt{\lambda},1\right)} &\leq 2 \: e^{\frac{- \: g(n)^2 \max \left(\lambda,1\right)}{2 \: (\lambda + g(n) \max\left(\sqrt{\lambda},1\right)) }} \\
		&\leq 2 \: e^{\frac{- \: g(n)}{2 \left(\frac{\lambda + g(n) \max\left(\sqrt{\lambda},1\right)}{g(n) \max \left(\lambda,1\right)} \right) }}.
	\end{align*}
	We only have left to show that 
	$$  \frac{\lambda + g(n) \max\left(\sqrt{\lambda},1\right)}{g(n) \max \left(\lambda,1\right)}= \OO{1}.$$
	First it is readily seen that we have that ($g(n) = \omega(n)$)
	$$  \frac{\lambda}{g(n) \max \left(\lambda,1\right)}  = \OO{1},$$
	and second, 
	$$ \frac{g(n) \max\left(\sqrt{\lambda},1\right)}{g(n) \max \left(\lambda,1\right)} = \begin{cases}
		1=\OO{1} & \mbox{ if } \lambda \leq 1 \\
		\frac{1}{\sqrt{\lambda}} = \OO{1} & \mbox{ if } \lambda > 1.
	\end{cases}$$
	which concludes the proof of Equation \eqref{lemma:proof_bound_bias_1_poisson_bound}. Equation \eqref{eq:lem_ass_min_1} directly follows from Equation \eqref{lemma:proof_bound_bias_1_poisson_bound}. Equation \eqref{eq:lem_ass_min_2} also directly follow sfrom Equation  \eqref{lemma:proof_bound_bias_1_poisson_bound} by noticing that $N_{i,j}= v \sim \mathrm{Poisson}\left(v \: \overline{N_i}\right)$ when $V_{j} = v$.
\end{proof}
\begin{proof}[Proof of Proposition \ref{prop:model_conjecture}] Apply successively Lemma \ref{lem:poisson_conj2} and Proposition \ref{prop:imply_conj2_conj1}.
\end{proof} 
}{
	\newpage
	\appendix
\section{Proof of Proposition \ref{prop:biasCodedRLPN}}\label{app:a}
 
 Let us first recall this proposition
 \propbiasCodedRLPN*
 
The proof of this Proposition \ref{prop:biasCodedRLPN} will be a consequence of the following lemma. 
\begin{restatable}{lemma}{lemmaVar}\label{lemma:varianceEb}
	Let $w,s,k,\kbis,t_{\textup{aux}},n \in \mathbb{N}$ be such that (for some constant $a > 0$)
	\begin{equation}\label{eq:constraintsVar} 
			\frac{\binom{n-s}{w} \binom{s}{t_{\textup{aux}}}}{2^{k}} = \OO{n^{\alpha}} \quad \mbox{and} \quad \frac{\binom{s}{t_{\textup{aux}}}}{2^{s-\kbis}} = \OO{n^{\alpha}}
		\end{equation}

	Let $\sN$ be a fixed set of $n-s$ positions in $\IInt{1}{n}$ and $\sP \eqdef \llbracket 1,n \rrbracket \backslash \sN$. Let $\vec{e}\in\F_{2}^{n}$ be some error of weight $u$ on $\sN$.

	Assume that $\CC$ is an $[n,k]$-linear code chosen by picking an $(n-k) \times n$ binary parity-check matrix $\Hm$ uniformly at random and that $\CCbis$ is chosen by picking an $(s-\kbis) \times s$ binary parity-check matrix $\Hmbis$ uniformly at random.
	Let us define for $b \in \{0,1\}$,
	\begin{align*}
			E_b &\eqdef \card{\left\{(\ev',\hv) \in   \mathcal{S}_{\tbis}^s \times \CC^\perp:\; \;\left|\hv_{\sN}\right| = w \; \mbox{ and } \; \ev' + \hv_{\sP} \in \CCbis, \; \langle \ev',\vec{e}_{\sP} \rangle +  \langle \vec{e}_{\sN},  \vec{h}_{\sN}\rangle=b \right\}}, \\
			E'_b &\eqdef \card{\left\{ \left(\vec{e}',\vec{h}\right) \in  \mathcal{S}_{\tbis}^{s} \times  \F_{2}^{n} :\; |\vec{h}_{\sN}| = w \; \mbox{ and } \; \;\langle \vec{e}',\vec{e}_{\sP} \rangle +  \langle \vec{e}_{\sN},  \vec{h}_{\sN}\rangle =b\right \}}. 
		\end{align*}
	
	Then,
	\begin{eqnarray}
			\mathbb{E}_{\vec{H},\Hmbis}(E_b) &=&  \Th{1}\frac{E'_b}{2^{k+s-\kbis}} \label{eq:espL}  \\
			\mathbf{Var}_{\vec{H},\Hmbis}(E_b) &=& \OO{n^{\alpha}}\frac{E'_b}{2^{k+s-\kbis}}  \label{eq:varL}
		\end{eqnarray}
\end{restatable}
	\begin{proof}
	Let $\mathbf{1}_{\ev',\hv}$ be the indicator function of the event ``$\hv \in \CC^{\bot}$ and $\ev'+\hv_{\sP} \in \CCbis$''. Define 
	$$
	\mathcal{E}_{b} \eqdef \left\{ \left( \vec{e}',\vec{h} \right) \in \mathcal{S}_{t_{\textup{aux}}}^{s} \times \F_{2}^{n} \mbox{: }  |\vec{h}_{\sN}| = w \mbox{ and }  \langle \ev',\vec{e}_{\sP} \rangle + \langle \vec{e}_{\sN},\vec{h}_{\sN} \rangle=b \right\}.
	$$
	Notice that $E_{b}'  = \card{\mathcal{E}_{b}}$. 
By definition and linearity of the expectation
	\begin{align}
		\mathbb{E}_{\vec{H},\Hmbis}\left( E_{b} \right) &=\mathbb{E}_{\vec{H},\Hmbis} \left( \sum_{(\vec{e}',\vec{h})\in \mathcal{E}_{b}} \mathbf{1}_{\ev',\hv} \right) \nonumber \\
		&= \sum_{(\vec{e}',\vec{h})\in \mathcal{E}_{b}} \mathbb{E}_{\vec{H},\Hmbis}(\mathbf{1}_{\ev',\hv}) \nonumber \\ 
		&=  \sum_{(\vec{e}',\vec{h})\in \mathcal{E}_{b}} \mathbb{P}_{\vec{H},\Hmbis}\left( \hv \in \CC^{\bot}, \hv_{\sP} + \ev' \in \CCbis \right)\label{eq:expectation}
	\end{align}
	We have that,
	\begin{equation}\label{eq:probBasic}
		\mathbb{P}_{\vec{H},\Hmbis}\left( \hv \in \CC^{\bot}, \hv_{\sP} + \ev' \in \CCbis \right) = \left\{
		\begin{array}{cl}
		\frac{1}{2^{k + s - \kbis}} & \;\; \mbox{if } \vec{h}_{\sP}+\vec{e}' \neq \vec{0} \\
		\frac{1}{2^{k}} & \;\; \mbox{otherwise.}
		\end{array}
		\right.
	\end{equation}
Indeed, notice that events ``$\hv \in \CC^{\bot}$'' and ``$\hv_{\sP} + \ev' \in \CCbis $'' are independent. Furthermore, $\vec{h}$ cannot be equal to $\vec{0}$ as $|\vec{h}_{\sN}| = w > 0$.  Therefore,
	\begin{align*}
		\mathbb{P}_{\vec{H},\Hmbis}\left( \hv \in \CC^{\bot},\hv_{\sP} + \ev' \in \CCbis \right) &= \mathbb{P}_{\vec{H}}\left( \hv \in \CC^{\bot}\right)  \; \mathbb{P}_{\Hmbis}\left(\hv_{\sP} + \ev' \in \CCbis \right) \nonumber \\
		&=  \frac{1}{2^{k}}\;\mathbb{P}_{\Hmbis}\left( \hv_{\sP} + \ev' \in \CCbis  \right)
\end{align*}
	Now, plugging Equation \eqref{eq:probBasic} in \eqref{eq:expectation} leads to,
	\begin{align}
		\mathbb{E}_{\vec{H},\Hmbis}\left( E_{b} \right) &= \sum_{\substack{(\vec{e}',\vec{h})\in \mathcal{E}_{b}\\ \vec{h}_{\sP} \neq \vec{e}' }} \frac{1}{2^{k + s - \kbis}} + \sum_{\substack{(\vec{e}',\vec{h})\in \mathcal{E}_{b}\\ \vec{h}_{\sP} = \vec{e}' }} \frac{1}{2^{k}} \nonumber \\
		&= \sum_{\vec{e}' \in \cS_{\tbis}^{s}} \left( \frac{\card{\sE_{b,1}^{\vec{e}'} }}{2^{k+s-\kbis}} + \frac{\card{\sE_{b,2}^{\vec{e}'} }}{2^k} \right) \label{eq:Eb12}
\end{align} 
	where for a fixed $\vec{e}'$, 
	\begin{equation}\label{eq:Eb1} 
	\sE_{b,1}^{\vec{e}'} \eqdef \left\{ \vec{h} \in \mathbb{F}_{2}^{n}: \; (\vec{e}',\vec{h}) \in \sE_{b} \mbox{ and } \vec{h}_{\sP} \neq \vec{e}' \right\},
	\end{equation}
	\begin{equation}\label{eq:Eb2} 
		\sE_{b,2}^{\vec{e}'} \eqdef \left\{ \vec{h} \in \mathbb{F}_{2}^{n}: \; (\vec{e}',\vec{h}) \in \sE_{b} \mbox{ and } \vec{h}_{\sP} = \vec{e}' \right\}.
	\end{equation} 
	Notice that for any $\vec{e}'$, 
	\begin{equation}\label{eq:cardEb12}  
	\card{\sE_{b,1}} = \frac{\card{\sE_{b,2}}}{2^{s}} 
	\end{equation} 
	It is readily seen that,
	$$
	\card{\sE_b} = \sum_{\vec{e}' \in \cS_{\tbis}^{s}} \left( \card{\sE_{b,1}^{\vec{e}'}} +  \card{\sE_{b,2}^{\vec{e}'}} \right) 
	$$
	which implies that 
	$$
	\sum_{\vec{e}' \in \cS_{\tbis}^{s}} \card{\sE_{b,1}^{\vec{e}'}} = \frac{\card{\sE_{b}}}{1+\frac{1}{2^{s}}} = \Th{1} \card{\sE_b} \quad \mbox{and} \quad \sum_{\vec{e}' \in \cS_{\tbis}^{s}} \card{\sE_{b,2}^{\vec{e}'}} = \frac{\card{\sE_{b}}}{2^{s}-1} = \Th{2^{-s}} \card{\sE_{b}}
	$$
	Plugging this into Equation \eqref{eq:Eb12} leads to,
	\begin{equation}\label{eq:expE12b} 
	\mathbb{E}_{\vec{H},\Hmbis}\left( E_{b} \right) = \Th{1} \frac{\card{\sE_{b}}}{2^{k+s-\kbis}} + \Th{1} \frac{\card{\sE_{b}}}{2^{k+s}} = \Th{1}\frac{\card{\sE_{b}}}{2^{k+s-\kbis}}
	\end{equation} 
	which shows Equation \eqref{eq:espL}. 	
	Let us show now Equation \eqref{eq:varL}. By definition 
	\begin{align}
		\mathbf{Var}\left( E_{b} \right) &= \sum_{(\vec{e}',\vec{h}) \in \mathcal{E}_{b}}\mathbf{Var}\left( \mathbf{1}_{\vec{e}',\vec{h}} \right) + \sum_{\substack{(\vec{e}_{0}',\vec{h}^{0}),(\vec{e}_{1}',\vec{h}^{1}) \in \mathcal{E}_{b} \\ (\vec{e}_{0}',\vec{h}^{0}) \neq (\vec{e}_{1}',\vec{h}^{1})}} \mathbb{E}\left( \mathbf{1}_{\vec{e}_{0}',\vec{h}^{0}} \; \mathbf{1}_{\vec{e}_{1}',\vec{h}^{1}} \right) - \mathbb{E}\left( \mathbf{1}_{\vec{e}_{0}',\vec{h}^{0}} \right) \mathbb{E}\left( \mathbf{1}_{\vec{e}_{1}',\vec{h}^{1}} \right) \nonumber \\
		&\leq \frac{E_{b}'}{2^{k+s-\kbis}}+ \sum_{\substack{(\vec{e}_{0}',\vec{h}^{0}),(\vec{e}_{1}',\vec{h}^{1}) \in \mathcal{E}_{b} \\ (\vec{e}_{0}',\vec{h}^{0}) \neq (\vec{e}_{1}',\vec{h}^{1})}} \mathbb{E}\left( \mathbf{1}_{\vec{e}_{0}',\vec{h}^{0}} \; \mathbf{1}_{\vec{e}_{1}',\vec{h}^{1}} \right) - \mathbb{E}\left( \mathbf{1}_{\vec{e}_{0}',\vec{h}^{0}} \right) \mathbb{E}\left( \mathbf{1}_{\vec{e}_{1}',\vec{h}^{1}} \right) \nonumber
	\end{align}
	where we used that $\mathbf{Var}(\mathbf{1}_{\vec{e}',\vec{h}}) \leq \mathbb{E}\left( \mathbf{1}_{\vec{e},\vec{h}}^{2} \right) =  \mathbb{E}\left( \mathbf{1}_{\vec{e},\vec{h}} \right)$ and Equation \eqref{eq:expE12b}.

	To compute the above expectations, we will split in two cases according to $\sE_{b,1}^{\vec{e}'}$ and $\sE_{b,2}^{\vec{e}'}$ which are respectively defined in Equations \eqref{eq:Eb1} and \eqref{eq:Eb2}. More precisely, we will fix $\vec{e}_{b}'$ and suppose that $\vec{h}^{b}$ belongs to $\sE_{b,1}$ or $\sE_{b,2}$. We will treat the following disjoint cases.

	\begin{enumerate}[label=\textcolor{blue}{\arabic*.}, ref=\arabic*]
		\item\label{case1} $\sF_{0}^{\vec{e}_{0}',\vec{e}_{1}'} = \left\{(\vec{h}^{0},\vec{h}^{1}): \;  \vec{h}^{0} \in \sE_{b,1}^{\vec{e}_{0}'} \mbox{ and }  \vec{h}^{1} \in \sE_{b,1}^{\vec{e}_{1}'}\right\}$,

		\item\label{case2} $\sF_{1}^{\vec{e}_{0}',\vec{e}_{1}'} = \left\{(\vec{h}^{0},\vec{h}^{1}): \;  \vec{h}^{0} \in \sE_{b,2}^{\vec{e}_{0}'} \mbox{ and }  \vec{h}^{1} \in \sE_{b,1}^{\vec{e}_{1}'}\right\}$

		\item\label{case3} $\sF_{2}^{\vec{e}_{0}',\vec{e}_{1}'} = \left\{(\vec{h}^{0},\vec{h}^{1}): \;  \vec{h}^{0} \in \sE_{b,1}^{\vec{e}_{0}'} \mbox{ and }  \vec{h}^{1} \in \sE_{b,2}^{\vec{e}_{1}'}\right\}$

		\item\label{case4} $\sF_{3}^{\vec{e}_{0}',\vec{e}_{1}'} = \left\{(\vec{h}^{0},\vec{h}^{1}): \;  \vec{h}^{0} \in \sE_{b,2}^{\vec{e}_{0}'} \mbox{ and }  \vec{h}^{1} \in \sE_{b,2}^{\vec{e}_{1}'}\right\}$
	\end{enumerate}
	In particular,
	\begin{multline}\label{eq:ineqVariant} 
		\mathbf{Var}\left( E_{b} \right)  \leq \frac{E_{b}'}{2^{k+s-\kbis}} + \sum_{\vec{e}_{0}',\vec{e}_{1}'}\Big( \underbrace{\sum_{\substack{\vec{h}^{0},\vec{h}^{1} \in \sF_{0}^{\vec{e}_{0}',\vec{e}_{1}'} \\ (\vec{e}_{0}',\vec{h}^{0}) \neq (\vec{e}_{1}',\vec{h}^{1}) }}  \mathbb{E}\left( \mathbf{1}_{\vec{e}_{0}',\vec{h}^{0}} \; \mathbf{1}_{\vec{e}_{1}',\vec{h}^{1}} \right) - \mathbb{E}\left( \mathbf{1}_{\vec{e}_{0}',\vec{h}^{0}} \right) \mathbb{E}\left( \mathbf{1}_{\vec{e}_{1}',\vec{h}^{1}} \right)}_{\eqdef F_0}  \\
		+ \underbrace{\sum_{\substack{\vec{h}^{0},\vec{h}^{1} \in \sF_{1}^{\vec{e}_{0}',\vec{e}_{1}'} \\ (\vec{e}_{0}',\vec{h}^{0}) \neq (\vec{e}_{1}',\vec{h}^{1}) }}  \mathbb{E}\left( \mathbf{1}_{\vec{e}_{0}',\vec{h}^{0}} \; \mathbf{1}_{\vec{e}_{1}',\vec{h}^{1}} \right) - \mathbb{E}\left( \mathbf{1}_{\vec{e}_{0}',\vec{h}^{0}} \right) \mathbb{E}\left( \mathbf{1}_{\vec{e}_{1}',\vec{h}^{1}} \right)}_{\eqdef F_1} \\
		+
		\underbrace{\sum_{\substack{\vec{h}^{0},\vec{h}^{1} \in \sF_{2}^{\vec{e}_{0}',\vec{e}_{1}'} \\ (\vec{e}_{0}',\vec{h}^{0}) \neq (\vec{e}_{1}',\vec{h}^{1}) }}  \mathbb{E}\left( \mathbf{1}_{\vec{e}_{0}',\vec{h}^{0}} \; \mathbf{1}_{\vec{e}_{1}',\vec{h}^{1}} \right) - \mathbb{E}\left( \mathbf{1}_{\vec{e}_{0}',\vec{h}^{0}} \right) \mathbb{E}\left( \mathbf{1}_{\vec{e}_{1}',\vec{h}^{1}} \right)}_{\eqdef F_2}  \\
		+
		\underbrace{\sum_{\substack{\vec{h}^{0},\vec{h}^{1} \in \sF_{3}^{\vec{e}_{0}',\vec{e}_{1}'} \\ (\vec{e}_{0}',\vec{h}^{0}) \neq (\vec{e}_{1}',\vec{h}^{1}) }}  \mathbb{E}\left( \mathbf{1}_{\vec{e}_{0}',\vec{h}^{0}} \; \mathbf{1}_{\vec{e}_{1}',\vec{h}^{1}} \right) - \mathbb{E}\left( \mathbf{1}_{\vec{e}_{0}',\vec{h}^{0}} \right) \mathbb{E}\left( \mathbf{1}_{\vec{e}_{1}',\vec{h}^{1}} \right)}_{\eqdef F_3} 
		\Big) 
	\end{multline}
	Let,
	\begin{equation}\label{eq:cov} 
		\textup{Cov} \eqdef F_{0} + F_{1} + F_{2} + F_{3}
	\end{equation}

	\noindent For each cases, we will split according to the following subcases. 
		\begin{enumerate}[label=\textcolor{blue}{\roman*.}, ref=\roman*]\setlength{\itemsep}{5pt}
		\item \label{subcase:1} $ \vec{e}_{0}' = \vec{e}_{1}', \; \vec{h}^{0} \neq \vec{h}^{1}$ and $\vec{h}_{\sP}^{0} = \vec{h}_{\sP}^{1}$,

		\item \label{subcase:2} $\vec{e}_{0}' = \vec{e}_{1}', \; \vec{h}^{0} \neq \vec{h}^{1}$ and $\vec{h}_{\sP}^{0} \neq \vec{h}_{\sP}^{1}$,

		\item\label{subcase:3} $\vec{e}_{0}' \neq \vec{e}_{1}'$ and $\vec{h}^{0} = \vec{h}^{1}$

		\item\label{subcase:4} $\vec{e}_{0}' \neq \vec{e}_{1}', \; \vec{h}^{0} \neq \vec{h}^{1}$ and $\vec{h}_{\sP}^{0} = \vec{h}_{\sP}^{1}$

		\item\label{subcase:5} $\vec{e}_{0}' \neq \vec{e}_{1}',  \vec{h}^{0} \neq \vec{h}^{1}, \hv_{\sP}^{0} \neq \vec{h}_{\sP}^{1}$ and $\vec{h}_{\sP}^{0} + \vec{e}_{0}' = \vec{h}_{\sP}^{1} + \vec{e}_{1}'$

		\item\label{subcase:6} $\vec{e}_{0}' \neq \vec{e}_{1}',  \vec{h}^{0} \neq \vec{h}^{1}, \hv_{\sP}^{0}\neq \vec{h}_{\sP}^{1}$ and $\vec{h}_{\sP}^{0} + \vec{e}_{0}' \neq  \vec{h}_{\sP}^{1} + \vec{e}_{1}'$ 
		\newline
	\end{enumerate}

	{\bf\noindent Case \ref{case1}:} Recall that in this case we have 
	\begin{equation}\label{eq:case1}  
		\vec{h}^{0}_{\sP} \neq \vec{e}_{0}' \quad \mbox{and} \quad \vec{h}^{1}_{\sP} \neq \vec{e}_{1}'
	\end{equation}   
	We have, 
	\begin{equation*} 
	\mathbb{E}\left( \mathbf{1}_{\vec{e}_{0}',\vec{h}^{0}} \right) \mathbb{E}\left( \mathbf{1}_{\vec{e}_{1}',\vec{h}^{1}} \right) = \left( \frac{1}{2^{k+s-\kbis}} \right)^{2}
	\end{equation*} 
	Let us compute $\mathbb{E}_{\vec{H},\Hmbis}\left( \mathbf{1}_{\vec{e}_{0}',\vec{h}^{0}} \; \mathbf{1}_{\vec{e}_{1}',\vec{h}^{1}} \right)$ when $(\vec{e}_{0}',\vec{h}^{0}) \neq (\vec{e}_{1}',\vec{h}^{1})$. By definition
	\begin{align*}
		\mathbb{E}_{\vec{H},\Hmbis}\left( \mathbf{1}_{\vec{e}_{0}',\vec{h}^{0}} \; \mathbf{1}_{\vec{e}_{1}',\vec{h}^{1}} \right) &= \mathbb{P}_{\vec{H},\Hmbis}\left( \hv^{0},\vec{h}^{1} \in \CC^{\bot}, \vec{h}^{0}_{\sP} + \ev_{0}' \in \CCbis, \vec{h}^{1}_{\sP} + \ev_{1}' \in \CCbis \right) \\
		&= \mathbb{P}_{\vec{H},\Hmbis}\left( \hv^{0},\vec{h}^{1} \in \CC^{\bot}\right) \mathbb{P}_{\vec{H},\Hmbis} \left(  \vec{h}^{0}_{\sP} + \ev_{0}' \in \CCbis, \vec{h}^{1}_{\sP} + \ev_{1}' \in \CCbis\right)
	\end{align*}
Therefore,
	\begin{align} 
	\sum_{\vec{e}_{0}',\vec{e}_{1}'} F_{0} &= \sum_{\vec{e}_{0}',\vec{e}_{1}'}\sum_{\substack{\vec{h}^{0},\vec{h}^{1} \in \sF_{0}^{\vec{e}_{0}',\vec{e}_{1}'} \\ (\vec{e}_{0}',\vec{h}^{0}) \neq (\vec{e}_{1}',\vec{h}^{1}) }}  \mathbb{E}\left( \mathbf{1}_{\vec{e}_{0}',\vec{h}^{0}} \; \mathbf{1}_{\vec{e}_{1}',\vec{h}^{1}} \right) - \left( \frac{1}{2^{k+s-\kbis}}\right)^{2} \nonumber\\
	&= \sum_{\vec{e}_{0}',\vec{e}_{1}'}\sum_{\substack{\vec{h}^{0},\vec{h}^{1} \in \sF_{0}^{\vec{e}_{0}',\vec{e}_{1}'} \\ (\vec{e}_{0}',\vec{h}^{0}) \neq (\vec{e}_{1}',\vec{h}^{1}) }} \textup{Cov}^{(0)}\label{eq:cov0} 
	\end{align}
	where, 
	\begin{equation*}
	\textup{Cov}^{(0)} \eqdef \mathbb{P}_{\vec{H}}\left( \hv^{0},\vec{h}^{1} \in \CC^{\bot}\right) \mathbb{P}_{\Hmbis} \left( \vec{h}_{\sP}^{0} + \ev_{0}' \in \CC_{\textup{aux}}, \vec{h}_{\sP}^{1} + \ev_{1}' \in \CC_{\textup{aux}} \right) - \left(\frac{1}{ 2^{k+s-\kbis}}\right)^{2}
	\end{equation*} 
	Our aim now it to upper-bound $\textup{Cov}^{(0)}$ according to the above $6$ sub-cases.
	\newline

	{\bf \noindent Sub-case \ref{subcase:1}.} Suppose that $\vec{e}_{0}' = \vec{e}_{1}'$, $\vec{h}^{0} \neq \vec{h}^{1}$ and $\vec{h}_{\sP}^{0} = \vec{h}_{\sP}^{1}$. We have
	\begin{align*}
		\textup{Cov}^{(0)} &= \frac{1}{2^{2k}} \; \mathbb{P}_{\Hmbis}\left( \vec{h}^{0}_{\sP} + \ev_{0}' \in \CCbis, \vec{h}^{1}_{\sP} + \ev_{1}' \in \CCbis \right) - \left(\frac{1}{ 2^{k+s-\kbis}}\right)^{2} \\
		&= \frac{1}{2^{2k}}  \; \mathbb{P}_{\Hmbis}\left( \vec{h}^{0}_{\sP} + \ev_{0}' \in \CCbis \right) 	- \left(\frac{1}{ 2^{k+s-\kbis}}\right)^{2} \qquad (\mbox{as }\vec{h}^{0}_{\sP} + \ev_{0}'  = \vec{h}^{1}_{\sP} + \ev_{1}' )  \\
		&= \frac{1}{2^{2k}}  \; \frac{1}{2^{s-\kbis}} - \left(\frac{1}{ 2^{k+s-\kbis}}\right)^{2}
	\end{align*}
		where in the last line we used Equation \eqref{eq:case1}.
	Therefore, in that case 
	$$
	\textup{Cov}^{(0)} \leq   \frac{1}{2^{2k}} \; \frac{1}{2^{s-\kbis}}.
	$$

	{\bf \noindent Sub-case \ref{subcase:2}.} Suppose that $\vec{e}_{0}'=\vec{e}_{1}', \vec{h}^{0}\neq \vec{h}^{1}$ and $\vec{h}_{\sP}^{0} \neq \vec{h}_{\sP}^{1}$. We have 
	\begin{align}
		\textup{Cov}^{(0)} &= \frac{1}{2^{2k}} \;\mathbb{P}_{\Hmbis} \left(\vec{h}^{0}_{\sP} + \ev_{0}' \in \CCbis, \vec{h}^{1}_{\sP} + \ev_{1}' \in \CCbis \right) - \left(\frac{1}{ 2^{k+s-\kbis}}\right)^{2}  \nonumber \\
		&=  \frac{1}{2^{2k}} \;\mathbb{P}_{\Hmbis} \left(\vec{h}^{0}_{\sP} + \ev_{0}' \in \CCbis \right) \mathbb{P}_{\Hmbis} \left( \vec{h}^{1}_{\sP} + \ev_{1}' \in \CCbis \right) - \left(\frac{1}{ 2^{k+s-\kbis}}\right)^{2}\label{eq:cov2}
	\end{align}
	where in the last line we used that $\vec{h}^{0}_{\sP} + \ev_{0}'  \neq \vec{h}^{1}_{\sP} + \ev_{1}'$showing that these two vectors are linearly independent (we work in $\F_{2}$), and thus that both events are independent. But, as they are different from $\vec{0}$ (according to Equation \eqref{eq:case1}) we have for ($b \in \{0,1\}$)
	$$
	\mathbb{P}_{\Hmbis} \left( \vec{h}^{b}_{\sP} + \ev_{b}' \in \CCbis \right)  = \frac{1}{2^{s-k_{\textup{aux}}}}
	$$
	Therefore, plugging this in Equation \eqref{eq:cov2} leads to 
	$$
	\textup{Cov}^{(0)} = 0.
	$$

	{\bf \noindent Sub-case \ref{subcase:3}.} Suppose that $\vec{e}_{0}' \neq \vec{e}_{1}'$, $\vec{h}^{0} = \vec{h}^{1}$. In that case, 
	\begin{align*} 
		\textup{Cov}^{(0)} &=	\mathbb{P}_{\Hmbis} \left(\hv^0 \in \CC^\perp \right) \mathbb{P}_{\Hmbis} \left(\vec{h}^{0}_{\sP} + \ev_{0}' \in \CCbis, \vec{h}^{0}_{\sP} + \ev_{1}' \in \CCbis  \right)  - \left(\frac{1}{ 2^{k+s-\kbis}}\right)^{2}\\
		&=\frac{1}{2^{k}} \; \mathbb{P}_{\Hmbis} \left(\vec{h}^{0}_{\sP} + \ev_{0}' \in \CCbis \right) \mathbb{P}_{\Hmbis} \left(\vec{h}^{0}_{\sP} + \ev_{1}' \in \CCbis \right)  - \left(\frac{1}{ 2^{k+s-\kbis}}\right)^{2}\\
		&= \frac{1}{2^{k}} \; \left(\frac{1}{2^{s-\kbis}} \right)^2  - \left(\frac{1}{ 2^{k+s-\kbis}}\right)^{2}
	\end{align*} 
	where in the second equality we used that $\vec{h}_{\sP}^{0} + \vec{e}_{0}' \neq \vec{h}_{\sP}^{1} + \vec{e}_{1}'$ and are different from $\vec{0}$ (according to Equation \eqref{eq:case1}) which implies that both events are independent.  
	Therefore, in that case  
	\begin{equation}\label{eq:subcase03} 
	\textup{Cov}^{(0)} \leq   \frac{1}{2^{k + 2s-2\kbis}} 
	\end{equation}

	{\bf \noindent Sub-case \ref{subcase:4}.} Suppose that $\vec{e}_{0}' \neq \vec{e}_{1}'$, $\vec{h}^{0} \neq \vec{h}^{1}$ and $\vec{h}_{\sP}^{0} = \vec{h}_{\sP}^{1}$. In that case, 
	\begin{align*}
		\textup{Cov}^{(0)} &=\frac{1}{2^{2k}} \; \mathbb{P}_{\Hmbis} \left(\vec{h}^{0}_{\sP} + \ev_{0}' \in \CCbis, \vec{h}^{0}_{\sP} + \ev_{1}' \in \CCbis  \right) - \left(\frac{1}{ 2^{k+s-\kbis}}\right)^{2} \\
		&=\frac{1}{2^{2k}}\mathbb{P}_{\Hmbis} \left(\vec{h}^{0}_{\sP} + \ev_{0}' \in \CCbis \right) \mathbb{P}_{\Hmbis} \left(\vec{h}^{0}_{\sP} + \ev_{1}' \in \CCbis \right) - \left(\frac{1}{ 2^{k+s-\kbis}}\right)^{2} \\
		&=  \left(\frac{1}{2^{k+ s-\kbis}} \right)^2 - \left(\frac{1}{ 2^{k+s-\kbis}}\right)^{2} \\
		&= 0
	\end{align*}
	where in the second equality we used that $\vec{h}_{\sP}^{0} + \vec{e}_{0}' = \vec{h}_{\sP}^{1} + \vec{e}_{0}' \neq \vec{h}_{\sP}^{1} + \vec{e}_{1}'$ which implies that both events are independent.  
	\newline

	{\bf \noindent Sub-case \ref{subcase:5}.} Suppose that $\vec{e}_{0}' \neq \vec{e}_{1}'$, $\vec{h}^{0} \neq \vec{h}^{1}$, $\vec{h}_{\sP}^{0} \neq \vec{h}_{\sP}^{1}$ and $\hv_{\sP}^0 + \ev_0' = \hv_{\sP}^1 + \ev_1'$. We have
	\begin{align*}
		\textup{Cov}^{(0)} &=\frac{1}{2^{2k}}\; \mathbb{P}_{\Hmbis} \left(\vec{h}^{0}_{\sP} + \ev_{0}' \in \CCbis \right) - \left(\frac{1}{ 2^{k+s-\kbis}}\right)^{2} \\
		&=\frac{1}{2^{2k}}\; \mathbb{P}_{\Hmbis} \left(\vec{h}^{0}_{\sP} + \ev_{0}' \in \CCbis \right)  - \left(\frac{1}{ 2^{k+s-\kbis}}\right)^{2} \\
		&=  \frac{1}{2^{2k}} \; \frac{1}{2^{s-\kbis}}  - \left(\frac{1}{ 2^{k+s-\kbis}}\right)^{2}
	\end{align*}
	Therefore, 
	$$
	\textup{Cov}^{(0)} \leq \frac{1}{2^{2k + s - \kbis}} 
	$$

	{\bf \noindent Sub-case \ref{subcase:6}.} Suppose that $\vec{e}_{0}' \neq \vec{e}_{1}'$, $\vec{h}^{0} \neq \vec{h}^{1}$, $\vec{h}_{\sP}^{0} \neq \vec{h}_{\sP}^{1}$ and $\hv_{\sP}^0 + \ev_0' \neq \hv_{\sP}^1 + \ev_1'$.
	In that case we can write
	\begin{align*}
		\textup{Cov}^{(0)} &=\frac{1}{2^{2k}} \; \mathbb{P}_{\Hmbis} \left(\vec{h}^{0}_{\sP} + \ev_{0}' \in \CCbis, \vec{h}^{1}_{\sP} + \ev_{1}' \in \CCbis  \right) - \left(\frac{1}{ 2^{k+s-\kbis}}\right)^{2} \\
		&=\frac{1}{2^{2k}} \; \mathbb{P}_{\Hmbis} \left(\vec{h}^{0}_{\sP} + \ev_{0}' \in \CCbis \right) \mathbb{P}_{\Hmbis} \left(\vec{h}^{1}_{\sP} + \ev_{1}' \in \CCbis \right) - \left(\frac{1}{ 2^{k+s-\kbis}}\right)^{2} \\
		&=  \left(\frac{1}{2^{k+ s-\kbis}} \right)^2 - \left(\frac{1}{ 2^{k+s-\kbis}}\right)^{2}
	\end{align*}
	Therefore we obtain,  
	$$
	\textup{Cov}^{(0)} = 0.
	$$

		{\bf\noindent Case \ref{case2}:} Recall that in this case we have 
	\begin{equation}\label{eq:case2}  
		\vec{h}^{0}_{\sP} = \vec{e}_{0}' \quad \mbox{and} \quad \vec{h}^{1}_{\sP} \neq \vec{e}_{1}'
	\end{equation}  
	We have, 
	\begin{equation*} 
		\mathbb{E}\left( \mathbf{1}_{\vec{e}_{0}',\vec{h}^{0}} \right) \mathbb{E}\left( \mathbf{1}_{\vec{e}_{1}',\vec{h}^{1}} \right) =  \frac{1}{2^{k}}  \; \frac{1}{2^{k+s-\kbis}}= \frac{1}{2^{2k+s-\kbis}}
	\end{equation*} 
	Let us compute $\mathbb{E}_{\vec{H},\Hmbis}\left( \mathbf{1}_{\vec{e}_{0}',\vec{h}^{0}} \; \mathbf{1}_{\vec{e}_{1}',\vec{h}^{1}} \right)$ when $(\vec{e}_{0}',\vec{h}^{0}) \neq (\vec{e}_{1}',\vec{h}^{1})$. By definition
	\begin{align*}
		\mathbb{E}_{\vec{H},\Hmbis}\left( \mathbf{1}_{\vec{e}_{0}',\vec{h}^{0}} \; \mathbf{1}_{\vec{e}_{1}',\vec{h}^{1}} \right) &= \mathbb{P}_{\vec{H},\Hmbis}\left( \hv^{0},\vec{h}^{1} \in \CC^{\bot}, \vec{h}^{0}_{\sP} + \ev_{0}' \in \CCbis, \vec{h}^{1}_{\sP} + \ev_{1}' \in \CCbis \right) \\
		&= \mathbb{P}_{\vec{H},\Hmbis}\left( \hv^{0},\vec{h}^{1} \in \CC^{\bot}\right) \mathbb{P}_{\vec{H},\Hmbis} \left(  \vec{h}^{0}_{\sP} + \ev_{0}' \in \CCbis, \vec{h}^{1}_{\sP} + \ev_{1}' \in \CCbis\right)
	\end{align*}
	Therefore,
	\begin{align} 
		\sum_{\vec{e}_{0}',\vec{e}_{1}'} F_{1} &= \sum_{\vec{e}_{0}',\vec{e}_{1}'} \mathbb{E}\left( \mathbf{1}_{\vec{e}_{0}',\vec{h}^{0}} \; \mathbf{1}_{\vec{e}_{1}',\vec{h}^{1}} \right) -  \frac{1}{2^{2k+s-\kbis}}\nonumber \\
		&= \sum_{\vec{e}_{0}',\vec{e}_{1}'}\sum_{\substack{\vec{h}^{0},\vec{h}^{1} \in \sF_{1}^{\vec{e}_{0}',\vec{e}_{1}'} \\ (\vec{e}_{0}',\vec{h}^{0}) \neq (\vec{e}_{1}',\vec{h}^{1}) }}  \textup{Cov}^{(1)} \label{eq:cov1} 
	\end{align}
	where, 
		\begin{equation*}
	\textup{Cov}^{(1)} \eqdef \mathbb{P}_{\vec{H}}\left( \hv^{0},\vec{h}^{1} \in \CC^{\bot}\right) \mathbb{P}_{\Hmbis} \left( \vec{h}_{\sP}^{0} + \ev_{0}' \in \CC_{\textup{aux}}, \vec{h}_{\sP}^{1} + \ev_{1}' \in \CC_{\textup{aux}} \right) - \frac{1}{ 2^{2k+s-\kbis}}
	\end{equation*} 
	Our aim now it to upper-bound $\textup{Cov}^{(1)}$ according to the above $6$ sub-cases.
	\newline

	{\bf \noindent Sub-case \ref{subcase:1}.} Suppose that $\vec{e}_{0}' = \vec{e}_{1}'$, $\vec{h}^{0} \neq \vec{h}^{1}$ and $\vec{h}_{\sP}^{0} = \vec{h}_{\sP}^{1}$. This subcase is impossible according to Equation \eqref{eq:case2}. Therefore, 
	$$
	\textup{Cov}^{(1)} =0.
	$$

	{\bf \noindent Sub-case \ref{subcase:2}.} Suppose that $\vec{e}_{0}'=\vec{e}_{1}', \vec{h}^{0}\neq \vec{h}^{1}$ and $\vec{h}_{\sP}^{0} \neq \vec{h}_{\sP}^{1}$. We have 
	\begin{align}
		\textup{Cov}^{(1)} &= \frac{1}{2^{2k}} \;\mathbb{P}_{\Hmbis} \left(\vec{h}^{0}_{\sP} + \ev_{0}' \in \CCbis, \vec{h}^{1}_{\sP} + \ev_{1}' \in \CCbis \right) - \frac{1}{ 2^{2k+s-\kbis}}  \nonumber \\
		&=  \frac{1}{2^{2k}}\; \mathbb{P}_{\Hmbis} \left( \vec{h}^{1}_{\sP} + \ev_{1}' \in \CCbis \right) - \frac{1}{ 2^{2k+s-\kbis}}\label{eq:cov22}
	\end{align}
	where in the last line we used that $\vec{h}^{0}_{\sP} + \ev_{0}'  \neq \vec{h}^{1}_{\sP} + \ev_{1}'$ showing that these two vectors are linearly independent (we work in $\F_{2}$), and thus that both events are independent. Furthermore, we also used that $\vec{h}^{0}_{\sP} + \vec{e}_{0}' = \vec{0}$ according to Equation \eqref{eq:case2}. But,
	$$
	\mathbb{P}_{\Hmbis} \left( \vec{h}^{1}_{\sP} + \ev_{1}' \in \CCbis \right)  = \frac{1}{2^{s-k_{\textup{aux}}}}
	$$
	Therefore, plugging this in Equation \eqref{eq:cov22} leads to 
	$$
	\textup{Cov}^{(1)} = 0.
	$$

	{\bf \noindent Sub-case \ref{subcase:3}.} Suppose that $\vec{e}_{0}' \neq \vec{e}_{1}'$, $\vec{h}^{0} = \vec{h}^{1}$. In that case, 
	\begin{align*} 
		\textup{Cov}^{(1)} &=	\mathbb{P}_{\Hmbis} \left(\hv^0 \in \CC^\perp \right) \mathbb{P}_{\Hmbis} \left(\vec{h}^{0}_{\sP} + \ev_{0}' \in \CCbis, \vec{h}^{1}_{\sP} + \ev_{1}' \in \CCbis  \right)  - \frac{1}{ 2^{2k+s-\kbis}}\\
		&=\frac{1}{2^{k}} \; \mathbb{P}_{\Hmbis} \left(\vec{h}^{0}_{\sP} + \ev_{0}' \in \CCbis \right) \mathbb{P}_{\Hmbis} \left(\vec{h}^{0}_{\sP} + \ev_{1}' \in \CCbis \right)  - \frac{1}{ 2^{2k+s-\kbis}} \\
		&= \frac{1}{2^{k}} \;\frac{1}{2^{s-\kbis}}  - \frac{1}{ 2^{2k+s-\kbis}}
	\end{align*} 
	where in the second equality we used that $\vec{h}_{\sP}^{0} + \vec{e}_{0}' \neq \vec{h}_{\sP}^{1} + \vec{e}_{1}'$ which implies that both events are independent. Furthermore, we also used that $\vec{h}^{0}_{\sP} + \vec{e}_{0}' = \vec{0}$ according to Equation \eqref{eq:case2}.
	Therefore, in that case  
	\begin{equation}\label{eq:subcase13}
	\textup{Cov}^{(1)} \leq \frac{1}{2^{k + s-\kbis}}
	\end{equation}

	{\bf \noindent Sub-case \ref{subcase:4}.} Suppose that $\vec{e}_{0}' \neq \vec{e}_{1}'$, $\vec{h}^{0} \neq \vec{h}^{1}$ and $\vec{h}_{\sP}^{0} = \vec{h}_{\sP}^{1}$. In that case, 
	\begin{align*}
		\textup{Cov}^{(1)} &=\frac{1}{2^{2k}} \; \mathbb{P}_{\Hmbis} \left(\vec{h}^{0}_{\sP} + \ev_{0}' \in \CCbis, \vec{h}^{0}_{\sP} + \ev_{1}' \in \CCbis  \right) - \frac{1}{ 2^{2k+s-\kbis}} \\
		&=\frac{1}{2^{2k}} \mathbb{P}_{\Hmbis} \left(\vec{h}^{0}_{\sP} + \ev_{1}' \in \CCbis \right) - \frac{1}{ 2^{2k+s-\kbis}}\\
		&= \frac{1}{ 2^{2k+s-\kbis}} - \frac{1}{ 2^{2k+s-\kbis}} \\
		&= 0
	\end{align*}
	where in the second equality we used that $\vec{h}_{\sP}^{0} + \vec{e}_{0}' = \vec{h}_{\sP}^{1} + \vec{e}_{0}' \neq \vec{h}_{\sP}^{1} + \vec{e}_{1}'$ which implies that both events are independent. Furthermore, we also used that $\vec{h}^{0}_{\sP} + \vec{e}_{0}' = \vec{0}$ according to Equation \eqref{eq:case2}.
	\newline

	{\bf \noindent Sub-case \ref{subcase:5}.} Suppose that $\vec{e}_{0}' \neq \vec{e}_{1}'$, $\vec{h}^{0} \neq \vec{h}^{1}$, $\vec{h}_{\sP}^{0} \neq \vec{h}_{\sP}^{1}$ and $\hv_{\sP}^0 + \ev_0' = \hv_{\sP}^1 + \ev_1'$. This subcase is impossible according to Equation \eqref{eq:case2}. Therefore,
	$$
	\textup{Cov}^{(1)} = 0
	$$

	{\bf \noindent Sub-case \ref{subcase:6}.} Suppose that $\vec{e}_{0}' \neq \vec{e}_{1}'$, $\vec{h}^{0} \neq \vec{h}^{1}$, $\vec{h}_{\sP}^{0} \neq \vec{h}_{\sP}^{1}$ and $\hv_{\sP}^0 + \ev_0' \neq \hv_{\sP}^1 + \ev_1'$.
	In that case we can write
	\begin{align*}
		\textup{Cov}^{(0)} &=\frac{1}{2^{2k}} \; \mathbb{P}_{\Hmbis} \left(\vec{h}^{0}_{\sP} + \ev_{0}' \in \CCbis, \vec{h}^{1}_{\sP} + \ev_{1}' \in \CCbis  \right) - \frac{1}{ 2^{2k+s-\kbis}}\\
		&=\frac{1}{2^{2k}} \; \mathbb{P}_{\Hmbis} \left(\vec{h}^{1}_{\sP} + \ev_{1}' \in \CCbis \right) - \frac{1}{ 2^{2k+s-\kbis}} \\
		&=  \frac{1}{ 2^{2k+s-\kbis}} - \frac{1}{ 2^{2k+s-\kbis}}
	\end{align*}
	Therefore we obtain,  
	$$
	\textup{Cov}^{(1)} = 0.
	$$	
	{\bf\noindent Case \ref{case3}:} This situation is symmetric to Case $2$.  
	\newline

	{\bf\noindent Case \ref{case4}:} Recall that in this case we have 
	\begin{equation}\label{eq:case4}  
		\vec{h}^{0}_{\sP} = \vec{e}_{0}' \quad \mbox{and} \quad \vec{h}^{1}_{\sP} = \vec{e}_{1}'
	\end{equation}  
	We have, 
	\begin{equation*} 
		\mathbb{E}\left( \mathbf{1}_{\vec{e}_{0}',\vec{h}^{0}} \right) \mathbb{E}\left( \mathbf{1}_{\vec{e}_{1}',\vec{h}^{1}} \right) =  \frac{1}{2^{2k}} 
	\end{equation*} 
	Let us compute $\mathbb{E}_{\vec{H},\Hmbis}\left( \mathbf{1}_{\vec{e}_{0}',\vec{h}^{0}} \; \mathbf{1}_{\vec{e}_{1}',\vec{h}^{1}} \right)$ when $(\vec{e}_{0}',\vec{h}^{0}) \neq (\vec{e}_{1}',\vec{h}^{1})$. By definition
	\begin{align*}
		\mathbb{E}_{\vec{H},\Hmbis}\left( \mathbf{1}_{\vec{e}_{0}',\vec{h}^{0}} \; \mathbf{1}_{\vec{e}_{1}',\vec{h}^{1}} \right) &= \mathbb{P}_{\vec{H},\Hmbis}\left( \hv^{0},\vec{h}^{1} \in \CC^{\bot}, \vec{h}^{0}_{\sP} + \ev_{0}' \in \CCbis, \vec{h}^{1}_{\sP} + \ev_{1}' \in \CCbis \right) \\
		&= \mathbb{P}_{\vec{H},\Hmbis}\left( \hv^{0},\vec{h}^{1} \in \CC^{\bot}\right) \mathbb{P}_{\vec{H},\Hmbis} \left(  \vec{h}^{0}_{\sP} + \ev_{0}' \in \CCbis, \vec{h}^{1}_{\sP} + \ev_{1}' \in \CCbis\right)
	\end{align*}
	Therefore,
	\begin{align} 
		\sum_{\vec{e}_{0}',\vec{e}_{1}'} F_{3} &= \sum_{\vec{e}_{0}',\vec{e}_{1}'}\sum_{\substack{\vec{h}^{0},\vec{h}^{1} \in \sF_{3}^{\vec{e}_{0}',\vec{e}_{1}'} \\ (\vec{e}_{0}',\vec{h}^{0}) \neq (\vec{e}_{1}',\vec{h}^{1}) }}  \mathbb{E}\left( \mathbf{1}_{\vec{e}_{0}',\vec{h}^{0}} \; \mathbf{1}_{\vec{e}_{1}',\vec{h}^{1}} \right) -  \frac{1}{2^{2k}}\nonumber \\
		&= \sum_{\vec{e}_{0}',\vec{e}_{1}'} \sum_{\substack{\vec{h}^{0},\vec{h}^{1} \in \sF_{3}^{\vec{e}_{0}',\vec{e}_{1}'} \\ (\vec{e}_{0}',\vec{h}^{0}) \neq (\vec{e}_{1}',\vec{h}^{1}) }}\textup{Cov}^{(3)}\label{eq:cov3} 
	\end{align}
	where, 
		\begin{equation*}
	\textup{Cov}^{(3)} \eqdef \mathbb{P}_{\vec{H}}\left( \hv^{0},\vec{h}^{1} \in \CC^{\bot}\right) \mathbb{P}_{\Hmbis} \left( \vec{h}_{\sP}^{0} + \ev_{0}' \in \CC_{\textup{aux}}, \vec{h}_{\sP}^{1} + \ev_{1}' \in \CC_{\textup{aux}} \right) - \frac{1}{2^{2k}}
	\end{equation*} 
	Our aim now it to upper-bound $\textup{Cov}^{(1)}$ according to the above $6$ sub-cases.
	\newline

	{\bf \noindent Sub-case \ref{subcase:1}.} Suppose that $\vec{e}_{0}' = \vec{e}_{1}'$, $\vec{h}^{0} \neq \vec{h}^{1}$ and $\vec{h}_{\sP}^{0} = \vec{h}_{\sP}^{1}$. We have
	\begin{align*}
		\textup{Cov}^{(3)} &= \frac{1}{2^{2k}} \; \mathbb{P}_{\Hmbis}\left( \vec{h}^{0}_{\sP} + \ev_{0}' \in \CCbis, \vec{h}^{1}_{\sP} + \ev_{1}' \in \CCbis \right) - \frac{1}{2^{2k}} \\
		&= \frac{1}{2^{2k}}  \; \mathbb{P}_{\Hmbis}\left( \vec{0} \in \CCbis \right) 	- \frac{1}{2^{2k}}\\
		&= \frac{1}{2^{2k}}  - \frac{1}{ 2^{2k}}
	\end{align*}
	where in the second equality we used Equation \eqref{eq:case4}.
	Therefore, in that case 
	$$
	\textup{Cov}^{(3)} = 0.
	$$
	
	{\bf \noindent Sub-case \ref{subcase:2}.} Suppose that $\vec{e}_{0}'=\vec{e}_{1}', \vec{h}^{0}\neq \vec{h}^{1}$ and $\vec{h}_{\sP}^{0} \neq \vec{h}_{\sP}^{1}$. According to Equation \eqref{eq:case4} this sub-case is impossible. Therefore, 
	$$
	\textup{Cov}^{(3)} = 0.
	$$

	{\bf \noindent Sub-case \ref{subcase:3}.} Suppose that $\vec{e}_{0}' \neq \vec{e}_{1}'$, $\vec{h}^{0} = \vec{h}^{1}$. According to Equation \eqref{eq:case4} this sub-case is impossible. Therefore, 
	$$
	\textup{Cov}^{(3)} = 0.
	$$

	{\bf \noindent Sub-case \ref{subcase:4}.} Suppose that $\vec{e}_{0}' \neq \vec{e}_{1}'$, $\vec{h}^{0} \neq \vec{h}^{1}$ and $\vec{h}_{\sP}^{0} = \vec{h}_{\sP}^{1}$. According to Equation \eqref{eq:case4} this sub-case is impossible. Therefore, 
	$$
	\textup{Cov}^{(3)} = 0.
	$$

	{\bf \noindent Sub-case \ref{subcase:5}.} Suppose that $\vec{e}_{0}' \neq \vec{e}_{1}'$, $\vec{h}^{0} \neq \vec{h}^{1}$, $\vec{h}_{\sP}^{0} \neq \vec{h}_{\sP}^{1}$ and $\hv_{\sP}^0 + \ev_0' = \hv_{\sP}^1 + \ev_1'$. We have
	\begin{align*}
		\textup{Cov}^{(3)} &=\frac{1}{2^{2k}}\; \mathbb{P}_{\Hmbis} \left(\vec{h}^{0}_{\sP} + \ev_{0}' \in \CCbis \right) - \frac{1}{2^{2k}} \\
		&=\frac{1}{2^{2k}}\; \mathbb{P}_{\Hmbis} \left( \vec{0} \in \CCbis \right)  - \frac{1}{2^{2k}} \\
		&=  \frac{1}{2^{2k}}  - \frac{1}{2^{2k}}
	\end{align*}
	where in the second equality we used Equation \eqref{eq:case4}. 
	Therefore, 
	$$
	\textup{Cov}^{(3)} = 0
	$$

	{\bf \noindent Sub-case \ref{subcase:6}.} Suppose that $\vec{e}_{0}' \neq \vec{e}_{1}'$, $\vec{h}^{0} \neq \vec{h}^{1}$, $\vec{h}_{\sP}^{0} \neq \vec{h}_{\sP}^{1}$ and $\hv_{\sP}^0 + \ev_0' \neq \hv_{\sP}^1 + \ev_1'$. According to Equation \eqref{eq:case4} this sub-case is impossible. Therefore, 
	$$
	\textup{Cov}^{(3)} = 0.
	$$

	We are now ready to gather Cases \ref{case1}, \ref{case2}, \ref{case3} and \ref{case4} according to Subcases \ref{subcase:1}, \ref{subcase:2}, \ref{subcase:3}, \ref{subcase:4}, \ref{subcase:5} and \ref{subcase:6}. Our aim is to bound 
	$
	\textup{Cov}
	$
	that were defined in Equation \ref{eq:cov}.
	We can already notice that Case \ref{case4} has no impact on this sum while Cases \ref{case2} and \ref{case3} have an influence only in Subcase \ref{subcase:3}. Furthermore, \ref{subcase:2}, \ref{subcase:4} and \ref{subcase:6} have no contribution to this sum, whatever is the considered case.

	Let us upper-bound $\textup{Cov}$ according to the different subcases where $\textup{Cov}_{i}$ denotes the terms involved in $\textup{Cov}$ coming from Subcase $i$ (in particular $\textup{Cov}_{i}$ is defined as a certain sum of $\textup{Cov}^{(i)}$, see Equations \eqref{eq:cov0}, \eqref{eq:cov1} and \eqref{eq:cov2}).
	\newline

	{\bf \noindent Subcase \ref{subcase:1}:} We have,  
		\begin{align*}
		\textup{Cov}_{1} 
		&\leq \sum_{(\vec{e}_{0}',\vec{h}^{0}) \in \mathcal{E}_{b}} \sum_{\substack{\vec{h}^{1}: (\vec{e}_{0}',\vec{h}^{1})\in \mathcal{E}_{b} \\ \vec{h}^{1} \neq \vec{h}^{0}, \vec{h}^{0}_{\sP} = \vec{h}^{1}_{\sP}}} \frac{1}{2^{2k}} \frac{1}{2^{s-\kbis}}
		&\leq \sum_{(\vec{e}_{0},\vec{h}^{0}) \in \mathcal{E}_{b}} \frac{\binom{n-s}{w}}{2^{2k + s - \kbis}} 
		&= \frac{\binom{n-s}{w}}{2^{k}} \frac{E_{b}'}{2^{k+s-\kbis}}
	\end{align*}

	{\bf \noindent Subcase \ref{subcase:2}:} We have,
	$$
	\textup{Cov}_{2} = 0. 
	$$

	{\bf \noindent Subcase \ref{subcase:3}:} We have here to split the computation here between Cases \ref{case1} and \ref{case2}. Recall that they are given by
		$$\sF_{0}^{\vec{e}_{0}',\vec{e}_{1}'} = \left\{(\vec{h}^{0},\vec{h}^{1}): \;  \vec{h}^{0} \in \sE_{b,1}^{\vec{e}_{0}'} \mbox{ and }  \vec{h}^{1} \in \sE_{b,1}^{\vec{e}_{1}'}\right\},
	$$
	$$
	\sF_{1}^{\vec{e}_{0}',\vec{e}_{1}'} = \left\{(\vec{h}^{0},\vec{h}^{1}): \;  \vec{h}^{0} \in \sE_{b,2}^{\vec{e}_{0}'} \mbox{ and }  \vec{h}^{1} \in \sE_{b,1}^{\vec{e}_{1}'}\right\}
	$$
	where,
	$$
		\sE_{b,1}^{\vec{e}'} \eqdef \left\{ \vec{h} \in \mathbb{F}_{2}^{n}: \; (\vec{e}',\vec{h}) \in \sE_{b} \mbox{ and } \vec{h}_{\sP} \neq \vec{e}' \right\},
	$$
	$$
		\sE_{b,2}^{\vec{e}'} \eqdef \left\{ \vec{h} \in \mathbb{F}_{2}^{n}: \; (\vec{e}',\vec{h}) \in \sE_{b} \mbox{ and } \vec{h}_{\sP} = \vec{e}' \right\}.
	$$
	But recall according to Equation \eqref{eq:cardEb12} that, 
	$$
	\card{\sE_{b,1}} = \frac{\card{\sE_{b,2}}}{2^{s}}
	$$  
	Therefore, in Subcase \ref{subcase:3},
	\begin{align*} 
	 \textup{Cov}_{3} &= \sum_{\vec{e}_{0}'\neq\vec{e}_{1}'}\sum_{\substack{\vec{h}^{0}=\vec{h}^{1} \in \sF_{0}^{\vec{e}_{0}',\vec{e}_{1}'}}} \textup{Cov}^{(0)} + 2\sum_{\vec{e}_{0}'\neq\vec{e}_{1}'}\sum_{\substack{\vec{h}^{0}=\vec{h}^{1} \in \sF_{1}^{\vec{e}_{0}',\vec{e}_{1}'}  }} \textup{Cov}^{(1)}\\ 
	 &= \sum_{\vec{e}_{0}'\neq\vec{e}_{1}'}\sum_{\substack{\vec{h}^{0}=\vec{h}^{1} \in \sF_{0}^{\vec{e}_{0}',\vec{e}_{1}'} }} \frac{1}{2^{k+2s-2\kbis}} + 2\sum_{\vec{e}_{0}'\neq\vec{e}_{1}'}\sum_{\substack{\vec{h}^{0}=\vec{h}^{1} \in \sF_{1}^{\vec{e}_{0}',\vec{e}_{1}'} }} \frac{1}{2^{k+s-\kbis}} \quad \mbox{(Equations \eqref{eq:subcase03}) and \eqref{eq:subcase13}} \\
	 &=  \sum_{\vec{e}_{0}'\neq\vec{e}_{1}'}\sum_{\substack{\vec{h}^{0}=\vec{h}^{1} \in \sF_{0}^{\vec{e}_{0}',\vec{e}_{1}'} }} \frac{1}{2^{k+2s-2\kbis}}  + \frac{1}{2^{s}}  \sum_{\vec{e}_{0}'\neq\vec{e}_{1}'}\sum_{\substack{\vec{h}^{0}=\vec{h}^{1} \in \sF_{0}^{\vec{e}_{0}',\vec{e}_{1}'} }} \frac{1}{2^{k+s-\kbis}} \\  
	 &= \OO{1}  \sum_{\vec{e}_{0}'\neq\vec{e}_{1}'}\sum_{\substack{\vec{h}^{0}=\vec{h}^{1} \in \sF_{0}^{\vec{e}_{0}',\vec{e}_{1}'}}} \frac{1}{2^{k+2s-2\kbis}} 
	\end{align*} 
	where we basically use the same reasoning than for proving Equation \eqref{eq:expE12b}. There in this subcase,
	\begin{align*} 
		\textup{Cov}_{3}
		&\leq \sum_{(\vec{e}_{0}',\vec{h}^{0}) \in \mathcal{E}_{b}} \sum_{\substack{ (\vec{e}_{1}',\vec{h}^{0})\in \mathcal{E}_{b} \\ \vec{h}^{1} = \vec{h}^{0}, \ev_0' \neq \ev_1' }} \frac{1}{2^{k+2s-2\kbis}} 
		&\leq \sum_{(\vec{e}_{0}',\vec{h}^{0}) \in \mathcal{E}_{b}} \frac{\binom{s}{t_{\textup{aux}}}}{2^{k + 2s - 2\kbis}} 
		&= \frac{\binom{s}{t_{\textup{aux}}}}{2^{s-\kbis}} \; \frac{E_{b}'}{2^{k+s-\kbis}}
	\end{align*}
	\newline

	{\bf \noindent Subcase \ref{subcase:4}:} We have,
	$$
	\textup{Cov}_{4} = 0. 
	$$

	{\bf \noindent Subcase \ref{subcase:5}:} We have,
			\begin{align*}
		\textup{Cov}_{5} &\leq \sum_{(\vec{e}_{0}',\vec{h}^{0}) \in \mathcal{E}_{b}} \sum_{\substack{(\vec{e}_{1}',\vec{h}^{1})\in \mathcal{E}_{b} \\ \vec{h}^{1} \neq \vec{h}^{0},\; \vec{h}^{0}_{\sP} \neq \vec{h}^{1}_{\sP}, \;\ev_0' \neq \ev_1' \\ \ev_0' + \hv^0_{\sP} =  \ev_1' + \hv^1_{\sP}}}   \frac{1}{ 2^{2k+s-\kbis}}
		\leq  \sum_{(\vec{e}_{0},\vec{h}^{0}) \in \mathcal{E}_{b}}  \frac{\binom{s}{t_{\textup{aux}}} \binom{n-s}{w}}{ 2^{2k+s-\kbis}} \\
		&= \frac{\binom{n-s}{w}\binom{s}{t_{\textup{aux}}}}{ 2^{k}} \; \frac{E_b'}{2^{k+s-\kbis}}
	\end{align*}

	{\bf \noindent Subcase \ref{subcase:6}:} We have,
	$$
	\textup{Cov}_{6} = 0.
	$$

		Plugging all these bounds on the $\textup{Cov}_{i}$ together and using that $\textup{Cov} = \sum_{i} \textup{Cov}_{i}$ leads to 
	\begin{align*} 
		\textup{Cov} &\leq \left( \frac{\binom{n-s}{w}  }{2^{k}} +  \frac{\binom{s}{t_{\textup{aux}}}}{2^{s - \kbis}}+ \frac{\binom{n-s}{w}  \binom{s}{t_{\textup{aux}}}  }{2^{k}}\right)\; \frac{E_{b}'}{2^{k+s-\kbis}} \\
		&= \OO{n^{\alpha}} \frac{E_{b}'}{2^{k+s-\kbis}} 
	\end{align*} 
	where in the last lines we used the constraints \eqref{eq:constraintsVar} given in the proposition. Plugging this equation in  Equation \eqref{eq:ineqVariant} concludes the proof.
	\end{proof}

We are now ready to prove our proposition:
\begin{proof}[Proof of Proposition \ref{prop:biasCodedRLPN}] Let $E_{b}$ and $E_{b}'$ (for $b \in \{0,1\}$) be defined as in Lemma \ref{lemma:varianceEb}. 
	By using the Bienaym\'e-Tchebychev inequality, we obtain for any function $f$ mapping the positive integers to positive real numbers:
	\begin{equation*}
		\mathbb{P}_{\vec{H},\Hmbis}\left( |E_b- \mathbb{E}\left( E_b \right)|\geq  \sqrt{f(n) \mathbb{E}\left(E_b\right)}\right) \leq  \frac{\mathbf{Var}(E_b)}{f(n) \mathbb{E}(E_b)} 
		= \frac{\OO{n^{\alpha}}}{f(n)}
	\end{equation*}
	where the last inequality is a consequence of Lemma \ref{lemma:varianceEb}.		
	Since,
	$$
	\bias_{\left(\hv,\cvbis\right) \drawn \widetilde{\sH}}\left(  \langle \cvbis + \hv_{\sP},\ev_{\sP} \rangle +  \langle \ev_{\sN},  \hv_{\sN}\rangle \right) = \frac{E_0-E_1}{E_0+E_1},
	$$ 
	we have with probability greater than
	$1 - \frac{\OO{n^{\alpha}}}{f(n)}$ that
	\begin{multline}\label{eq:complicated}
		\frac{\mu_0-\mu_1 - \sqrt{2f(n)} \sqrt{ \mu_0+\mu_1} }{\mu_0+\mu_1 + \sqrt{2f(n)} \sqrt{ \mu_0+\mu_1}}	 \leq \bias_{\left(\hv,\cvbis\right) \drawn \widetilde{\sH}}\left(  \langle \cvbis + \hv_{\sP},\ev_{\sP} \rangle +  \langle \ev_{\sN},  \hv_{\sN}\rangle \right) \\
		  \leq  
		\frac{\mu_0-\mu_1 + \sqrt{2f(n)}\sqrt{ \mu_0+\mu_1} }{\mu_0+\mu_1 - \sqrt{2f(n)}\sqrt{ \mu_0+\mu_1}}	 \end{multline}
	where 
	$$
	\mu_i \eqdef \mathbb{E}\left(E_i\right)
	$$	
	and where we used that for all positive $x$ and $y$, $\sqrt{x}+\sqrt{y} \leq \sqrt{2(x+y)}$. 
	Let, 
	$$
	N = \mu_{0} + \mu_{1} 
	$$
	It is readily seen that,
	$$
	N = \frac{\binom{n-s}{w} \binom{s}{\tbis}}{2^{k-\kbis}} 
	$$
	We let 
	$f(n) = \delta\sqrt{N}/2$. Since $N= \mu_0+\mu_1$ this implies $f(n) = \delta\sqrt{\mu_0+\mu_1}/2$. By Equation \eqref{eq:cstPropo35}, note that $\frac{f(n)}{\OO{n^{\alpha}}}$ tends to infinity as 
	$n$ tends to infinity.  We notice that 
	\begin{align*}
		\sqrt{2f(n)} \sqrt{ \mu_0+\mu_1} &=  \delta^{1/2}(\mu_0+\mu_1)^{3/4} \nonumber \\
		&= o\left( \delta(\mu_0+\mu_1)\right) 
	\end{align*}
	because 
	\begin{equation*}
		\frac{\delta^{1/2}(\mu_0+\mu_1)^{3/4}}{\delta(\mu_0+\mu_1)} 
		= 	\frac{1}{\sqrt{\delta\sqrt{\mu_0+\mu_1}}} 
		=  \frac{1}{\sqrt{2f(n)}} 
		\mathop{\longrightarrow}\limits_{n\to +\infty}  0.
	\end{equation*}
	Equation \eqref{eq:complicated} can now be rewritten as
	\begin{multline}
		\label{eq:simpler}
		\frac{\mu_0-\mu_1 - o\left(\delta(\mu_0+\mu_1)\right) }{\mu_0+\mu_1 + o\left(\delta(\mu_0+\mu_1)\right) }	 \leq \bias\left( \langle \ev'',\vec{e}_{\sP} \rangle +  \langle \vec{e}_{\sN},  \vec{h}_{\sN}\rangle  \right) \\
		 \leq  
		\frac{\mu_0-\mu_1 + o\left(\delta(\mu_0+\mu_1)\right)  }{\mu_0+\mu_1 - o\left(\delta(\mu_0+\mu_1)\right)}	
	\end{multline}Now on the other hand
	\begin{align*}
		\delta  &= \bias_{\left(\hv,\cvbis\right) \drawn \widetilde{\sH}}\left(  \langle \cvbis + \hv_{\sP},\ev_{\sP} \rangle +  \langle \ev_{\sN},  \hv_{\sN}\rangle \right)   \\
		&= \frac{E'_0-E'_1}{E'_0+E'_1} \\
		&= \frac{\frac{E'_0}{2^{k+s-\kbis}}-\frac{E'_1}{2^{k+s-\kbis}}}{\frac{E'_0}{2^{k+s-\kbis}}+\frac{E'_1}{2^{k+s-\kbis}}} \\
		&= \frac{\mu_0-\mu_1}{\mu_0+\mu_1}
	\end{align*}
	where the last equality is a consequence of Lemma \ref{lemma:varianceEb}, in particular Equation \eqref{eq:espL}. 
	From this it follows that we can rewrite \eqref{eq:simpler} as
	\begin{multline*}
		\frac{\delta}{1+o(\delta)}- o(\delta) \leq \bias_{\left(\hv,\cvbis\right) \drawn \widetilde{\sH}}\left(  \langle \cvbis + \hv_{\sP},\ev_{\sP} \rangle +  \langle \ev_{\sN},  \hv_{\sN}\rangle \right) \\
		   \leq 
		\frac{\delta}{1-o(\delta)}+ o(\delta) 
	\end{multline*}
	from which it follows immediately that 
	$$
	\bias_{\left(\hv,\cvbis\right) \drawn \widetilde{\sH}}\left(  \langle \cvbis + \hv_{\sP},\ev_{\sP} \rangle +  \langle \ev_{\sN},  \hv_{\sN}\rangle \right)  = \delta(1+o(1))
	$$
	which concludes the proof. 
\end{proof} \section{Correctness and Running-Time of the \nRLPNs algorithm (Algorithm \ref{alg:codedRLPN})}\label{app:complexity}
In this section we prove the correctness of Algorithm \ref{alg:codedRLPN} in Subsection \ref{subsec:Correc}. Furthermore, we give its running-time in Subsection \ref{sec:app_final_comp}. To this aim, we instantiate Instructions \ref{state:PCE} ($\Call{ParityCheckEquations}{}$) \ref{state:Dec} ($\Call{Decode}{}$) of Algorithm \ref{alg:sample}.
\begin{notation}
	\begin{itemize}
		\item Instantiation of Algorithm \ref{alg:codedRLPN}.
	\begin{itemize}
		\item We instantiate $\Call{ParityCheckEquations}{}$ instruction with the technique devised in \cite[\S 5]{CDMT22} to compute all the parity-checks of a given weight in a code (which is inspired from \cite{BJMM12}). Its asymptotic complexity is recalled in Proposition \ref{prop:asym_BJMM}.
		\item The family of auxiliary $[s,\kbis]$ linear codes $\CCbis$ used will be product of $\log s$ random codes as devised in Section \ref{sec:SDC} . Using these the $\Call{Decode}{}$ procedure outputs almost all codeword at distance $\tbis$ in time $2^{o(s)} \max\left(\frac{\binom{s}{\tbis}}{2^{s-\kbis}} , 1 \right)$.
\end{itemize}
	\item Framework for the analysis of Algorithm \ref{alg:codedRLPN}
	\begin{itemize}
			\item We prove the correctness (Proposition \ref{prop:correctness}) and we make the complexity analysis (Proposition \ref{prop:asym_comp_doubleRLPN}) in the framework of Proposition \ref{prop:exp_sizeS}. More specifically, analysis is made for 
$\CC$ and $\CCbis$ being random $\lbrack n,k \rbrack$ and $\lbrack s,k\bis \rbrack$ codes. We argue in Section \ref{sec:SDC} that the proof would be roughly the same (but more complicated) if we were to make it using $\CCbis$ being random product codes. The complexity of Algorithm \ref{alg:codedRLPN} would only grow by a factor of $2^{o(s)}$ when using these codes.
		\item Note that with the $\Call{ParityCheckEquations}{}$ and $\Call{Decode}{}$ procedures we have chosen, Algorithm \ref{alg:codedRLPN} computes in fact all the possible LPN samples, namely we have (as required in Proposition \ref{prop:exp_sizeS}):
		$$ \sH = \widetilde{\sH} $$
		\item We reuse notation introduced in Proposition \ref{prop:exp_sizeS}: the set $\sS^{(j)}$ of candidates for the $j$'th auxiliary code $\CCbis[(j)]$ is defined by
		\begin{equation} \label{eq:prop_def_S2}
			\sS^{(j)} \eqdef \{\sv \in \F_2^\kbis \: : \widehat{f_{\yv,\widetilde{\sH},\Gmbis[(j)] }  }\left(\sv\right) \geq  \frac{\delta}{2} \: \widetilde{H}\},
		\end{equation}
		where 
		\begin{equation}
			\widetilde{H} \eqdef \frac{\binom{n-s}{w} \binom{s}{\tbis}}{2^{k-\kbis}}, 	\qquad \delta \eqdef \frac{K_w^{(n-s)}\left(u\right) K_{\tbis}^{(s)}\left(t-u\right) }{\binom{n-s}{w} \binom{s}{\tbis}}.
		\end{equation}
	\end{itemize}
	\end{itemize}
\end{notation}
\subsection{Correctness of the algorithm}\label{subsec:Correc}
The goal of this section is to prove that  \nRLPN, namely Algorithm \ref{alg:codedRLPN}, outputs the desired error vector $\ev$ after essentially $\Niter \approx \frac{\binom{n}{t}}{\binom{s}{t-u} \binom{n-s}{u} }$ iterations of the outer loop (Line \eqref{lst:line:while} of Algorithm \ref{alg:codedRLPN}). This is given by the following proposition.
\begin{proposition} \label{prop:correctness}
	Let $\CC$ be a code taken uniformly at random among the $[n,k]$ linear codes and $\CCbis[(1)],\dots, \CCbis[(N_{\textup{aux}})]$ which are $N_{\textup{aux}}$ codes taken uniformly at random among the $[s,\kbis]$ linear codes.  Let $\yv \eqdef \cv + \ev$ where $\cv \in \CC$ and where $\ev \in \mathcal{S}_t^{n}$ is a fixed error vector of weight $t$.
	As long as the parameters $s,\kbis,\tbis,w,u$ verify the Parameters constraint \eqref{ass:parameters_doubleRLPN} and as long as $N_{\textup{iter}} = \omega \left(\frac{\binom{n}{t}}{\binom{s}{t-u} \binom{n-s}{u} }\right)$ and $N_{\textup{aux}} = \OO{1}$, Algorithm \ref{alg:codedRLPN} outputs the error vector $\ev$ with probability $1-o(1)$. \end{proposition}

It is readily seen that when $\Niter = \omega \left(\frac{\binom{n}{t}}{\binom{s}{t-u} \binom{n-s}{u} }\right)$ then, with probability $1-o(1)$ over the choice of $\sP$ there exists an iteration such that $|\ev_{\sN}|=u$. We only have left to show that for such an iteration we have with high probability that $\ev_{\sP} \left(\Gmbis[(j)]\right)^{\top} \in \sS^{(j)}$ for $j=1, \dots,\Naux$ which is the purpose of the following lemma.
\begin{lemma}
	Let us reuse the setting of Proposition \ref{prop:correctness}. Moreover, let us fix $\sP$ and $\sN$ two complementary sets of $\llbracket 1,n\rrbracket$ of size $s$ and $n-s$ respectively and such that $|\ev_{\sN} | = u$. Let us denote by  $\Gmbis[(1)] , \dots, \Gmbis[(N_{\textup{aux}})]$  the generators matrices of the codes $\CCbis[(1)] , \dots, \CCbis[(N_{\textup{aux}})]$ respectively. Then,\begin{equation} \label{eq:eq_correctness}
		\prob{}{\bigcap_{j=1}^{N_{\textup{aux}}} ``\ev_{\sP} \: \transp{\Gmbis[(j)]}\in \sS^{(j)}"} = 1-o(1).
	\end{equation}
\end{lemma}
\begin{proof}
	First, notice that
	\begin{align*}
		\prob{}{\bigcap_{j=1}^{\Naux}  ``\ev_{\sP} \: \transp{\Gmbis[(j)]}\in \sS^{(j)}"} &= 1 - \prob{}{\bigcup_{j=1}^{\Naux}  ``\ev_{\sP} \: \transp{\Gmbis[(j)]}\notin \sS^{(j)}"} \\
		&\geq 1 - \sum_{j=1}^{\Naux} \prob{}{``\ev_{\sP} \: \transp{\Gmbis[(j)]}\notin \sS^{(j)}"}
	\end{align*}
		where we used the union-bound. 
	 Now, as $\Naux = \OO{1}$ we only have to show that $\prob{}{``\ev_{\sP} \: \transp{\Gmbis[(j)]}\notin \sS^{(j)}"}=o(1)$ to prove Equation \eqref{eq:eq_correctness}. By using Fact  \ref{fact:fundamental1},
	\begin{equation*}
		"\ev_{\sP} \: \transp{\Gmbis[(j)]}\notin \sS^{(j)}" \Longleftrightarrow  \bias_{\left(\hv,\cvbis\right) \drawn \widetilde{\sH}} \left(\langle \yv, \hv \rangle  + \langle \ev_{\sP}, \cvbis \rangle\right)  < \frac{\delta}{2} \frac{\widetilde{H}}{\card{\widetilde{\sH}}}
\end{equation*}
	Our aim is to show,
	$$
	\prob{}{\bias_{\left(\hv,\cvbis\right) \drawn \widetilde{\sH}} \left(\langle \yv, \hv \rangle  + \langle \ev_{\sP}, \cvbis \rangle\right)  < \frac{\delta}{2} \frac{\widetilde{H}}{\card{\widetilde{\sH}}} } = o(1).
	$$ 
	To simplify notation let $b \eqdef \bias_{\left(\hv,\cvbis\right) \drawn \widetilde{\sH}} \left(\langle \yv, \hv \rangle  + \langle \ev_{\sP}, \cvbis \rangle\right)$.
It is readily seen that $\expect{\CC,\CCbis}{\card{\widetilde{\sH}}} = \widetilde{H} \left(1+o(1)\right)$ and $\mathbf{Var}_{\CC,\CCbis}\left( \card{\widetilde{\sH}} \right) \leq \widetilde{H}(1+o(1))$. Therefore, by using Bienaymé-Tchebychev inequality, for any $a$,
	$$
	\mathbb{P}\left( \left| \card{\widetilde{\sH}} - (1+o(1))\widetilde{H} \right| \geq \sqrt{a \widetilde{H}} \right) \leq \frac{1+o(1)}{a}
	$$
	We have the following computation,
	\begin{multline*}
			\prob{}{b < \frac{\delta}{2} \frac{\widetilde{H}}{\card{\widetilde{\sH}}}} = 	\mathbb{P}\left( b < \frac{\delta}{2} \frac{\widetilde{H}}{\card{\widetilde{\sH}}} \mid  \left| \card{\widetilde{\sH}} - (1+o(1))\widetilde{H}\right| \geq \sqrt{a\widetilde{H}} \right) 	\mathbb{P}\left( \left| \card{\widetilde{\sH}} - (1+o(1))\widetilde{H} \right| \geq \sqrt{a \widetilde{H}} \right) \\
			+ 	\mathbb{P}\left( b < \frac{\delta}{2} \frac{\widetilde{H}}{\card{\widetilde{\sH}}} \mid  \left| \card{\widetilde{\sH}} - (1+o(1))\widetilde{H}\right| < \sqrt{a\widetilde{H}} \right) 	\mathbb{P}\left( \left| \card{\widetilde{\sH}} - (1+o(1))\widetilde{H} \right| < \sqrt{a \widetilde{H}} \right)
	\end{multline*}
	Therefore,
	\begin{align*}
			\prob{}{b < \frac{\delta}{2} \frac{\widetilde{H}}{\card{\widetilde{\sH}}}} &\leq 	\mathbb{P}\left( \left| \card{\widetilde{\sH}} - (1+o(1))\widetilde{H} \right| \geq \sqrt{a \widetilde{H}} \right) + \mathbb{P}\left( b < \frac{\delta}{2} \frac{\widetilde{H}}{(1+o(1))\widetilde{H} - \sqrt{a\widetilde{H}} } \right) \\
			&\leq \frac{1+o(1)}{a} + \mathbb{P}\left( b < \frac{\delta}{2} \frac{1}{(1+o(1)) - \sqrt{\frac{a}{\widetilde{H}} }} \right)
	\end{align*}
	Let us choose $a = \widetilde{H}^{1/2}$. Recall that $\widetilde{H} = \om{\frac{n^{\alpha + 8}}{\delta^2}}$ where $\delta \leq 1$. Therefore, 
	\begin{align*} 
	\prob{}{b < \frac{\delta}{2} \frac{\widetilde{H}}{\card{\widetilde{\sH}}}}  &= o(1) + \mathbb{P}\left( b <  \frac{\delta}{2}(1+ o(1)) \right) \\
	&=o(1)
	\end{align*} 
	where in the last equality we used Proposition \ref{eq:cstPropo35}. It concludes the proof. 
\end{proof}
\subsection{Asymptotic complexity of \nRLPNs} \label{sec:app_final_comp}
We now have every tool to give the complexity of our algorithm, namely, we can compute the expected number of candidates at each iteration given by Proposition \ref{prop:exp_sizeS} and we have the correctness of our algorithm which is given by Proposition \ref{prop:correctness}.
\begin{proposition}{Asymptotic complexity exponent of the \nRLPNs algorithm.} \label{prop:asym_comp_doubleRLPN} Define
	\begin{align*}
	 &R \eqdef \lim\limits_{n \to \infty}\frac{k}{n}, \quad \tau \eqdef \lim\limits_{n \to \infty}\frac{t}{n},   \quad \sigma \eqdef  \lim\limits_{n \to \infty} \frac{s}{n}, \quad R_{\textup{aux}} \eqdef  \lim\limits_{n \to \infty} \frac{\kbis}{n}\\ &  \quad \tau_{\textup{aux}} \eqdef  \lim\limits_{n \to \infty} \frac{\tbis}{n}, \omega \eqdef  \lim\limits_{n \to \infty} \frac{w}{n}, \quad \mu \eqdef  \lim\limits_{n \to \infty} \frac{u}{n}
	\end{align*}
	Suppose that de $\Call{Decode}{}$ procedure has an expected time complexity of $2^{n \cdot \sigma \cdot o(1)}$.
	\noindent The expected complexity of the \nRLPNs algorithm to decode a code of rate $R$ at relative distance $\tau$ is upper bounded by
		$2^{n \: \left(\alpha_{\textup{\nRLPNs}} + o(1)\right)}$
where
\begin{eqnarray*}
\alpha_{\textup{\nRLPNs}}&\leq & - \pi + \max\left( \left(1-\sigma\right) \: \cdot \:\alpha_{\textup{eq}}\left(\frac{R-\sigma}{1-\sigma}, \frac{\omega}{1-\sigma}\right) , \; \nu_{\textup{sample}}, \; R_{\textup{aux}}, \; \nu_{\textup{candidate}} \: \cdot \:  N_{\textup{aux}} +  \alpha_{\textup{ISD}} \right)
	\end{eqnarray*}
	where
	\begin{eqnarray*}
	\pi & \eqdef & h \left(\tau \right) - \sigma  \: \cdot \: h_2 \left(\frac{\tau - \mu }{\sigma}\right) -  (1- \sigma)  \: \cdot \: h \left(\frac{\mu }{1-\sigma}\right), \\
	\nu_{\textup{samples}} &\eqdef& \left(1 - \sigma\right)  \: \cdot \: h_2 \left(\frac{\omega}{1-\sigma}\right)  + \sigma  \; \cdot \; h_2 \left(\frac{\tau_{\textup{aux}}}{\sigma}\right)- (R - R_{\textup{aux}}), \\
	\alpha_{\textup{ISD}} & \eqdef & \max\left(\sigma \: \cdot \: \alpha_{\textup{ISD-Dumer}}\left( 1 - \frac{N_{\textup{aux}} \cdot \tau_{\textup{aux}} }{\sigma}, \frac{\tau - \mu}{\sigma}\right), \nu_{\textup{ISD}} + \left(1- \sigma\right)\; \cdot \; \alpha_{\textup{ISD-Dumer}}\left(\frac{R-\sigma}{1- \sigma}, \frac{\mu}{1-\sigma}\right) \right), \\
	 \nu_{\textup{ISD}} &\eqdef& \max\left(\sigma  \: \cdot \: h_2 \left(\frac{\tau - \mu }{\sigma}\right) -N_{\textup{aux}} \: \cdot \: \tau_{\textup{aux}},\: 0  \right)\\
	 		\nu_{\textup{candidates}} &\eqdef& \max \left(\max_{(\eta,\zeta) \in \mathcal{A}} \sigma \: \cdot \: h_2\left(\frac{\zeta}{\sigma}\right) + (1-\sigma) \: \cdot \:  h_2 \left(\frac{\eta}{1-\sigma}\right) - (1-R), \; 0\right) ,
\end{eqnarray*}
with
\begin{multline*}
	\mathcal{A}  \eqdef \Bigg\{ (\eta,\zeta) \in \left[0, 1-\sigma \right]\times \left[ 0,\sigma \right] \: : \sigma \left[\widetilde{\kappa} \left(\frac{\tbis}{\sigma},\frac{\tau-\mu}{\sigma}\right)  - \widetilde{\kappa} \left(\frac{\tbis}{\sigma},\frac{\zeta}{\sigma}\right) \right] +  \\   (1-\sigma) \left[\widetilde{\kappa} \left(\frac{\omega}{1-\sigma},\frac{\mu}{1-\sigma}\right)  - \widetilde{\kappa} \left(\frac{\omega}{1-\sigma},\frac{\eta}{1-\sigma}\right)  \right] \leq 0\Bigg\},
\end{multline*}
and
\begin{itemize}
	\item $\alpha_{\textup{eq}}(R',\tau')$ is the complexity exponent of $\Call{ParityCheckEquations}{}$ to compute all parity-checks of relative weight $\tau'$ of a code of rate $R'$. It is instantiated here with a technique devised in \cite[\S 5 Eq. (5.4)]{CDMT22} and its complexity is recalled in Proposition \ref{prop:asym_BJMM}.
	\item  $\alpha_{\textup{ISD-Dumer}}(R',\tau')$ is the complexity exponent of $\Call{Decode-Dumer}{}$ to return all the solutions to the decoding problem in a code of rate $R'$ at relative distance $\tau'$. Its complexity is recalled in Proposition \ref{prop:comp_ISD_decoder_dumer_asym}.
\end{itemize}
Moreover, $\sigma, \: R_{\textup{aux}}, \: \tau_{\textup{aux}}, \: \omega, \: \mu,$ are non-negative and such that
$$
\sigma \leq R, \quad \tau- \sigma \leq \mu\leq \tau, \quad \omega \leq 1 - \sigma,
$$
\begin{align}
\nu_{\textup{samples}} &\geq -2 \: \varepsilon_{\textup{bias}} , \\
0 &\geq (1-\sigma) \:  h_2 \left(\frac{\omega}{1-\sigma}\right) + \sigma h_2 \left(\frac{\tau_{\textup{aux}}}{\sigma}\right)  - R \\
0 &\geq \sigma \:  h_2 \left(\frac{\tau_{\textup{aux}}} {\sigma}\right) - (\sigma - R_{\textup{aux}})
\end{align}
where $\widetilde{\kappa}$ is the function defined in Proposition \ref{prop:expansion} and $$\varepsilon_{\textup{bias}} \eqdef  \: \sigma \: \left[ \widetilde{\kappa} \left(\frac{\tbis}{\sigma},\frac{\tau-\mu}{\sigma}\right) - h_2 \left(\frac{\tau_{\textup{aux}}}{\sigma}\right)\right] + \: (1-\sigma) \: \left[\widetilde{\kappa} \left(\frac{\omega}{1-\sigma},\frac{\mu}{1-\sigma}\right) - h_2 \left(\frac{\omega}{1-\sigma}\right)\right].$$ 
Finally, we require that $N_{\textup{aux}}=\OO{1}$.
\end{proposition}
\begin{remark}
	In practice our parameters are such that we decode the auxiliary code at Gilbert-Varshamov distance, namely $\tau_{\textup{aux}} = \sigma \: h_2^{-1} \left(1-R_{\textup{aux}}\right)$.
\end{remark}
While initially Dumer's decoder \cite{D91} is designed to produce only one solution to the decoding problem it suffices to re-run it as many time as the number of solutions we expect from the decoding problem to find all of them. We get the following proposition giving the asymptotic complexity of the $\Call{Decode-Dumer}{}$ procedure.
\begin{proposition}[Asymptotic time complexity of ISD Decoder \cite{D91} to produce \textbf{all} solutions to the decoding problem] \label{prop:comp_ISD_decoder_dumer_asym}
	Let
	$ R \eqdef \lim\limits_{n \to \infty}\frac{k}{n}, \tau \eqdef \lim\limits_{n \to \infty}\frac{t}{n}$. Let $\ell$ and $w$ be two (implicit) parameters of the algorithm and define $ \lambda \eqdef \lim\limits_{n \to \infty}\frac{\ell}{n}, \omega \eqdef \lim\limits_{n \to \infty}\frac{w}{n}$.
	The time and space complexities of \cite{D91} to find a proportion $1-o(1)$ of all solutions to the decoding problem at distance $t$ on an $[n,k]$ linear code are given by $2^{n \: \left(\alpha_{\textup{ISD-Dumer}} + o(1)\right)}$ and $2^{n \: \left(\beta_{\textup{ISD-Dumer}}+ o(1)\right)}$ respectively where
	\begin{align}
		\alpha_{\textup{ISD-Dumer}} &\eqdef \min_{\omega,\lambda} \left( \pi +  \max\left(\frac{R + \lambda}{2}h_2 \left(\frac{\omega}{ R+\lambda}\right) , \left(R + \lambda\right)h_2 \left(\frac{\omega}{ R+\lambda}\right) - \lambda\right) \right), \\
		\pi &\eqdef h_2(\tau) - (1- R-\lambda) h_2 \left(\frac{\tau - \omega}{1- R-\lambda}\right)-  (R + \lambda) h_2 \left(\frac{\omega}{R+\lambda}\right), \\
		\nu_{\textup{sol}} &\eqdef \max \left(h_2\left(\tau\right) - (1-R), 0 \right), \\
			\beta_{\textup{ISD-Dumer}} &\eqdef \frac{R + \lambda}{2}h_2 \left(\frac{\omega}{R+\lambda}\right).
	\end{align}
	Moreover $\lambda$ and $\omega$ must verify the following constraints:
	$$ 0 \leq \lambda \leq 1- R, \qquad \max \left(R + \lambda + \tau - 1, 0\right)\leq \omega \leq \min \left(\tau,R + \lambda \right).$$
	The expected number of solutions is given by $2^{n \: \left(\nu_{\textup{sol}} + o(1)\right)}$.
\end{proposition}
We recall here the asymptotic complexity of the technique devised in \cite[\S 5, Equation (5.4)]{CDMT22} based on \cite{BJMM12} decoder to produce all parity-checks of low weight of a code.
\begin{proposition}{Asymptotic time complexity exponent of BJMM technique \cite{BJMM12}, \cite[\S 5, Equation (5.4)]{CDMT22} to produce all parity-checks of relative weight $\omega$ of a code of rate $R$} \label{prop:asym_BJMM}
	\begin{equation}
		\alpha_{\text{BJMM}}(R,\omega) \eqdef \min_{\pi_1,\pi_2, \lambda_1,\lambda_2} \gamma
	\end{equation}
	where
	\begin{align*}
		&\gamma \eqdef \max(\gamma_1,\gamma_2,\gamma_3) \\
			&\gamma_1 \eqdef \max\left(\nu_0, \: 2 \: \nu_0 - \lambda_1 \right), \quad \gamma_2\eqdef  \max\left(\nu_1, \: 2 \: \nu_1 - \left( \lambda_2 - \lambda_1\right) \right), \quad	\gamma_3 \eqdef  \max\left(\nu_2, \: 2 \: \nu_2 - \left( \lambda - \lambda_2\right) \right) \\
		&\nu_0 \eqdef \frac{h(\pi_1)}{2}, \quad \nu_1 \eqdef h(\pi_1) - \lambda_1, \quad\nu_2 \eqdef h(\pi_2) - \lambda_2, \quad \nu_3 \eqdef h(\omega) - \lambda.
	\end{align*}

	and the constraint region $\sR$ is defined by the sub-region of nonnegative tuples $(\pi_1,\pi_2,\lambda_1,\lambda_2)$ such that
	$$
	\lambda_1 \leq \lambda_2 \leq \lambda, \quad \pi_1 \leq \pi_2 \leq \pi, \quad
	\pi_2 \leq 2 \pi_1, \quad \pi \leq 2\pi_2, \quad 
	\pi_2 \leq \lambda_1, \quad \pi \leq \lambda_2, 
	$$
	and
	\begin{eqnarray}
		\lambda_1 & = &  \pi_2 + (1 - \pi_2)h\left( \frac{\pi_1-\pi_2/2}{1-\pi_2}\right),\\
		\lambda_2 & = &  \omega + (1 - \omega)h\left( \frac{\pi_2-\omega/2}{1-\omega}\right)
	\end{eqnarray}
		The expected number of parity-checks computed is given by $2^{n \: \left(\nu_3 + o(1) \right)}$
\end{proposition}

 \section{Proof of Proposition \ref{prop:exp_bias}} \label{app:proofBias}

Let us recall Proposition \ref{prop:exp_bias}.
\propDualityBias*
Let us devise a more convenient expression for $\langle \yv, \hv \rangle  + \langle \xv, \cvbis \rangle$. By noticing that $\hv_{\sP}$ and $\hv_{\sN}$ are linearly linked we get the following lemma.
\begin{lemma} \label{lem:expr_error_doubleRLPN}
	Let $\sP$ and $\sN$ be two complementary subsets of $\llbracket 1 , n \rrbracket$ of size $s$ and $n-s$ respectively. Let $\CC$ and $\CCbis$ be two $[n,k]$ linear and $[s,\kbis]$ linear codes respectively such that $\CC_{\sP}$ is of dimension $s$. Let $\xv \in \F_2^s$ and $\left(\hv,\cvbis\right) \in \widetilde{\sH}$ where recall that 
	$$ 	\widetilde{\sH} \eqdef \{ \left(\hv,\cvbis\right) \in \CC^{\perp} \times \CCbis \: : |\hv_{\sN}|=w \mbox{ and } |\hv_{\sP} + \cvbisperp|=\tbis \}. $$
	We have that \begin{equation} \label{rel:noise_doubleRLPN}
		\langle \yv, \hv \rangle + \langle \xv, \cvbis \rangle = \langle \left(\xv + \ev_{\sP} \right) \Rm + \ev_{\sN}, \hv_{\sN}\rangle + \langle \xv , \hv_{\sP} + \cvbis \rangle
	\end{equation} 
	where $\Rm \in \F_2^{s \times (n-s)}$ is (independently of the parity-check $\vec{h}$) such that
	\begin{equation}\label{rel:parity}
		\hv_{\sP}= \hv_{\sN} \Rm^{\transpose}.
	\end{equation}
\end{lemma}
\begin{proof}
	First, let us show Equation \eqref{rel:parity}. Suppose without loss of generality that $\sP = \llbracket 1,s \rrbracket$ and $\sN = \llbracket s+1,n \rrbracket$. Let $\Gm \in \F_2^{k \times n}$ be a generator matrix of $\CC$. Because $\CC_{\sP}$ is of dimension $s$ there exists an invertible $\Jm \in \F_2^{k \times k}$ such that
	$$\Jm \Gm =  \begin{pmatrix}
		\vec{Id}_{s} & \Rm \\
		\vec{0}_{k-s} & \Rm'
	\end{pmatrix} $$
	where $\Rm \in \F_2^{s \times (n-s)}$ and $\Rm' \in \F_2^{(k-s)\times (n-s)}$. Furthermore, $\Jm \Gm$ is another generator matrix for $\CC$. Therefore for any $\hv \in \CC^{\perp}$ we have $\Jm \Gm \: \transp{\hv} = \vec{0}$. Since $\Jm \Gm \transp{\hv}= \transp{\hv_{\sP}}+\Rm \transp{\hv_{\sN}}$, this gives \eqref{rel:parity}.
	Now, let us prove \eqref{rel:noise_doubleRLPN}.
	Recall that, using Equation \eqref{rel:parity} we have:
	\begin{align*}
		\langle \yv, \hv\rangle &= \langle \ev_{\sP}, \hv_{\sP}\rangle + \langle \ev_{\sN}, \hv_{\sN}\rangle \\
		&= \langle \ev_{\sP}, \hv_{\sN} \Rm^{\top} \rangle + \langle \ev_{\sN}, \hv_{\sN}\rangle  \\
		&= \langle \ev_{\sP} \Rm + \ev_{\sN}, \hv_{\sN}\rangle,
	\end{align*}
	and
	\begin{align*}
		\langle \xv, \cvbis \rangle &=  \langle \xv, \hv_{\sP}\rangle +  \langle \xv,  \hv_{\sP} + \cvbis \rangle \\
		&= \langle \xv \Rm , \hv_{\sN} \rangle + \langle \xv ,  \hv_{\sP} + \cvbis \rangle
	\end{align*}
	where in the last equality we used Equation \eqref{rel:parity}. 
	This concludes the proof.
\end{proof}

\begin{proof}[Proof of Proposition \ref{prop:exp_bias}]
	Let us consider $\Rm \in \F_2^{s \times (n-s)}$ as in Lemma \ref{lem:expr_error_doubleRLPN} and let us prove Equation \eqref{eq:prop_bias_doubleRLPN}.
	By definition of the bias and $\widetilde{\sH}$ given Equation \eqref{rel:parity}  we have the following computation and ,
	\begin{align*}
	\bias_{\left(\hv,\cvbis \right)) \drawn \widetilde{\sH}}&   \left(	\langle \yv, \hv \rangle + \langle \xv, \cvbis \rangle \right) \\
	&= \frac{1}{\card{\widetilde{\sH}}} \sum_{ \left( \hv, \cvbis \right) \in \widetilde{\sH}} (-1)^{\langle \yv, \hv \rangle + \langle \xv, \cvbis \rangle} \\
	&=  \frac{1}{\card{\widetilde{\sH}}} \sum_{ \substack{ \left( \hv_{\sN}, \cvbis\right) \in \left(\CC^{\perp}\right)_{\sN} \times \CCbis  \\   \left|\hv_{\sN}\right| = w ,\left| \Rm \hv_{\sN} + \cvbis \right| = \tbis}} (-1)^{\langle \left(\xv + \ev_{\sP} \right) \Rm + \ev_{\sN}, \hv_{\sN}\rangle + \langle \xv , \hv_{\sN} \transp{\Rm} + \cvbis \rangle} 
	\end{align*}
	where in the last equality we used Lemma \ref{lem:expr_error_doubleRLPN}. Therefore, 
	\begin{equation}\label{eq:bias} 
		\bias_{\left(\hv,\cvbis \right)) \drawn \widetilde{\sH}} = \frac{1}{\card{\widetilde{\sH}}} \sum_{ \substack{ \left( \hv_{\sN}, \cvbis\right) \in \left(\CC^{\perp}\right)_{\sN} \times \CCbis  } } f(\hv_{\sN}, \cvbis) 
	\end{equation}
	where,
	$$
	f(\hv_{\sN}, \cvbis) \eqdef   (-1)^{\langle \left(\xv + \ev_{\sP} \right) \Rm + \ev_{\sN}, \hv_{\sN} \rangle + \langle \xv ,  \hv_{\sN} \transp{\Rm} + \cvbis \rangle} \mathbf{1}_{ \{\left|\hv_{\sN}\right| = w ,\left| \hv_{\sN} \transp{\Rm} + \cvbis\right| = \tbis \}}.
	$$
Using Equation \ref{eq:shortPunct} we have that $\left(\CC^{\perp}\right)_{\sN} =  \left(\CC^{\sN}\right)^{\perp}$ and thus $\left( \left(\CC^{\sN}\right)^{\perp} \times \CCbis \right)^{\perp} =\CC^{\sN} \times \CCbisperp $. By using the Poisson formula
	(see \cite[Lemma 2, Ch. 5.2]{MS86}), together with the fact that $\dim\left(\CC^{\sN} \times \CCbisperp \right) = k-\kbis$, 
	we get
	\begin{align}\label{eq:Poisson} 
		\sum_{ \substack{ \left( \hv_{\sN}, \cvbis\right) \in \left(\CC^{\sN} \right)^{\perp} \times \CCbis  } } f(\hv_{\sN}, \cvbis) &= \frac{1}{2^{k-\kbis}}\sum_{ \substack{ \left( \cv^{\sN}, \cvbisperp\right) \in \CC^{\sN} \times \CCbisperp  } } \widehat{f}(\cv^{\sN}, \cvbisperp).
	\end{align}
	Let us compute the right-hand term. By definition of $f$, it is readily seen that 
	\begin{align*} 
	\widehat{f}(\vec{y}_{1},\vec{y}_{2}) &=\sum_{\substack{\vec{z}_{1}\in \F_{2}^{n-1},\vec{z}_{2}\in \F_{2}^{s} \\ |\vec{z}_{1}| = w, |\vec{z}_{1}\vec{R}^{\top}+\vec{z}_{2}| = \tbis } } (-1)^{\langle \vec{y}_{1},\vec{z}_{1} \rangle + \langle \vec{y}_{2},\vec{z}_{2}\rangle} (-1)^{\langle \left(\xv + \ev_{\sP} \right) \Rm + \ev_{\sN}, \vec{z}_{1} \rangle + \langle \xv ,  \vec{z}_{1} \transp{\Rm} + \vec{z}_{2} \rangle} \\ 
	&= \sum_{\vec{z}_{1}\in \F_{2}^{n-s}:|\vec{z}_{1}|=w}  (-1)^{\langle \vec{y}_{1}+\left(\xv + \ev_{\sP}+\vec{y}_{2} \right) \Rm + \ev_{\sN}, \vec{z}_{1} \rangle} \sum_{\substack{\vec{z}_{2}\in \F_{2}^{s}: |\vec{z}_{1}\transp{\vec{R}}+\vec{z}_{2}}| = \tbis } (-1)^{\langle\vec{y}_{2}+\vec{x},\vec{z}_{1}\transp{\vec{R}} +\vec{z}_{2} \rangle} \\
	&= K^{(n-s)}_{w}\left(\left| \vec{y}_{1} + (\vec{y}_{2}+\vec{x}+\vec{e}_{\sP})\vec{R}+\vec{e}_{\sN} \right| \right) \; K^{(s)}_{\tbis}\left( \left|\vec{y}_{2}+\vec{x} \right| \right)
	\end{align*} 
	where in the last equality we used Fact \ref{fact:krawtchouk_bias}. Plugging this into Equation \eqref{eq:Poisson}  and then into Equation \eqref{eq:bias} concludes the proof. 
\end{proof} \section{Proof of Proposition \ref{prop:exp_sizeS}} \label{app:appendix_sec_tools}
The proof of this Appendix is to prove Proposition \ref{prop:exp_sizeS} which we recall is given by
\propExpSizeS*
The proof is divided in the following steps.
\newline

{\bf \noindent Step 1:} in Lemma \ref{lem:link_S_bias} we show that the expected size of $\sS$ is related to the probability that the bias of $\langle \yv, \hv \rangle  + \langle \xv, \hv_{\sP} \rangle$ is superior to the threshold $\approx \frac{\delta}{2}$.
\newline

{\bf \noindent Step 2:}  We give an exponential bound on the aforementioned probability by using Poisson summation formula as it was done in the proof of Proposition \ref{prop:exp_bias}.

\subsubsection*{\textbf{Step 1}} \label{sec:link_S_bias}
Recall that we have from Equation $\ref{eq:prop_def_S}$ that 
$$ 	
\sS \eqdef \left\{\sv \in \F_2^\kbis \: : \widehat{f_{\yv,\widetilde{\sH},\Gmbis}}\left(\sv\right) \geq  \frac{\delta}{2} \: \widetilde{H}\right\}. 
$$
By using Lemma \ref{lem:fundamental} we get the following condition for an element $\sv$ to be a candidate:
\begin{fact} \label{fact:fundamental1} 	 Let $\sv \in \F_2^{\kbis}$ and $\xv \in \F_2^{s}$ such that $\xv \Gmbis[\transpose] = \sv$. We have,
$$
	\sv \in \sS \; \Longleftrightarrow  \bias_{\left(\hv,\cvbis\right) \drawn \widetilde{\sH}}\left(\langle \yv,\hv \rangle + \langle \xv,\cvbis \rangle\right)  \geq \frac{\delta}{2} \; \frac{\widetilde{H}}{\card{\widetilde{\sH}}}.
	$$
\end{fact}
From there, we can derive the following lemma linking the expected size of $\sS$ and the previous bias. 
\begin{lemma} \label{lem:link_S_bias}
	Under Distribution \ref{notation:randomCodes}, using notation of Proposition \ref{prop:exp_sizeS} and under the constraint of Proposition \ref{prop:exp_sizeS} that $\sH = \widetilde{\sH}$ we have
	$$ 
	\mathbb{E}_{\CC,\CCbis}\; \left(\card{\sS}\right) \leq 2^{\kbis} \prob{\CC,\CCbis,\xv}{\bias_{(\hv,\cvbis)\drawn \widetilde{\sH}}(\langle \yv, \hv \rangle  + \langle \xv, \cvbis \rangle) \geq \frac{\delta}{2} \; \frac{\widetilde{H}}{\card{\widetilde{\sH}}} } + 1$$
	where $\xv$ is taken uniformly at random in $\F_2^s \setminus \{\CCbisperp + \ev_{\sP}\}$.
\end{lemma}
\begin{proof} We have the following computation, 
	\begin{align*} 
		\expect{\CC,\CCbis}{\card{\sS}}&= \expect{}{\sum_{\sv \in \F_2^{\kbis}} \mathbf{1}_{\sv \in \sS}} \\
		&\leq1 + \expect{}{\sum_{\sv \in \F_2^{\kbis} \: : \sv \neq \ev_{\sP} \Gmbis[\transpose]} \mathbf{1}_{\sv \in \sS}} \\
		&=1 + \expect{}{\sum_{\sv \in \F_2^{\kbis} \: : \sv \neq \ev_{\sP} \Gmbis[\transpose]} \frac{1}{2^{s-\kbis}} \sum_{ \zv \in \F_2^s \: : \zv \: \Gmbis[\transpose] = \sv}\mathbf{1}_{\zv \Gmbis[\transpose] \in \sS}}
	\end{align*} 
where in the last equality we used that $\Gmbis$, which is a generator matrix of $\CCbis$, has rank $s-\kbis$ according to Distribution \ref{notation:randomCodes}.
Now, from the linearity of the expectation we get,
	\begin{align*}
		\expect{\CC,\CCbis}{\card{\sS}}&\leq  1 + \frac{1}{2^{s-\kbis}} \sum_{\zv \in \F_2^s \: : \zv \notin \CCbisperp + \ev_{\sP}} \prob{\CC,\CCbis}{\zv \Gmbis[\transpose] \in \sS} \\
&= 1 + 2^{\kbis} \:  \prob{\CC,\CCbis,\xv}{\xv \Gmbis[\transpose] \in \sS} 
	\end{align*}
	where in the last line we used that $\Gmbis$ has full rank, $\card{\CCbisperp + \ev_{\sP}} = 2^{s}$ and $\xv$ is taken uniformly at random in $\F_2^s \setminus \{\CCbisperp + \ev_{\sP}\}$. Using Fact \ref{fact:fundamental1} concludes the proof.
\end{proof}
\subsubsection*{Step 2}\label{sec:distrib_bias} The following lemma relates the upper-bound given in Lemma \ref{lem:link_S_bias} to the involved probability in Conjecture \ref{ass:bias_dom}. 
\begin{lemma} \label{lem:simpl_bias}
		Using Distribution \ref{notation:randomCodes} and notation of Proposition \ref{prop:exp_sizeS} we have
		\begin{multline*}
\prob{\CC,\CCbis,\xv}{\bias_{\left(\hv,\cvbis \right) \drawn \widetilde{\sH}} \left(\langle \yv, \hv \rangle  + \langle \xv, \cvbis \rangle\right) \geq \frac{\delta}{2} \frac{\widetilde{H}}{\card{\widetilde{\sH}}}  }   = \\ \prob{\CC,\CCbis,\xv}{ \sum_{j = 0}^{s} \sum_{i = 0}^{n-s} K_{\tbis}^{(s)} \left(j\right) K_w^{(n-s)} \left(i\right) N_{i,j} \geq \frac{1}{2} \: K_w^{(n-s)}\left(u\right) K_{\tbis}^{(s)}\left(t-u\right) }.
		\end{multline*}
		where $\xv$ is taken uniformly at random in $\F_2^s \setminus \{ \CCbisperp + \ev_{\sP}\}$.
	\end{lemma}
	\begin{proof}
		Recall that we have that from notation of Proposition \ref{prop:exp_sizeS},
		\begin{equation} \label{lem:simpl_bias_bias_def}
			\widetilde{H} \eqdef \frac{\binom{n-s}{w} \binom{s}{\tbis}}{2^{k-\kbis}}, \qquad \delta \eqdef \frac{K_w^{(n-s)}(u) K_{\tbis}^{(s)}(t-u) }{\binom{n-s}{w} \binom{s}{\tbis}}.
		\end{equation}
		According to Proposition \ref{prop:exp_bias} we have the following computation,
		\begin{align*}
			\mathbb{P}&\left( \bias_{\left(\hv,\cvbis \right) \drawn \widetilde{\sH}} \left(\langle \yv, \hv \rangle  + \langle \xv, \cvbis \rangle\right) \geq \frac{\delta}{2} \frac{\widetilde{H}}{\card{\widetilde{\sH}}} \right) \\ 
			&= \prob{}{\frac{1}{2^{k-\kbis}} \frac{1}{\card{\widetilde{\sH}}} \sum_{j = 0}^{s} \sum_{i = 0}^{n-s} K_{\tbis}^{(s)} \left(j\right) K_w^{n-s} \left(i\right) N_{i,j} \geq \frac{\delta}{2} \frac{\widetilde{H}}{\card{\widetilde{\sH}}}  } \\
			&=   \prob{}{\sum_{j = 0}^{s} \sum_{i = 0}^{n-s} K_{\tbis}^{(s)} \left(j\right) K_w^{n-s} \left(i\right) N_{i,j} \geq \frac{K_w^{(n-s)}\left(u\right) K_{\tbis}^{(s)}\left(t-u\right)}{2}}
		\end{align*}
		where in the last equality we used Equation \eqref{lem:simpl_bias_bias_def}. It concludes the proof. 
	\end{proof}
We will now use Conjecture \ref{ass:bias_dom} to bound the right-hand term of Lemma \ref{lem:simpl_bias} by the probability of the  event ``$N_{i,j} \neq 0$'' for some low $i$ and $j$ (more precisely when $K_w^{(n-s)}(u) K_{\tbis}^{(s)}(t-u) \leq K_w^{(n-s)}(i) K_{\tbis}^{(s)}(j)$). 
\begin{restatable}{lemma}{lemprobNij}\label{lem:probNij}
		Under Distribution \ref{notation:randomCodes} we have that for $(i,j) \in \llbracket 0, n-s  \rrbracket \times \llbracket 0, s\rrbracket$,
		$$ 
		\prob{\CC,\CCbis,\xv}{N_{i,j} \neq 0} = \OO{ \frac{\binom{s}{j} \binom{n-s}{i}}{2^{\kbis + n-k}} }
		$$
		where $\xv$ is taken uniformly at random in $\F_2^s \setminus \{ \CCbisperp + \ev_{\sP}\}$.
	\end{restatable}
	\begin{proof}
		This is proved in the first lemma of Appendix \ref{app:prop_Nij}.
	\end{proof}
	We are now ready to prove Proposition \ref{prop:exp_sizeS}.
	\begin{proof}[proof of Proposition \ref{prop:exp_sizeS}]
		Recall that we want to show, 
		$$  \mathbb{E}_{\CC,\CCbis} \left(\card{\sS} \right)  = \OOt{ \max_{(i,j) \in \mathcal{A}} \frac{\binom{s}{j} \binom{n-s}{i}}{2^{n-k}}} + 1. 
			$$
		Lemma \ref{lem:link_S_bias} gives us that 
		\begin{equation}\label{eq:expBia} 
			\mathbb{E}_{\CC,\CCbis} \left(\card{\sS }\right) \leq 2^{\kbis} \prob{\CC,\CCbis,\xv}{\bias_{(\hv,\cvbis)\drawn \widetilde{\sH}}(\langle \yv, \hv \rangle  + \langle \xv, \cvbis \rangle) \geq \frac{\delta}{2} \; \frac{\widetilde{H}}{\card{\widetilde{\sH}}} } + 1
		\end{equation} 
		Now using Lemma \ref{lem:simpl_bias} we get that
		\begin{multline}\label{eq:biasConj}
			\prob{}{\bias_{\left(\hv,\cvbis \right) \drawn \widetilde{\sH}} \left(\langle \yv, \hv \rangle  + \langle \xv, \cvbis \rangle\right) \geq \frac{\delta}{2} \; \frac{\widetilde{H}}{\card{\widetilde{\sH}}}  }  = \\ \prob{}{ \sum_{j = 0}^{s} \sum_{i = 0}^{n-s} K_{\tbis}^{(s)} \left(j\right) K_w^{(n-s)} \left(i\right) N_{i,j} \geq \frac{1}{2} \: K_w^{(n-s)}\left(u\right) K_{\tbis}^{(s)}\left(t-u\right) }.
		\end{multline}
		From Conjecture \ref{ass:bias_dom},
		\begin{align*}
			\prob{}{ \sum_{j = 0}^{s} \sum_{i = 0}^{n-s} K_{\tbis}^{(s)} \left(j\right) K_w^{(n-s)} \left(i\right) N_{i,j} \geq \frac{1}{2} \: K_w^{(n-s)}\left(u\right) K_{\tbis}^{(s)}\left(t-u\right) } = \OOt{ \max_{(i,j) \in \mathcal{A}}\prob{\CC,\CCbis,\xv}{N_{i,j} \neq 0} }
		\end{align*}
		Therefore, using Lemma \ref{lem:probNij} we get 
		$$
		\prob{}{ \sum_{j = 0}^{s} \sum_{i = 0}^{n-s} K_{\tbis}^{(s)} \left(j\right) K_w^{(n-s)} \left(i\right) N_{i,j} \geq \frac{1}{2} \: K_w^{(n-s)}\left(u\right) K_{\tbis}^{(s)}\left(t-u\right) } = \OOt{ \frac{\binom{s}{j} \binom{n-s}{i}}{2^{\kbis + n-k}} }
		$$
		Plugging this into Equation \eqref{eq:biasConj} and then in \eqref{eq:expBia} concludes the proof. 
		which proves our result.
	\end{proof}
 \section{About the distribution of the $N_{i,j}$}\label{app:prop_Nij}
This appendix is dedicated to studying the distribution of $N_{i,j}$ and in particular to prove Lemma~\ref{lem:probNij} which was used in Appendix \ref{app:proofBias} to prove Proposition \ref{prop:exp_bias}. We recall that it is given by
 \lemprobNij*
 First, using the union bound we can devise the following upper bound on our target probability:
\begin{fact} \label{fact:Nij_toexpect}
	$$
	\prob{}{N_{i,j} \neq 0}\leq \mathbb{E}\left(N_{i,j}\right).
	$$
\end{fact}
To compute this expected value, we first need to give the following lemma giving the probability that a word belongs to a random code.
\begin{lemma} \label{lem:probaRandom}
	Let $\CC$ be chosen uniformly at random among the $[n,k]$ linear codes. Let $\cv \in \F_2^n \setminus \{ \vec{0}\}$ we have
	\begin{align}
		\prob{\CC}{\cv \in \CC} &= \frac{2^k - 1}{2^{n} - 1}, \\
		\prob{\CC}{\vec{0} \in \CC} &= 1.
	\end{align}
\end{lemma}
\begin{proof}
	This lemma directly follows from \cite[\S 3, Lemma 9.3.2, (iii)]{BCN89}.
\end{proof}
We now give the preliminary lemma which breaks down the expected value of $N_{i,j}$ on the $\CCbisperp$ part and on the $\CC^{\sN}$ part.
\begin{restatable}{lemma}{lemmaNbar} \label{lem:Nbar}
Under Distribution \ref{notation:randomCodes} we denote by
\begin{align*}
	\overline{N_j^{(\CCbisperp)} }&\eqdef \expect{\CC,\CCbis,\xv} { N_j \left(\CCbisperp + \xv\right)}
\end{align*}
where $\xv$ is taken uniformly at random in $\F_2^{s} \setminus \{\CCbisperp + \ev_{\sP}\}$ and $\vec{r} \in \CCbisperp + \vec{x}$. We denote by fixing~$\vec{x}$,
\begin{align}
	\overline{N_i^{(\CC^{\sN})} }&\eqdef \expect{\CC} {N_i \left( \CC^{\sN} + \left(\rv+\ev_{\sP}\right)\Rm + \ev_{\sN} \right) }
\end{align}
We have that
	\begin{align}
\overline{N_i^{(\CC^{\sN})}} &=\frac{\binom{n-s}{i}}{2^{n-k}} \left(1 + \OO{2^{- (s-k)}} \right) \label{eq:expected_ni_bar}\\
	 \overline{N_j^{(\CCbisperp)}}  &\leq  	 \frac{\binom{s}{j}}{2^{\kbis}} \left(1 + \OO{ 2^{-  \left(s-\kbis\right)}}\right) \label{eq:expected_nj_bar}
\end{align}
Furthermore, the term $\OO{2^{- (s-\kbis)}}$ in Equation \eqref{eq:expected_ni_bar} does not depend on $i$·
\end{restatable}
\begin{proof}

	Under Distribution \ref{notation:randomCodes}, $\CC$ is taken at random among the $[n,k]$-codes that are such that $\CC_{\sP}$ is of full rank dimension $s$. Therefore, it is the same as if $\CC$ was chosen by taking its generator matrix $\Gm \in \F_2^{k \times n}$ as follows: 
	
$$ 
				\Gm_{\sP} \eqdef \begin{pmatrix}
						\Id_{s} \\
						\vec{0}_{k-s}
					\end{pmatrix}, \qquad \Gm_{\sN} \eqdef \begin{pmatrix}
					 \Rm \\
					 \Gm^{\sN}
					\end{pmatrix}. 
					$$
				where $\Rm$ is chosen uniformly at random among matrices of $\F_2^{s \times (n-s)}$ and $\Gm^{\sN}$ is any generator matrix of a code chosen uniformly at random among the $[n-s,k-s]$-codes. In particular, $\Gm^{\sN}$ and $\vec{R}$ are independent.

		 		Now, by definition $\xv\in \F_2^{s} \setminus \{\CCbisperp + \ev_{\sP}\}$ and $\vec{r} \in \CCbisperp + \vec{x}$, therefore $\rv + \ev_{\sP} \neq \vec{0}$. We deduce that, $ \left(\rv+\ev_{\sP}\right)\Rm$ is uniformly distributed in $\F_2^{n-s}$ as a non-zero sum of uniformly distribution vectors.
To simplify the notations let us define the uniformly distributed vector $$\vv \eqdef  \left(\rv+ \ev_{\sP}\right)\Rm + \ev_{\sN}$$ which is independent of $\CC^{\sN}$ (by construction $\Gm^{\sN}$ and $\vec{R}$ are independent). Now, let us show Equation \eqref{eq:expected_ni_bar}. We have by linearity of the expected value that
	\begin{align*}
		\overline{N_i^{(\CC^{\sN})}} &= \sum_{\zv \in \mathcal{S}_i^{n-s}} \prob{ \CC,\vv}{\zv \in \CC^{\sN} + \vv} \\
		&= \sum_{\zv \in \mathcal{S}_i^{n-s}} \prob{ \CC,\vv}{\zv \in \CC^{\sN} + \vv \middle| \vv = \zv} \prob{\CC,\vv}{\vv = \zv} + \prob{ \CC,\vv}{\zv \in \CC^{\sN} + \vv \middle| \vv \neq \zv} \prob{\CC,\vv}{\vv \neq \zv} \\
		&=  \sum_{\zv \in \mathcal{S}_i^{n-s}} \frac{1}{2^{n-s}}  + \frac{2^{k-s} - 1}{2^{n-s} - 1} \left( 1 - \frac{1}{2^{n-s}} \right) \qquad \mbox{(Lemma \ref{lem:probaRandom})} \\
		&= \frac{\binom{n-s}{i}}{2^{n-k}} \left( \frac{2^{n-k}}{2^{n-s}} + 2^{n-k} \; \frac{2^{k-s} -1}{2^{n-s}-1} \OO{1} \right) \\
		&= \frac{\binom{n-s}{i}}{2^{n-k}} \left( \frac{1}{2^{k-s}} + \frac{2^{-s} - 2^{-k}}{2^{-s}-2^{-n}} \; \OO{1}\right) \\
		&= \frac{\binom{n-s}{i}}{2^{n-k}} \left( \frac{1}{2^{k-s}} + \frac{1- 2^{-k+s}}{1-2^{-n+s}} \OO{1} \right) \\
		&= \frac{\binom{n-s}{i}}{2^{n-k}} \left( 1 + \frac{1}{2^{k-s}} + \OO{2^{-k+s}} \right) \quad \mbox{($s \leq k \leq n$)} 
\end{align*}
	Let us now show Equation \eqref{eq:expected_nj_bar}. Recall that 
	$$
	\overline{N_j^{(\CCbisperp)}}  \eqdef \expect{\CCbis, \xv} {N_j \left(\CCbisperp + \xv \right)}
	$$ 
	where $\xv$ is taken uniformly at random in $\F_2^{s} \setminus \{\CCbisperp + \ev_{\sP} \}$. By definition,
	\begin{align*}
		\expect{\CCbis, \xv} {N_j \left(\CCbisperp + \xv \right)} &= \sum_{\vec{z} \in  \mathcal{S}_{j}^{s}} \mathbb{P}_{\CCbis, \xv} \left( \vec{z} \in \CCbisperp + \xv \right) \\
		&= \sum_{\vec{z} \in  \mathcal{S}_{j}^{s}, \vec{x}_{0} \in \F_2^{s}} \mathbb{P}_{\CCbis} \left( \vec{z} - \vec{x}_{0} \in \CCbisperp \right)\mathbb{P}_{\vec{x},\CCbis}\left( \vec{x} = \vec{x}_{0} \right) \\
		&\leq \sum_{\vec{z} \in  \mathcal{S}_{j}^{s}}\left(  \sum_{\vec{x}_{0} \in \F_2^{s}} \frac{2^{s-\kbis}-1}{2^{s}-1} \;  \mathbb{P}_{\vec{x},\CCbis}\left( \vec{x} = \vec{x}_{0} \right) + \mathbb{P}_{\vec{x},\CCbis}\left( \vec{x} = \vec{z} \right) \right) \\
		&\leq \binom{s}{j} \frac{2^{s-\kbis}-1}{2^{s}-1} + \frac{\binom{s}{j}}{2^{s} - 2^{\kbis}} \\
		&\leq \frac{\binom{s}{j}}{2^{\kbis}} \;   2^{\kbis} \; \frac{2^{s-\kbis}-1}{2^{s}-1} + \frac{\binom{s}{j}}{2^{\kbis}} \; \frac{2^{\kbis}}{2^{s} - 2^{\kbis}} \\
		&\leq \frac{\binom{s}{j}}{2^{\kbis}} \;   \frac{2^{s}-2^{\kbis}}{2^{s}-1} + \frac{\binom{s}{j}}{2^{\kbis}} \; \frac{1}{2^{s-\kbis} - 1}  \\
		&\leq \frac{\binom{s}{j}}{2^{\kbis}} \; \left( 1 + \OO{2^{\kbis-s}} \right)
	\end{align*}
	which completes the proof. 
\end{proof}
	We can now show that the expected value of $N_{i,j}$ is the product of the two previously computed quantities:
	\begin{lemma}  \label{fact:product_expect}
		Under Distribution \ref{notation:randomCodes} and when $\xv$ is taken uniformly at random in $\F_2^{s} \setminus \{ \CCbisperp + \ev_{\sP}\}$ we have
		$$ 
		\expect{\CC,\CCbis,\xv}{N_{i,j}} = \overline{N_j^{(\CCbisperp)}}  \; 	\overline{N_i^{(\CC^{\sN})}}
		$$
		where, 
		$$
		N_{i,j} \eqdef \card{ \left\{  \left(\rv,\cv^{\sN}\right) \in  (\xv +\CCbisperp)  \times \CC^{\sN} \: : \left| \rv \right|=j \mbox{ and } \left| \left( \rv + \ev_{\sP}\right) \Rm + \ev_{\sN} + \cv^{\sN} \right| = i\right\} }
		$$
	and,
		$$
			\overline{N_j^{(\CCbisperp)} }\eqdef \expect{\CC,\CCbis,\xv} { N_j \left(\CCbisperp + \xv\right)} \quad \mbox{;} \quad 	\overline{N_i^{(\CC^{\sN})} }\eqdef \expect{\CC} {N_i \left( \CC^{\sN} + \left(\rv+\ev_{\sP}\right)\Rm + \ev_{\sN} \right) }.
		$$
	\end{lemma}
	\begin{proof} By definition,
		$$
		N_{i,j} = \sum_{u = 0}^{N_j\left(\CCbisperp + \xv\right)} N_i \left(\left(\rv^{(u)} + \ev_{\sP}\right) \Rm + \ev_{\sN} + \CC^{\sN}\right)
		$$
		where $\rv^{(u)}$ is the $u$'th codeword of weight~$j$ of $\CCbisperp + \xv$ and $N_j(\CCbisperp + \xv)$ counts the number of elements in $\CCbisperp + \xv$ of Hamming weight~$j$. Therefore,
		\begin{align*}
		N_{i,j} &= \sum_{\vec{z}\in \mathcal{S}_{j}^{s}} N_i \left(\left(\vec{z}+ \ev_{\sP}\right) \Rm + \ev_{\sN} + \CC^{\sN}\right)  \;\mathds{1}_{``\vec{z} \in \CCbisperp + \xv''}  \\
			&= \sum_{\vec{z}\in \mathcal{S}_{j}^{s},\vec{w} \in \mathcal{S}_{i}^{n-s}}  \mathds{1}_{``\vec{w} \in \left(\vec{z}+ \ev_{\sP}\right) \Rm + \ev_{\sN} + \CC^{\sN}"}  \;\mathds{1}_{``\vec{z} \in \CCbisperp + \xv''}  \\
		\end{align*} 
		We deduce that,
		\begin{align*}
			\expect{\CC,\CCbis,\xv}{N_{i,j}} &= \sum_{\vec{z}\in \mathcal{S}_{j}^{s},\vec{w} \in \mathcal{S}_{i}^{n-s}} \mathbb{P}\left( \vec{w} \in \left(\vec{z}+ \ev_{\sP}\right) \Rm + \ev_{\sN} + \CC^{\sN}, \; \vec{z} \in \CCbisperp + \xv \right) \\
			 &= \sum_{\vec{z}\in \mathcal{S}_{j}^{s},\vec{w} \in \mathcal{S}_{i}^{n-s}} \mathbb{P}\left( \vec{w} \in \left(\vec{z}+ \ev_{\sP}\right) \Rm + \ev_{\sN} + \CC^{\sN} \mid \vec{z} \in \CCbisperp + \xv  \right) \mathbb{P}\left(  \vec{z} \in \CCbisperp + \xv \right) \\
			 &= \sum_{\vec{z}\in \mathcal{S}_{j}^{s}}\left( \sum_{\vec{w} \in \mathcal{S}_{i}^{n-s}} \mathbb{P}\left( \vec{w} \in \left(\vec{z}+ \ev_{\sP}\right) \Rm + \ev_{\sN} + \CC^{\sN} \mid \vec{z} \in \CCbisperp + \xv  \right) \right)\mathbb{P}\left(  \vec{z} \in \CCbisperp + \xv \right) \\
			 &= \sum_{\vec{z}\in \mathcal{S}_{j}^{s}} \overline{N_i^{(\CC^{\sN})} }\; \mathbb{P}\left(  \vec{z} \in \CCbisperp + \xv \right) \\
			 &=  \overline{N_i^{(\CC^{\sN})} } \; 	\overline{N_j^{(\CCbisperp)} }
		\end{align*}
		which completes the proof. 
\end{proof}
\begin{proof}[Proof of Lemma \ref{lem:probNij}]
	We prove our result by using Fact \ref{fact:Nij_toexpect}, then Lemmas \ref{lem:Nbar} and \ref{fact:product_expect}.
\end{proof} \section{Proof of Proposition \ref{prop:model_conjecture}}\label{app:model_conjecture}

\subsection{A more minimalistic conjecture} \label{sec:min_ass}
The goal of this section is to devise a more minimalistic but stronger conjecture which implies Conjecture \ref{ass:bias_dom} (in the sense that it involves only concentration bound of the weight enumerator of some random linear codes). Furthermore, the aforementioned implication is a key step of the proof of Proposition \ref{prop:model_conjecture} which shows that the experimental model implies Conjecture \ref{ass:bias_dom}.
\begin{conjecture} \label{ass:min}
	For any $(i,j) \in \llbracket 0, n-s \rrbracket \times \llbracket 0, s \rrbracket $ and any~$v \in \left \llbracket 0, \overline{V_j}  +  n^{1.1} \max\left(\sqrt{\overline{V_j}},1\right) \right \rrbracket$, we have under Distribution \ref{notation:randomCodes},
	\begin{align}
		&\prob{\CCbis, \xv}{\left|V_j - \overline{V_j}\right|>  n^{1.1} \max\left(\sqrt{\overline{V_j}},1\right)} &&= \OOt{2^{-n}}, \label{eq:ass_min_1} \\
&\prob{\CC,\CCbis,\xv}{ \left|N_{i,j} - v \: \overline{N_i} \right| > n^{1.1} \max \left(\sqrt{v \: \overline{N_i}}, \:1\right) \middle | \: N_j = v} &&= \OOt{2^{-n}}. \label{eq:ass_min_2}
	\end{align}
	where
	\begin{equation*}
		V_j \eqdef N_j \left(\CCbisperp + \xv\right)\end{equation*}
and
	\begin{align}
		\overline{N_i} \eqdef \frac{\binom{n-s}{i}}{2^{n-k}}, \quad \overline{V_j} \eqdef \frac{\binom{s}{j}}{2^{\kbis}} \label{def:N_ibar}.
	\end{align}
\end{conjecture}
\begin{remark}
	Recall that $N_{i,j} \eqdef \sum_{u = 1}^{V_j} N_i^{(u)}$ where $N^{(u)}_i \eqdef N_i\left(\left(\rv^{(u)} + \ev_{\sP}\right) \Rm + \ev_{\sN} + \CC^{\sN}\right)$ and $\vec{r}^{(u)}$ is the $u$'th codeword of weight $j$ of $\CCbisperp + \xv$. From this and using Lemma \ref{lem:Nbar} it is readily seen that we have that $$\expect{\CC,\CCbis,\xv}{N_{i,j} \middle| \: V_j = v} = v \: \overline{N_i} \left(1+o(1)\right).$$ As such, Equation \eqref{eq:ass_min_2} in the previous Conjecture can also really be seen as a concentration inequality.
\end{remark}
\begin{proposition} \label{prop:imply_conj2_conj1}
	Conjecture \ref{ass:min} implies Conjecture \ref{ass:bias_dom}.
\end{proposition}
The following lemmas will be useful to prove the this proposition.
\begin{lemma}[Centering Lemma.]\label{lem:centering} We have,
	$$ 
	\sum_{j = 0}^{s} \sum_{i = 0}^{n-s} K_{\tbis}^{(s)} \left(j\right) K_w^{(n-s)} \left(i\right) N_{i,j} = \sum_{j = 0}^{s} \sum_{i = 0}^{n-s} K_{\tbis}^{(s)} \left(j\right) K_w^{(n-s)} \left(i\right) \left(N_{i,j} - V_j \overline{N_i} \right).
	$$
\end{lemma}
\begin{proof}
	From the orthogonality of Krawtchouk polynomials relatively to the measure $\mu\left(i\right) = \binom{n-s}{i}$ \cite[Ch. 5. \S 7. Theorem 16]{MS86} we have:
	$$
	\sum_{i = 0}^{n-s} \binom{n-s}{i} K_w^{(n-s)}\left(i\right) = 0.
	$$
	By using the definition of $\overline{N_i}$ in Equation \eqref{def:N_ibar}, we get,
	\begin{align*}
		0 &= \sum_{j = 0}^s V_j \: K^{(s)}_{\tbis} \left(j\right)  \sum_{i = 0}^{n-s} K^{(n-s)}_w \left(i\right) \overline{N_{i}}\\
		&= \sum_{j = 0}^s  \sum_{i = 0}^{n-s} K^{(n-s)}_w \left(i\right) K^{(s)}_{\tbis} \left(j\right) V_j \: \overline{N_{i}}.
	\end{align*}
	Therefore,
	$$  
	\sum_{j = 0}^{s} \sum_{i = 0}^{n-s} K_{\tbis}^{(s)} \left(j\right) K_w^{(n-s)} \left(i\right) N_{i,j} = \sum_{j = 0}^{s} \sum_{i = 0}^{n-s} K_{\tbis}^{(s)} \left(j\right) K_w^{(n-s)} \left(i\right) \left(N_{i,j} - V_j \overline{N_i} \right)
	$$
	which completes the proof. 
\end{proof}
\begin{corollary}\label{coro:probaKrawNij}We have,
	\begin{multline*}
		\mathbb{P}\left( \sum_{j = 0}^{s} \sum_{i = 0}^{n-s} K_{\tbis}^{(s)} \left(j\right) K_w^{(n-s)} \left(i\right) N_{i,j} \geq \frac{1}{2} K_w^{(n-s)}\left(u\right) K_{\tbis}^{(s)}\left(t-u\right)\right) = \\
		\OOt{\max_{(i,j) \in \llbracket 0,n-s\rrbracket \times \llbracket 0,s \rrbracket} \prob{}{\left| N_{i,j} -  V_j \: \overline{N_i}\right| \geq \frac{1}{2 \: (n+1)^2}  \left|\frac{K_w^{(n-s)}\left(u\right) K_{\tbis}^{(s)}\left(t-u\right)}{ K_w^{(n-s)}\left(i\right) K_{\tbis}^{(s)}\left(j\right) }\right| } }
	\end{multline*}
	
\end{corollary}
\begin{proof}
	The event 
	$$ 
	\sum_{j = 0}^{s} \sum_{i = 0}^{n-s} K_{\tbis}^{(s)} \left(j\right) K_w^{n-s} \left(i\right) \left( N_{i,j} - V_j \overline{N_i}\right) \geq \frac{1}{2} K_w^{(n-s)}\left(u\right) K_{\tbis}^{(s)}\left(t-u\right) 
	$$ 
	implies that it exists $j \in  \llbracket 0, s \rrbracket$ and $i \in \llbracket 0, n-s \rrbracket$ such that 
	\begin{equation} \label{eq:proof_bias_to_NI_doubleRLPN}
		K_{\tbis}^{(s)} \left(j\right) K_w^{n-s}\left(i\right) \left( N_{i,j} - V_j \overline{N_i}\right)  \geq \frac{K_w^{(n-s)}\left(u\right) K_{\tbis}^{(s)}\left(t-u\right) }{2 \left(n-s+1\right) \left(s+1\right)} .
	\end{equation}
	Therefore,
	\begin{align*}
		&	\mathbb{P}\left( \sum_{j = 0}^{s} \sum_{i = 0}^{n-s} K_{\tbis}^{(s)} \left(j\right) K_w^{(n-s)} \left(i\right) N_{i,j} \geq \frac{1}{2} K_w^{(n-s)}\left(u\right) K_{\tbis}^{(s)}\left(t-u\right)\right) \\ 
		&=	\mathbb{P}\left( \sum_{j = 0}^{s} \sum_{i = 0}^{n-s} K_{\tbis}^{(s)} \left(j\right) K_w^{n-s}\left(i\right) \left( N_{i,j} - V_j \overline{N_i}\right)   \geq \frac{1}{2} K_w^{(n-s)}\left(u\right) K_{\tbis}^{(s)}\left(t-u\right)\right)  \quad (\mbox{Lemma } \ref{lem:centering}) \\ 
		&\leq 	 \prob{\CC,\CCbis}{ \bigvee_{j = 0}^{s} \bigvee_{i = 0}^{n-s} \left( K_{\tbis}^{(s)} \left(j\right) K_w^{n-s} \left(i\right) \left( N_{i,j} - V_j \overline{N_i} \right)  \geq \frac{K_w^{(n-s)}\left(u\right) K_{\tbis}^{(s)}\left(t-u\right) }{2 \left(n-s+1\right) \left(s+1\right)}  \right)} \quad (\mbox{using } \eqref{eq:proof_bias_to_NI_doubleRLPN})\\
		&\leq \sum_{j = 0}^s \sum_{i = 0}^{n-s}  \prob{}{ K_{\tbis}^{(s)} \left(j\right) K_w^{n-s} \left(i\right) \left( N_{i,j} - V_j \overline{N_i} \right) \geq \frac{K_w^{(n-s)}\left(u\right) K_{\tbis}^{(s)}\left(t-u\right) }{2 \left(n-s+1\right) \left(s+1\right)} } \quad (\mbox{union bound}) \\
&= \OO{n^2} \max_{\substack{i = 0 \dots n-s\\j = 0\dots s}}  \prob{}{  \left| N_{i,j} - V_j \overline{N_i} \right|\geq  \frac{1}{2 \: (n+1)^2}  \left|\frac{K_w^{(n-s)}\left(u\right) K_{\tbis}^{(s)}\left(t-u\right)}{ K_w^{(n-s)}\left(i\right) K_{\tbis}^{(s)}\left(j\right) }\right| }
	\end{align*}
	which completes the proof. 
\end{proof}
\begin{lemma} \label{lem:ineq_NI}
	Under Parameters constraint \ref{ass:parameters_doubleRLPN} we have that
	$$ \left(  \frac{ K_w^{(n-s)}\left(u\right) K_{\tbis}^{(s)} \left(t-u\right)}{K_w^{(n-s)}\left(i\right) K_{\tbis}^{(s)} \left(j\right)  } \right)^2 = \omega \left(n^{\alpha + 8}\right)\overline{N_i} \max \left(\overline{V_{j}},\frac{1}{\OO{n^{\alpha}}}\right) 
	$$
\end{lemma}
\begin{proof}
	First, we simplify $ \left(K_w^{(n-s)}\left(u\right) K_{\tbis}^{(s)} \left(t-u\right) \right)^2$ by using Constraint $(i)$ 
of Parameters constraint \ref{ass:parameters_doubleRLPN} which we recall is given by:
	$$ 
	\frac{\binom{n-s}{w} \binom{s}{\tbis}}{2^{k-\kbis}} = \frac{\omega \left(n^{\alpha+8} \right)}{\delta^2}, \quad \mbox{where }\; \delta \eqdef \frac{K_w^{(n-s)}\left(u\right) K_{\tbis}^{(s)} \left(t-u\right) }{\binom{n-s}{w}\binom{s}{\tbis}}.
	$$
	By reordering the terms in the previous equation we get:
	\begin{align*}
		\left(K_w^{(n-s)}\left(u\right) K_{\tbis}^{(s)} \left(t-u\right) \right)^2= 2^{k-\kbis}\binom{s}{\tbis} \binom{n-s}{w} \omega(n^{\alpha + 8}).
	\end{align*}
	Now, to show the lemma we only have to show that
	\begin{align}
		\frac{ 2^{k-\kbis} \binom{s}{\tbis} \binom{n-s}{w} }{ \left( K_w^{(n-s)}\left(i\right) K_{\tbis}^{(s)} \left(j\right) \right)^2 } )\geq \overline{N_i} \max \left(\overline{V_{j}},\frac{1}{\OO{n^{\alpha}}}\right). \label{eq:proof_ineq_Ni}
	\end{align}
	First ,let us lower bound $\frac{1}{\left(K_w^{(n-s)}(i) \right)^2}$.  From the orthonormality relations of the Krawtchouk polynomials \cite[Ch. 5. \S 7. Theorem 16]{MS86} with the measure $\mu_1 \left(v\right) = \frac{\binom{n-s}{v}}{2^{n-s}}$ we have that
	\begin{align*}
		\sum_{v=0}^{n-s} \binom{n-s}{v} \left(K_w^{(n-s)}\left(v\right)\right)^2 &= \binom{n-s}{w} 2^{n-s},
	\end{align*}
	and thus, as the previous sum is composed of positive terms, we have that
	\begin{align}
		\frac{1}{\left(K_w^{(n-s)}(i) \right)^2} \geq \frac{\binom{n-s}{i}}{\binom{n-s}{w}2^{n-s}}. \label{eq:proof_ineq_NI_eq1}
	\end{align}
	Now, let us lower-bound $\frac{1}{\left(K_{\tbis}^{(s)}(j) \right)^2} $. Using the same orthonormality argument relatively to the measure $\mu_2 \left(v\right) = \frac{\binom{s}{v}}{2^{s}}$ we get
	\begin{align*}
		\frac{1}{\left(K_{\tbis}^{(s)}(j) \right)^2} \geq \frac{\binom{s}{j}}{\binom{s}{\tbis}2^{s}}.
	\end{align*}
	Furthermore, from Fact \ref{fact:krawtchouk_bias} we can also deduce the following inequality \begin{equation*}  \label{eq:proof_lem_ineq_NI_3}
		\left(K_{\tbis}^{(s)}\left( j\right)\right)^2 \leq \binom{s}{\tbis}^2.
	\end{equation*} 
	Combining the last two equations we get that
	\begin{align}
		\frac{1}{\left(K_{\tbis}^{(s)}(j) \right)^2} \geq \min \left(\frac{\binom{s}{j}}{\binom{s}{\tbis}2^{s}} , \frac{1}{\binom{s}{\tbis}^2} \right)  \label{eq:proof_ineq_NI_eq2}.
	\end{align}
	Finally let us show Equation \eqref{eq:proof_ineq_Ni} by using Equation \eqref{eq:proof_ineq_NI_eq1} and \eqref {eq:proof_ineq_NI_eq2} :
	\begin{align*}
		\frac{ 2^{k-\kbis} \binom{s}{\tbis} \binom{n-s}{w} }{ \left( K_w^{(n-s)}\left(i\right) K_{\tbis}^{(s)} \left(j\right) \right)^2 }  &\geq  2^{k-\kbis} \binom{s}{\tbis} \binom{n-s}{w}  \frac{\binom{n-s}{i}}{\binom{n-s}{w}2^{n-s}} \:  \min \left(\frac{\binom{s}{j}}{\binom{s}{\tbis}2^{s}} , \frac{1}{\binom{s}{\tbis}^2} \right) \\
		&= \frac{\binom{n-s}{i}}{2^{n-k}} \min \left(\frac{\binom{s}{j}}{2^{\kbis}}, \quad \frac{2^{s-\kbis}}{\binom{s}{\tbis}}\right) \\
		&= \frac{\binom{n-s}{i}}{2^{n-k}} \min \left(\frac{\binom{s}{j}}{2^{\kbis}}, \frac{1}{\OO{n^{\alpha}}}\right) \quad \mbox{(Constraint $(iii)$ of Eq. \eqref{constraints:constraint})}\\
		&= \overline{N_i} \min \left(\overline{V_j}, \frac{1}{\OO{n^{\alpha}}} \right).
	\end{align*}
	This completes the proof. 
\end{proof}
\begin{proof}[Proof of Proposition \ref{prop:imply_conj2_conj1}.]
	Let us suppose that the Parameter constraint \ref{ass:parameters_doubleRLPN} is verified and that Conjecture \ref{ass:min} is true. We want to show Conjecture \ref{ass:bias_dom} holds, namely that
	\begin{multline*} 
	\mathbb{P}\left( \sum_{j = 0}^{s} \sum_{i = 0}^{n-s} K_{\tbis}^{(s)} \left(j\right) K_w^{(n-s)} \left(i\right) N_{i,j} \geq \frac{1}{2} K_w^{(n-s)}\left(u\right) K_{\tbis}^{(s)}\left(t-u\right)\right) \\ 
	=
	\OOt{ \max_{(i,j) \in \mathcal{A}}\prob{}{N_{i,j} \neq 0}  + 2^{-n}}. 
	\end{multline*} 
	Using Corollary \ref{coro:probaKrawNij}, we only have to show that for any $(i,j) \in \llbracket 0,n-s\rrbracket \times \llbracket 0,s \rrbracket$ we have
	$$ \prob{}{\left| N_{i,j} -  V_j \: \overline{N_i}\right| \geq \frac{1}{2 \: (n+1)^2}  \left|\frac{K_w^{(n-s)}\left(u\right) K_{\tbis}^{(s)}\left(t-u\right)}{ K_w^{(n-s)}\left(i\right) K_{\tbis}^{(s)}\left(j\right) }\right| }  = \OOt{ \max_{(i^*,j^*) \in \mathcal{A}}\prob{}{N_{i^*,j^*} \neq 0} + 2^{-n}}.$$
	To ease up the notations let us denote by 
	\begin{equation} \label{def:R_ij}
		R_{i,j} \eqdef  \frac{1}{2 \: (n+1)^2}  \left|\frac{K_w^{(n-s)}\left(u\right) K_{\tbis}^{(s)}\left(t-u\right)}{ K_w^{(n-s)}\left(i\right) K_{\tbis}^{(s)}\left(j\right) }\right|.
	\end{equation}
	Thus, we only have to show that 
	\begin{equation} \label{eq:proof_conj2_to_prove}
		\prob{}{\left| N_{i,j} -  V_j \: \overline{N_i}\right| \geq R_{i,j} }  = \OOt{ \max_{(i^*,j^*) \in \mathcal{A}}\prob{}{N_{i^*,j^*} \neq 0} + 2^{-n}}.
	\end{equation}
	We prove the previous equality for each cases: $(i,j) \in \mathcal{A}$ or $(i,j) \notin \mathcal{A}$, where recall that
	$$
		\mathcal{A} \eqdef \left\{ \left(i,j\right) \in \llbracket 0, n-s \rrbracket \times \llbracket0,s \rrbracket, \iftoggle{llncs}{}{\;} \left| \frac{ K_w^{(n-s)}\left(u\right) K_{\tbis}^{(s)}\left(t-u\right) }{ K_w^{(n-s)}\left(i\right) K_{\tbis}^{(s)}\left(j\right)}\right| \leq n^{3.2}\right\}.
	$$

	\noindent \textbf{Cases 1:} Here we suppose that  $(i,j) \in \mathcal{A}$. Let us prove Equation \eqref{eq:proof_conj2_to_prove}. Using the law of total probability we have that
	\begin{align*}
		\prob{}{\left| N_{i,j} -  V_j \: \overline{N_i}\right| \geq R_{i,j} } \leq \prob{}{N_{i,j} \neq 0} +  \prob{}{\left| N_{i,j} -  V_j \: \overline{N_i}\right| \geq R_{i,j}, \; N_{i,j} = 0}.
	\end{align*}
	As $(i,j) \in \mathcal{A}$, 
	$$  
	\prob{}{N_{i,j} \neq 0} = \OOt{\max_{(i^*,j^*) \in \mathcal{A}}\prob{}{N_{i^*,j^*} \neq 0}}
	$$
	we only have left to show that:
	\begin{equation*}
		\prob{}{\left| N_{i,j} -  V_j \: \overline{N_i}\right| \geq R_{i,j} , \: \; N_{i,j} = 0} = \OOt{\max_{(i^*,j^*) \in \mathcal{A}}\prob{}{N_{i^*,j^*} \neq 0} + 2^{-n}}.
	\end{equation*}
	We now show the previous equation, by proving that,
		\begin{equation} \label{proof:conj2_target_1}
		\prob{}{\left| N_{i,j} -  V_j \: \overline{N_i}\right| \geq R_{i,j} , \: \; N_{i,j} = 0} = \OOt{ 2^{-n}}.
	\end{equation}
	We have:
	\begin{align}
		\prob{}{\left| N_{i,j} -  V_j \: \overline{N_i}\right| \geq R_{i,j} , \: \; N_{i,j} = 0} &= \prob{}{ V_j \: \overline{N_i} \geq R_{i,j}} \nonumber \quad \mbox{(we used that $V_{j},\overline{N_{i}} > 0$)} \\
		&= \prob{}{ V_j\geq \frac{R_{i,j}}{\overline{N_i}}} \nonumber \\
		&= \prob{}{V_j - \overline{V_j}\geq \frac{R_{i,j}}{\overline{N_i}} - \overline{V_j}} \nonumber \\
		&\leq \prob{}{\left|V_j - \overline{V_j} \right|\geq \frac{R_{i,j}}{\overline{N_i}} - \overline{V_j}}. \nonumber
	\end{align}
	Recall that from Equation \eqref{eq:ass_min_1} of Conjecture \ref{ass:min} we have that
	\begin{equation*} \label{redit:conj2_eq_1}
		\prob{\CCbis, \xv}{\left|V_j - \overline{V_j}\right|>  n^{1.1} \max\left(\sqrt{\overline{V_j}},1\right)} = \OOt{2^{-n}}.
	\end{equation*}
	Therefore, we only have to show that for $n$ big enough we have 
	\begin{equation} \label{proof:conj2_ineq_1}
		\frac{R_{i,j}}{\overline{N_i}} - \overline{V_j} \geq n^{1.1} \max \left(\sqrt{\overline{V_j}, \: 1} \right)
	\end{equation}
	to prove Equation \eqref{proof:conj2_target_1}. Let us prove Equation \eqref{proof:conj2_ineq_1}. By definition of $R_{i,j}$ in Equation \eqref{def:R_ij} and using Lemma~\ref{lem:ineq_NI} we have that
	\begin{align}
		R_{i,j}^2 &=   \frac{1}{4 \: (n+1)^4}  \left(\frac{K_w^{(n-s)}\left(u\right) K_{\tbis}^{(s)}\left(t-u\right)}{ K_w^{(n-s)}\left(i\right) K_{\tbis}^{(s)}\left(j\right) }\right)^2 \nonumber \\
		&= \frac{1}{4 \: (n+1)^4} \; \omega \left(n^{\alpha + 8}\right)\overline{N_i} \max \left(\overline{V_{j}},\frac{1}{\OO{n^{\alpha}}}\right) \nonumber \\
&=  f(n) \max \left(\overline{N_i} \: \overline{V_j}, n^{-\alpha}\overline{N_i}\right) \label{eq:fn}
	\end{align}
	where $f(n) = \omega\left(n^{\alpha + 4}\right)$.
Therefore,
	\begin{align}
		\frac{R_{i,j}^2}{\overline{N_i}^2} \; \frac{1}{n^{2.4} \max \left(\overline{V_j}^2, \: 1\right)}& =  \frac{f(n) \max \left(\overline{N_i} \: \overline{V_j}, n^{-\alpha} \: \overline{N_i}\right)}{n^{2.4} \; \overline{N_i}^2 \max \left(\overline{V_j}^2, \: 1\right)} \nonumber \\
		&=   \frac{f(n) \max \left(\frac{\overline{V_j}}{\overline{N_i}}, \frac{n^{-\alpha}}{\overline{N_i}} \: \right)}{ n^{2.4} \: \max \left(\overline{V_j}^2, \: 1\right)} \nonumber\\
		&= \begin{cases}
			\frac{1}{n^{2.4}}f(n)\max \left(\frac{1 }{\overline{N_i} \overline{V_j}}, \frac{n^{-\alpha}}{\overline{N_i} \overline{V_j}^2} \: \right) & \mbox{ if } \overline{V_j} > 1\\
			\frac{1}{n^{2.4}} f(n)\max \left(\frac{\overline{V_j}}{\overline{N_i}}, \frac{n^{-\alpha}}{\overline{N_i}} \: \right) & \mbox{ if } \overline{V_j} \leq 1
		\end{cases} \nonumber \\
		&\geq  	 \frac{1}{n^{2.4}} f(n) \min \left(\frac{1}{\overline{N_i} \overline{V_j}}, \frac{n^{-\alpha}}{\overline{N_i}}\right) \nonumber\\
		&\geq  \frac{n^{-2 \alpha}}{n^{2.4}} \; f(n) \; \min \left(\frac{1}{\overline{N_i} \overline{V_j}}, \frac{1}{n^{-\alpha} \: \overline{N_i}}\right) \nonumber \\
		&= \frac{n^{-2 \alpha}}{n^{2.4}} \; \frac{f(n)}{\max \left( \overline{N_i} \overline{V_j}, n^{-\alpha} \: \overline{N_i}\right)} \nonumber \\
		&=   \frac{n^{-2 \alpha}}{n^{2.4}} \;  \frac{f(n)^2}{R_{i,j}^2}\nonumber \qquad (\mbox{By Equation \eqref{eq:fn}}) \\
		&=  \frac{\omega\left(n^{5.6}\right)}{R_{i,j}^2} \qquad (f(n) = \omega\left(n^{\alpha + 4}\right)) \nonumber\\
		&= \omega(1) \label{eq:proof_conj2_c1_f}
	\end{align}
	where in the last line we used the fact that $(i,j) \in \mathcal{A}$: by definition, 
	$$
	R_{i,j} = \frac{1}{2 \: (n+1)^2} \; \left|\frac{K_w^{(n-s)}\left(u\right) K_{\tbis}^{(s)}\left(t-u\right)}{ K_w^{(n-s)}\left(i\right) K_{\tbis}^{(s)}\left(j\right) }\right|\leq \frac{1}{2 \: (n+1)^2} \; n^{3.2}
	$$ 
	and thus
	\begin{align*}
		\frac{1}{R_{i,j}^2} &\geq \frac{ 2 \: \left(n+1\right)^2}{n^{3.2}}.
	\end{align*}
	Finally, Equation \eqref{eq:proof_conj2_c1_f} shows that for $n$ big enough
	$$  \frac{R_{i,j}}{\overline{N_i}} \geq n^{1.2} \max \left(\overline{V_j}, \: 1 \right)$$
	And as such, for $n$ big enough
	\begin{align*}
		\frac{R_{i,j}}{\overline{N_i}} - \overline{V_j} &\geq n^{1.1}\max \left(\overline{V_j}, \: 1 \right) \\
		&\geq n^{1.1}\max \left(\sqrt{\overline{V_j}}, \: 1 \right)
	\end{align*}
	which proves Equation \eqref{proof:conj2_ineq_1}. Therefore we have just proved Equation \eqref{eq:proof_conj2_to_prove} in the case where $(i,j) \in \mathcal{A}$.
	\newline
	
	\noindent \textbf{Case 2: }. Here we suppose that  $(i,j) \notin \mathcal{A}$. Let us prove Equation \eqref{eq:proof_conj2_to_prove}.
	We only have to prove that:
	\begin{equation*}
		\prob{}{\left| N_{i,j} -  V_j \: \overline{N_i}\right| \geq R_{i,j} }  = \OOt{2^{-n}}.
	\end{equation*}
	 Let $M$ be defined as
	\begin{equation} \label{eq:def_M}
		M \eqdef \overline{V_j} + n^{1.1} \max \left(\sqrt{\overline{V_j}}, 1\right).
	\end{equation} 
	By the law of total probability we have that
	\begin{multline*}
		\prob{\CC,\CCbis,\xv}{  \left| N_{i,j} -  V_j\overline{N_{i}}  \right| > R_{i,j} } = \\  \prob{}{ \left| N_{i,j} -  V_j \overline{N_{i}}  \right| > R_{i,j} | V_j > M}\prob{}{V_j > M} + \\ \sum_{v = 0}^{M} \prob{}{ \left| N_{i,j} -  V_j \overline{N_{i}}  \right| > R_{i,j} | V_j =  v}\prob{}{V_j= v}.
	\end{multline*}
	Which we can upper bound by
	\begin{multline*}
		\prob{\CC,\CCbis,\xv}{  \left| N_{i,j} -  V_j \overline{N_{i}}  \right| > R_{i,j} } \leq \prob{}{V_j > M} +  \max_{v= 0...M }\prob{}{ \left| N_{i,j} -  v \:  \overline{N_{i}}  \right|> R_{i,j} | V_j =  v}.
	\end{multline*}
	By definition of $M$ in Equation \eqref{eq:def_M} and using Equation \eqref{eq:ass_min_1} of Conjecture \ref{ass:min} we get that 
	$$ \prob{}{V_j > M}  = \OOt{2^{-n}}.$$
	Now, we only have left to prove that for any $v \in \llbracket 0, M \rrbracket$ we have
	$$ 
	\prob{}{ \left| N_{i,j} -  v\; \overline{N_{i}}  \right|> R_{i,j} |V_j =  v}  = \OOt{2^{-n}}.
	$$
	Let us consider $v \in \llbracket 0, M \rrbracket$. Let us first show that for $n$ big enough we have that
	\begin{align} \label{proof:ineq_Rij}
		n^{1.1}\max \left(\sqrt{v \: \overline{N_i}}, \:1\right) \leq R_{i,j}.
	\end{align}We have
	\begin{align*}
		n^{2.2}\max \left(v \: \overline{N_i}, \:1\right)  &\leq n^{2.2}\max \left(M \: \overline{N_i}, \:1\right)  \\
		&\leq n^{2.2}\max \left(\overline{V_j} \: \overline{N_i} + n^{1.1} \overline{N_i} \max \left(\sqrt{\overline{V_j}},1\right), \:1\right) && \mbox{(Using Equation \eqref{eq:def_M})} \\
		&\leq  \max \left(2 \: n^{2.2} \: \overline{V_j} \: \overline{N_i}, \; 2\: n^{3.3} \overline{N_i} \sqrt{\overline{V_j}}, \: 2\: n^{3.3} \overline{N_i} ,\:  n^{2.2}\right) 
	\end{align*}
	To show Equation \eqref{proof:ineq_Rij}, ne only have left to prove that, for $n$ big enough, each term in the previous maximum is smaller than $R_{i,j}^2$.
	First let us recall that by definition of $R_{i,j}$ in Equation \eqref{def:R_ij} and from Lemma \ref{lem:ineq_NI},
$$R_{i,j}^2 = \max \left(\omega \left(n^{\alpha + 4}\right) \overline{N_i} \: \overline{V_j}, \omega \left(n^{4}\right) \overline{N_i}\right).$$ 
	For $n$ big enough we have that:
	\begin{align*}
		&2 \: n^{2.2} \: \overline{V_j} \: \overline{N_i} &&\leq  n^{\alpha + 4}\overline{N_i} \: \overline{V_j} \leq R_{i,j}^2, \\
		&2\: n^{3.3} \overline{N_i} \sqrt{\overline{V_j}} &&\leq \begin{cases}
			n^{4}\overline{N_i}\leq R_{i,j}^2 &\mbox{ when } \overline{V_j} \leq 1 \\
			n^{\alpha + 4}\overline{N_i} \: \overline{V_j} \leq R_{i,j}^2 &\mbox{ when } \overline{V_j} > 1
		\end{cases},  \\
		&	2 \: n^{3.3} \: \overline{N_i} &&\leq  n^{4} \overline{N_i} \leq R_{i,j}^2, \\
		& n^{2.2} &&\leq R_{i,j}^2.
	\end{align*}
	Where in the last equation we used the fact that $(i,j) \notin \mathcal{A}$, thus $R_{i,j} \geq \frac{n^{3.2}}{2 \left(n+1\right)^2}$ and thus $R_{i,j}^2 \geq n^{2.3}$ for $n$ big enough. We have shown that $$n^{2.2}\max \left(v \: \overline{N_i}, \:1\right) \leq R_{i,j}^2$$ and thus we have shown Equation \eqref{proof:ineq_Rij}.
	Finally we have
	\begin{align*}
		\prob{}{ \left| N_{i,j} -  V_j \overline{N_{i}}  \right|> R_{i,j} |V_j =  v} &=	\OO{\prob{}{ \left| N_{i,j} -  V_j \overline{N_{i}}  \right|> n^{1.1}\max \left(\sqrt{v \: \overline{N_i}}, \:1\right) |V_j =  v}} \\
		&= \OOt{2^{-n}}
	\end{align*}
	where in the last line we used Equation \eqref{eq:ass_min_2} of Conjecture \ref{ass:min}. This concludes the proof.
\end{proof}

\begin{lemma} \label{lem:poisson_conj2}
	The Poisson Model \ref{model:poisson} imply Conjecture \ref{ass:min}.
\end{lemma}
\begin{proof}
	Let $M$ be defined as
	\begin{equation} \label{lemma:proof_bound_bias_1_Mn}
	M \eqdef \overline{V_j} + n^{1.1} \max \left(\sqrt{\overline{V_j}}, 1\right)
	\end{equation} 
	Recall that to show Conjecture \ref{ass:min} we only have to show that
	\begin{align}
	\prob{\CCbis, \xv}{ \left| V_j - \overline{V_j} \right| \geq  n^{1.1} \max\left(\sqrt{\overline{V_j}},1\right)} &&=& \OOt{2^{-n}}, \label{eq:lem_ass_min_1} \\
	\forall v \in \llbracket 0, M \rrbracket, \:\; \prob{\CC,\CCbis,\xv}{ \left|N_{i,j} - v \: \overline{N_i} \right| > n^{1.1} \max\left(\sqrt{v \: \overline{N_i}}, 1\right)\: \middle | \: N_j = v} &&=& \OOt{2^{-n}} \label{eq:lem_ass_min_2}
\end{align}
	Under the Poisson Model \ref{model:poisson} we have that
	$$ N_{i,j} \sim \mathrm{Poisson}\left(V_j \: \overline{N_i}\right), \quad V_j \sim  \mathrm{Poisson}\left( \overline{V_j}\right).$$
	We will use the following fact: when $\vec{X}$ follows a Poisson distribution of parameter $\lambda$ and $g(n)=\omega(n)$, then we have that
	\begin{align} \label{lemma:proof_bound_bias_1_poisson_bound}
		\prob{}{|\vec{X} - \lambda| > g(n) \max \left(\sqrt{\lambda},1\right)} = 2^{-\omega(n)} .
		\end{align}
		Let us prove this claim. It is known \cite[Prop 11.15]{G17b} that we have the following exponential tail bound for $\vec{X}$:
		\begin{equation*}
			\prob{}{|\vec{X} - \lambda| > r} \leq 2 \: e^{\frac{- \: r^2}{2 \: (\lambda + r) }}.
		\end{equation*}
	Thus,
	\begin{align*}
		\prob{}{|\vec{X} - \lambda| > g(n) \max \left(\sqrt{\lambda},1\right)} &\leq 2 \: e^{\frac{- \: g(n)^2 \max \left(\lambda,1\right)}{2 \: (\lambda + g(n) \max\left(\sqrt{\lambda},1\right)) }} \\
		&\leq 2 \: e^{\frac{- \: g(n)}{2 \left(\frac{\lambda + g(n) \max\left(\sqrt{\lambda},1\right)}{g(n) \max \left(\lambda,1\right)} \right) }}.
	\end{align*}
	We only have left to show that 
	$$  \frac{\lambda + g(n) \max\left(\sqrt{\lambda},1\right)}{g(n) \max \left(\lambda,1\right)}= \OO{1}.$$
	First it is readily seen that we have that ($g(n) = \omega(n)$)
	$$  \frac{\lambda}{g(n) \max \left(\lambda,1\right)}  = \OO{1},$$
	and second, 
	$$ \frac{g(n) \max\left(\sqrt{\lambda},1\right)}{g(n) \max \left(\lambda,1\right)} = \begin{cases}
		1=\OO{1} & \mbox{ if } \lambda \leq 1 \\
		\frac{1}{\sqrt{\lambda}} = \OO{1} & \mbox{ if } \lambda > 1.
	\end{cases}$$
	which concludes the proof of Equation \eqref{lemma:proof_bound_bias_1_poisson_bound}. Equation \eqref{eq:lem_ass_min_1} directly follows from Equation \eqref{lemma:proof_bound_bias_1_poisson_bound}. Equation \eqref{eq:lem_ass_min_2} also directly follow sfrom Equation  \eqref{lemma:proof_bound_bias_1_poisson_bound} by noticing that $N_{i,j}= v \sim \mathrm{Poisson}\left(v \: \overline{N_i}\right)$ when $V_{j} = v$.
\end{proof}
\begin{proof}[Proof of Proposition \ref{prop:model_conjecture}] Apply successively Lemma \ref{lem:poisson_conj2} and Proposition \ref{prop:imply_conj2_conj1}.
\end{proof} \section{Instantiating the Auxiliary Code $\CCbis$ with an Efficient Decoder}\label{app:Caux}

We use here notation from \S\ref{sec:SDC}. In particular, we suppose the auxiliary code $\CCbis$ is a product of $b$ small random codes where
\begin{equation}
	b = \Theta(\log n).
\end{equation}

We have to show that for such $b$, the 
analyses 
from Propositions \ref{prop:biasCodedRLPN} and \ref{prop:exp_sizeS} still hold. Indeed, these analyses were done as if $\CCbis$ were a random code equipped with genie aided decoders. Here we compute $\sH$ as a subset of
$$
\widetilde{\sH} \subseteq \left\{ \left( \hv, \cvbis \right) \in \CC^{\perp} \times \CCbis\: : \: \forall i \in \IInt{1}{b}, \; |\hv_{\sN}(i)| = \tfrac{w}{b} \mbox{ and } |\hv_{\sP}(i) + \cv_i| = \tfrac{\tbis}{b}  \right\}
$$
by decoding each parity-check performing an exhaustive search on each block. We have therefore to show that in this case,

\begin{enumerate}[label=\textcolor{blue}{\roman*.}, ref=\roman*]
\item\label{SDC1} the bias
	$$
	\bias_{\left(\hv,\cvbis\right) \drawn \widetilde{\sH}}\left(  \langle \cvbis + \hv_{\sP},\ev_{\sP} \rangle +  \langle \ev_{\sN},  \hv_{\sN}\rangle \right)
	$$
is of the same order as that given by Proposition \ref{prop:biasCodedRLPN},
\item\label{SDC2} and the number of candidates to test 
$$
\mathbb{E }_{\CC,\CCbis}\left(\card{\sS }\right)
$$
is of the same order as that given by Proposition \ref{prop:exp_sizeS}.
\end{enumerate}

To prove item \ref{SDC1}., we first suppose that $\ev_{\sP}$ and $\ev_{\sN}$ have a weight which is fairly distributed, that is:
\begin{equation}
\label{eq:distrib_e}
\forall i \in \IInt{1}{b}, \; \hw{\ev_{\sP}(i)} = \tfrac{t - u}{b} \mbox{ and } \hw{\ev_{\sN}(i)} = \tfrac{u}{b}.
\end{equation}
This happens with a probability:
\begin{align}
\mathbb{P}_{\mathrm{succ}} = \dfrac{ \binom{s/b}{(t - u)/b}^b\binom{(n-s)/b}{u/b}^b}{\binom{s}{t - u}\binom{n-s}{u}} &= \Om{n \left(\tfrac{b}{c  n}\right)^{b}}
\end{align}
where $c$ is constant in $n$.
So we only need to iterate the whole double-RLPN algorithm a sub-exponential number of times, namely at most $\tfrac{1}{\mathbb{P}_{\mathrm{succ}}}$ times. Note that \eqref{eq:distrib_e} is not a necessary condition to achieve our decoding so this overcost is overestimated.

Now, assuming Condition \eqref{eq:distrib_e}, then we can see the bias above as the product of $b$ independent biases involving smaller vectors. More formally, we have
\begin{equation}
\begin{array}{l}
\displaystyle{\bias_{\left(\hv, \cvbis\right) \drawn \overline{\sH}}\left(  \langle \cvbis + \hv_{\sP},\ev_{\sP} \rangle +  \langle \ev_{\sN},  \hv_{\sN}\rangle \right)} \\
\ \ \ \ \ \ \ \ \ \ \ \ \ \ \ \ \ \ \ = \displaystyle{\prod_{i = 1}^{b}\bias_{\left(\hv(i), \cvbis(i)\right) \drawn \widetilde{\sH}_i}\left(  \langle \cvbis(i) + \hv_{\sP}(i),\ev_{\sP}(i) \rangle +  \langle \ev_{\sN}(i),  \hv_{\sN}(i)\rangle \right)}
\end{array}
\end{equation}
where 
\begin{equation}
\begin{array}{l}
\widetilde{\sH}_i \eqdef \Big\{ \left( \hv_{\sP}(i), \hv_{\sN}(i), \cvbis(i) \right) \in \left(\CC_{\sP(i) \cup \sN(i)}\right)^{\perp} \times \CC_i \; \\
\ \ \ \ \ \ \ \ \ \ \ \ \ \ \ \ \ \ \ \ \ \ \ \ \ \ \ \ \ \ : \; |\hv_{\sN}(i)| = \tfrac{w}{b} \mbox{ and } |\hv_{\sP}(i) + \cv_i| = \tfrac{\tbis}{b}  \Big\}
\end{array}
\end{equation}

Moreover, let us degrade the Constraints \eqref{eq:cstPropo35} of Proposition \ref{prop:biasCodedRLPN} by replacing the polynomial factor $n^{\alpha}$ by a super-polynomial 
\begin{equation}
A \eqdef \frac{n^{\alpha - 1 + \log(c) + \log(n)}}{\log(n)^{\log(n)}}.
\end{equation}
On the one hand, this super-polynomial factor is multiplied to the final complexity, but on the other hand the new constraint (with the original one \eqref{eq:constraints}) induces:
\begin{align}
\frac{\binom{(n-s)/b}{w/b} \binom{s/b}{\tbis/b}}{2^{(k-\kbis)/b}}  
&= \om{\frac{n^{\alpha/\log(n)}}{\delta^{2/\log(n)}}}
\end{align}
and
\begin{equation}
\frac{\binom{(n-s)/b}{w/b} \binom{s/b}{\tbis/b}}{2^{k/b}} = \OO{n^{\alpha/\log(n)}} \quad \mbox{and} \quad \frac{\binom{s/b}{\tbis/b}}{2^{(s-\kbis)/b}} = \OO{n^{\alpha/\log(n)}}.
\end{equation}

Which allows us to say, using Proposition \ref{prop:biasCodedRLPN}, that for all $i \in \IInt{1}{b}$ and for a proportion $1 - \oo{1}$ of codes $\CC_i$ and $\CC_{\sP(i) \cup \sN(i)}$:

\begin{equation}
\bias_{\left(\hv(i), \cvbis(i)\right) \drawn \overline{\sH}_i}\left(  \langle \cvbis(i) + \hv_{\sP}(i),\ev_{\sP}(i) \rangle +  \langle \ev_{\sN}(i),  \hv_{\sN}(i)\rangle \right) = \delta^{1/\log(n)}(1-\oo{1}).
\end{equation}
By specifying the values of both $\oo{1}$ (see proof of Proposition \ref{prop:biasCodedRLPN} in Appendix \ref{app:a}), we can deduce that \ref{SDC1}. is verified.\\

To verify item \ref{SDC2}., we can adapt Section \ref{sec:distrib_bias} to show that
\begin{align}
\mathbb{E }_{\CC,\CCbis}\left(\card{\sS}\right) & = \displaystyle{\OOt{ \max_{(i,j) \in \mathcal{A}} \frac{\binom{s/b}{j}^b \binom{(n-s)/b}{i}^b}{2^{n-k}} } + 1} 
\end{align}
where
\begin{equation}
				\mathcal{A} \eqdef \left\{ \left(i,j\right) \in \llbracket 0, \tfrac{n-s}{b} \rrbracket \times \llbracket0,\tfrac{s}{b} \rrbracket, \frac{ K_{w/b}^{((n-s)/b)}\left(u/b\right) K_{\tbis/b}^{(s/b)}\left((t-u)/b\right) }{ K_{w/b}^{((n-s)/b)}\left(i\right) K_{\tbis/b}^{(s/b)}\left(j\right)} \leq n^{2/b} \right\},
\end{equation}
\begin{equation}
\sS \eqdef \{\sv \in \F_2^\kbis \: : \widehat{f_{\yv,\widetilde{\sH},\Gmbis}}\left(\sv\right) \geq  \frac{\delta}{2} \: \widetilde{H}\},
\end{equation}
and
\begin{equation}
\widetilde{H} \eqdef \frac{\binom{(n-s)/b}{w/b}^b \binom{s/b}{\tbis/b}^b}{2^{k-\kbis}}.
\end{equation}
		
Finally, up to a sub-exponential factor, the above expectation is of the same order as in Proposition \ref{prop:exp_sizeS}. \section{Proofs of the statements made in Section \ref{sec:lattice}}\label{sec:app_lattice}

{\bf \noindent Proof of Proposition \ref{propo:dualityLattice}}
\propdualityLattice*

It is helpful to notice before the following link between the Bessel functions and the Fourier transform of the indicator function 
$\un_{\leq w}$
of the words of Euclidean norm $\leq w$
in $\mathbb{R}^n$ (see \cite[Fact 4.11]{DDRT23})
\begin{lemma}
	We have for any positive integer $n$, any $w\geq 0$, any $\xv$ in $\mathbb{R}^n$ 
	$$
	\widehat{\un_{\leq w}}(\xv) = \left( \frac{w}{\| \vec{x}\|_{2}} \right)^{n/2} J_{n/2}\left(2\pi w \|\vec{x}\|_{2}\right)
	$$
where $\widehat{f}(\vec{x}) = \int_{\mathbb{R}^{n}} f(\vec{x})e^{-2i\pi \langle \vec{x},\vec{y} \rangle} d\vec{y}$ for $f : \mathbb{R}^{n} \rightarrow \mathbb{C}$. 
\end{lemma}

\begin{proof}[Proof of Proposition \ref{propo:dualityLattice}]
	First, notice that,
	\begin{align}
		f_{\widetilde{\sW}}(\yv) &= \frac{1}{2} \sum_{\vec{w} \in \widetilde{\sW}} \left( e^{2i\pi \langle \vec{w},\vec{y} \rangle} + e^{-2i\pi \langle \vec{x},\vec{y} \rangle} \right) \nonumber \\
		&=  \sum_{\vec{w} \in \widetilde{\sW}} e^{-2i\pi \langle \vec{w},\vec{y} \rangle} \quad (\mbox{$\vec{w} \mapsto -\vec{w}$ is a bijection in $\sW$}) \nonumber \\
		&= \sum_{\vec{w} \in \Lambda^{\vee}} \left( \un_{\leq w+\varepsilon}(\vec{w}) - \un_{\leq w-\varepsilon}(\vec{w}) \right) e^{-2i\pi \langle \vec{w},\vec{y} \rangle} \label{eq:diffUn} 
	\end{align}
	Recall now the Poisson summation formula, for any $\vec{y} \in \Lambda + \vec{e}$ and sufficiently regular function $f$,
	$$
	\sum_{\vec{x} \in \Lambda^{\vee}} f(\vec{x})e^{-2i\pi \langle \vec{x},\vec{y} \rangle} = \frac{1}{|\Lambda^{\vee}|} \; \sum_{\vec{x}\in \Lambda + \vec{e}} \widehat{f}(\vec{x})
	$$
	Plugging this formula into Equation \eqref{eq:diffUn} yields to, 
	\begin{align}
			f_{\widetilde{\sW}}(\yv) &= \frac{1}{|\Lambda^{\vee}|} \sum_{\vec{x \in \Lambda+\vec{e}}} \left( \widehat{\un_{\leq w+\varepsilon}}(\vec{x}) - \widehat{\un_{\leq w-\varepsilon}}(\vec{x}) \right) \nonumber \\
			&= \frac{1}{|\Lambda^{\vee}|}\sum_{j\geq 0}  \frac{N_{j}}{j^{n/2}} \left( (w+\varepsilon)^{n/2} J_{n/2}\left( 2\pi (w+\varepsilon) j \right) - (w-\varepsilon)^{n/2}J_{n/2}(2\pi (w-\varepsilon) j) \right) \nonumber \\
			&= \frac{1}{|\Lambda^{\vee}|}\sum_{ j \geq 0} \frac{N_j}{(2\pi)^{n/2}\; j^{n}} \Big( \left( 2\pi (w+\varepsilon) j \right)^{n/2} J_{n/2}\left( 2\pi (w+\varepsilon) j \right) \nonumber\\
			&\qquad\qquad\qquad\qquad\qquad\qquad - \left( 2\pi (w-\varepsilon)j \right)^{n/2}J_{n/2}(2\pi (w-\varepsilon) j) \Big)   \label{eq:concl}
	\end{align}
	which 
	 concludes the proof. 
\end{proof} 

{\noindent \bf An Approximation.}
We have also an approximate form for $f_{\widetilde{\sW}}(\yv)$ which is given by
\begin{equation}
f_{\widetilde{\sW}}(\yv) 
\approx \frac{4\pi \varepsilon}{|\Lambda^{\vee}|}\; \sum_{j \geq 0} j N_{j} \;\left( \frac{w}{j} \right)^{n/2} J_{n/2-1}(2\pi  w j) \label{eq:approxDerivBessel}.
\end{equation}
This follows from the fact that 
\begin{equation}\label{eq:derivative} 
	\frac{d}{dx} \left( x^{n/2} J_{n/2}(x) \right) = x^{n/2} J_{n/2-1}(x). 
	\end{equation} 
	Let,
	$$
	X \eqdef 2\pi w j \quad \mbox{ and } \quad h \eqdef 2\pi \varepsilon j
	$$	
	Notice that,
	\begin{multline*}
		\left( \left( 2\pi (w+\varepsilon) j \right)^{n/2} J_{n/2}\left( 2\pi (w+\varepsilon) j \right) - \left( 2\pi (w-\varepsilon)j \right)^{n/2}J_{n/2}(2\pi (w-\varepsilon) j) \right)   \\
		= (X+h)^{n/2}J_{n/2}(X+h) - (X-h)^{n/2} J_{n/2}(X-h)
	\end{multline*}
	From Equation \eqref{eq:derivative},
	\begin{align*} 
	(X+h)^{n/2} J_{n/2}(X+h) - (X-h)^{n/2} J_{n/2}(X-h) &\approx 2h \; \frac{d}{dX} \left( X^{n/2} J_{n/2}(X) \right) \\
	&= 2h \; X^{n/2} J_{n/2-1}(X) \\
	&= (4\pi \varepsilon j) \; (2\pi w j)^{n/2}J_{n/2-1}(2\pi w j)
	\end{align*} 
	Plugging this into Equation \eqref{eq:concl} yields \eqref{eq:approxDerivBessel}.

We also recall that we make the approximation
\begin{equation}\label{eq:fW}
f_{{{\sW}}}(\yv) \approx \frac{N}{\card{\widetilde{\sW}}} \; f_{\widetilde{\sW}}(\yv)
\end{equation}

The number $N_{\leq x}^{\vee}$ of dual vectors of length $\leq x$ can be approximated using the Gaussian heuristic:
\begin{equation*} 
N_{\leq x}^{\vee} \approx \frac{x^{n}}{\sqrt{n\pi}\; \left(\frac{n}{2\pi e} \right)^{n/2}}\;  \frac{1}{|\Lambda^{\vee}|}
\end{equation*} 
Thus we have:
\begin{align}
\card{\widetilde{\sW}} &= N_{\leq w + \varepsilon}^{\vee} -  N_{\leq w - \varepsilon}^{\vee} \nonumber \approx \frac{(w + \varepsilon)^n - (w - \varepsilon)^n}{\sqrt{\tfrac{n}{2 \pi e}}^n \cdot \sqrt{\pi n} \cdot | \Lambda^{\vee}|} \nonumber\approx \frac{2 n \varepsilon w^{n-1}}{\sqrt{\tfrac{n}{2 \pi e}}^n\;  \sqrt{\pi n} \; | \Lambda^{\vee}|}
\end{align}
Putting this into Equation \eqref{eq:fW} shows,
\begin{equation*}
	f_{\sW}(\vec{y}) \approx N \; \frac{\card{\Lambda^{\vee}}\; \sqrt{n\pi}\left( \frac{n}{2\pi e} \right)^{n/2}}{2n \varepsilon w^{n-1} } \; f_{\widetilde{\sW}}(\vec{y})
\end{equation*}
and after some further computation
\begin{align*}
	f_{\sW}(\vec{y}) &\approx   N \; \card{\Lambda^{\vee}} \frac{\sqrt{n\pi}\left( \frac{n}{2\pi e} \right)^{n/2}}{2n \varepsilon w^{n-1} } \; \frac{4\pi \varepsilon}{|\Lambda^{\vee}|}\; \sum_{j \geq 0} j N_{j} \;\left( \frac{w}{j} \right)^{n/2} J_{n/2-1}(2\pi  w j)  \\ 
	&= N \; \frac{\sqrt{n \pi}}{2n w^{n-1}} \left( \frac{n}{2\pi e} \right)^{n/2} \;4\pi \sum_{j\geq 0} j N_{j} \;\left( \frac{w}{j} \right)^{n/2} J_{n/2-1}(2\pi  w j)  \\ 
	&= N \; \frac{\sqrt{n}\; n^{n/2}}{n} \; \frac{1}{w^{n-1}}\; \frac{\sqrt{\pi} \; \pi}{\pi^{n/2}} \; \frac{4}{2 \; 2^{n/2}} \; \frac{1}{e^{n/2}} \; \sum_{j\geq 0} N_{j} \; w^{n/2} \frac{1}{j^{n/2-1}} \; J_{n/2-1}(2\pi wj) \\
	&= N \;\sqrt{n}\; \sqrt{\pi} \;e^{-1}\; n^{n/2-1} \; \pi^{-(n/2-1)} \; 2^{-(n/2-1)} e^{-(n/2-1)} \sum_{j\geq 0} N_{j} \left( \frac{1}{wj} \right)^{n/2-1} \; J_{n/2-1}(2\pi w j) \\
	&= N \frac{\sqrt{n \pi}}{e} \sum_{j\geq 0} N_{j} \left( \frac{n}{2\pi e wj} \right)^{n/2-1} J_{n/2-1}(2\pi wj)
\end{align*}

 }
\end{document}